\documentclass[12pt]{article}
\pdfoutput=1

\usepackage[hmargin=1.60cm,vmargin=1.60cm]{geometry}

\usepackage[title]{appendix}
\usepackage{amsmath,amsfonts,amssymb,amsthm, setspace,rotating} 
\usepackage{lscape}
\usepackage{graphicx,siunitx,booktabs}
\usepackage{natbib}     
\usepackage[outdir=./]{epstopdf}
\usepackage[outdir=./]{epsfig} 
\usepackage{pdflscape}
\usepackage[section]{placeins}    
\usepackage{color, threeparttable, bm, url, enumitem, comment, multirow, accents}  
\usepackage{caption, subcaption, todonotes, hyperref}
\usepackage{xr-hyper}  

\usepackage{xr}  

\makeatletter

\DeclareMathOperator*{\plim}{plim}

\newtheorem{theorem}{Theorem}
\numberwithin{theorem}{section} 
\newtheorem{proposition}{Proposition}
\numberwithin{proposition}{section} 
\newtheorem{assumption}{Assumption} 
\numberwithin{assumption}{section} 
\newtheorem{lemma}{Lemma}
\numberwithin{lemma}{section}

\numberwithin{corollary}{section}

\def\be{\begin{equation}}
	\def\ee{\end{equation}}
\def\bea{\begin{eqnarray}}
	\def\eea{\end{eqnarray}}
\def\beAA{\begin{align}}
	\def\eeAA{\end{align}}

\newcommand*{\dt}[1]{%
	\accentset{\mbox{\large\bfseries .}}{#1}}

\newcommand{\pto}{\overset{p}{\to}} 
 
\theoremstyle{definition}
\numberwithin{remark}{section}

\def\1{1\!{\rm l}}

\def \R {\mathbb{R}}
\def \E {\mathbb{E}}

\hypersetup{
	colorlinks   = true, 
	urlcolor     = blue, 
	linkcolor    = blue, 
	citecolor    = black 
}

\graphicspath{{./plot/}} 

\usepackage{setspace} 
\begin{document}

	\begin{center}
		
		{\large \bf Quantile Random-Coefficient Regression with Interactive Fixed Effects: \\Heterogeneous Group-Level Policy Evaluation}
		\medskip
		
		Ruofan Xu$^{\dag}$, Jiti Gao$^{\dag}$, Tatsushi Oka$^{\star}$and Yoon--Jae Whang$^{*}$\\
		$^{\dag}$Department of Econometrics and Business Statistics, Monash University, \\ 
		$^{\star}$Department of Economics, Keio University, and \\
		$^{*}$Department of Economics, Seoul National University
		
	\end{center}

\begin{abstract}
	We propose a quantile random-coefficient regression with interactive fixed effects to study the effects of group-level policies that are heterogeneous across individuals. Our approach is the first to use a latent factor structure to handle the unobservable heterogeneities in the random coefficient. The asymptotic properties and an inferential method for the policy estimators are established. The model is applied to evaluate the effect of the minimum wage policy on earnings between 1967 and 1980 in the United States. Our results suggest that the minimum wage policy has significant and persistent positive effects on black workers and female workers up to the median. Our results also indicate that the policy helps reduce income disparity up to the median between two groups: black, female workers versus white, male workers. However, the policy is shown to have little effect on narrowing the income gap between low- and high-income workers within the subpopulations. 
	\smallskip
	
	\smallskip
	
	
\end{abstract}


\section{Introduction}

Hierarchical, or nested, data structure is natural in many research fields, including policy evaluation, educational research and meta-analysis, to name a few. In this paper, we are particularly interested in the hierarchical data observed over multiple periods, where the group level consists of a number of cross-sections (such as industries, and states) across time, and the individual level comprises random samples within each group unit. A typical example in survey is the repeated cross-sectional data.

Random-coefficient regression models prevail in modeling hierarchical data, as they not only allow the heterogeneity across groups but also link the individual- and group-level explanations.
The development of random-coefficient models in the conditional mean paradigm can be dated back to the 1970s \citep[see, e.g.,][]{swamy1970efficient,hsiao1975some,borjas1994two,zhang2019identification}.  However, the above models are misspecified when unobserved heterogeneous individual effects exist, in which case quantile regression serves as an alternative approach since the seminal work by \cite{koenker1978regression}. Several papers study the quantile random coefficient models, such as \cite{kim2011semiparametric}, \cite{graham2018quantile}, \cite{pitselis2020multi} and \cite{chetverikov2016iv}.

The aforementioned literature studies the hierarchical data at a single period. These models fail to capture the dynamics across time, and thus, lead to biased estimates when applying to data with multiple periods. Indeed, limited research can be found tailored for the hierarchical data with multiple periods. In this paper, we consider a latent factor structure to flexibly capture the unobserved time and group effects simultaneously. The latent factor structure is prevailingly used in the panel data literature to control the time-varying common shocks that are distinct across groups, which is known as the interactive fixed effects \citep[e.g.,][]{bai2009panel, pesaran2006estimation}. Quantile extension of the latent factor structure has recently been studied by \cite{ando2020quantile,chen2021quantile}. However, to the best of our knowledge, this is the first paper that allows for interactive fixed effects in the quantile random-coefficient model. 

To this end, we proposed a quantile random-coefficient model, where the random coefficients are heterogeneous across both time and cross-sections and modelled by a linear regression with interactive fixed effects.\footnote{As cross-sections are the same across time, the group-level regression framework can be viewed as a panel regression.} For this model, we introduce a simple two-step estimation procedure. First, we estimate the quantile regression model, introduced by \cite{koenker1978regression}, for each group-time pair. Then, we estimate the group-level model on quantile random coefficients using the estimation procedure proposed by \cite{bai2009panel}. This model is capable of documenting the observed and unobserved heterogeneities at both individual and group levels, as well as the interactions between the two levels. 
We note that the model is not only of methodological importance but also of empirical interest. We show that the model provides us with a way to study how group-level policies affect the marginal effects differently among observationally identical individuals and uncover policy effects depending on individual observed and unobserved heterogeneities, which is important yet ``somewhat neglected'' in the policy evaluation \citep{koenker2017quantile}. 

In summary, the contributions of this paper are threefold. 

First, from the methodological point of view, we contribute to the literature on quantile random-coefficient models by (i) allowing the random coefficients to be heterogeneous across both groups and time and (ii) characterizing the unobserved group and time heterogeneities within the random coefficients flexibly via interactive fixed effects. A closely related paper to the empirical aspect of our paper is \cite{oka2021heterogeneous}. They characterize the unobserved time and group heterogeneity via two-way fixed effects in their quantile random coefficients, which can be viewed as a special case of our proposed model. In addition, our paper takes one step further to establish asymptotic results and identification conditions for a number of treatment parameters that quantify the heterogeneous effects of group-level policy. We extend the scope of policy evaluation approaches \citep[e.g.,][]{gobillon2016regional, Arellano2016} by establishing the identification strategy of the distributional group-level policy effect which is heterogeneous across the individuals' characteristics.

Second, from the theoretical perspective, 
we establish the consistency and limiting distribution of the proposed group-level coefficients estimator when the number of observations per group-time pair $N_{st}$, cross-sections $S$, and time units $T$ all go to infinity simultaneously. The estimation error from the first-step quantile regression imposes a great challenge in deriving the asymptotic expression of the group-level estimators. However, we show that the first-step estimation error is negligible, as long as the relative growth rate between individual-level and group-level sample sizes is sufficiently large. 

Last but not least, from the empirical perspective, 
we apply the model to contribute to the debate on racial income disparity during 1960s and 1970s in the United States \citep[see, e.g.,][]{freeman1973changes,card1992school,derenoncourt2021minimum}.
Unlike \cite{derenoncourt2021minimum}, our method estimates the heterogeneous policy effects, particularly on racial income inequality, under a unified framework without the need for alternating responses and selecting sub-samples. In addition, the interactive fixed effects in our model are suitable for controlling for time-varying unobserved common shocks, such as macroeconomic shocks, which could affect industries differently. Our estimation results support the core conclusion of \cite{derenoncourt2021minimum} that minimum wage policy helps reduce the racial income gap and provide additional findings that the minimum wage policy has a significant negative impact on between-inequality but little effect on within-inequality.

The paper proceeds as follows. Section~\ref{sec:model} introduces the model and proposes a two-step estimation method. Section~\ref{sec:4}
presents the asymptotic properties of the estimators. Section~\ref{sec:treatment} 
discusses the identification of treatment effects for group-level policies using our proposed model.
In Section~\ref{sec:empirical}, we apply the proposed method to analyze
the effect of minimum wage on earnings under 
the 1966 Labour Standards Act.
Section~\ref{sec:conclusion} concludes.
Online supplementary material encompasses simulation results, additional empirical analysis, technical assumptions for asymptotic theories and proofs of the theorems presented in the main text.

\section{Model and Estimation}\label{sec:model}

In this section, we first provide the model setup and then introduce the estimation method.
Before preceding, we introduce some notations.
Throughout the paper, 
let $\| \cdot \|$ denote the Euclidean norm
for vectors
and 
the spectral norm
for matrices, 
that is $\| a\| :=\sqrt{a'a}$
and 
$\|A\|:= \sup_{a \not = 0}\| A a \| /\| a \|$
for column vector $a$
and matrix $A$. 
Let $I_p$ denote the $p$-dimensional identity matrix. Let $1\{\cdot\}$ denote the indicator function.
Let $e_k$  be a unit column vector having 1 at the $k^\text{th}$ entry
and 0 for the others,
and the dimension of $e_k$ is allowed to vary according to the context. 
Let $\text{diag}(\cdot)$ denotes the diagonal matrix, whose diagonal entries are given in the parenthesis.
We denote 
$a \vee b:=\max \{a, b\}$
for scalars $a, b$.
Let $C_M$ and $c_M$ be some pre-determined positive real numbers which are independent of the sample.

\subsection{Model}
\label{setup}

Given group $s=1, \dots S$ and
time $t=1, \dots, T$,
let  
$\{(y_{ist}, z_{ist})\}_{i=1}^{N_{st}}$
be 
repeated cross-sectional observations
of 
a scalar outcome
$y_{ist}$
and
a $J \times 1$ regressor vector $z_{ist}$
, which includes a constant 1 if needed, for individual $i$
with the sample size $N_{st}$.
We denote 
supports of $y_{ist}$ and $z_{ist}$
by 
$\mathcal{Y} \subseteq \R$
and 
$\mathcal{Z} \subseteq \R^J$,
respectively.\footnote{
	The supports $\mathcal{Y}$ and $\mathcal{Z}$ can be allowed to depend on
	group and time, while we suppress
	the dependency for notational simplicity.
}
Also,
we observe a $K \times 1$ vector of group-level covariates 
$x_{st}$, including a constant 1 if necessary, whose support  is $\mathcal{X} \subseteq \R^{K}$,
and
a dummy variable $d_{st}$ which
takes 1 when some policy is employed
in group $s$ and time $t$ and 0 otherwise.

We assume that 
the repeated cross-sectional observations
are randomly sampled within each group-time pair, conditional on group-level information. That is to say, interdependency between individauls within paris and dependency across pairs, characterized by both the observed and unobserved group-level information, are allowed. Specifically, we assume that the $u^{\text{th}}$ quantile of the conditional distribution of $y_{ist}$ is given by
\begin{align}
	& 
	Q_{y_{ist}|z_{ist},\alpha_{st}}(u)
	=
	z_{ist}'{\alpha}_{st}(u),
	\quad 
	\alpha_{st}(u)\equiv[\alpha_{1st}(u), \dots, \alpha_{Jst}(u)]',
	\label{Q_y}\\
	& 
	\alpha_{jst}(u)
	=
	\delta_{jt}(u)
	d_{st}
	+
	x_{st}'\beta_j(u) +
	f_{jt}(u)' {\lambda}_{js}(u) +
	\eta_{jst}(u), \label{alpha_st}
\end{align}
where $Q_{y_{ist}|z_{ist},\alpha_{st}}(u)$ is the $u^{\text{th}}$ conditional quantile of $y_{ist}$ given $(z_{ist},\alpha_{st})$ for $u \in \mathcal{U} \subseteq (0,1)$,
$\alpha_{jst}(u)$ is the scalar random quantile coefficient corresponding to the $j^{\text{th}}$ component of the individaul-level observable $z_{ist}$, $\delta_{jt}(u)$
is a scalar coefficient
for the policy effect at time $t$,
$\beta_{j}(u)$
is a $K \times 1$ vector of coefficients measuring the marginal effect of the group-level observable $x_{st}$ on $\alpha_{jst}(u)$,
$f_{jt}(u)$ is an $r \times 1$ vector of unobservable group-level factors, which are heterogeneous across time $t$, ${\lambda}_{js}(u)$ is the corresponding factor loading vector, and
$\eta_{jst}(u)$ is an idiosyncratic group-level error satisfying $\E[\eta_{jst}(u)|d_{st},x_{st},f_{jt}(u),\lambda_{js}(u)] =0$. For notational simplicity, we write $Q_{st}(u | z_{ist}) \equiv Q_{y_{ist}|z_{ist},\alpha_{st}}(u)$ in what follows. 

In this paper, we consider a binary group-level policy that is employed at a known time $T_0$ onward with $1 < T_0 < T$. Then, we may write model~(\ref{alpha_st}) in the vector form as 
\begin{equation*}
	A_{js}(u) 
	= 
	D_s{\delta}_{j}(u) + X_s \beta_j(u) + 
	F_j(u) {\lambda}_{js} (u) +
	{\eta}_{js}(u),
\end{equation*}
where 
$A_{js}(u) := [\alpha_{js1}(u),...,\alpha_{jsT}(u)]'$,
$D_s := d_s[e_{T_0},...,e_{T}]$  with  $d_s=1$ if group $s$ is treated after $T_0$ and $0$ otherwise, and $e_{k}$ being a unit column vector having $1$ at the $k^\text{th}$ entry and 0 for the others, 
$\delta_{j} := [{\delta}_{jT_0}(u),...,\delta_{jT}(u)]'$,
$X_{s}:=[x_{s1}, \dots, x_{sT}]'$,
$F_{j}(u):=[f_{j1}(u), \dots, f_{jT}(u)]'$, 
and 
$\eta_{js}(u) := [\eta_{js1}(u),...,\eta_{jsT}(u) ]'$.

It is known that, the quantile coefficient $\alpha_{st}(u)$ can be interpreted as the marginal effect of individual covariates $z_{ist}$ on the $u^{\text{th}}$ quantile of outcome variables. Moreover, 
when the underlying structural model depends on multi-dimension unobservables, 
\cite{sasaki2015quantile} shows that the quantile regression coefficients can be interpreted as the marginal effects of $z_{ist}$ on $y_{ist}$ averaged over the unobserved variables that satisfy mild regularity conditions. Thus, quantile regression coefficients $\alpha_{st}(u)$ can succinctly summarize the marginal effect of observed individual characteristics on the outcome, averaged over individual unobserved heterogeneity among observationally equivalent individuals within each group and time. That said, although model (\ref{Q_y}) doesn't include an individual fixed effect for each given pair of $(s,t)$, it allows for the unobserved heterogeneity across individuals within the group-time pair.

The key interest of this paper lies in studying how these marginal effects depend on the group-level information, especially the group-level policy. To this end, we impose a linear regression model (\ref{alpha_st})
with interactive fixed effects on each $\alpha_{jst}(u)$.
The interactive fixed effects structure ${f}_{jt}(u)' {\lambda}_{js}(u)$
accounts for
unobserved group and time effects in a flexible way.
For example, it captures time-varying macro shocks ${f}_{jt}(\cdot)$
affecting industry or regions differently via ${\lambda}_{js}(\cdot)$.
Also, the two-way fixed effects model
is included as a special case
if
${f}_{jt}(u)=[1, \nu_{jt}(u)]'$
and 
${\lambda}_{js}(u)=[\phi_{js}(u), 1]'$.

As an example to showcase the usefulness of the above hierarchical modelling framework, consider the empirical study in this paper where we aim to quantify the effects of a policy $d_{st}$, which varies at the industry-by-year level, on the distribution of individuals' wages $y_{ist}$ across individuals with innate heterogeneities $z_{ist}$, such as race, gender, etc. A policy may have differential effects on lower wage quantiles for black workers than for white workers; the specification of (\ref{alpha_st}) captures this idea by allowing the researcher to specify the coefficient for the racial indicator as a function of the policy, together with other group-level variables. Furthermore, we show in Section~\ref{sec:treatment} that our model is capable of quantifying the differentials in the policy effects across different subpopulations and quantiles, which facilitates the analysis of policy effects on inequality measures.

In addition, considering a special case where (\ref{alpha_st}) only applies to the intercept in (\ref{Q_y}), the group and time fixed effects are additive and no group-level unobservables $\epsilon_{st}(u)$ presents, model (\ref{Q_y})-(\ref{alpha_st}) is reduced to a 
two-way fixed effects model:
$$Q_{y_{ist}|z_{ist},\alpha_{st}}(u)
	= \delta_{jt}(u) d_{st} + x_{st}'\beta(u) + f_t(u) + \lambda_s(u) + z_{ist}'{\gamma}_{st}(u),
	$$
which has been used for estimating distributional policy effects in empirical studies
\citep[e.g.][]{angrist2004does},

\subsection{Estimation}\label{sec:estimation}
We propose a computationally simple two-step estimation approach for model~(\ref{Q_y})-(\ref{alpha_st}).
In what follows, we consider a finite set of probability levels $\mathcal{U}$,
as our analysis mainly focuses on several quantiles and their spreads. 
Details of the algorithm for each $u \in \mathcal{U}$ are outlined below.

\vspace{0.3cm}
\noindent 
\textit{Step 1}:
Using 
the individual-level data $\{(y_{ist},{z}_{ist})\}_{i=1}^{N_{st}}$
for each pair of group and time $(s,t)$ separately,
we obtain
the estimator 
$\widehat{{\alpha}}_{st}(u)$
of ${\alpha}_{st}(u)$
as the solution of the following quantile minimization problem:
\begin{equation*}
	\min_{\alpha \in \R^{J}}
	\sum_{i=1}^{N_{st}}
	\varrho_{u}(y_{ist} - z_{ist}'\alpha),
\end{equation*}
where $\varrho_{u}(v) := (u - 1\{v < 0\})v$ for $v \in \R$.

\vspace{0.5cm}
\noindent 
\textit{Step 2}: Let $\Lambda_j(u) := [\lambda_{j1}(u),...,\lambda_{jS}(u)]'$.
Given a collection of the estimators
$\widehat{A}_{js}(u) := [\widehat{{\alpha}}_{js1}(u),...,\widehat{{\alpha}}_{jsT}(u)]'$,
we obtain the estimator of
$\big (
\delta_{j}(u),\beta_j(u),F_j(u),\Lambda_{j}(u) \big )$
by minimizing the following sum of squared residuals:
\begin{align} \label{SSR}
	&\text{SSR}_{u}
	\big (
	\delta_{j},\beta_j,F_j,\Lambda_{j}
	\big )  
	:=
	\sum_{s=1}^{S}
	\big \| 
	\widehat{A}_{js}(u) 
	-
	D_{s}'{\delta}_{j}
	- 
	X_{s}'\beta_j
	-
	F_j {\lambda}_{js} 
	\big \|^2,
\end{align}
with the normalization condition in Assumption~\ref{a2.2}.(i),
which ensures the identification of factors and their loadings up to an orthogonal rotation matrix with column sign change \citep{bai2009panel}.
The least squares estimators are obtained using an iterated procedure described below.
The superscript $m$ indicates the number of iterations and
the converged estimator
is represented as 
{\small $\big(\widehat{\delta}_{j}(u), \widehat{\beta}_j(u), \widehat{F}_j(u),\widehat{\Lambda}_j(u)\big)$. 
\begin{enumerate}[label=(\roman*)]
\item
  Using the LS estimator without the factor components,
  we obtain an initial estimator
  $(\widehat{\delta}^{(0)}_{j}(u),\widehat{\beta}^{(0)}_j(u))$, which minimizes 
	$\sum_{s=1}^{S}\|\widehat{A}_{js}(u) -    D_{s}{\delta}_{j}    -     X_{s}'\beta_j\|^2$.

      \item Given $\big(\widehat{\delta}^{(m-1)}_{j}(u),\widehat{\beta}^{(m-1)}_j(u)\big)$ for $m \ge 1$,
	we obtain
	$\big(\widehat{F}_j^{(m)}(u), \widehat{\Lambda}_j^{(m)}(u) \big)$ as the solution to 
	$$\min_{(F_j,\Lambda_j)}\text{SSR}_{u}
	\big (
	\widehat{\delta}^{(m-1)}_{j}(u),\widehat{\beta}^{(m-1)}_j(u),F_j,\Lambda_j
	\big )$$ 
	by applying the principle component analysis (PCA)
	with the normalization conditions in Assumption~\ref{a2.2}.(i).

	\item
	Given 
	$\big(
	\widehat{F}_j^{(m)}(u), \widehat{\Lambda}_j^{(m)}(u)
	\big)$, 
	we obtain $\big(\widehat{\delta}_{j}^{(m)}(u),\widehat{\beta}_{j}^{(m)}(u)\big)$ as the minimizer of the objective function
	$\text{SSR}_{u}
	\big (
	\delta_{j},\beta_j,
	\widehat{F}_j^{(m)}(u), \widehat{\Lambda}_j^{(m)}(u)
	\big ) $.

	\item Repeat (ii)-(iii) until numerical convergence is reached. Specifically,
	we stop
	the algorithm if 
	$
	\big \|
	\widehat{{\delta}}^{(m)}_{j}(u) - \widehat{\delta}^{(m-1)}_{j}(u)
	\big \|  \leq 10^{-5}
	$,
	$
	\big \|
	\widehat{{\beta}}^{(m)}_j(u) - \widehat{{\beta}}^{(m-1)}_j(u)
	\big \|  \leq 10^{-5}
	$,
	and
	$
	\big \|\widehat{F}^{(m)}_j(u)\widehat{\Lambda}^{(m)}_j(u)' - \widehat{F}^{(m-1)}_j(u)\widehat{\Lambda}^{(m-1)}_j(u)' \big \|  \leq 10^{-5}
	$.
\end{enumerate}

}

In practice, as the number of factors $r$ is unknown,
we use a popular eigen-ratio criterion in PCA for selection. That is, for each $m$ and $j$,
we select the number of factors by minimizing the modified eigen-ratio criterion of \cite{casas2021time}
as follows:
\begin{align*}
	\min_{1 \leq r \le r_{\max}}
	\bigg(
	&
	\frac{\widehat{\rho}^{(m)}_{j,r+1}(u)}{\widehat{\rho}^{(m)}_{j,r}(u)}
	\cdot 
	1\bigg\{\frac{\widehat{\rho}^{(m)}_{j,r}(u)}{\widehat{\rho}^{(m)}_{j,1}(u)}
	\geq
	\frac{1}
	{
		\ln \big(S \vee \widehat{\rho}^{(m)}_{j,1}(u) \big)
	}
	\bigg\}
	+
	1\bigg\{\frac{\widehat{\rho}^{(m)}_{j,r}(u)}{\widehat{\rho}^{(m)}_{j,1}(u)} <
	\frac{1}
	{
		\ln \big (S \vee \widehat{\rho}^{(m)}_{j,1}(u) \big )
	}
	\bigg\}\bigg),
\end{align*}
where
$r_{\max}$ is a pre-specified integer
and 
$\widehat{\rho}^{(m)}_{j,1}(u),...,\widehat{\rho}^{(m)}_{j,T}(u)$ are the estimated eigenvalues of the $T \times T$ matrix $\widehat{L}_j\big(\widehat{\delta}^{(m-1)}_j(u),\widehat{\beta}^{(m-1)}_j(u)\big)$ in descending order, where
\begin{align}
	\label{eq:W}
	\widehat{L}_j( \delta, \beta) :=
	\frac{1}{ST}\sum_{s=1}^{S}
	\big(\widehat{A}_{js}(u) - D_s {\delta}  -X_s \beta \big)
	\big(\widehat{A}_{js}(u) - D_s {\delta}  -X_s \beta \big)'.
\end{align}
Given that a relatively large $r_{\max}$ suffices, 
we set it to the cardinality of the set
$\{\widehat{\rho}^{(m)}_{j,r}(u): \widehat{\rho}^{(m)}_{j,r}(u)>T^{-1}\sum_{r=1}^{T}\widehat{\rho}^{(m)}_{j,r}(u),
r=1, \dots, T\}$ in the empirical and simulation studies.

The proposed two-step estimation method is tailored for the survey data, with the number of individuals within group-time pairs, groups, and time periods being large. It is computationally less demanding relative to jointly estimating all coefficients for large hierarchical data as it only requires estimating a small number of coefficients each time \citep{chetverikov2016iv}.
However, as the proposed model~(\ref{Q_y})-(\ref{alpha_st}) is versatile, alternative one-step estimation methods can also be proposed, especially for the case when the number of individuals per group-time is limited such that the first step of the above algorithm is infeasible.

One major difficulty in the one-step estimation is addressing the omitted-variable bias. 
Let us consider a simple case where $j=1$. We combine (\ref{Q_y})-(\ref{alpha_st}) and write
$Q_{y_{ist}|z_{ist},\alpha_{st}}(u)
	=
	z_{ist}d_{st}
	\delta_{t}(u)
	+
	z_{ist}x_{st}'\beta(u) +
	z_{ist}f_{t}(u)' {\lambda}_{s}(u) + z_{ist}\eta_{st}(u)$.\footnote{We compress the subscript $j$ for all group-level coefficients and idiosyncratic error for notational simplicity.} For this model, one may consider a conventional quantile estimator of \cite{koenker1978regression} through an iterative algorithm similar to \cite{ando2020quantile}, which updates $\widehat{\delta}_t(u), \widehat{f}_t(u)$ and $\widehat{\beta}_{s}(u),\widehat{\lambda}_s(u)$ recursively until convergence and then obtains $\widehat{\beta}(u)$ as the average of $\widehat{\beta}_s(u)$. Unfortunately, such an algorithm is biased as it ignores the interaction term between $z_{ist}$ and the group-level idiosyncratic error $\eta_{st}(u)$. However, how to deal with the endogeneity issue in the presence of interactive fixed effects remains an open question for future research.

\section{Asymptotic Properties}\label{sec:4}

In this section, we first introduce the main assumptions
and then present asymptotic properties of the proposed estimators. The technical assumptions are provided in Online Supplement~\ref{sec:tech_assumptions}.

\subsection{Assumptions}\label{sec:assumptions}
This subsection provides the necessary assumptions for deriving the asymptotic properties of the recursive estimator along with some detailed explanations. 
In what follows, we consider the case where the set $\mathcal{U}$ consists of finite points, since
our empirical application mainly focuses on multiple quantiles and their spreads, instead of the entire distribution.

\begin{assumption} \label{a2.1}
	For each fixed 
	$s,t \geq 1$, 
	\begin{enumerate}[label=(\roman*)]
		\item  Individual observations 
		$\{(y_{ist}, z_{ist})\}_{i=1}^{N_{st}}$
		are
		independent and identically distributed (i.i.d.) across $i=1,...,N_{st}$ conditional on $(\alpha_{st},x_{st},d_{st},f_t,\lambda_s)$ . 
		The regressor vector
		$z_{ist}$ satisfies 
		$\|z_{ist}\|<C_M$ almost surely.
		\item All eigenvalues of $\E[z_{1st}z_{1st}']$ are bounded from below by $c_M>0$.

	\end{enumerate}
	
\end{assumption}

Assumption~\ref{a2.1}.(i) imposes a conditional i.i.d. assumption within group-time pairs. It allows for interdependency within group-time pairs that are fully controlled by group-level information, as well as the dependency between individual observations and group-level characteristics. Such an assumption is typically applicable to survey studies characterized by hierarchical data and is also considered in \cite{chetverikov2016iv}. Assumption~\ref{a2.1}.(ii) is a familiar identification condition in regression analysis. 

\begin{assumption} \label{ass.dens}
	Consider $(s,t) \in \{1,...,S\} \times \{1,...,T\}$, for all $y \in (z'_{1st}\alpha_{st}(u)-c_M,z'_{1st}\alpha_{st}(u)+c_M)$ for some $c_M>0$, it satisfies
	\begin{enumerate}[label=(\roman*)] 
		\item the conditional density function $g_{st}(y)$ is continuously differentiable with the first-order derivative $g_{st}'(\cdot)$ satisfying $|g_{st}'(y)| \leq C_M$ and $|g_{st}'({z}_{1st}'{\alpha}_{st}(u))| \allowbreak\geq c_M$. 
		\item $g_{st}(y) \leq C_M$, and $g_{st}({z}_{1st}'{\alpha}_{st}(u)) \geq c_M$ for some $c_M>0$.
	\end{enumerate}
\end{assumption}
Assumptions~\ref{ass.dens} is a set of mild regularity conditions that are typically imposed in the quantile regression literature \citep{koenker2004quantile,chetverikov2016iv}.

\begin{assumption}\label{a2.2} 
Let $\Lambda_j(u):=[{\lambda}_{j1}(u),...,{\lambda}_{jS}(u) ]'$.
	For all 
	$(s, t) \in \{1, \dots, S\} \times \{1, \dots, T\}$, 
	\begin{enumerate}[label=(\roman*)]
		\item
		$T^{-1}F_j(u)'F_j(u) = I_r$ and $S^{-1}\Lambda_j(u)'\Lambda_j(u)$ is a positive-definite diagonal matrix.
		\item $\mathbb{P}(d_{s}=1)$ is bounded away from below by $c_M>0$ and from above by $C_M<1$.
		\item The eigenvalues of $\E[x_{st}x_{st}'] $ are bounded away from zero.
	\end{enumerate}
\end{assumption}
Assumption~\ref{a2.2} provides the identification conditions of the regression coefficients, factor and loadings. Specifically, Assumption~\ref{a2.2}.(i) guarantees the identification of factor and the loadings up to an orthogonal rotation matrix. Such assumption is standard in the literature of the mean panel factor models \citep{bai2009panel,jiang2020recursive} and quantile factor models \citep[][]{ando2020quantile,chen2021quantile}. Assumption~\ref{a2.2}.(ii) ensures the identifiability of the policy parameter $\delta_{jt}(u)$, which is standard for policy evaluation \citep{hsu2022estimation,noh2023nonparametric}. Assumption~\ref{a2.2}.(iii) is a conventional assumption in regression analysis \citep{bai2009panel} to guarantee the identification of regression coefficients $\beta_{j}(u)$.

\begin{assumption}\label{ass.cov}
	For all $(s,t) \in \{1,...,S\}{\times}\{1,...,T\}$,
	\begin{enumerate}[label=(\roman*)]

		\item $\mathbb{E}[||{\lambda}_{js}(u)||^4] \leq C_M$.
		\item  $\mathbb{E}[\eta_{j s t}(u) | d_{gl}, x_{g  l}, {\lambda}_{jg}(u), f_{jl}(u)]=0$ for all $(g,l) \in \{1,...,S\}\times \{1,...,T\}$.
		\item  The largest eigenvalue of the $T \times T$ matrix $\mathbb{E}[{\eta}_{js}(u){\eta}_{js}(u)'] $ is bounded uniformly in $s$ and $T$.  
	\end{enumerate}
\end{assumption}

Assumption~\ref{ass.cov}.(i) requires standard moment conditions for our analysis. 
Assumption~\ref{ass.cov}.(ii) and (iii) impose weak restriction on the correlation among the idiosyncratic error components, group-level regressors and common factors. These assumptions are often imposed in the factor model literature \citep[e.g.,][]{bai2009panel,jiang2020recursive}.

\begin{assumption} \label{ass.growth}
	Let $N_{\min} := \min\{N_{st},s=1,...,S,t=1,...,T\}$. As $S,T \rightarrow \infty$, we have (i) $T/S \to \kappa > 0$ and (ii) $(ST)^{3/4}(\ln(N_{\min})/N_{\min})^{1/2} \leq C_M$.
\end{assumption}

Assumption~\ref{ass.growth} controls the diverging rates of the number of groups $S$, time $T$ and individuals per group and time $N_{st}$. Assumption~\ref{ass.growth}.(i) is standard in the panel factor model literature \citep[][]{bai2009panel}.  Assumption~\ref{ass.growth}.(ii) requires that the number of individuals per group-time pair grows sufficiently fast as $S$ and $T$ jointly go to infinity, such that the estimation error from the quantile estimation in the first-step is negligible. Compared with Assumption 3 of \cite{chetverikov2016iv}, Assumption~\ref{ass.growth}.(ii) imposes a more explicit yet comparable growth rate, which is necessary in analyzing the limiting property of the interactive fixed effects estimator. Although the growth rate seems restrictive, we have shown with both simulation and empirical study that the estimation procedure is valid in practice as long as the minimum number of individuals per group $N_{\min}$ used to perform the first step estimation is comparable to the total number of groups and time ($S \times T$). For example, in the empirical study, we have $N_{\min} = 229$ while $S \times T = 323$, and the estimation algorithm converges within a small number of iterations.

Define the $J \times 1$ vector
$K_t(u) := [K_{1t}(u),...,K_{Jt}(u)]'$,
whose $j^{\text{th}}$ element is given by  
$$
K_{jt}(u) :=
\big(S^{-1}\sum_{s=1}^{S} R_{js}(u)^2\big)^{-1} 
S^{-1/2}
\sum_{s=1}^{S}R_{js}(u)\eta_{jst}(u),
$$
where
$R_{js}(u) := d_{s} - S^{-1}\sum_{g=1}^{S}\omega_{j,sg}(u)d_g$
		with 
		$\omega_{j,sg}(u) := {\lambda}_{jg}(u)'
		\big(
		S^{-1}\Lambda_j(u)' \Lambda_j(u)
		\big)^{-1} {\lambda}_{js}(u)$.

\begin{assumption}\label{ass.clt_delta}
	For any $u_1,u_2 \in \mathcal{U}$ and $t \geq T_0$, as $S \to \infty$, we have
	\begin{align*}
		\bigg [
		\begin{array}{c}
			K_t(u_1) \\[-0.5cm]
			K_t(u_2)
		\end{array}
		\bigg ]
		\overset{d}{\rightarrow} \mathcal{N} 
		\Bigg(
		\mathbf{0},
		\bigg[
		\begin{array}{c}
			\Sigma_t(u_1,u_1), \Sigma_t(u_1,u_2)\\[-0.5cm]
			\Sigma_t(u_2,u_1), \Sigma_t(u_2,u_2)
		\end{array}
		\bigg]
		\Bigg ), 
	\end{align*}
	where $\mathbf{0}$ is a $(2J) \times 1$ vector and
	$\Sigma_t(u_1,u_2) := \underset{S \to \infty}{\lim} \E[ K_t(u_1) K_t(u_2)']$.
	
\end{assumption}

Assumption~\ref{ass.clt_delta}  is required to derive the joint Central Limit Theorem (CLT) in Theorem~\ref{t6} for the convergence of the estimator of the policy parameter, given quantile levels $u_1,u_2 \in \mathcal{U}$. The assumption shares the same idea as Assumption E of \cite{bai2009panel}.

\subsection{Asymptotic Results} \label{sec.asymp}

In this subsection, we present the asymptotic properties of our proposed estimators.
To maintain focus on the key parameters of interest, we omit the CLT for $\beta_{j}(u)$, though it can be derived similarly to that for $\delta_{jt}(u)$. 
Also, for empirical interests, we derive the corresponding consistent estimator of the asymptotic covariance matrix to facilitate the construction of confidence intervals.

\begin{theorem} \label{t5}
	Suppose that 
	Assumptions~\ref{a2.1}-\ref{ass.growth} and \ref{ass.mixing}-\ref{ass.norm} hold. 
	Then, for any fixed $u\in \mathcal{U}$, $j=1,...,J$ and $m \geq 0$, as $S,T \to \infty$,
	we have
	\begin{enumerate}[label=(\roman*)]
		\item $
		\sqrt{S}
		\big (
		\widehat{\delta}^{(m)}_{jt}(u) - {\delta}_{jt}(u)
		\big )=    O_P(1)$ \ for each $t \geq T_0$.
		\item
		$
		\sqrt{ST}
		\big (
		\widehat{\beta}^{(m)}_j(u) - {\beta}_j(u) 
		\big ) 
		    =  O_P(1)$.
	\end{enumerate}
\end{theorem}

The time-varying policy effects are estimated for each time period following the policy intervention. The convergence rate of the policy effect estimator depends solely on the group size $S$. On the other hand, the remaining regression coefficients are estimated using the full sample, and their convergence rate depends on $ST$. The convergence rates are slower than those of the conventional one-step quantile estimator, which could typically achieve $\sqrt{STN_{\min}}$. However, as explained at the end of the Section~\ref{sec:estimation}, the existing one-step quantile estimator suffers from the omitted-variable bias.

For the purpose of inference, we then establish the joint CLT of the estimator $(\widehat{\delta}_t(u_1)',\widehat{\delta}_t(u_2)')'$ in the following theorem.

\begin{theorem}\label{t6}
	Suppose that 
	Assumptions~\ref{a2.1}-\ref{ass.clt_delta} and \ref{ass.mixing}-\ref{ass.norm} hold.
	Let $\widehat{\delta}_{\cdot t}(u) := [\widehat{\delta}_{1t}(u),...,\widehat{\delta}_{Jt}(u)]'$ denote the estimator of $\delta_{\cdot t}(u)$. Then, for any $u_1,\ u_2 \in \mathcal{U}$ and $t \geq T_0$, we have, as $S,T \to \infty$,
	\begin{align*}
		\sqrt{S}
		\bigg[
		\begin{array}{l}
			\widehat{\delta}_{\cdot t}(u_1) - {\delta}_{\cdot t}(u_1)     \\[-0.5cm]
			\widehat{{\delta}}_{\cdot t}(u_2) - {\delta}_{\cdot t}(u_2)
		\end{array}
		\bigg]
		\overset{d}{\rightarrow} 
		\mathcal{N}
		\Bigg(
		\bigg [
		\begin{array}{l}
			B_t(u_1)  \\[-0.5cm]
			B_t(u_2) 
		\end{array}
		\bigg ]
		, 
		\bigg[
		\begin{array}{c}
			\Sigma_t(u_1,u_1), \Sigma_t(u_1,u_2)\\[-0.5cm]
			\Sigma_t(u_2,u_1), \Sigma_t(u_2,u_2)
		\end{array}
		\bigg]
		\Bigg ),
	\end{align*}
	where
	$\Sigma_t(u_1,u_2)$ is defined in Assumption \ref{ass.clt_delta},
	and
	$B_t(u) :=
	\underset{S,T \to \infty}{\plim}
	[\widetilde{B}_{1t}(u),...,\widetilde{B}_{Jt}(u)]'$ is the bounded asymptotic bias, whose $j^{\text{th}}$ component is given by
	\begin{align*}
		\widetilde{B}_{jt}(u)
		:=  
		- \bigg(\frac{1}{S}\sum_{s=1}^{S} R_{js}(u)^2\bigg)^{-1} 
		\frac{1}{S^{3/2}T}\sum_{s,g=1}^{S}d_s \mathbb{E}[{\eta}_{jgt}(u) {\eta}_{jg}(u)']  F_j(u) \bigg(\frac{\Lambda_j(u)'\Lambda_j(u)}{S}\bigg)^{-1}\lambda_{js}(u), 
	\end{align*}
	for $j=1,...,J$,
	where $R_{js}(u)$ is defined above Assumption~\ref{ass.clt_delta}.
\end{theorem} 

In view of this theorem, under general cases, the asymptotic distribution of the recursive estimator $\widehat{\delta}_{jt}(u)$ depends on: (i) the quantiles, (ii) the accuracy of the first-step estimation (i.e., $\widehat{\alpha}_{jst}(u) - \alpha_{jst}(u)$), (iii) the consistency of the initial estimation of the second-step, and (iv) the degenerating estimation error of the regression coefficient, factor and loadings carried over from the iterative steps. Therefore, although the CLT is derived per given time $t$, we still require $S,T \to \infty$ jointly, as the convergence of the policy parameter estimator relies on the convergence of the estimators of the factors and loadings, which only hold when $S,T \to \infty$ jointly. 

In addition, we note that the estimation error of the first-step $\widehat{\alpha}_{jst}(u) - \alpha_{jst}(u)$ depends on the sample size of individual observations within each group-time pair ($N_{st}$). Hence, by controlling the relative growth rate between individual-level and group-level sample size (Assumption~\ref{ass.growth}.(ii)), the asymptotic first-step estimation error becomes negligible in the asymptotic representation of $\widehat{\delta}_{jt}(u)$, and the estimation error from the group-level regression contributed to the asymptotic bias and covariance, similar to \cite{bai2009panel}. We also note that the growth rate of $N_{\min}$ relative to $S$ and $T$ can be relaxed, but at the expense of more complicated asymptotic expressions.

We now provide a way to estimate the asymptotic bias and covariance given in Theorem~\ref{t6}, and subsequently,
Proposition~\ref{clt_estimate} below establishes their consistency.

Following \cite{bai2009panel}, we define $\widehat{B}_t(u) := [\widehat{B}_{1t}(u),...,\widehat{B}_{Jt}(u)  ]'$, where, for $j=1,...,J$,
\begin{align*}
	\widehat{B}_{jt}(u) := 
	&
	-\bigg(\frac{1}{S}\sum_{s=1}^{S} \widehat{R}_{js}(u)^2\bigg)^{-1}  \cdot \frac{1}{S^{3/2}T}\sum_{s,g=1}^{S}d_s \big(\widehat{\eta}_{jgt}(u)\big)^2 \widehat{f}_{jt}(u)' \bigg(\frac{\widehat{\Lambda}_j(u)'\widehat{\Lambda}_j(u)}{S}\bigg)^{-1}\widehat{\lambda}_{js}(u), 
\end{align*}
where $\widehat{R}_{js}(u) := d_s -S^{-1}\sum_{g=1}^{S}d_g \widehat{\lambda}_{jg}(u)'
\big(
S^{-1}\widehat{\Lambda}_j(u)' \widehat{\Lambda}_j(u)
\big)^{-1} \widehat{\lambda}_{js}(u)$.
We construct an estimator of the asymptotic covariance matrix $\Sigma_t(u_1,u_2)$, denoted as $\widehat{\Sigma}_t(u_1,u_2)$, by its empirical counterpart. $\widehat{\Sigma}_t(u_1,u_2)$ is a $J \times J$ block matrix, whose
$(j,k)^{\text{th}}$ block is given by
\begin{align*}
	\bigg(\frac{1}{S}\sum_{s=1}^{S} \widehat{R}_{j,s}(u_1)^2\bigg)^{-1} 
	\bigg(\frac{1}{S}\sum_{s=1}^{S} \widehat{R}_{k,s}(u_2)^2\bigg)^{-1} 
	\frac{1}{S}
	\sum_{s=1}^{S}\widehat{R}_{j,s}(u_1)
	\widehat{R}_{k,s}(u_2)
	\widehat{\eta}_{jst}(u_1)\widehat{\eta}_{kst}(u_2).
\end{align*}

\begin{proposition}\label{clt_estimate}
Suppose that the conditions of Theorem~\ref{t6} hold. In addition, we assume that for any fixed $u_1,u_2 \in \mathcal{U}$ and $j,k=1,...,J$, for $t,l \in \{1,...,T\}$ and $s,g \in \{1,...,S\}$, 
$$\mathbb{E}\big[\eta_{jst}(u_1) \eta_{kgl}(u_2)\big|D_s,D_g,W_s,W_g, {\Lambda}_{j}(u_1), {F}_{j}(u_1),{\Lambda}_{k}(u_2), {F}_{k}(u_2)\big] = 0$$
if $s \neq g$ or $t \neq l$.
Then, for any given $u_1, u_2 \in \mathcal{U}$ and $t \geq T_0$, we have
(i) $\widehat{B}_t(u) \overset{p}{\rightarrow} B_t(u)$,
and 
(ii) $\widehat{\Sigma}_t(u_1,u_2) \overset{p}{\rightarrow} \Sigma_t(u_1,u_2)$
as $ S,T \rightarrow \infty$.
\end{proposition}

Proposition~\ref{clt_estimate} assumes that the idiosyncratic errors are uncorrelated across groups and over time,
after conditioning group-level regressors and interactive fixed effects. 
For the correlated idiosyncratic errors, the analytical expression for the consistent estimator is hard to derived.
\cite{bai2009panel} provides some conjectures for bias-correction and covariance estimators using the partial sample method together with the Newey-West procedure. 
Recently, \cite{yan2023bootstrap} propose a wide dependent bootstrap method to consistently estimate the asymptotic covariance matrix when both serial and cross-sectional dependences exist.

\section{Treatment Effects for Group-Level Policy}\label{sec:treatment}
Measuring the impact of policy interventions is a central interest in economic and social studies. Our proposed modeling framework is capable to serve this purpose. To see this, we consider the true data generating process under the potential outcome framework
of \cite{rubin1974estimating},
and provide identification results for several policy effect parameters. 

We consider a binary group-level policy
that is employed at a known time $T_{0}$
onward
with
$1 < T_{0} < T$.
Thus, the sample periods
can be divided into 
the before-period  
($t < T_{0}$)
and after-period ($t \geq T_0$).
Also, 
let
$d_{s} = 1$
if group $s$ is treated
after $T_{0}$
and
0 otherwise, which implies that $d_{st} = d_s 1\{t \geq T_0\}$.\footnote{
  Our analysis and application concentrate on a single policy change event, rather than staggered or sequential policy changes. The latter setup, requiring the causally interpretable estimates framework as discussed in   
  \cite{callaway2021difference, sun2021estimating, athey2022design} among others, fall outside the scope of our study.}
Let $y_{ist}^{1}$ and $y_{ist}^{0}$ denote the individual potential outcomes with and without exposure to the group-level policy ($d_{st} = 1$ and $d_{st} = 0$), respectively.
Then, the observed outcome is written as
\begin{eqnarray*}
	y_{ist} = (1-d_{st}) y_{ist}^{0} + d_{st} y_{ist}^{1}.
\end{eqnarray*}
Correspondingly, under treatment status $d_{st} = d \in \{0,1\}$,
the $u$\textsuperscript{th} conditional quantile of the
potential outcome $y_{ist}^{d}$
is given by
\begin{equation}
	\label{Q_std}
	Q_{st}^{d}(u|z_{ist}) \equiv Q_{y^d_{ist}|z_{ist},\alpha_{st}^d}(u)
	=
	z_{ist}' \alpha_{st}^{d}(u).
\end{equation}
Also, we suppose that the group-level treatment affects the potential conditional quantile through the potential (random) marginal effects $\alpha_{st}^{d}(u)$. That is,
we specify the $j^{\text{th}}$ element of $ \alpha_{st}^{d}(u)$ given the treatment status $d_{st} = d$ as
\begin{eqnarray}
	\label{eq:delta1}
	\alpha_{jst}^{d}(u)
	=  \Delta_{jst}(u)d +
	x_{st}^{\prime}\beta_j(u) +
	f_{jt}(u)' \lambda_{js}(u) +
	\eta_{jst}^d(u), 
	\quad j=1,...,J,
\end{eqnarray}
where
$\Delta_{jst}(u)$ represents the random 
policy effects. 
The correlation between $d_{st}$ and the factor component is unrestricted so that the selection into the treatment can be correlated with factor loadings $\lambda_{js}(u)$. Additionally, the implementation of the policy can be dependent on the aggregate shocks characterized by $f_{jt}(u)$.
Below, we first introduce three treatment parameters-of-interest, and then propose the identification results.

\subsection{Treatment Parameters}

For a given probability level $u$ and individual-level regressor $z$,
the conditional quantile $Q_{st}^d(u|z)$ is a random variable depending on the policy status and additional group-level characteristics.
\cite{Arellano2016} propose an average of conditional quantile treatment effect
to measure treatment effects in non-linear response models.
Extending their idea,
we consider the average quantile treatment effect on the treated (AQTT) at time $t \ge T_{0}$,
defined as:
\begin{eqnarray*}
	\Delta_{t}^{AQTT}(u|z)
	:=
	\E[
	Q_{st}^{1}(u|z)
	-
	Q_{st}^{0}(u|z)
	| d_{s} = 1
	].
\end{eqnarray*}
AQTT can be considered as an extension of the treatment effect measure on unconditional quantiles,
which is employed in numerous empirical studies \citep[for example, see][]{lee1999wage,angrist2004does,bitler2006mean}.
The extensive body of research on quantile treatment effects, as illustrated by 
\cite{callaway2019quantile,wuthrich2020comparison},
typically measures treatment effects as the difference of the quantile functions of responses from the treated and untreated groups.
However, the AQTT approach
differs by accounting for the heterogeneity in the conditional quantile function 
$Q_{st}^d(u|z)$ across various groups and time periods.

As an alternative measure,
we consider 
spreads of conditional quantile functions
to quantify inequality within and between collections of individuals characterized by the individual-level
regressors.
Given the policy status $d \in \{0, 1\}$,
we fix individual-level regressors $z \in \mathcal{Z}$
and 
consider two probability levels of interest
$u_{1}, u_{2} \in (0,1)$
with $u_{2} > u_{1}$.
Then, a within-inequality measure
under the policy status $d_{st} = d$ at time $t \ge T_{0}$
is defined as
the spread of conditional quantiles:
$$
\Delta_{st}^{W, d}(u_{1}, u_{2}|z)
:=
Q_{st}^{d}(u_{2}|z)
-
Q_{st}^{d}(u_{1}|z).
$$
Similarly, 
we fix individual attributes 
$z_{1}, z_{2} \in \mathcal{Z}$
and a probability level $u$
to define a between-inequality measure
under the policy status $d$ at time $t \ge T_{0}$ as 
$$
\Delta_{st}^{B, d}(u |z_{1}, z_{2})
:=
Q_{st}^{d}(u|z_{2}) 
-
Q_{st}^{d}(u|z_{1}).
$$

Figure \ref{fig:inequality}
illustrates the within- and between-inequalities.
The within-inequality measures the
dispersion of the distribution of the outcome conditional on
individual characteristics $z$ by using two conditional
quantile functions,
whereas
the between-inequality measures the distance
between two conditional distributions at a certain probability level.

\begin{figure}[ht]
	\centering
	\caption{Inequality Measures}
	\begin{subfigure}[b]{0.45\textwidth}
		\centering
		\caption{Within-Inequality} 
		\includegraphics[width=0.99\linewidth]{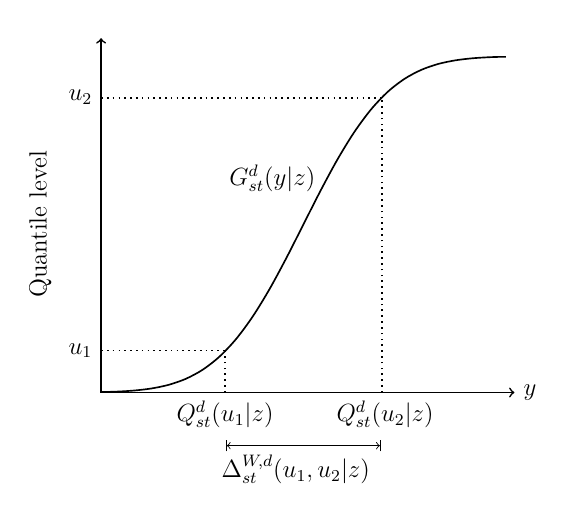}         
	\end{subfigure}
	~
	\begin{subfigure}[b]{0.45\textwidth}
		\centering
		\caption{Between-Inequality} 
		\includegraphics[width=0.99\linewidth]{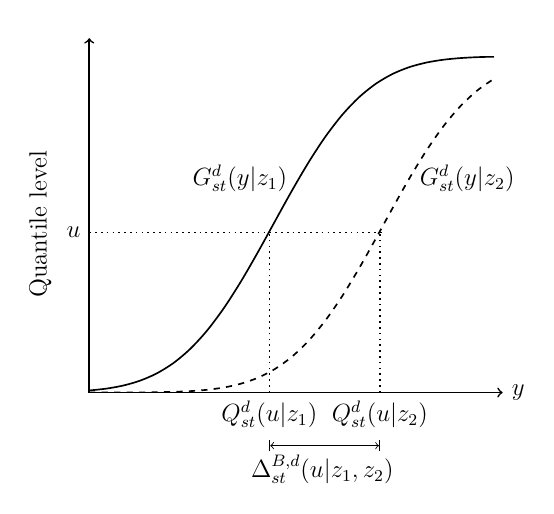} 
	\end{subfigure} \\
	\begin{minipage}{0.95\linewidth}
		\footnotesize
		\textit{Notes:}
                Under the treatment status $d$,
		Panel (a) illustrates
		the within-inequality measure
		as
		the speared of two conditional quantile functions
		at
		quantile levels $u_{1}, u_{2}$,
                while  
		Panel (b) shows
		the between-inequality measure
		as a distance
		of 
		two distributions functions
		conditional on two distinct set of individual attributes $z_{1}, z_{2}$,
		given the same quantile level $u$.
	\end{minipage}
	\label{fig:inequality}
\end{figure}

Group-level policies can affect these inequality measures and their impact can be
quantified
as changes in the inequality measures at time $t$
averaged over treated groups:
\begin{eqnarray*}
	\dt{\Delta}_{t}^{W}(u_{1}, u_{2}|z)
	&:=\E \big [
	\Delta_{st}^{W, 1}(u_{1}, u_{2}|z)
	-
	\Delta_{st}^{W, 0}(u_{1}, u_{2}|z)\big|d_{s}=1 \big], \\
	\dt{\Delta}_{t}^{B}  (u |z_{1}, z_{2})
	&:=
	\E
	\big [
	\Delta_{st}^{B, 1}
	(u |z_{1}, z_{2})
	-
	\Delta_{st}^{B, 0}(u |z_{1}, z_{2})
	\big| d_{s}=1
	\big].
\end{eqnarray*}


\subsection{Identification}\label{identification}

We now exhibit conditions, under which model (\ref{Q_y})-(\ref{alpha_st}) on the observed outcome allows for the identification for
treatment parameters.
Under a similar setup,  
\cite{gobillon2016regional} prove the identification
of average policy effects using the mean regression model.
We make the following assumptions. Throughout the assumptions, we fixed a given $j=1,...,J$ and $u \in \mathcal{U}$, unless stated otherwise.

\begin{assumption}\label{a2.3}
	For all 
	$(s, t) \in \{1, \dots, S\} \times \{1, \dots, T\}$ and $j=1,...,J$,
	\begin{enumerate}[label=(\roman*)]
		\item
		$\E[\eta_{jst}^d(u)|d_{st}, X_{s}, \lambda_{js}(u), F_{j}(u)] =
		\E[\eta_{jst}^d(u)|X_{s}, \lambda_{js}(u), F_{j}(u)] = 0$ for $d=0,1$.
		\item 
		$\E[\Delta_{jst}(u)|d_s=1,X_s] = \E[\Delta_{jst}(u)|d_s=1]$.
	\end{enumerate}
	
\end{assumption}

Assumption~\ref{a2.3}.(i) requires that the error term for the potential outcome is mean-zero and mean-independent of the treatment status, conditional on group-level observed and unobserved variables. 
Assumption~\ref{a2.3}.(i) implies a type of parallel-trend assumption. For clarification, we first note that the assumption implies that, for $t \geq T_0$, we have
$\E[\eta^0_{jst}(u) - \eta^0_{js,T_0-1}(u)|d_{s}=1,X_{s}, \lambda_{js}(u), F_{j}(u)]
= \E[\eta^0_{jst}(u) - \eta^0_{js,T_0-1}(u)|d_s=0,X_{s}, \lambda_{js}(u), F_{j}(u)]$. By the iterative law of expectation, this immediately leads to $
\E[\alpha_{jst}^0(u|z) - \alpha_{js,T_0-1}^0(u|z)|d_s = 1] = \E[\alpha_{jst}^0(u|z) - \alpha_{js,T_0-1}^0(u|z)|d_s = 0]
$. Essentially, it states that, without treatment, the changes in the marginal effects of any variable $j$ don't depend on whether the individual belongs to a treated group or not on average. Taking the analysis of the minimum wage policy as an example, this assumption suggests that the changes in the marginal effect of gender (or race, education, etc,) over time are the same for individuals in the treated and control industries, supposing that the policy is not introduced.

Assumption~\ref{a2.3}.(ii) is a technical assumption that is also considered by \cite{gobillon2016regional}. It is imposed to fulfil the technical requirement that $\E\big[\Delta_{jst}(u) - \E[\Delta_{jst}|d_s=1] \big| d_{st},x_{st}\big] = 0$. It assumes that the random policy effects are mean-independent of the group-level covariates within the treated group. We note that the same as the case of \cite{gobillon2016regional}, this assumption is stronger than necessary. However, as the empirical study in this paper is in the absence of $x_{st}$, generalization of this assumption is out of the scope of the current paper. One may generalize the model by interacting covariates with the treatment indicator, and this would substantially weaken this condition \citep{caetano2022difference}.

In this paper, we do not impose the standard yet restrictive rank invariance \citep{imbens2007nonadditive}, or less restrictive rank similarity \citep{chernozhukov2005iv} assumptions on the individual unobserved heterogeneities, which requires an individual's rank in the potential outcome distribution to be the same or has the same probability distribution across treatment status. Instead, our identification results directly rely on the quantile specifications of the potential outcome. However, we do note that, if the rank preservation assumption holds up, AQTT can be additionally interpreted as individual causal effect for (the same) individual at $u^{\text{th}}$ quantile before and after treatment, and quantile treatment effect parameters can be identified accordingly.

The theorem below shows that we can identify the time-varying distributional impact of a group-level policy using model (\ref{Q_y})-(\ref{alpha_st}).

\begin{theorem}
	\label{theorem:identification}
	Suppose that Assumptions~\ref{a2.1}-\ref{a2.2} and \ref{a2.3} hold. Then, for a given $t \geq T_0$,
        and $(u, z) \in \mathcal{U} \times \mathcal{Z}$, we have
	\begin{eqnarray*}
		\Delta_{t}^{AQTT}(u|z)
		=
		z'\delta_{\cdot t}(u),
	\end{eqnarray*}
        and, for $u_{1}, u_{2} \in \mathcal{U}$ and $z_{1}, z_{2} \in \mathcal{Z}$,
	\begin{eqnarray*}
		\dt{\Delta}_{t}^{B}(u| z_{1}, z_{2})
		= 
		(z_{2} - z_{1})'\delta_{\cdot t}(u)
		\ \ \ \ \ \mathrm{and} \ \ \  \ \
		\dt{\Delta}_{t}^{W}(u_{1}, u_{2}| z)
		= 
		z'\big (\delta_{\cdot t}(u_{2}) - \delta_{\cdot t}(u_{1}) \big).
	\end{eqnarray*}
	Here, $\delta_{\cdot t}(u) := [\delta_{1t}(u),...,\delta_{Jt}(u)]'$,
	whose $j^{\text{th}}$ element $\delta_{jt} := \E[\Delta_{jst}(u) | d_s = 1]$
	can be identified   as
	$ 
	\delta_{jt}(u) = \E[d_{st}\Pi_{st}]^{-1}\E[\Pi_{st}(\alpha_{jst}(u) -  f_{jt}(u)'\lambda_{js}(u))]
	$
	with $\Pi_{st} := d_{st} - \E[d_{st}x_{st}']\E[x_{st}x_{st}']^{-1}x_{st} $.
\end{theorem}

The above result shows that we can identify the treatment effects of group-level policy which are allowed to vary according to individuals' observed and unobserved characteristics. Taking into account the interplay between a group-level policy and individuals' characteristics,
our framework can explicitly identify heterogeneous impacts of the policy across individuals sharing the same observed regressors $z$
and also the impact on the within- and between-inequalities among individuals. 
To simplify the proof, we treat factor and loadings as observed. When they are unobserved, the iterative estimation approach, proposed in Section~\ref{sec:estimation}, can be used to obtain their estimators. 

According to the identification Theorem~\ref{theorem:identification}, we define the estimators of the treatment parameters as $\widehat{{\Delta}}^{AQTT}_t(u) := z'\widehat{\delta}_{\cdot t}(u)$, 
$\widehat{\dt{\Delta}}_{t}{\hspace{-0.1cm}}^{B}(u| z_1,z_2) := (z_2-z_1)'\widehat{\delta}_{\cdot t}(u)$,
$\widehat{\dt{\Delta}}_{t}{\hspace{-0.1cm}}^{W}(u_{1}, u_{2}| z) := z'(\widehat{\delta}_{\cdot t}(u_2) - \widehat{\delta}_{\cdot t}(u_1))$. Then, we establish the CLT results for the estimators of the treatment parameters in the next theorem.

\begin{theorem}
\label{CLT_inequality}
	Under the conditions of Theorems~\ref{t6} and \ref{theorem:identification}, for any $u, u_1,\ u_2 \in \mathcal{U}$, $z,z_1,z_2 \in \mathcal{Z}$ and $t \geq T_0$, we have, as $S,T \to \infty$,
	\begin{align*}
		\sqrt{S}  \Big (    \widehat{{\Delta}}^{AQTT}_t(u|z)  - {\Delta}^{AQTT}_t(u|z)  \Big )
		&
		\overset{d}{\rightarrow} 
		\mathcal{N}  \Big (   z' B_t(u)  , z'\Sigma_t(u,u)z  \Big ), \\
		\sqrt{S}  \Big (
		\widehat{\dt{\Delta}}_{t}{\hspace{-0.1cm}}^{B}(u, | z_1,z_2)
		- \dt{\Delta}_{t}^{B}(u| z_1,z_2)
		\Big ) 
		&  \overset{d}{\rightarrow} 
		\mathcal{N}  \Big (
		\big(z_2-z_1\big)' B_t(u),
		\sigma_{B,t}^{2}(u, z_{1}, z_{2})
		\Big ),\\
		\sqrt{S}  \Big ( 
		\widehat{\dt{\Delta}}_{t}{\hspace{-0.1cm}}^{W}(u_{1}, u_{2}| z)
		- \dt{\Delta}_{t}^{W}(u_{1}, u_{2}| z)
		\Big )  
		&  \overset{d}{\rightarrow} 
		\mathcal{N}  \Big (
		z'\big( B_t(u_2) - B_t(u_1) \big), 
		\sigma_{W,t}^{2}(u_{1},u_{2}, z)
		\Big ),
	\end{align*}
	where 
	$B_t(u)$ is defined in Theorem~\ref{t6}, 
	$
	\sigma_{B,t}^{2}(u, z_{1}, z_{2})
	:=
	(z_2-z_1)' \Sigma_t(u,u)(z_2-z_1)
	$
	and
	$
	\sigma_{W,t}^{2}(u_{1},u_{2}, z)
	:=
	z'\big(\Sigma_t(u_1,u_1)-\Sigma_t(u_1,u_2)-\Sigma_t(u_2,u_1)+\Sigma_t(u_2,u_2)\big)z 
	$
	with 
	$\Sigma_t(u_1,u_2)$ in Assumption~\ref{ass.clt_delta}.
\end{theorem}

Recall that Proposition~\ref{clt_estimate} proposes the consistent estimators of the asymptotic bias and covariance.
Using this result with Theorem~\ref{CLT_inequality}, it is straightforward to construct the confidence intervals for the time-varying AQTT, changes in between- and within-inequality measures.


\section{Empirical Analysis}\label{sec:empirical}

\subsection{Background on Racial Income-Inequality}

Racial economic inequalities
have persisted in the United States over long periods.
Among these inequalities,
the income gap between black and white workers
is evident.
As in Figure \ref{fig:gap}, 
the income gap, measured by 
the average annual earnings, 
was around 
20-30\% for the last two decades,
whereas 
the gap significantly dropped
during the late 1960s and early 1970s.
The empirical literature has explored factors that
narrowed the racial income gap during those periods,
including federal anti-discrimination legislation \citep{smith1984affirmative} and improvements in education \citep{smith1977black,card1992school}.


\begin{figure}[ht]
	\centering
	\caption{White-Black Unadjusted Wage Gap in the Long Run}
	\includegraphics[width=0.7\textwidth]{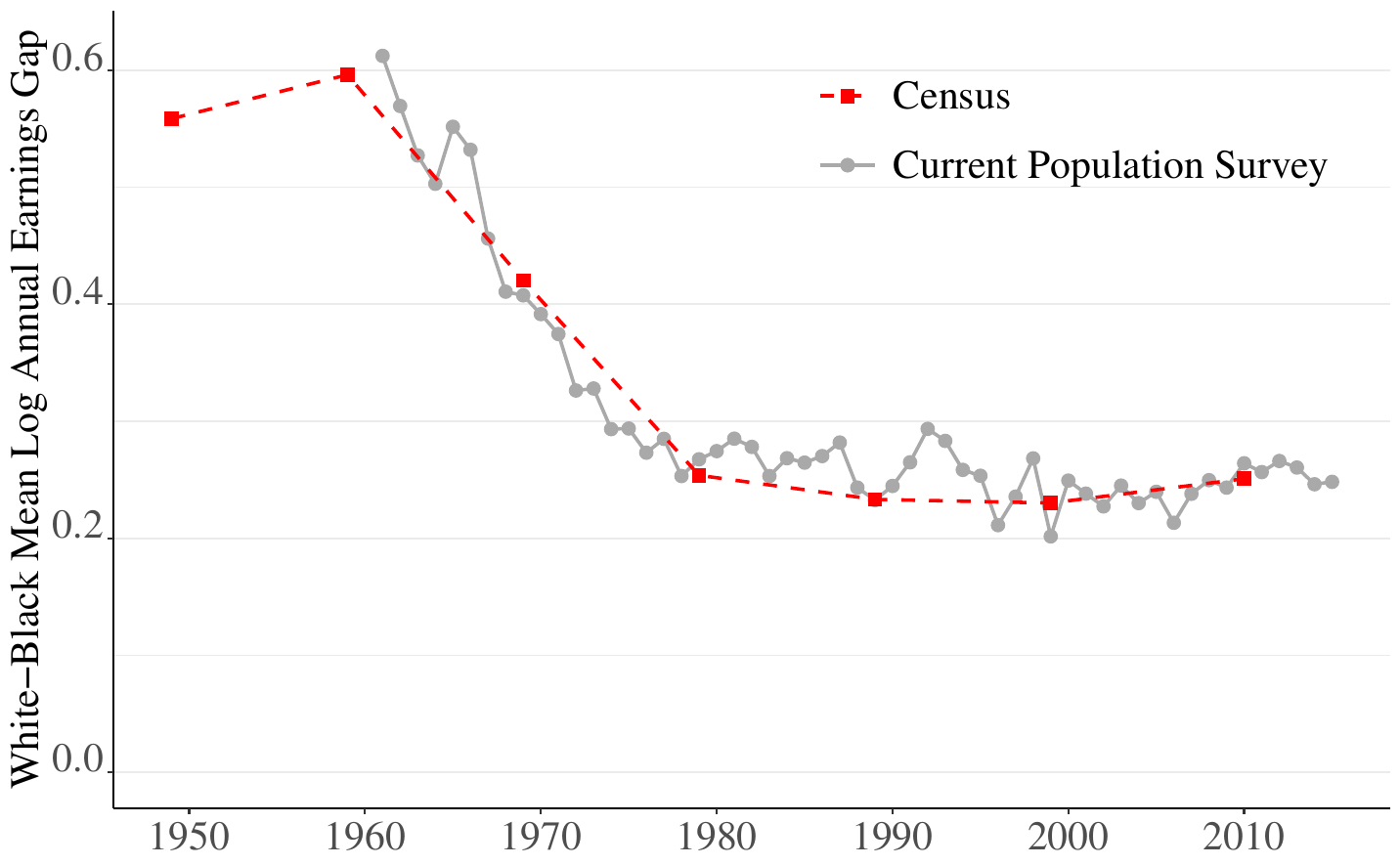}         
	\begin{minipage}{0.95\linewidth} 
		\footnotesize
		\textit{Notes:}  This figure is
		a replication of
		Figure 1 in \cite{derenoncourt2021minimum}.
		The data sources are 
		the Current Population Survey 1962--2016,
		U.S. Census from 1950 to 2000, and
		American Community Survey data in 2010 and 2017.
		The sample includes black or white adults aged 25--65,
		who worked more than 13 weeks last year,
		worked three hours last week,
		do not live in group quarters,
		are not self-employed and  not unpaid family worker
		with no missing industry or occupation code.
	\end{minipage}  
	\label{fig:gap}   
\end{figure}  


Recently,
\cite{derenoncourt2021minimum} put forward a new explanation: the extension of the federal minimum wage to some industries.
The 1966 FLAS established a federal minimum wage (effective February 1967) in previously unregulated industries,
which employed about $20\%$ of the total workforce in the US and nearly a third of all black workers.
They 
evaluate the minimum wage policy effect on earnings, 
using a cross-industry difference-in-differences design,
in which
eight treated and eight control industries
were subject to the minimum wage
under the 1966 and 1938 FLSA, respectively. Additional information and background about the dataset are provided in the Online Supplement~\ref{sec:additional_empirical_cpt2}.

Using repeated cross-sections
of black and white workers aged between 25 and 55
for years 1961 and 1963--1980,
extracted from March CPS,\footnote{
	Since the March CPS of year $t$ contains information in calendar year $t-1$,
	the data source is the 1962, 1964--1981 March CPS.
	The 1963 March CPS is excluded due to the lack of observations and missing demographic information.
}
they
estimate the following two-way fixed effects mean regression model:
\begin{eqnarray}
	\label{eq:meanDID}
	y_{ist} = \alpha + \delta_{t} d_{st} + {z}_{ist}'{\beta} + \phi_{s} + \nu_{t} + \eta_{ist},
\end{eqnarray}
for worker $i$ in industry $s=1, \dots, 16$ and time
$t = 1961,1963 \dots, 1980$.
Here, 
$y_{ist}$ is the log annual earnings deflated by annual CPI-U-RS (\$2017)\footnote{
	Using March CPS data in the 1960s and early 1970s, we only directly
	observe annual earnings, but not hourly wages, whereas 
	the CPS contains more detailed individual worker–level information 
	the Bureau of Labour Statistics data.
	See Section III.B.~of \cite{derenoncourt2021minimum}.
},
and 
$d_{st}$ denotes a dummy variable 
taking 1 if
industrial sector $s$ is subject to 
the federal minimum wage,
and 0 otherwise.
Also, $z_{ist}$ is a vector of worker's characteristics, and 
the unobserved random variables consist of
industry fixed effect $\phi_{s}$, time fixed effect $\nu_{t}$
and 
an idiosyncratic error $\eta_{ist}$.
The parameter of interest is 
$\delta_{t}$, which measures dynamic policy effects.

Their result shows that, 
after controlling for individual characteristics, the average wage of workers in the newly covered industries is around $5\%$ higher relative to that in control industries in 1967--1980 compared with the pre-period 1961--1966,
and the effect of minimum wage reform on workers' log-earning is more than twice as large for black workers as that for white workers on average.
In addition to the above regression,
they present several regression results to uncover intricate facets of the effects of minimum wage by taking various variables as the dependent variable in
(\ref{eq:meanDID}), including log annual wage or its unconditional quantiles, and also
selecting sub-samples based on workers' characteristics.

\subsection{Model and Practical Implementation}\label{sec:implementation}

In this paper, we analyze time-varying policy effects
from 1967 to 1980 at quantile $u \in \{0.1, 0.3, 0.5,  0.7, 0.9\}$. 

	We use the same dataset as \cite{derenoncourt2021minimum}. However, we note that as the Forestry and Fishing industry only contains 35 individual observations per year on average, which could lead to large first-step quantile estimation error, we remove it from the treated industries. Thus, there are seven treated industries and eight controlled industries in the following study. We provide an elementary explanatory data analysis in Appendix~\ref{sec:empirical_summary_stat}.
Given the model of \cite{derenoncourt2021minimum},
we consider the following $u^{\mathrm{th}}$ quantile regression model
in (\ref{Q_y})
with 
$\alpha_{st}(u) = [\alpha_{1st}(u), \dots, \alpha_{Jst}(u)]'$,
where,
for $j=1, \dots, J$, 
\begin{eqnarray}
	\label{quantileDiD.eq2}
	\alpha_{jst}(u)
	= \beta_{j}(u) + 
	\delta_{jt}(u)    d_{st}
	+
	{f}_{jt}(u)' {\lambda}_{js}(u) +
	\eta_{jst}(u). 
\end{eqnarray}
Here,
a set of coefficients
$\{\delta_{jt}(u)\}_{t=1967}^{1980}$ 
measures the time-varying policy effect.
Model (\ref{quantileDiD.eq2}) does not involve any individual-level covariates since such information is not available in the March CPS dataset. We incorporate interactive fixed effects into the model to account for the industrial and temporal heterogeneities. The industry-specific loadings can be interpreted as latent industrial factors.

For simplicity of interpretation,
we treat some of the original covariates as
ordered variables rather than dummy variables for categories.
More precisely,
$z_{ist}$ includes a constant $1$,
dummy variables for race (white/black), gender (male/female)
and work type (full-time/part-time), and ordered variables
including years of schooling, experience, experience squared, the number of weeks worked in a year, and the number of hours worked in a week.\footnote{
	\cite{derenoncourt2021minimum} use dummy variables
	to
	control for
	the number of weeks worked in a year and
	the number of hours worked in a week,  
	because hourly wage is not available in the CPS data
	during the periods of interest.}
This selection yields a very similar result to that
of the mean-regression in (\ref{eq:meanDID}) as
the one in \cite{derenoncourt2021minimum}. See Figure~\ref{fig:pre_analysis}. In this empirical study, we are particularly interested in the heterogeneous policy-effects varies between gender and race.

\begin{figure}[!htb]
	\centering
	\caption{Estimation Results of Mean Regression}\label{fig:pre_analysis}
	\includegraphics[width=0.7\textwidth]{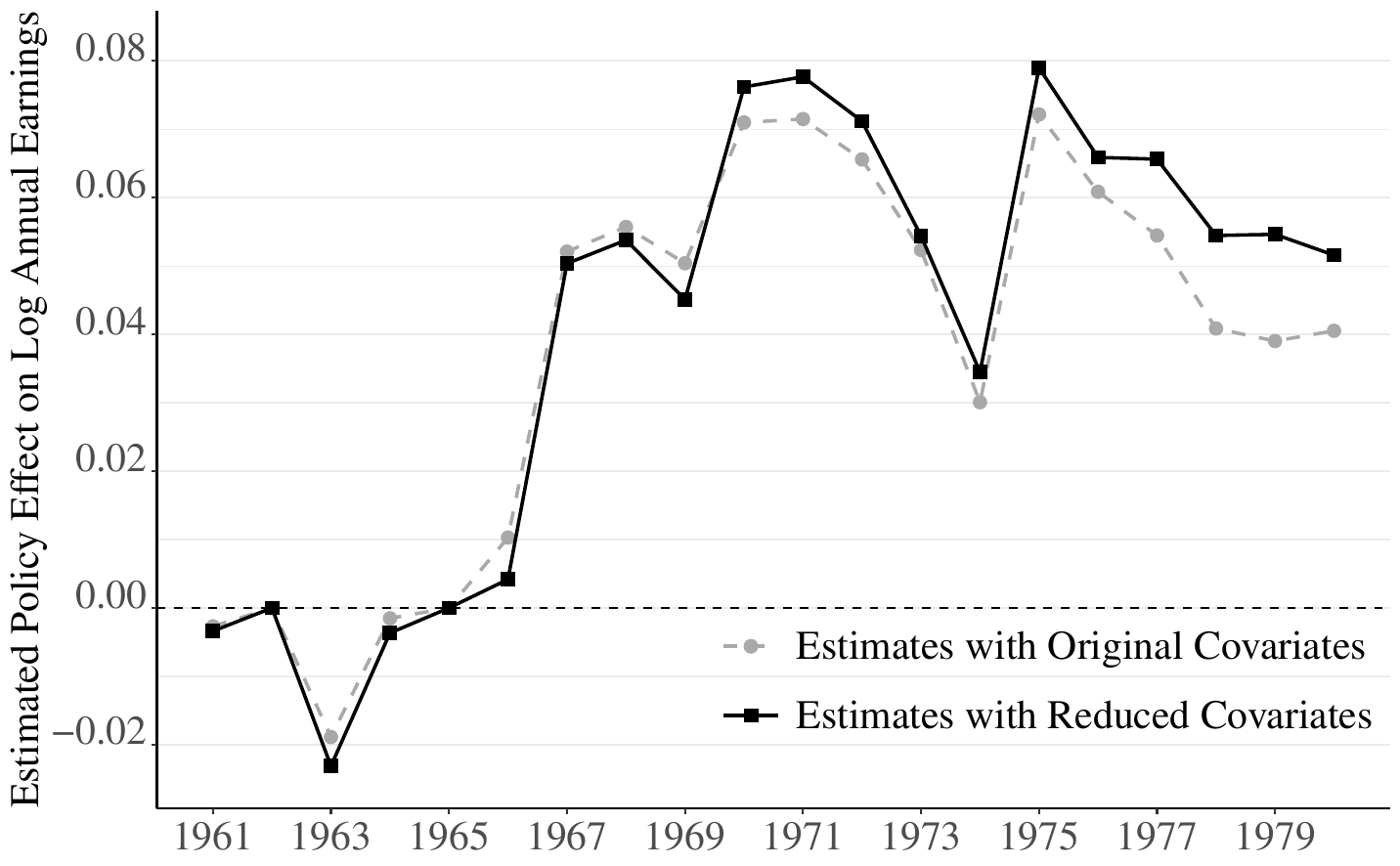}
	\begin{minipage}{0.95\linewidth}
		\footnotesize
		\textit{Notes:}
		We plot the estimates of time-varying policy effect $\delta_t$ in Model (\ref{eq:meanDID}) given the original set of controlled variables (in dashed grey line) and the estimates given the reduced set of covariates that used in our empirical analysis (in solid black line).
	\end{minipage}
\end{figure}

\subsection{Results Analysis}

In Figure \ref{fig:marginal}, we present time-varying policy effects $\delta_{jt}(u)$ in (\ref{quantileDiD.eq2}) with $95\%$ confidence intervals for the categorical individual-level covariates. The policy effects for the continuous covariates: education and experience are insignificant across quantiles. That said, after taking the differentiation between gender and race into account, the policy effect is insignificant across the levels of skill (measured by education and experience). For presentation conciseness, those figures are omitted. 

Panel (a) reports the effects on the intercept coefficients, which correspond to white, male, full-time workers, which are insignificantly different from zero for most of the estimates.
Panel (b) shows statistically significant positive policy effects for black workers in the majority of the years across all quantiles. Especially, the policy effects are most significant at the 0.1th conditional quantile, which is 15--20\% (0.15--0.2 log points).
Panel (c) presents the estimated policy effects on the female dummy's coefficient. 
The effects in the late 1970s are positive and significant up to 0.7th quantile with a magnitude of 5\%. 
However, caution is warranted in interpreting the results, which may be an integrated impact of the 1966 Fair Labour Standards Act and two pieces of important legislation that targeted labour market discrimination against women: the Equal Pay Act of 1963 and Title VII of the Civil Rights Act of 1964.\footnote{
	The Equal Pay Act of 1963 is a federal law that amends the Fair Labour Standards Act and prohibits wage disparity based on gender. Title VII of The Civil Rights Act of 1964 more broadly prohibits discrimination in employment on the basis of race, colour, religion, national origin, and gender.}
Although the gender gap of median wages for full-time, full-year workers was unchanged over the 1960s--1970s \cite[see][]{blau2017gender}, \cite{bailey2021changes} recently document sharp increases in women's wages relative to men's
below median during the 1960s. They underscore the importance of minimum wage policy and the laws to target gender-based workplace discrimination. Our result is consistent with their findings and further suggests long-term positive effects even at 0.5th and 0.7th quantiles, conditional on individual attributes.

\begin{figure}[!htb]
	\centering
	\caption{Time-Varying Policy Effect Estimates ($\delta_{jt}$)}
	\label{fig:marginal}
	\begin{subfigure}{\textwidth}
		\centering
		\caption{Intercept (Baseline: White, Male, Full-Time Workers) }
		\includegraphics[width=\textwidth]{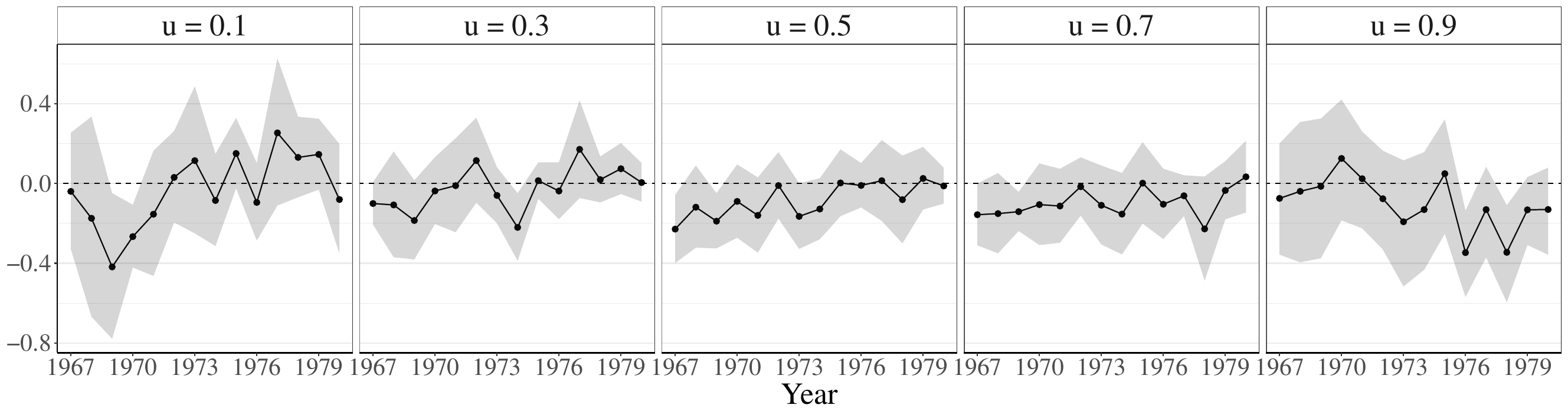} 
	\end{subfigure}\\
	\begin{subfigure}{\textwidth}
		\centering
		\caption{Black}
		\includegraphics[width=\textwidth]{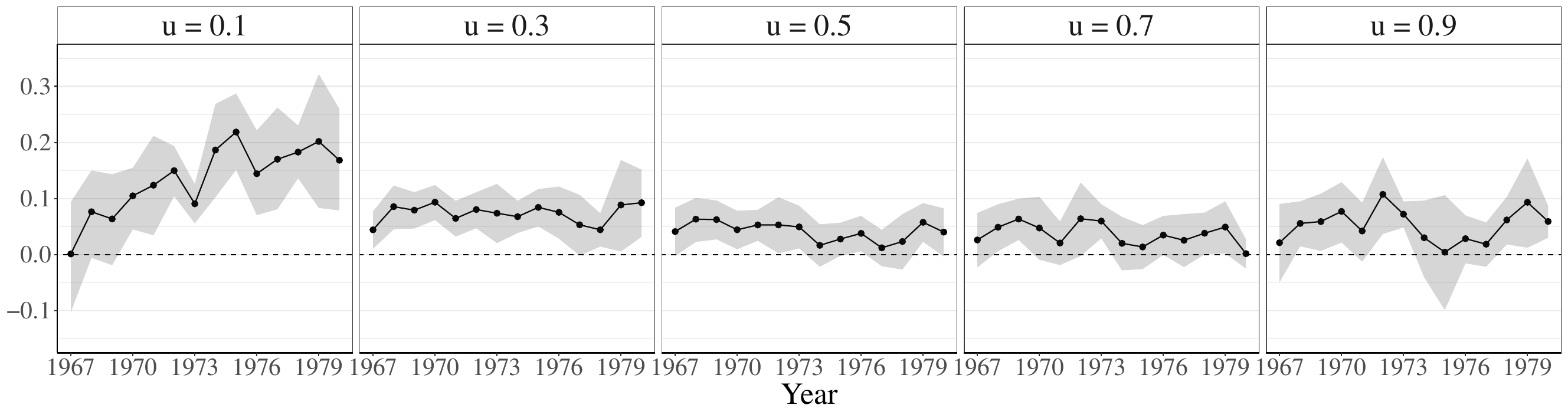} 
	\end{subfigure}\\ 
	\begin{subfigure}{\textwidth}
		\centering
		\caption{Female}
		\includegraphics[width=\textwidth]{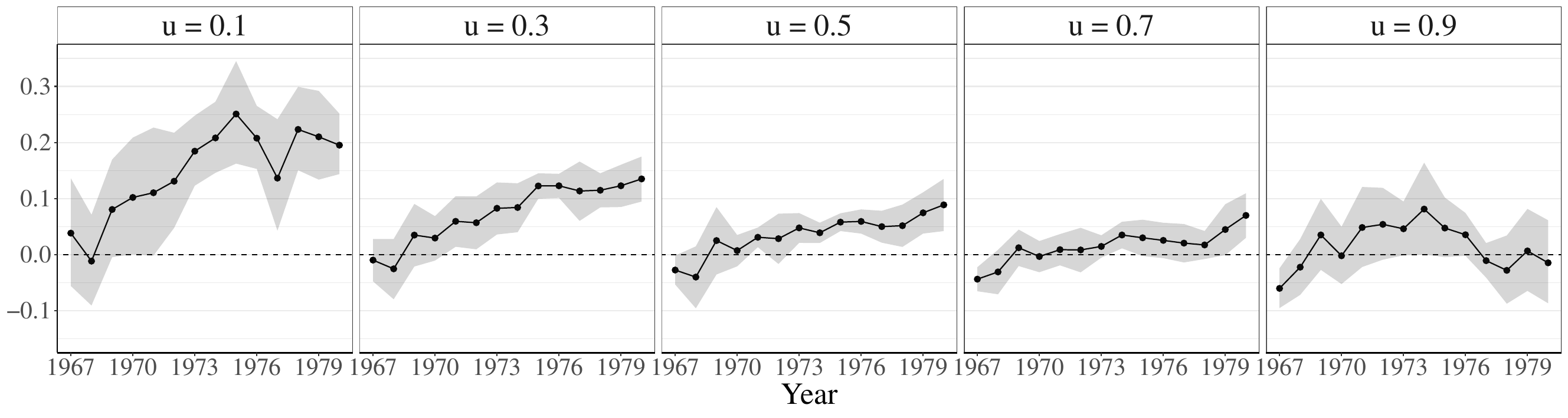}  
	\end{subfigure} 
	\begin{minipage}{0.95\linewidth}
		\footnotesize
		\textit{Notes:} Panels (a)-(c) present the estimated time-varying marginal policy effect
		$\delta_{jt}(u)$
		for $t=1967,\dots,1980$. From left to right, figures correspond to the estimates at quantiles $u = 0.1, 0.3, 0.5, 0.7, 0.9$. Point estimates are plotted with solid black lines, and the pointwise $95\%$ confidence interval are shown as grey shaded area.
	\end{minipage}
\end{figure}

In Figure \ref{fig:within},
we present estimated policy effects on changes in the within-inequality measure
$\dt{\Delta}_{t}^{W}(u_{1}, u_{2}|z)$
in the year $1980$.
For quantile pairs 
$(u_{1}, u_{2}) = (0.1, 0.9)$
or  
$(0.1, 0.5)$,
we measure how much the minimum wage policy changes
the conditional quantile spread.
In the following analysis, we consider the labour with an average skill level
(12 years of education and 20 years of experience),
while reporting
three pairs of
categorical individual attributes,
as shown in the horizontal axis.
The estimates suggest that the introduction of
minimum wage reduces the within-inequality, while the reduction is insignificantly different from 0 for all subpopulations. In addition, we also note that the same conclusion can be drawn when considering less-skilful or more-skilful labours within the three sub-populations considered in Figure~\ref{fig:within}, since the estimated $\delta_{jt}(u)$ associated with education and experience are insignificantly different from 0.


\begin{figure}[!htb]
	\centering
	\caption{Changes in Within-Inequality in 1980}
	\label{fig:within}
	\begin{subfigure}{0.495\textwidth}
		\centering
		\caption{$u_1=0.1,\ u_2=0.9$}
		\includegraphics[width=0.95\textwidth]{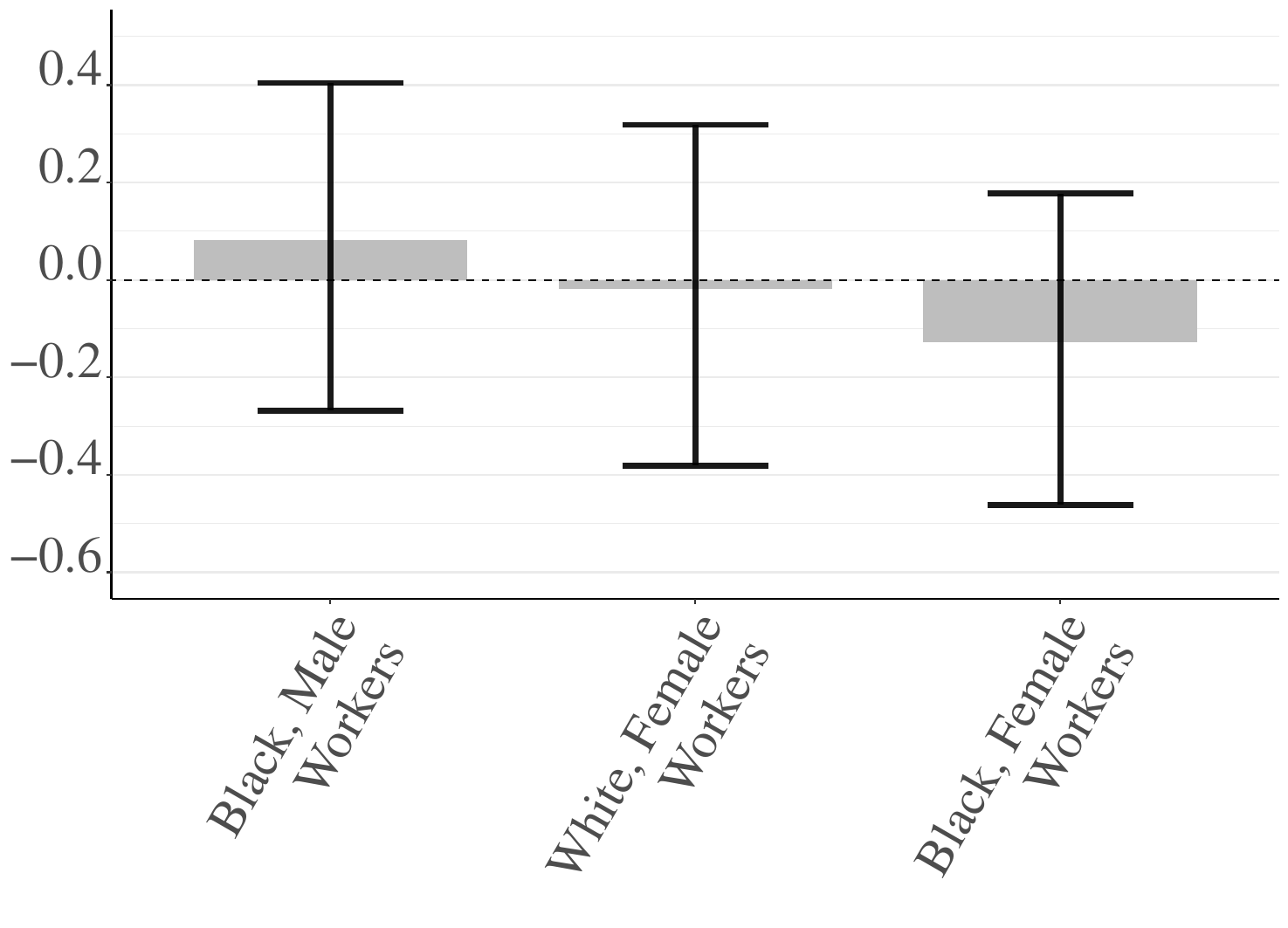} 
	\end{subfigure}
	\begin{subfigure}{0.495\textwidth}
		\centering
		\caption{$u_1=0.1,\ u_2=0.5$}
		\includegraphics[width=0.95\textwidth]{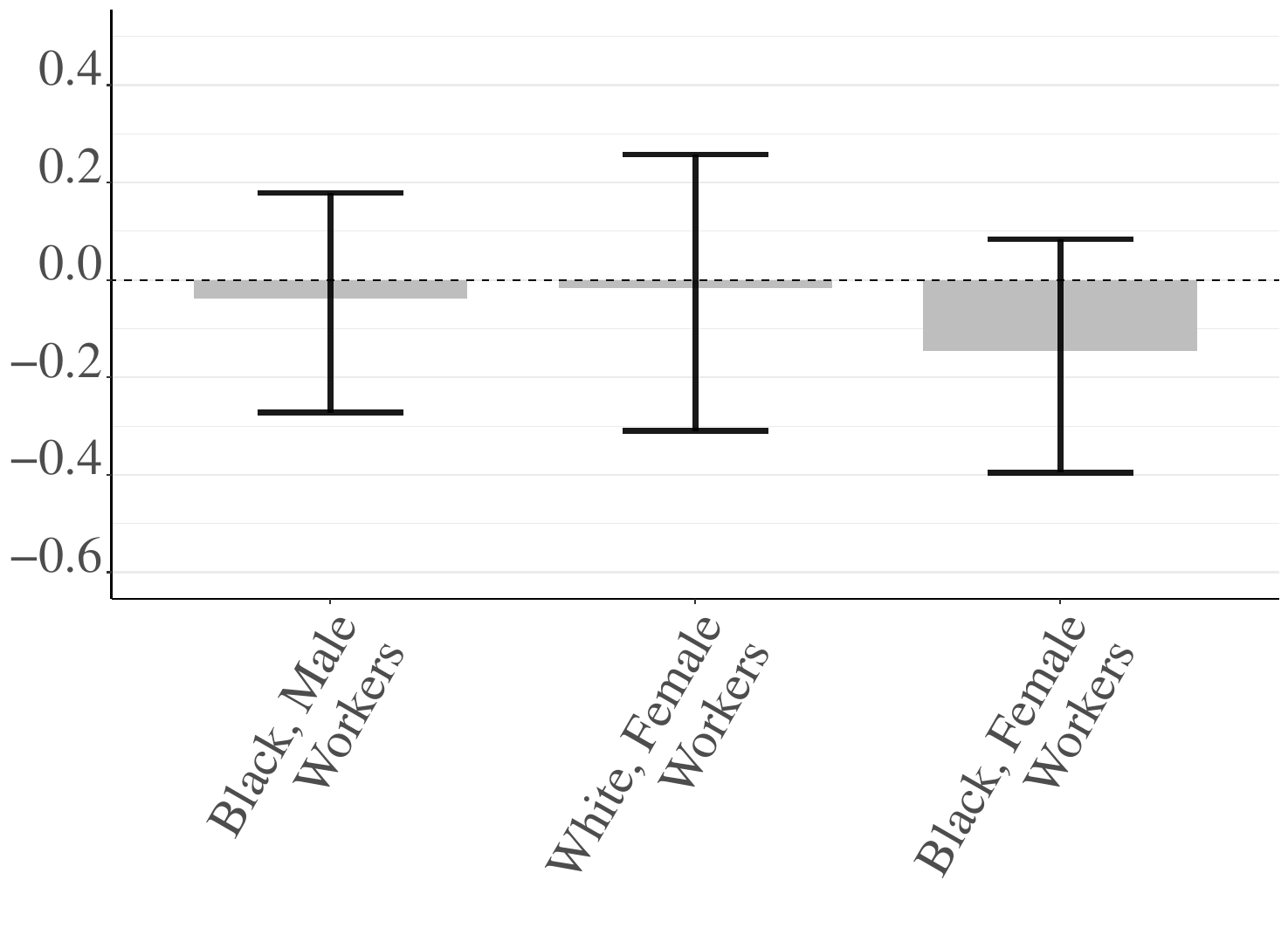} 
	\end{subfigure} \\
	\begin{minipage}{0.95\linewidth}
		\footnotesize
		\textit{Notes:}
		Panels (a)-(b) report
		estimated changes in the within-inequality
		$\dt{\Delta}_{t}^{W}(u_{1}, u_{2}|z)$
		among individuals with attributes $z$
		in year $t=1980$,
		with 
		$(u_{1}, u_{2}) = (0.1, 0.9)$
		in Panel (a)
		and 
		$(u_{1}, u_{2}) = (0.1, 0.5)$
		in Panel (b). 
		The estimates are presented by grey bars
		with the $95\%$ confidence intervals in black. 
		The horizontal axis shows the three subpopulations based on race and gender categories,
		while we fix
		12 years of education and 20 years of experience.

	\end{minipage}
\end{figure}

Figure \ref{fig:between} reports the policy effects on the changes in the between-inequality measure,
$\dot{\Delta}_{t}^{B}(u|z_{1}, z_{2})$,
for $t =1980$
and $u \in \{0.1, 0.3, 0.5, 0.7, 0.9\}$.
The baseline $z_{2}$ is fixed to include white, male workers
and $z_{1}$ changes over the three pairs of categorical attributes
as in Figure~\ref{fig:within}, 
while the remaining variables are the same in $z_{1}$ and $z_{2}$.\footnote{According to the identification result in Theorem \ref{theorem:identification},
$\dot{\Delta}_{t}^{B}(u|z_{1}, z_{2}) = (z_{2} - z_{1})'\delta(u)$, the common values in $z_{1}$ and  $z_{2}$ cancel each other and do not affect the conclusions.}
Panel (a) suggests significant negative impacts on the
between-inequality for black, male workers, compared to white, male workers, with magnitudes 5--20\% for all quantiles except the 0.7th conditional quantile. Panel (b) also shows reduction (0.05--0.20) in the between-inequality for white, female workers at the 0.7th conditional quantile and below.
Panels (c) plots results for female, black workers,
and the estimated changes in the between-inequality
range from -0.05 to -0.35 at the 0.7th conditional quantile and below.

\begin{figure}[!htb]
	\centering
	\caption{Changes in Between-Inequality in 1980 }  
	\label{fig:between}
	\begin{subfigure}{0.33\textwidth}
		\centering
		\caption{Black, Male Workers}
		\includegraphics[width=1\textwidth]{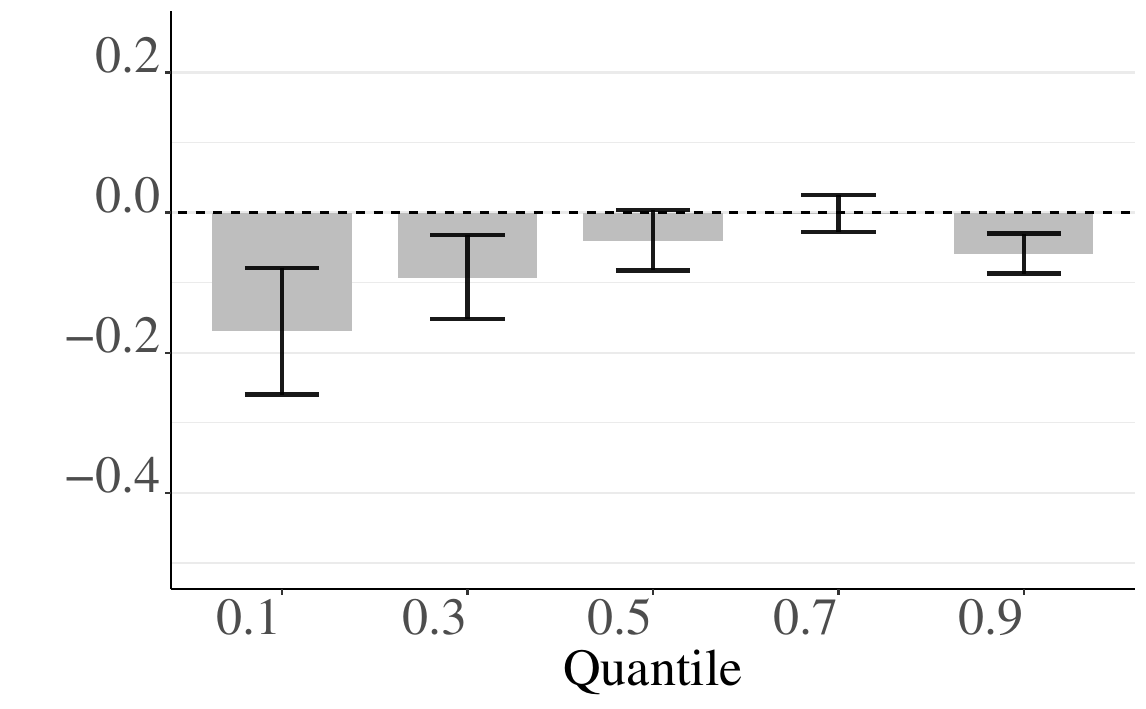} 
	\end{subfigure}%
	\begin{subfigure}{0.33\textwidth}
		\centering    
		\caption{White, Female Workers}
		\includegraphics[width=1\textwidth]{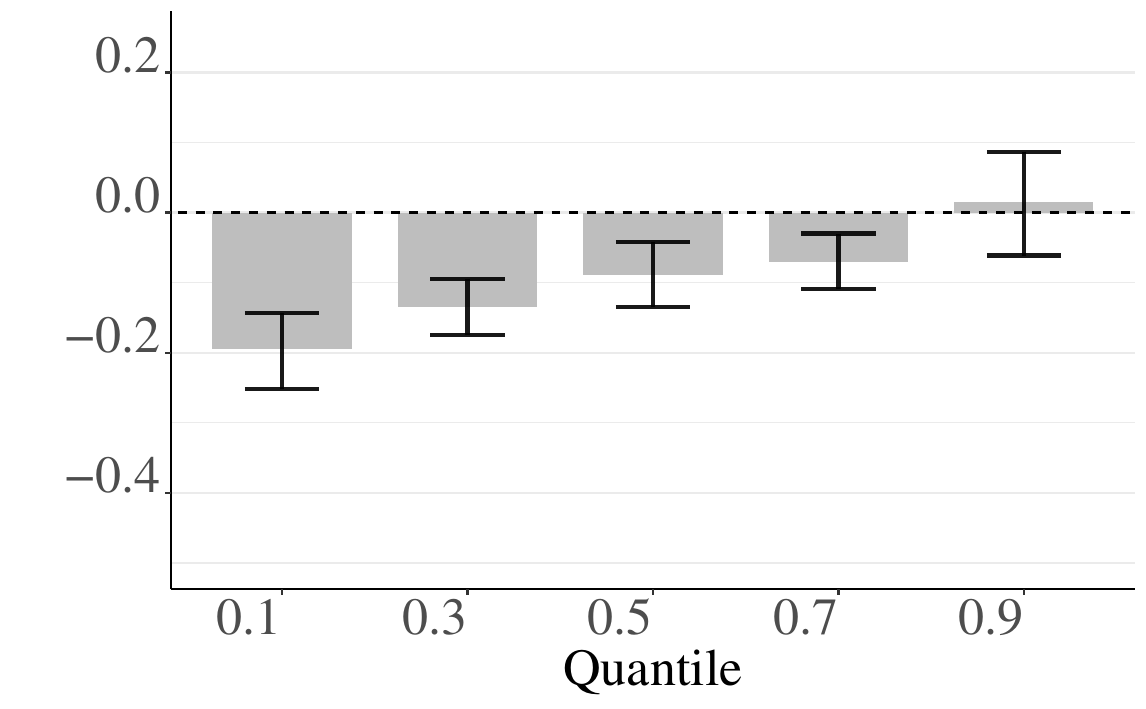} 
	\end{subfigure}%
	\begin{subfigure}{0.33\textwidth}
		\centering
		\caption{Black, Female Workers}
		\includegraphics[width=1\textwidth]{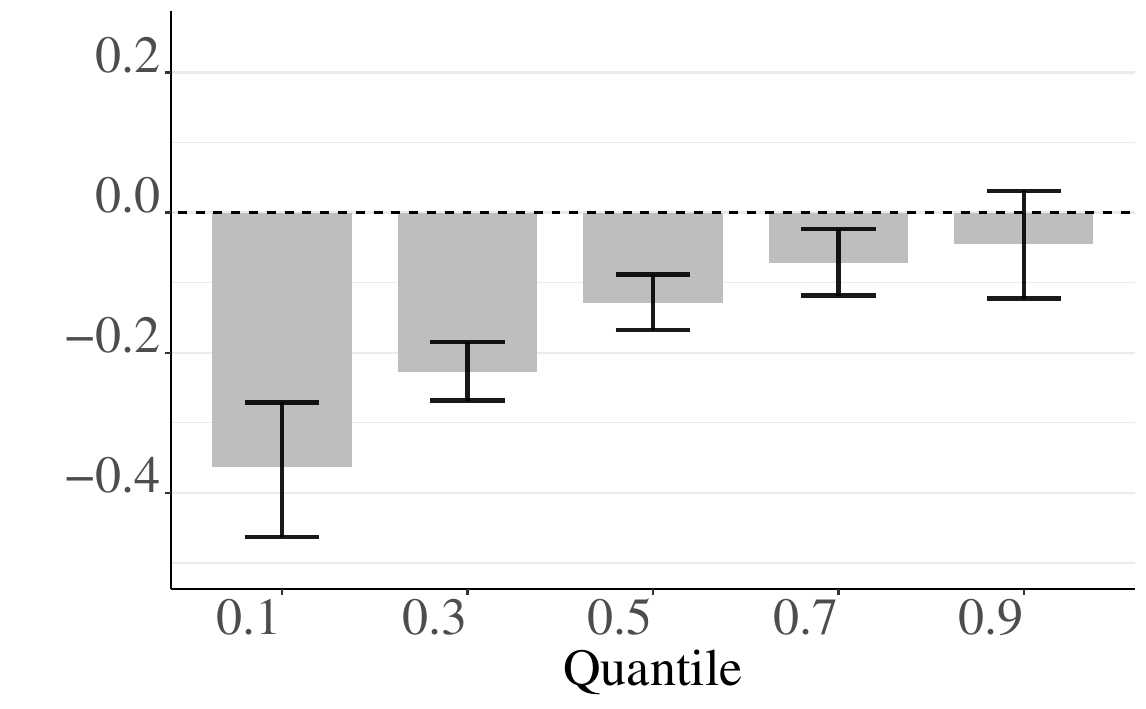} 
	\end{subfigure}
	\begin{minipage}{0.93\linewidth}
		\footnotesize
		\textit{Notes:} Panels (a)-(c) plot the changes in between-inequality ($\dt{\Delta}_{t}^{B}$) between multiple subpopulation groups and the base-level group: white male workers (with 12 years of education and 20 years of experience) while holding other individual covariates constant. The inequality measure are considered at quantiles $u = 0.1, 0.3, 0.5, 0.7, 0.9$ at year 1980. Point estimates are presented by grey bars with point-wise $95\%$ confidence intervals shown in black. 
	\end{minipage}
	
\end{figure}

To illustrate the robustness of the significant policy effects in reducing the between-equality, Figure~\ref{fig:between_alltime} plots the changes in between-inequality for black, female workers from 1967 to 1980. In quantiles up to medium, the magnitude of reduction in the between-inequality increases as time increases. In addition, such policy effects are significant after 1970. Similar patterns are also witnessed in the other two subpopulations presented in Figure~\ref{fig:between}. We choose not to report those plots for the concise of presentation.

\begin{figure}[!htb]
	\centering
	\caption{Changes in Between-Inequality for Black, Female Workers from 1967 to 1980}
	\label{fig:between_alltime}
	\begin{subfigure}{\textwidth}
		\centering
		\includegraphics[width=\textwidth]{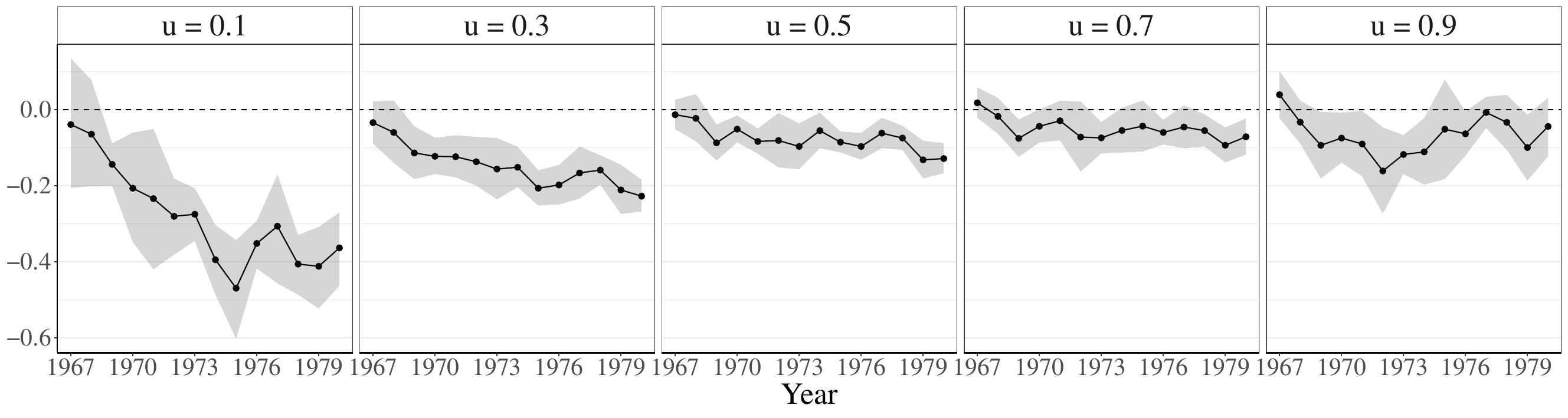} 
	\end{subfigure}
	\begin{minipage}{0.95\linewidth}
		\footnotesize
		\textit{Notes:} From left to right, figures correspond to the estimates at quantiles $u = 0.1, 0.3, 0.5, 0.7, 0.9$. Point estimates are plotted with solid black lines, and the pointwise $95\%$ confidence intervals are shown as grey shaded area.
	\end{minipage}
\end{figure}

Overall, the results above confirm the findings of \cite{derenoncourt2021minimum} that the reform was effective in improving the black economic status and reducing the racial income gap. In addition, we provide empirical evidence of a compounded impact of the policy effect in reducing the racial and gender income gap, which leads to a significant reduction in the between-inequality, at least up to the medium.


\section{Conclusion}\label{sec:conclusion}

In this paper, we introduce an estimation method for evaluating 
the effect of group-level policies
under the quantile random-coefficient regression framework with interactive fixed effects.
Our method can capture the heterogeneous policy effects through the interaction of policy variables and the individual observed and unobserved characteristics, while controlling the unobserved interactive fixed effects, and provides a straightforward way of identifying the policy effect on inequality measures. 
The consistency and limiting distribution of the proposed estimators are established.
Using our proposed model, we evaluate the effect of the minimum wage policy on earnings between 1967 and 1980 in the United States. Our analysis confirms the findings of \cite{derenoncourt2021minimum} that the policy helps reduce the racial income gap by improving the black economic status. On top of that, we provide empirical evidence of a compounded policy effect in narrowing the racial and gender gap, which contributes to the significant reduction in the between-inequality.

\section*{Acknowledgments}
This paper benefited greatly from our discussions with Dukpa Kim.
We are also grateful for comments
from 
Martin Huber,
Rustam Ibragimov,
Artem Prokhorov,
and participants 
at 
the Center for Econometrics and Business Analytics (CEBA) talk
and seminars.
We would also like to thank Ellora Derenoncourt and Claire Montialoux  for sharing the data and code for the empirical application.
Gao and Oka gratefully acknowledge financial support from the Australian Research Council Discovery Programs Scheme under Grant Numbers: DP200102769 and DP190101152, respectively, and Whang thanks financial support from the Korea Bureau of Economic Research and Innovation at the Institute of Economic Research in Seoul National University under Grant 0405-20220046. All errors are our own.

\section*{Declaration of Interest Statement}

The authors would like to declare that there is no conflict of interest.

\newpage
\appendix
\section*{Online Supplementary  Materials}

\setcounter{page}{1}
\setcounter{table}{1}
\setcounter{figure}{1}
\setcounter{assumption}{1}
\setcounter{equation}{1}
\setcounter{proposition}{1}

\renewcommand{\thepage}{S-\arabic{page}}
\renewcommand{\thesection}{S.\arabic{section}}
\renewcommand{\theassumption}{S.\arabic{assumption}} 
\renewcommand{\theproposition}{S.\arabic{proposition}} 
\numberwithin{proposition}{section} 
\renewcommand{\thelemma}{S.\arabic{lemma}}
\numberwithin{lemma}{section} 
\renewcommand{\thetheorem}{S.\arabic{theorem}}
\renewcommand{\theequation}{S.\arabic{equation}} 
\renewcommand{\thetable}{S.\arabic{table}} 
\renewcommand{\thefigure}{S.\arabic{figure}}

		This material consists of five sections. Section~\ref{sec:sim} examines the finite sample properties of the estimator
		through Monte Carlo simulation. 
		Section~\ref{sec:additional_empirical_cpt2} provides additional empirical analysis on the 1966 Fair Labour Standard Act. Section~\ref{sec:tech_assumptions} covers the technical assumptions required for the asymptotic theories.
		Section~\ref{sec:proof} establishes the proofs of the theoretical results given in the main text. Some theoretical results required in Section~\ref{sec:proof} are then provided in Section~\ref{online.theorem}, while Section~\ref{B.lemma} presents the preliminary lemmas and the proofs that are required in Section~\ref{online.theorem}.

		\section{Simulation Study}\label{sec:sim}
		
		In this section, we investigate the accuracy of both the point estimators of the policy parameter $\delta_{jt}(u)$ and the treatment parameters $\Delta_{t}^{AQTT}(u|z)$, $\dt{\Delta}_{t}^{B}(u| z_{1}, z_{2})$ and $\dt{\Delta}_{t}^{W}(u_{1}, u_{2}| z)$ through Monte Carlo simulations. We consider two data-generating processes (DGP). DGP I is a simplified scenario of our model, similar to the settings of \cite{chetverikov2016iv}. Using this setting, we illustrate the fast convergence of our proposed estimation method by documenting the estimation bias after given iteration steps. 
		In addition, we consider a location-scale model as the DGP II, in which case the expressions for the treatment parameters can be derived. We investigate the estimation accuracy of the treatment parameters in DGP II. For both DGPs, we separately consider two scenarios, where the unobserved factors are either correlated or uncorrelated with the group-level regressors. 
		
		
		\subsection{Data Generating Process I}\label{sec:dgp1}
		In this section, we investigate the convergence rate of our proposed two-step iterative estimation approach.
		We generate data according to the following simplified model, similar to \cite{chetverikov2016iv}. That is, for $i = 1,...,N,\ s=1,...,S$, and $t=1,...,T$,
		\begin{align*}
			& y_{ist} = \alpha_{st0}(u_{ist}) +
			z_{ist}\alpha_{st1}(u_{ist}) + \Phi^{-1}(u_{ist}) ,
			\\
			& \alpha_{st0}(u) =  \delta_0(u) + d_{st} \delta_t(u) + x_{st} \beta(u) +  f_{t}(u)'\lambda_{s}(u)+
			\eta_{st}(u), 
			\quad \alpha_{st1}(u) = 2+0.3u,\\
			& 
			\delta_0(u) = 2 + \frac{u^2}{4},
			\quad 
			\delta_t(u) = 2 + \frac{t}{2T}+ \frac{u^2}{4} ,
			\quad
			\beta(u) = 3+ \frac{u^2}{4},
			\quad 
			\eta_{st}(u) = (\xi_{st} - 0.5)u,
		\end{align*}
		where $\Phi(\cdot)$ is the cdf of $N(0,1)$,$\{u_{ist}\}$ and $\{\xi_{st}\}$ are i.i.d $\text{U}(0,1)$, $\{z_{ist}\}$ are i.i.d $\text{U}(0,1)$, $x_{st}$ are i.i.d $N(0,1)$, $(f_{t1},f_{t2})$ are generated orthogonally via the SVD of a $T \times T$ random matrix whose entries are i.i.d $N(0,1)$, and $(\lambda_{s1},\lambda_{s2})$ are generated from i.i.d ${\rm U}(0,2)$. The policy dummy variable $d_{st} = 1\{t \geq T/4\} \times 1\{s \geq S/4\} $, that is, we fixed the first one quarter of the cross-sections as the control groups, and the policy is implemented at $T_0 = T/4$.
		We consider the following two scenarios in terms of whether exists endogeneity in the group-level:
		\begin{enumerate}
			\item group-level observables $x_{st}$ are i.i.d $N(0,1)$.
			\item group-level observables $x_{st} =\zeta_{st} + 0.02 f_{t1}^2 + 0.02 \lambda_{s1}^2$, where $\{\zeta_{st}\}$ are i.i.d $N(0,1)$. This setting allows moderate endogeneity at group-level.
		\end{enumerate}

		We investigate the accuracy of the point estimators of the policy parameters $\delta_{T}(u)$ for both scenarios. The sample bias and standard deviation at finite iteration steps $m=2,5$ and after convergence criterion satisfied are reported in Table~\ref{table1}. 
		The simulation shows a few nice properties, which can be summarized as follows. First of all, the recursive estimator converges quite fast. In all cases, the mean bias and standard deviation of the coefficient estimation at the second iteration are reasonably small. In addition, the estimators remain valid even if the number of observations per group ($N_{st}$) is relatively small compared to the total number of groups and time ($S \times T$), which is a particularly attractive property in practice. 
		
		\begin{table}[!htbp]
			\centering
			\caption{
				Sample Bias and Standard Deviation of estimates of the Policy Parameters \label{table1}}
			\tabcolsep 0.10in
			\resizebox{\textwidth}{!}{%
				\begin{tabular}{@{} l l l l r  r r   r  r r  @{}} 
					\toprule
					&& & & \multicolumn{3}{c}{Scenario 1} & \multicolumn{3}{c}{Scenario 2}  \\
					\cmidrule(r){5-7}   \cmidrule{8-10}
					$N$ & $S$ & $T$ & $u$ & {$\widehat{\delta}^{(2)}_T(u)$}  & {$\widehat{\delta}^{(5)}_T(u)$} & {$\widehat{\delta}_T(u)$}  & {$\widehat{\delta}^{(2)}_T(u)$}  & {$\widehat{\delta}^{(5)}_T(u)$} & {$\widehat{\delta}_T(u)$} \\
					\midrule                
					1000	&	20	&	20	&	0.1	&$	-0.001	\ (	0.059	)$&$	-0.001	\ (	0.057	)$&$	-0.001	\ (	0.059	)$&$	-0.001	\ (	0.060	)$&$	-0.001	\ (	0.057	)$&$	-0.001	\ (	0.059	)$\\				
					&		&		&	0.5	&$	0.001	\ (	0.103	)$&$	0.001	\ (	0.086	)$&$	0.001	\ (	0.094	)$&$	0.001	\ (	0.103	)$&$	0.001	\ (	0.086	)$&$	0.001	\ (	0.095	)$\\				
					&		&		&	0.9	&$	-0.003	\ (	0.159	)$&$	-0.003	\ (	0.145	)$&$	-0.003	\ (	0.150	)$&$	-0.003	\ (	0.160	)$&$	-0.003	\ (	0.145	)$&$	-0.003	\ (	0.150	)$\\				
					&	40	&	40	&	0.1	&$	0.001	\ (	0.038	)$&$	0.001	\ (	0.035	)$&$	0.001	\ (	0.035	)$&$	0.001	\ (	0.038	)$&$	0.001	\ (	0.035	)$&$	0.001	\ (	0.035	)$\\				
					&		&		&	0.5	&$	0.001	\ (	0.078	)$&$	0.001	\ (	0.063	)$&$	0.001	\ (	0.055	)$&$	0.001	\ (	0.079	)$&$	0.001	\ (	0.063	)$&$	0.001	\ (	0.055	)$\\				
					&		&		&	0.9	&$	0.001	\ (	0.112	)$&$	0.001	\ (	0.096	)$&$	0.001	\ (	0.091	)$&$	0.001	\ (	0.112	)$&$	0.001	\ (	0.096	)$&$	0.001	\ (	0.091	)$\\				
					&	60	&	60	&	0.1	&$	-0.001	\ (	0.030	)$&$	-0.001	\ (	0.028	)$&$	-0.001	\ (	0.028	)$&$	-0.001	\ (	0.030	)$&$	-0.001	\ (	0.028	)$&$	-0.001	\ (	0.028	)$\\				
					&		&		&	0.5	&$	0.001	\ (	0.075	)$&$	0.001	\ (	0.055	)$&$	0.001	\ (	0.043	)$&$	0.001	\ (	0.075	)$&$	0.001	\ (	0.055	)$&$	0.001	\ (	0.043	)$\\				
					&		&		&	0.9	&$	0.001	\ (	0.093	)$&$	0.001	\ (	0.075	)$&$	0.001	\ (	0.071	)$&$	0.001	\ (	0.093	)$&$	0.001	\ (	0.075	)$&$	0.001	\ (	0.071	)$\\				
					\addlinespace[0.2cm]  																																			
					2000	&	20	&	20	&	0.1	&$	0.001	\ (	0.050	)$&$	0.001	\ (	0.045	)$&$	0.001	\ (	0.047	)$&$	0.001	\ (	0.050	)$&$	0.001	\ (	0.045	)$&$	0.001	\ (	0.047	)$\\				
					&		&		&	0.5	&$	-0.001	\ (	0.100	)$&$	0.001	\ (	0.082	)$&$	0.001	\ (	0.082	)$&$	-0.001	\ (	0.100	)$&$	0.001	\ (	0.082	)$&$	0.001	\ (	0.083	)$\\				
					&		&		&	0.9	&$	0.004	\ (	0.160	)$&$	0.004	\ (	0.144	)$&$	0.004	\ (	0.147	)$&$	0.004	\ (	0.160	)$&$	0.004	\ (	0.144	)$&$	0.004	\ (	0.147	)$\\				
					&	40	&	40	&	0.1	&$	0.001	\ (	0.036	)$&$	0.001	\ (	0.028	)$&$	0.001	\ (	0.026	)$&$	0.001	\ (	0.036	)$&$	0.001	\ (	0.028	)$&$	0.001	\ (	0.026	)$\\				
					&		&		&	0.5	&$	-0.001	\ (	0.072	)$&$	0.001	\ (	0.050	)$&$	0.001	\ (	0.049	)$&$	-0.001	\ (	0.072	)$&$	0.001	\ (	0.050	)$&$	0.001	\ (	0.049	)$\\				
					&		&		&	0.9	&$	0.001	\ (	0.121	)$&$	0.001	\ (	0.099	)$&$	0.001	\ (	0.085	)$&$	0.001	\ (	0.121	)$&$	0.001	\ (	0.099	)$&$	0.001	\ (	0.085	)$\\				
					&	60	&	60	&	0.1	&$	0.001	\ (	0.027	)$&$	0.001	\ (	0.022	)$&$	0.001	\ (	0.021	)$&$	0.001	\ (	0.027	)$&$	0.001	\ (	0.022	)$&$	0.001	\ (	0.021	)$\\				
					&		&		&	0.5	&$	0.001	\ (	0.066	)$&$	0.001	\ (	0.046	)$&$	0.001	\ (	0.039	)$&$	0.001	\ (	0.066	)$&$	0.001	\ (	0.047	)$&$	0.001	\ (	0.039	)$\\				
					&		&		&	0.9	&$	-0.001	\ (	0.099	)$&$	0.001	\ (	0.071	)$&$	0.001	\ (	0.068	)$&$	-0.001	\ (	0.099	)$&$	0.001	\ (	0.071	)$&$	0.001	\ (	0.068	)$\\				
					\bottomrule
			\end{tabular}}
			\begin{minipage}{1.0\linewidth}
				\footnotesize
				\textit{Notes:}
				The number of Monte-Carlo repetitions is 250. 
				We report bias averaged over 250 repetitions with
				the standard deviation in the parenthesis at finite iteration steps $m=2,5$ and after convergence. 
			\end{minipage}
			
		\end{table}

		
		\subsection{Data Generating Process II}\label{sec:dgp2}
		In this section, we focus on the estimation accuracy of the policy parameters and treatment parameters after convergence.
		We consider a location-scale model, whose conditional quantile is easy to derive. The setting of the factor structure is similar to that of \cite{chen2021quantile}. Specifically, we generate
		\begin{align*}
			& y_{ist} = \sum_{j=1}^{3}z_{ist(j)}\alpha_{1(j),st} + \sum_{j=1}^{3}z_{ist(j)}\alpha_{2(j),st}\epsilon_{ist},\\
			& \alpha_{1(j),st} 
			= 3 + d_{st}\delta_{1(j),t} + x_{st}'\beta_{1(j)}
			+ f_{1(j),t}'\lambda_{2(j),s} +  \eta_{1(j),st}\\
			&\alpha_{2(j),st} 
			= 0.5 + d_{st}\delta_{2(j),t} + x_{st}'\beta_{2(j)} 
			+ f_{2(j),t}'\lambda_{2(j),s}
			+ \eta_{2(j),st},
			\quad j=1,2,3,
		\end{align*}
		where 
		\begin{align*}
			& z_{ist(1)} \equiv 1,
			\quad 
			z_{ist(2)} \overset{iid}{\sim} Bern(0.6),
			\quad
			z_{ist(3)} \overset{iid}{\sim} N(0,1) \text{ truncated between }[0,3], \\
			& \epsilon_{ist} \overset{iid}{\sim} N(0,1),
			\quad
			\eta_{1(j),st} \overset{iid}{\sim} N(0,2),
			\quad
			\eta_{2(j),st} \overset{iid}{\sim} N(0,0.5) \text{ truncated between }[-0.1,0.1],\\
			& 
			d_{st} = 1\{s \geq S/3\}\times 1\{t \geq T/3\},
			\\
			& \beta_{1(j)} = 5+j/3,
			\quad \beta_{2(j)} =0.1,
			\quad
			\delta_{1(j),t} = 5+t/T,
			\quad
			\delta_{2(j),t} = 0.1, \\
			& 
			f_{1(1),t}  = 0.5 f_{1(1),t-1} + \xi_{1,t},
			\quad
			f_{1(2),t} = 0.75 f_{1(2),t-1} + \xi_{2,t}, 
			\quad
			\xi_{1,t}, \xi_{2,t} \overset{iid}\sim N(0,1), \\
			&
			f_{1(3),t} \overset{iid}{\sim}N(0,1), 
			\quad
			f_{2(j),t} \overset{iid}{\sim} |N(0,0.5)|,
			\quad
			\lambda_{1(j),s} \overset{iid}{\sim} N(0,1),
			\quad
			\lambda_{2(j),s} \overset{iid}{\sim} U(0,0.5).
		\end{align*}
		Again, we consider the following two scenarios in terms of whether exists endogeneity in the group-level:
		\begin{enumerate}
			\item group-level observables $x_{st}$ are i.i.d $\text{exp}(0.1  N(0,1))$.
			\item group-level observables $x_{st} =\zeta_{st} + 0.1 f_{t2}^2 + 0.2 \lambda_{s2}^2$, where $\{\zeta_{st}\}$ are i.i.d $\exp\big(0.25N(0,1)\big)$. This setting allows moderate endogeneity at group-level.
		\end{enumerate}
		
		For both scenarios, it is easy to check that $z_{ist}\alpha_{2(j),st} >0$ for all $j=1,2,3$. Thus, it is straightforward to obtain the conditional quantile function of $y_{ist}$ as
		\begin{align*}
			Q_{y_{ist}} (u &|d_{st},x_{it},z_t,f_t(u),\lambda_i(u)) = \sum_{j=1}^{3}z_{ist(j)}\alpha_{jst}(u), \\
			\alpha_{jst}(u)
			& = \alpha_{1(j),st} + \alpha_{2(j),st}Q_{\epsilon}(u) \\
			& = \delta_{j0}(u) 
			+ d_{st}\delta_{jt}(u)
			+ x_{st}'\beta_{j}(u)
			+ f_{jt}(u)'\lambda_{js}(u)
			+ \eta_{jst}(u),\footnotemark
		\end{align*}
		\footnotetext{The DGP is constructed such that the assumption 
			$\mathbb{E}[\eta_{jst}(u) | x_{st},d_{st},f_{jt},\lambda_{js}] = 0$ is satisfied.}
		where 
		\begin{align*}
			& \delta_{j0}(u) = 3 + 0.5Q_{\epsilon}(u),
			\quad
			\delta_{jt}(u) = \delta_{1(j),t} + \delta_{2(j), t}Q_{\epsilon}(u),
			\quad
			\beta_{j}(u) = \beta_{1(j)} + \beta_{2(j)}Q_{\epsilon}(u),\\
			& f_{jt}(u) = [f_{1(j),t},f_{2(j),t}]',
			\quad
			\lambda_{js}(u) = [\lambda_{1(j),s},\lambda_{2(j),s}Q_{\epsilon}(u)]', \\
			& \eta_{jst}(u) = \eta_{1(j),st}+\eta_{2(j),st}Q_{\epsilon}(u).
		\end{align*}

		We investigate the accuracy of the point estimators of the policy parameter $\delta_{jt}(u)$ and the treatment parameters $\Delta_{t}^{AQTT}(u|z)$, $\dt{\Delta}_{t}^{B}(u| z_{1}, z_{2})$ and $\dt{\Delta}_{t}^{W}(u_{1}, u_{2}| z)$ for both scenarios. The sample bias and standard deviation at time $t=T$ are reported in Tables~\ref{sim:table1}-\ref{sim:table4}. 
		
		Table~\ref{sim:table1} reports the estimation accuracy of $\delta_{jt}(u)$. The simulation shows a few nice properties similar to those found in Section~\ref{sec:dgp1}. First of all, the recursive estimator converges quite fast at both tail and central quantiles. In all cases, the mean bias and standard deviation of the coefficient estimation at the second iteration are reasonably small. An additional appealing feature is that the estimators remain accurate even if the number of observations per group ($N_{st} \equiv N$) is relatively small compared to the total number of groups and time ($S \times T$). Last but not least, we conclude that the moderate group-level endogeneity, captured by the common factor structure, does not diminish the precision of the estimators. This feature is aligned with our asymptotic analysis.
		
		\begin{table}[!htbp]
			\centering
			\caption{Point accuracy of $\widehat{\delta}_{jT}(u)$\label{sim:table1}}
			\tabcolsep 0.12in
			\resizebox{\textwidth}{!}{%
				\begin{tabular}{@{}l l l r r r r r r r   @{}} 
					\toprule
					& & & & & Scenario 1 & & & Scenario 2 & \\
					\cmidrule(r){5-7}   \cmidrule{8-10}
					N & S & T & $u$ & $j=1$ & $j=2$ & $j=3$ & $j=1$ & $j=2$ & $j=3$ \\
					\midrule
					1000 & 20 &20  &$  0.1 $&$ 0.046\ ( 0.769 )$&$  -0.048\  ( 0.882 )$&$  0.038\ ( 0.787 )$&$  0.042\ ( 0.766 )$&$  -0.033\  ( 0.872 )$&$  0.041\ ( 0.789 )$\\
					&&&$  0.5 $&$ 0.032\ ( 0.900 )$&$  -0.038\  ( 0.869 )$&$  0.063\ ( 0.780 )$&$  0.029\ ( 0.897 )$&$  -0.025\  ( 0.860 )$&$  0.066\ ( 0.782 )$\\
					&&&$  0.9 $&$ 0.028\ ( 1.149 )$&$  -0.014\  ( 0.898 )$&$  0.065\ ( 0.821 )$&$  0.022\ ( 1.149 )$&$  -0.001 \ ( 0.890 )$&$  0.069\ ( 0.824 )$\\
					& 40 &40  &$  0.1 $&$ -0.051\  ( 0.402 )$&$  0.096\ ( 0.534 )$&$  -0.017\  ( 0.436 )$&$  -0.052\  ( 0.402 )$&$  0.096\ ( 0.534 )$&$  -0.018\  ( 0.436 )$\\
					&&&$  0.5 $&$ -0.051\  ( 0.394 )$&$  0.081\ ( 0.472 )$&$  -0.029\  ( 0.425 )$&$  -0.051\  ( 0.394 )$&$  0.081\ ( 0.472 )$&$  -0.030\  ( 0.424 )$\\
					&&&$  0.9 $&$ -0.052\  ( 0.412 )$&$  0.079\ ( 0.488 )$&$  -0.043\  ( 0.437 )$&$  -0.052 \ ( 0.412 )$&$  0.080\ ( 0.490 )$&$  -0.043 \ ( 0.436 )$\\
					& 60 &60  &$  0.1 $&$ -0.005 \ ( 0.331 )$&$  0.013\ ( 0.381 )$&$  -0.014\  ( 0.329 )$&$  -0.006\  ( 0.331 )$&$  0.013\ ( 0.381 )$&$  -0.013\   ( 0.328 )$\\
					&& &$  0.5 $&$ -0.007\  ( 0.328 )$&$  0.012\ ( 0.364 )$&$  -0.020\  ( 0.306 )$&$  -0.008 \ ( 0.328 )$&$  0.012\ ( 0.364 )$&$  -0.019 \ ( 0.306 )$\\
					&&&$  0.9 $&$ -0.010\  ( 0.356 )$&$  0.014\ ( 0.379 )$&$  -0.027\  ( 0.316 )$&$  -0.010\  ( 0.355 )$&$  0.015\ ( 0.379 )$&$  -0.027\  ( 0.316 )$\\
					\addlinespace[0.2cm]  
					2000 & 20 &20  &$  0.1 $&$ -0.072\  ( 0.846 )$&$  -0.032\  ( 0.797 )$&$  0.144\ ( 1.340 )$&$  -0.080\  ( 0.824 )$&$  -0.041\  ( 0.786 )$&$  0.027\ ( 0.818 )$\\
					&&&$  0.5 $&$ -0.059\  ( 0.831 )$&$  -0.003\  ( 0.770 )$&$  0.055\ ( 0.677 )$&$  -0.070\  ( 0.797 )$&$  -0.012\  ( 0.760 )$&$  0.018\ ( 0.818 )$\\
					&&&$  0.9 $&$ -0.044\  ( 1.015 )$&$  0.022\ ( 0.776 )$&$  0.034\ ( 0.713 )$&$  -0.046\  ( 0.911 )$&$  0.013\ ( 0.765 )$&$  0.030\ ( 0.707 )$\\
					& 40 &40 &$  0.1 $&$ -0.072\  ( 0.424 )$&$  -0.037\  ( 0.438 )$&$  0.033\ ( 0.434 )$&$  -0.072\  ( 0.424 )$&$  -0.037\  ( 0.439 )$&$  0.033\ ( 0.435 )$\\
					&&&$  0.5 $&$ -0.066 \ ( 0.415 )$&$  -0.034\  ( 0.425 )$&$  0.022\ ( 0.422 )$&$  -0.066\  ( 0.415 )$&$  -0.035\  ( 0.426 )$&$  0.022\ ( 0.423 )$\\
					&&&$  0.9 $&$ -0.062\  ( 0.432 )$&$  -0.032\  ( 0.443 )$&$  0.011\ ( 0.434 )$&$  -0.062\  ( 0.431 )$&$  -0.032\  ( 0.444 )$&$  0.011\ ( 0.435 )$\\
					& 60 &60 &$  0.1 $&$ -0.032\  ( 0.346 )$&$  0.007\ ( 0.345 )$&$  -0.022\  ( 0.348 )$&$  -0.031\  ( 0.345 )$&$  0.006\ ( 0.345 )$&$  -0.022\  ( 0.349 )$\\
					&&&$  0.5 $&$ -0.024\  ( 0.322 )$&$  0.014\ ( 0.327 )$&$  -0.032\  ( 0.332 )$&$  -0.024\  ( 0.323 )$&$  0.013\ ( 0.328 )$&$  -0.032\  ( 0.333 )$\\
					&&&$  0.9 $&$ -0.019\  ( 0.336 )$&$  0.019\ ( 0.344 )$&$  -0.041 \ ( 0.347 )$&$  -0.019\  ( 0.335 )$&$  0.019\ ( 0.345 )$&$  -0.041\  ( 0.348 )$\\
					\bottomrule
			\end{tabular}}
			\begin{minipage}{1.0\linewidth}
				\footnotesize
				\textit{Notes:}
				The number of Monte-Carlo repetitions is 250. 
				We report bias averaged over 250 repetitions with
				the standard deviation in the parenthesis. 
			\end{minipage}
		\end{table}
		
		Furthermore, to illustrate the estimation accuracy of our proposed treatment parameters. Tables~\ref{sim:table2}-\ref{sim:table4} report the estimation accuracy of $\Delta_{t}^{AQTT}(u|z)$, $\dt{\Delta}_{t}^{B}(u| z_{1}, z_{2})$ and $\dt{\Delta}_{t}^{W}(u_{1}, u_{2}| z)$, respectively. Throughout the comparison, we fix $z := (1,1,1)'$, $z_1 := (1,1,0)'$ and $z_2 := (1,0,1)'$ to allow sufficient variation in terms of the individual characteristics. It is not surprising that the estimation results are relatively similar to those for $\delta_{jt}(u)$, as the treatment parameters are identified as linear combinations of $\delta_{jt}(u)$ according to Theorem~\ref{theorem:identification}.

		\begin{table}[!htbp]
			\centering
			\caption{Point accuracy of $\widehat{\Delta}_{T}^{AQTT}(u|z)$\label{sim:table2}}
			\tabcolsep 0.20in
			\resizebox{0.75\textwidth}{!}{%
				\begin{tabular}{@{}l l l r r r    @{}} 
					\toprule
					N & S & T & $u$ & Scenario 1 & Scenario 2 \\
					\midrule
					1000 & 20 & 20&$  0.1 $&$ 0.036\ ( 1.409 )$&$  0.050\ ( 1.399 )$\\
					&&&$  0.5 $&$ 0.056\ ( 1.497 )$&$  0.070\ ( 1.488 )$\\
					&&&$  0.9 $&$ 0.079\ ( 1.719 )$&$  0.090\ ( 1.718 )$\\
					& 40 & 40 &$  0.1 $&$ 0.027\ ( 0.817 )$&$  0.027\ ( 0.816 )$\\
					&&&$  0.5 $&$ 0.001\ ( 0.760 )$&$  0.001\ ( 0.760 )$\\
					&&&$  0.9 $&$ -0.016\  ( 0.783 )$&$  -0.016 \ ( 0.783 )$\\
					& 60 & 60 &$  0.1 $&$ -0.006\ ( 0.606 )$&$  -0.005\ ( 0.606 )$\\
					&&&$  0.5 $&$ -0.015\ ( 0.584 )$&$  -0.015\ ( 0.585 )$\\
					&&&$  0.9 $&$ -0.023\  ( 0.617 )$&$  -0.022\  ( 0.617 )$\\
					\addlinespace[0.2cm]  
					2000 & 20  & 20  &$  0.1 $&$ 0.040\ ( 1.731 )$&$  -0.094 \ ( 1.271 )$\\
					&&&$  0.5 $&$ -0.006\  ( 1.218 )$&$  -0.065\  ( 1.223 )$\\
					&&&$  0.9 $&$ 0.013\ ( 1.401 )$&$  -0.002\  ( 1.286 )$\\
					& 40  & 40  &$  0.1 $&$ -0.077\  ( 0.744 )$&$  -0.077\  ( 0.745 )$\\
					&&&$  0.5 $&$ -0.079\  ( 0.723 )$&$  -0.078\  ( 0.724 )$\\
					&&&$  0.9 $&$ -0.084\  ( 0.748 )$&$  -0.084\  ( 0.750 )$\\
					& 60 & 60 &$  0.1 $&$ -0.047\  ( 0.608 )$&$  -0.047\  ( 0.609 )$\\
					&&&$  0.5 $&$ -0.042\  ( 0.559 )$&$  -0.042\  ( 0.560 )$\\
					&&&$  0.9 $&$ -0.041\  ( 0.567 )$&$  -0.041 \ ( 0.568 )$\\	
					\bottomrule
			\end{tabular}}
			\begin{minipage}{0.75\linewidth}
				\footnotesize
				\textit{Notes:}
				The number of Monte-Carlo repetitions is 250. 
				We report bias averaged over 250 repetitions with
				the standard deviation in the parenthesis. We consider individual covariate $z := (1,1,1)'$.
			\end{minipage}
		\end{table}

		\begin{table}[!htbp]
			\centering
			\caption{Point accuracy of $\widehat{{\Delta}}_T^B(u|z_1,z_2)$ \label{sim:table3}}
			\tabcolsep 0.20in
			\resizebox{0.75\textwidth}{!}{%
				\begin{tabular}{@{}l l l r r r    @{}} 
					\toprule
					N & S & T & $u$ & Scenario 1 & Scenario 2 \\
					\midrule
					1000 & 20 & 20 &$  0.1 $&$ 0.086\ ( 1.122 )$&$  0.073\ ( 1.115 )$\\
					&&&$  0.5 $&$ 0.101\ ( 1.110 )$&$  0.091\ ( 1.105 )$\\
					&&&$  0.9 $&$ 0.079\ ( 1.163 )$&$  0.070\ ( 1.160 )$\\
					&  40  & 40  &$  0.1 $&$ -0.113\  ( 0.679 )$&$  -0.114\  ( 0.679 )$\\
					&&&$  0.5 $&$ -0.110\  ( 0.623 )$&$  -0.111\  ( 0.623 )$\\
					&&&$  0.9 $&$ -0.122\  ( 0.640 )$&$  -0.123\  ( 0.641 )$\\
					& 60  & 60  & $  0.1 $&$ -0.027\  ( 0.501 )$&$  -0.026\  ( 0.501 )$\\
					&&&$  0.5 $&$ -0.031\  ( 0.483 )$&$  -0.031\  ( 0.483 )$\\
					&&&$  0.9 $&$ -0.041 \ ( 0.506)$&$  -0.041\  ( 0.505 )$\\
					\addlinespace[0.2cm]  
					2000 & 20  &20  &$  0.1 $&$ 0.176\ ( 1.556 )$&$  0.068\ ( 1.197 )$\\
					&&&$  0.5 $&$ 0.058\ ( 1.107 )$&$  0.030\ ( 1.200 )$\\
					&&&$  0.9 $&$ 0.012\ ( 1.137 )$&$  0.017\ ( 1.130 )$\\
					& 40  & 40  &$  0.1 $&$ 0.070\ ( 0.625 )$&$  0.070\ ( 0.627 )$\\
					&&&$  0.5 $&$ 0.056\ ( 0.607 )$&$  0.057\ ( 0.608 )$\\
					&&&$  0.9 $&$ 0.043\ ( 0.626 )$&$  0.043\ ( 0.626 )$\\
					& 60  & 60 &$  0.1 $&$ -0.029\  ( 0.490 )$&$  -0.028\  ( 0.490 )$\\
					&&&$  0.5 $&$ -0.045\  ( 0.474 )$&$  -0.045\  ( 0.474 )$\\
					&&&$  0.9 $&$ -0.060\  ( 0.504 )$&$  -0.060\  ( 0.504 )$\\
					
					\bottomrule
			\end{tabular}}
			\begin{minipage}{0.75\linewidth}
				\footnotesize
				\textit{Notes:}
				The number of Monte-Carlo repetitions is 250. 
				We report bias averaged over 250 repetitions with
				the standard deviation in the parenthesis. We consider individual covariates $z_1 := (1,1,0)'$ and $z_2 := (1,0,1)'$.
			\end{minipage}
		\end{table}

		\begin{table}[!htbp]
			\centering
			\caption{Point accuracy of $\widehat{{\Delta}}_T^W(u_1,u_2|z)$\label{sim:table4}}
			\tabcolsep 0.20in
			\resizebox{0.75\textwidth}{!}{%
				\begin{tabular}{@{}l l l r r r    @{}} 
					\toprule
					N & S & T & $(u_1,u_2)$ & Scenario 1 & Scenario 2 \\
					\midrule
					1000  & 20  & 20  &$( 0.1 , 0.5 )$&$  0.021\ ( 0.444 )$&$  0.020\ ( 0.444 )$\\
					&&&$( 0.1 , 0.9 )$&$  0.043\ ( 0.894 )$&$  0.040\ ( 0.898 )$\\
					&&&$( 0.5 , 0.9 )$&$  0.022\ ( 0.497 )$&$  0.020\ ( 0.499 )$\\
					& 40  & 40  &$( 0.1 , 0.5 )$&$  -0.027\  ( 0.209 )$&$  -0.026\  ( 0.207 )$\\
					&&&$( 0.1 , 0.9 )$&$  -0.043\  ( 0.372 )$&$  -0.042\  ( 0.371 )$\\
					&&&$( 0.5 , 0.9 )$&$  -0.016\  ( 0.184 )$&$  -0.016\  ( 0.184 )$\\
					& 60  & 60  &$( 0.1 , 0.5 )$&$  -0.009\  ( 0.190 )$&$  -0.009\  ( 0.189 )$\\
					&&&$( 0.1 , 0.9 )$&$  -0.017\  ( 0.376 )$&$  -0.017\  ( 0.374 )$\\
					&&&$( 0.5 , 0.9 )$&$  -0.008\  ( 0.188 )$&$  -0.008\  ( 0.187 )$\\
					\addlinespace[0.2cm]  
					2000  & 20 & 20  & $( 0.1 , 0.5 )$&$  -0.046\  ( 1.187 )$&$  0.029\ ( 0.320 )$\\
					&&&$( 0.1 , 0.9 )$&$  -0.027\  ( 1.437 )$&$  0.091\ ( 0.909 )$\\
					&&&$( 0.5 , 0.9 )$&$  0.019\ ( 0.596 )$&$  0.062\ ( 0.737 )$\\
					& 40  & 40  &$( 0.1 , 0.5 )$&$  -0.002\  ( 0.182 )$&$  -0.002\  ( 0.181 )$\\
					&&&$( 0.1 , 0.9 )$&$  -0.007\  ( 0.365 )$&$  -0.007\  ( 0.363 )$\\
					&&&$( 0.5 , 0.9 )$&$  -0.005\  ( 0.186 )$&$  -0.005\  ( 0.185 )$\\
					& 60  & 60  &$( 0.1 , 0.5 )$&$  0.005\ ( 0.177 )$&$  0.004\ ( 0.176 )$\\
					&&&$( 0.1 , 0.9 )$&$  0.006\ ( 0.352 )$&$  0.005\ ( 0.350 )$\\
					&&&$( 0.5 , 0.9 )$&$  0.001\ ( 0.177 )$&$  0.001\ ( 0.176 )$\\
					
					\bottomrule
			\end{tabular}}
			\begin{minipage}{0.75\linewidth}
				\footnotesize
				\textit{Notes:}
				The number of Monte-Carlo repetitions is 250. 
				We report bias averaged over 250 repetitions with
				the standard deviation in the parenthesis.  We consider individual covariate $z := (1,1,1)'$.
			\end{minipage}
		\end{table}
		\clearpage


		\newpage
		\section{Additional Empirical Results}\label{sec:additional_empirical_cpt2}

		In this section, we first provide additional information about the March CPS dataset, and then show the estimation results based on model with additive fixed effects as a comparison with our proposed model.

		\subsection{Additional Details about the March CPS Dataset}\label{sec:empirical_summary_stat}

		Table~\ref{appendix:table1} summarizes the 16 industries covered by 1938 and 1966 FLSA, respectively. As the key interest of this empirical study is to quantify the minimum wage policy effect of 1966 FLSA, the eight industries covered by 1966 FLSA are considered as treated industries in this case study, while the eight industries covered by 1938 FLSA are classified as control industries. For each listed industry, the average number of observations per year is listed. We note that Forestry and Fishing only contains 35 individual observations per year. As the sample size is not sufficient to perform the first-step quantile estimation, we remove Forestry and Fishing from the treated industries to avoid estimation error, which leads to 7 treated industries in our study.
		
		Table~\ref{appendix:table2} reports the variation of employment shares and the average earnings for white, black, male and female subpopulations over time. Figures~\ref{fig:earnings_year} and \ref{fig:summary_stat} further visualize the time trends of those statistics for treated and control industries respectively. Following \cite{derenoncourt2021minimum}, the March CPS 1963 (i.e. corresponding to earnings earned in the year 1962) is excluded from our analysis as it suffers from a lower number of observations and lacks demographic information (such as education) for the entire population.
		Figure~\ref{fig:earnings_year} shows an upward trend in the average earnings for all subpopulations. Black workers and female workers have lower average earnings, compared to white workers and male workers. Furthermore, within black and female workers, the earning gap between workers in the treated and control industry remained rather stable before the 1966 FLSA but gradually reduced after 1967. Figure~\ref{fig:earnings_box} additionally summarizes the distribution of the earnings for the four major groups of individuals across treated and control industries before and after the introduction of the 1966 FLSA. 
		Figure~\ref{fig:summary_stat} presents the aggregate evidence of stable employment trends in the CPS. Figures~\ref{fig:summary_stat}.(a) and (b) show that the proportion of black employees and female employees across industry types (industries covered in 1938 v.s. covered in 1967) are relatively stable from the early 1960s to 1980. There is no discontinuity in the aggregate shares of workers in the treated vs. control industries around the reform. Notably, while black proportions are below 25\% for both treated and control industries, the female proportions are doubled in treated industries compared with those in control industries throughout the years. Figures~\ref{fig:summary_stat}.(c) and (d) further decompose these aggregate employment trends by gender and race and show a stable trend.

		Figure~\ref{fig:educ_exp} summarizes the distribution of education and experience level of the four major groups of individuals of interest. The average and medium education level is around 12 years of schooling for all groups, with heavy tails at the lower quantile, representing a large variation in education level among the society. There is vast literature analyzing the effect of improving education on reducing wage disparities, such as \citep{card1992school}. However, the education factor cannot explain the dramatic decline of the wage gap after 1967, as education levels remained stable over the years. In addition, the distributions of experience levels are rather similar for all groups, with the first quarter and medium level of experience being 10 and 20 years, respectively.
		
		\begin{table}[!htbp]
			\centering
			\caption{List of industries, and the average number of observations per year used in March CPS (1962, 1964--1981) \label{appendix:table1}}
			\tabcolsep 0.12in
			\resizebox{\textwidth}{!}{%
				\begin{tabular}{@{} l   c c c c c  @{}} 
					\toprule
					Industry &  Observations & \multicolumn{4}{c}{Subgroup Sizes} \\
					\cmidrule(r){3-6} 
					& & White & Black & Men & Women\\
					
					\midrule
					\textbf{Control industries (industries covered by 1938 FLSA):}	&		&				&				&				&				\\																			
					Business and Repair Services	&	841	&$	256	\ (	140	)$&$	585	\ (	217	)$&$	63	\ (	25	)$&$	778	\ (	331	)$\\
					Durable manufacturing	&	5107	&$	1060	\ (	333	)$&$	4048	\ (	956	)$&$	420	\ (	115	)$&$	4687	\ (	1136	)$\\
					Finance, Insurance, and Real Estate	&	1519	&$	760	\ (	345	)$&$	760	\ (	216	)$&$	85	\ (	43	)$&$	1435	\ (	506	)$\\
					Food manufacturing	&	726	&$	171	\ (	49	)$&$	555	\ (	115	)$&$	77	\ (	20	)$&$	649	\ (	141	)$\\
					Mining	&	304	&$	26	\ (	20	)$&$	278	\ (	124	)$&$	8	\ (	4	)$&$	296	\ (	140	)$\\
					Other non-durable manufacturing	&	2480	&$	988	\ (	225	)$&$	1493	\ (	308	)$&$	215	\ (	63	)$&$	2265	\ (	475	)$\\
					Transportation, Communication, and Other Utilities	&	2203	&$	439	\ (	195	)$&$	1764	\ (	423	)$&$	187	\ (	60	)$&$	2016	\ (	542	)$\\
					Wholesale Trade	&	1149	&$	240	\ (	97	)$&$	909	\ (	267	)$&$	68	\ (	15	)$&$	1081	\ (	349	)$\\
					\\												
					\textbf{Treated industries (industries covered by 1966 FLSA):}		&		&				&				&				&				\\
					Agriculture	&	351	&$	50	\ (	19	)$&$	301	\ (	82	)$&$	54	\ (	23	)$&$	296	\ (	95	)$\\
					Entertainment and Recreation Services	&	228	&$	82	\ (	44	)$&$	146	\ (	66	)$&$	18	\ (	7	)$&$	211	\ (	103	)$\\
					Forestry and Fishing	&	36	&$	6	\ (	5	)$&$	30	\ (	15	)$&$	2	\ (	2	)$&$	34	\ (	20	)$\\
					Hospitals	&	1134	&$	865	\ (	344	)$&$	270	\ (	107	)$&$	220	\ (	69	)$&$	914	\ (	386	)$\\
					Hotels, laundries, and other personal services	&	531	&$	331	\ (	102	)$&$	200	\ (	59	)$&$	99	\ (	29	)$&$	432	\ (	148	)$\\
					Nursing homes and other professional services	&	1515	&$	922	\ (	513	)$&$	593	\ (	279	)$&$	139	\ (	77	)$&$	1376	\ (	714	)$\\
					Restaurants	&	649	&$	429	\ (	141	)$&$	219	\ (	87	)$&$	68	\ (	16	)$&$	581	\ (	215	)$\\
					Schools and other educational services	&	2693	&$	1667	\ (	684	)$&$	1027	\ (	340	)$&$	266	\ (	84	)$&$	2427	\ (	939	)$\\
					\bottomrule
					
			\end{tabular}}
			\begin{minipage}{\linewidth} 
				\footnotesize
				\textit{Notes:}  
				The data sources are 
				the Current Population Survey (1962, 1964--1981).
				Sample includes black or white adults aged 25--65,
				who worked more than 13 weeks last year,
				worked three hours last week,
				do not live in group quarters,
				are not self-employed and  not unpaid family worker
				with no missing industry or occupation code within the treated and control industries. We report the number of observations, and the size of subpopulations (white, female, black and white), averaged over years and aggregated by industry, as well as the corresponding standard deviation over years within the parenthesis. 
			\end{minipage}  
		\end{table}
		
		\begin{table}[!htbp]
			\centering
			\caption{The employment shares and earnings by race and gender \label{appendix:table2}}
			\tabcolsep 0.12in
			\resizebox{\textwidth}{!}{%
				\begin{tabular}{@{} l   c c c c c c c c c @{}} 
					\toprule
					Year &  Observations  & \multicolumn{4}{c}{Employment Shares} & \multicolumn{4}{c}{Log-Earnings (\$2017)}\\
					\cmidrule(r){3-6}   \cmidrule{7-10}
					& & White & Black & Men & Women & White & Black & Men & Women\\
					
					\midrule
					1961	&	10086	&	0.913	&	0.087	&	0.679	&	0.321	&	10.329	&	9.722	&	10.506	&	9.791	\\
					1963	&	10535	&	0.911	&	0.089	&	0.675	&	0.325	&	10.378	&	9.843	&	10.573	&	9.826	\\
					1964	&	10561	&	0.901	&	0.099	&	0.665	&	0.335	&	10.394	&	9.891	&	10.606	&	9.824	\\
					1965	&	22587	&	0.902	&	0.098	&	0.669	&	0.331	&	10.434	&	9.891	&	10.641	&	9.855	\\
					1966	&	14313	&	0.904	&	0.096	&	0.665	&	0.335	&	10.459	&	9.923	&	10.676	&	9.874	\\
					1967	&	22867	&	0.905	&	0.095	&	0.650	&	0.350	&	10.474	&	10.012	&	10.702	&	9.926	\\
					1968	&	23228	&	0.904	&	0.096	&	0.646	&	0.354	&	10.505	&	10.089	&	10.742	&	9.961	\\
					1969	&	22452	&	0.903	&	0.097	&	0.641	&	0.359	&	10.549	&	10.136	&	10.796	&	9.996	\\
					1970	&	22043	&	0.902	&	0.098	&	0.633	&	0.367	&	10.554	&	10.175	&	10.801	&	10.028	\\
					1971	&	21363	&	0.903	&	0.097	&	0.625	&	0.375	&	10.547	&	10.174	&	10.798	&	10.033	\\
					1972	&	21407	&	0.908	&	0.092	&	0.621	&	0.379	&	10.576	&	10.264	&	10.846	&	10.058	\\
					1973	&	21176	&	0.905	&	0.095	&	0.614	&	0.386	&	10.582	&	10.261	&	10.863	&	10.055	\\
					1974	&	20349	&	0.908	&	0.092	&	0.603	&	0.397	&	10.549	&	10.271	&	10.833	&	10.052	\\
					1975	&	22011	&	0.908	&	0.092	&	0.593	&	0.407	&	10.515	&	10.237	&	10.799	&	10.037	\\
					1976	&	26094	&	0.911	&	0.089	&	0.590	&	0.410	&	10.517	&	10.276	&	10.813	&	10.037	\\
					1977	&	26059	&	0.908	&	0.092	&	0.583	&	0.417	&	10.527	&	10.275	&	10.821	&	10.061	\\
					1978	&	26598	&	0.908	&	0.092	&	0.571	&	0.429	&	10.525	&	10.309	&	10.828	&	10.074	\\
					1979	&	31723	&	0.915	&	0.085	&	0.559	&	0.441	&	10.516	&	10.291	&	10.825	&	10.081	\\
					1980	&	32438	&	0.916	&	0.084	&	0.557	&	0.443	&	10.496	&	10.260	&	10.798	&	10.072	\\
					\bottomrule
					
			\end{tabular}}
			\begin{minipage}{\linewidth} 
				\footnotesize
				\textit{Notes:}  
				The data sources are 
				the Current Population Survey (1962, 1964--1981).
				Sample includes black or white adults aged 25--65,
				who worked more than 13 weeks last year,
				worked three hours last week,
				do not live in group quarters,
				are not self-employed and  not unpaid family worker
				with no missing industry or occupation code within the treated and control industries. We report the number of observations, and the size of subpopulations (white, female, black and white), averaged over years and aggregated by industry, as well as the corresponding standard deviation over years within the parenthesis. 
			\end{minipage}  
		\end{table}

		\begin{figure}[!htb]
			\centering
			\caption{Average Log-Earnings (\$2017) per Year}
			\label{fig:earnings_year}
			\begin{subfigure}{0.465\textwidth}
				\centering
				\caption{White Workers}
				\includegraphics[width=0.95\textwidth]{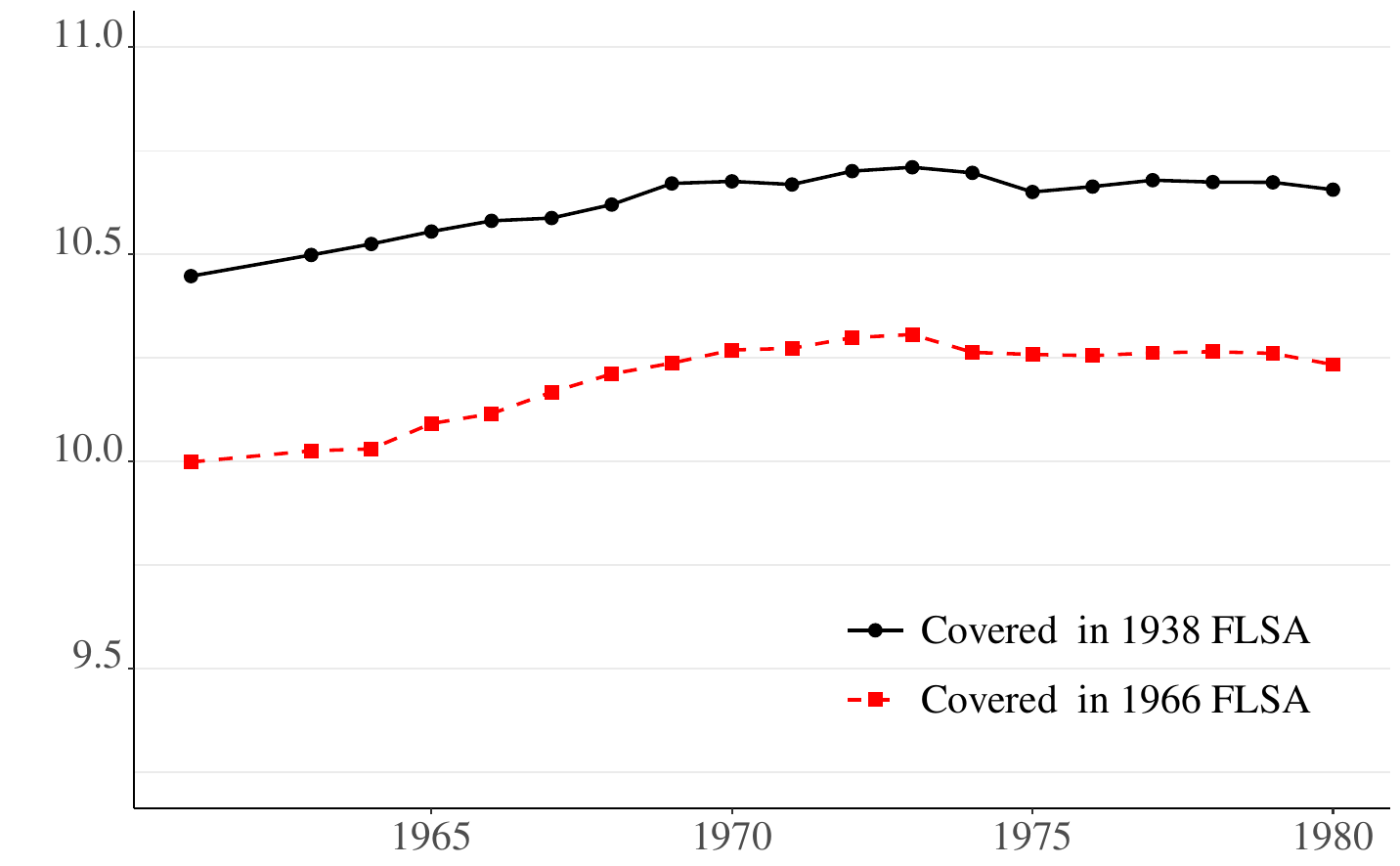} 
			\end{subfigure}
			\begin{subfigure}{0.465\textwidth}
				\centering
				\caption{Black Workers}
				\includegraphics[width=0.95\textwidth]{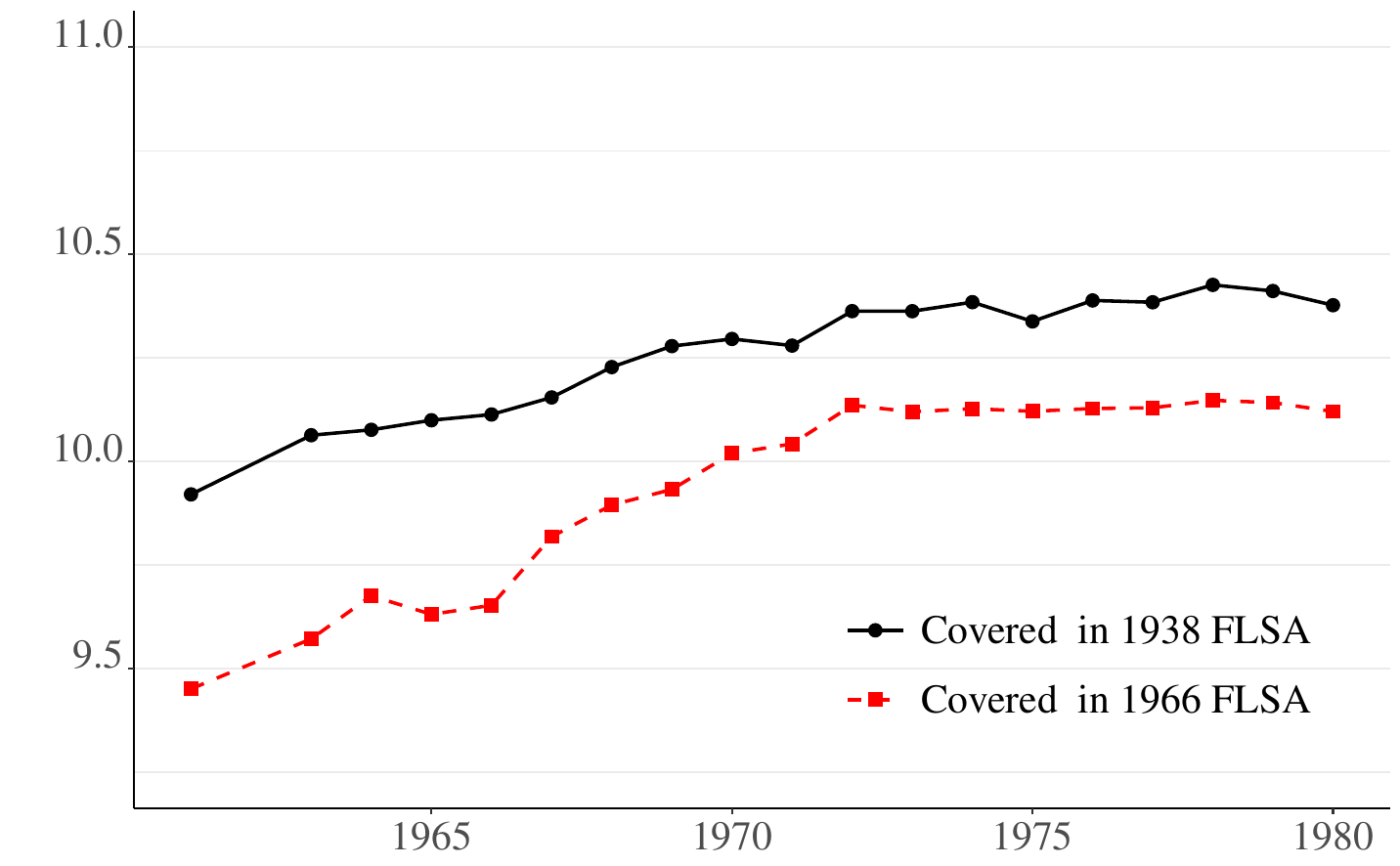} 
			\end{subfigure} \\
			\begin{subfigure}{0.465\textwidth}
				\centering
				\caption{Male Workers}
				\includegraphics[width=0.95\textwidth]{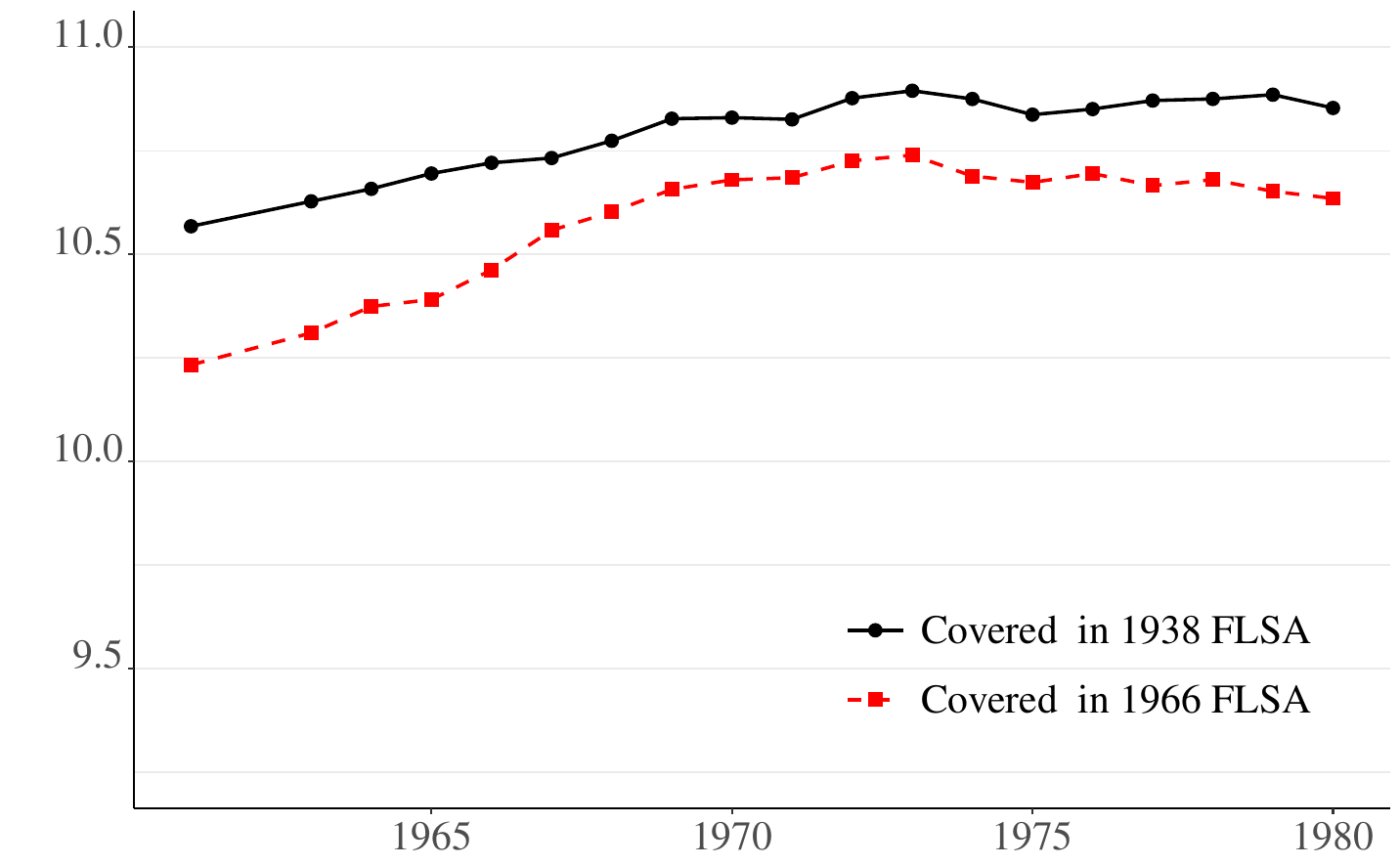} 
			\end{subfigure}
			\begin{subfigure}{0.465\textwidth}
				\centering
				\caption{Female Workers}
				\includegraphics[width=0.95\textwidth]{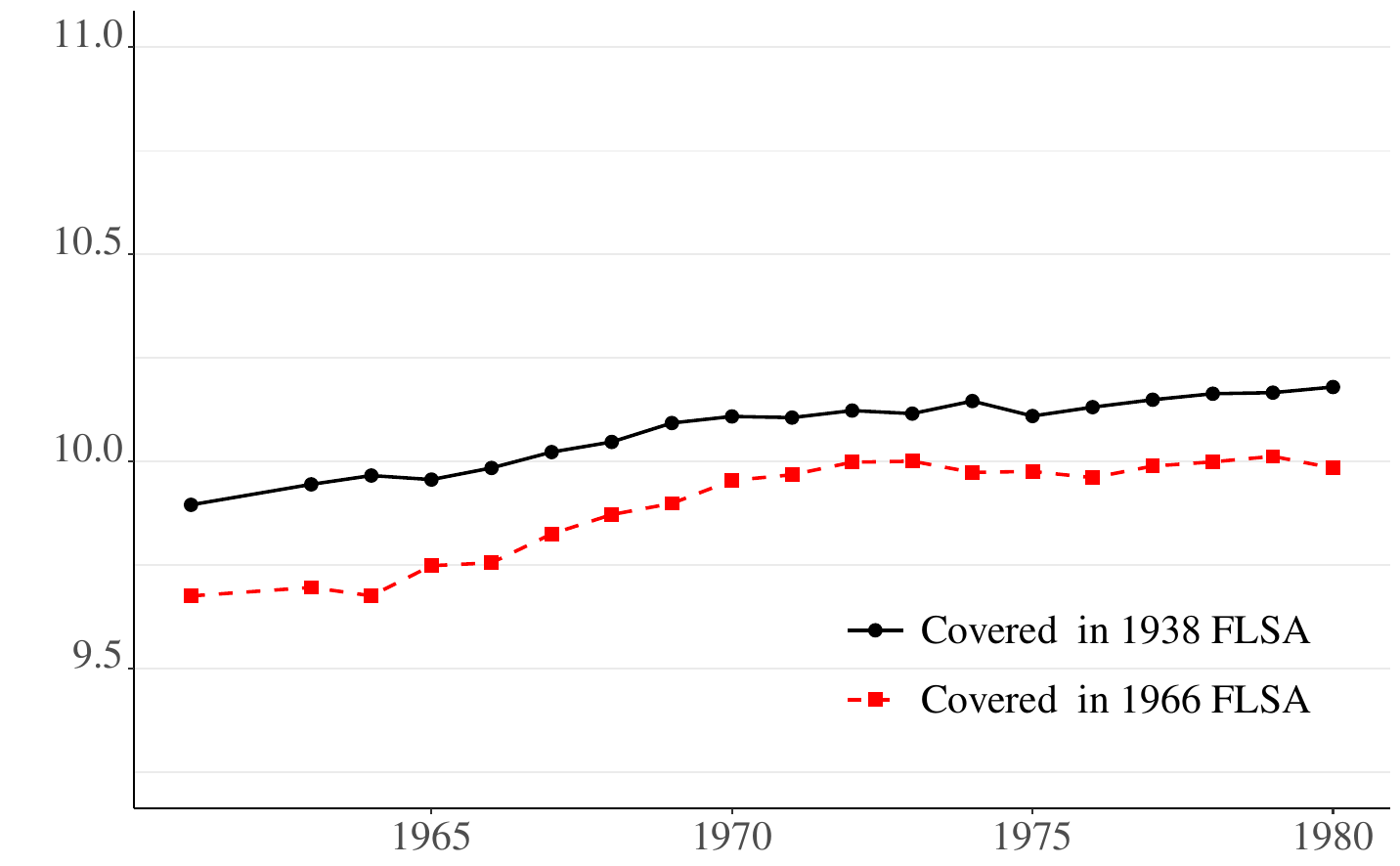} 
			\end{subfigure} 
			
			\begin{minipage}{\linewidth} 
				\footnotesize
				\textit{Notes:}  
				The data sources are 
				the Current Population Survey 1962, 1964--1981.
				Sample includes black or white adults aged 25--65,
				who worked more than 13 weeks last year,
				worked three hours last week,
				do not live in group quarters,
				are not self-employed and  not unpaid family worker
				with no missing industry or occupation code within the treated and control industries.
			\end{minipage}  
		\end{figure}

		\begin{figure}[!htb]
			\centering
			\caption{Log-Earnings (\$2017) by Race and Gender}
			\label{fig:earnings_box}
			
			\begin{subfigure}{0.465\textwidth}
				\centering
				\caption{Control industries, Year $< 1967$}
				\includegraphics[width=0.95\textwidth]{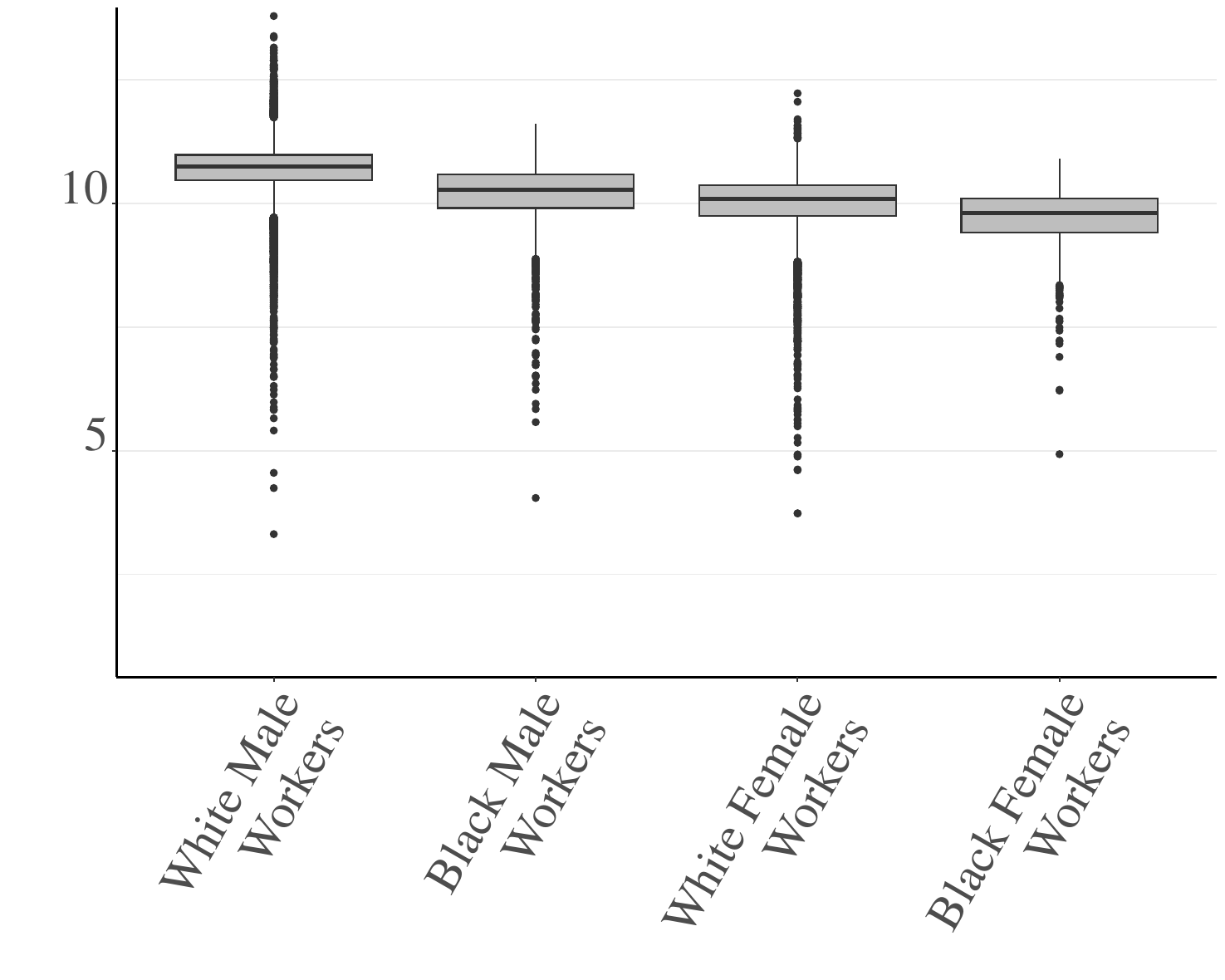} 
			\end{subfigure}
			\begin{subfigure}{0.465\textwidth}
				\centering
				\caption{Treated industries, Year $< 1967$}
				\includegraphics[width=0.95\textwidth]{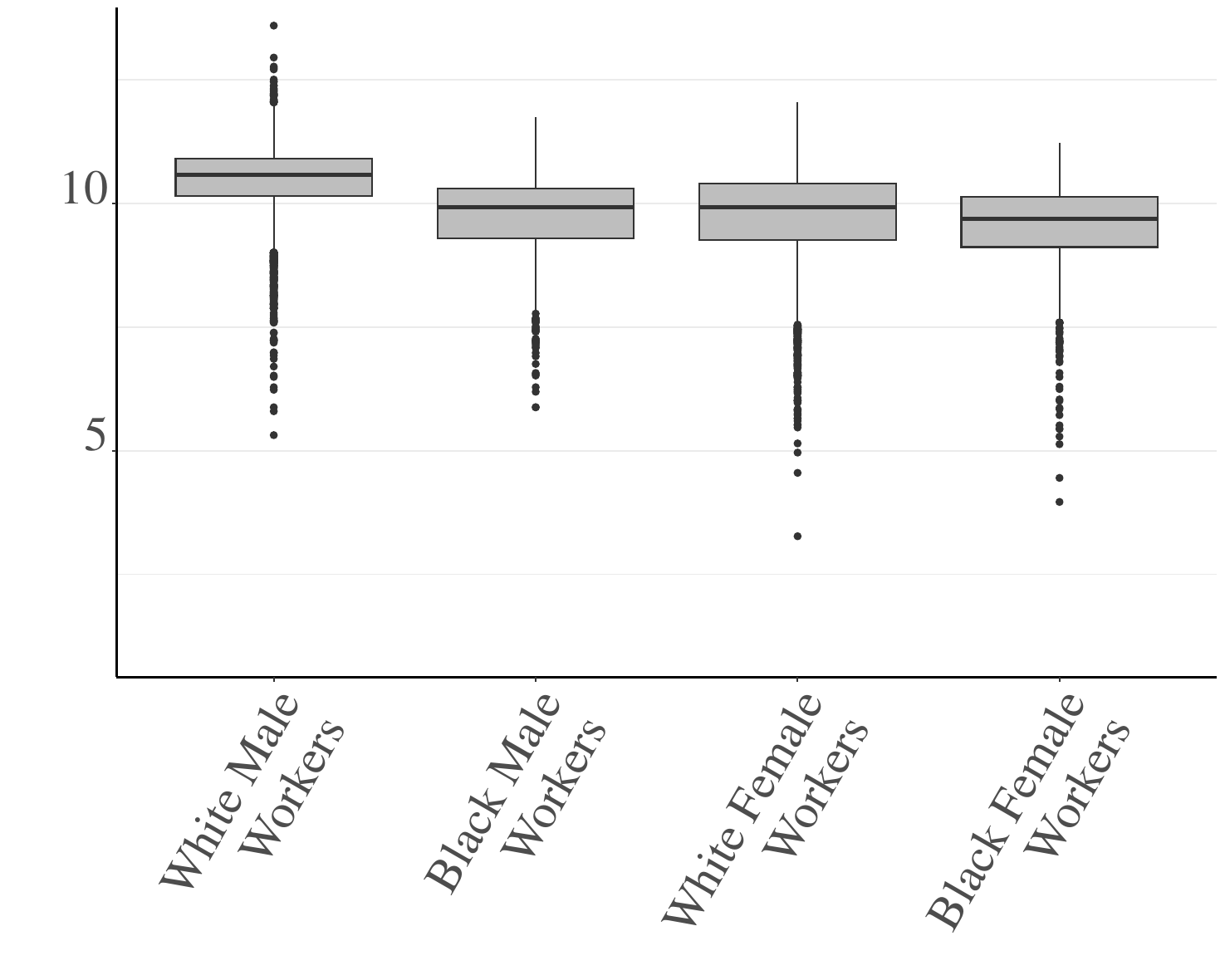} \\
			\end{subfigure}
			
			\begin{subfigure}{0.465\textwidth}
				\centering
				\caption{Control industries, Year $\geq 1967$}
				\includegraphics[width=0.95\textwidth]{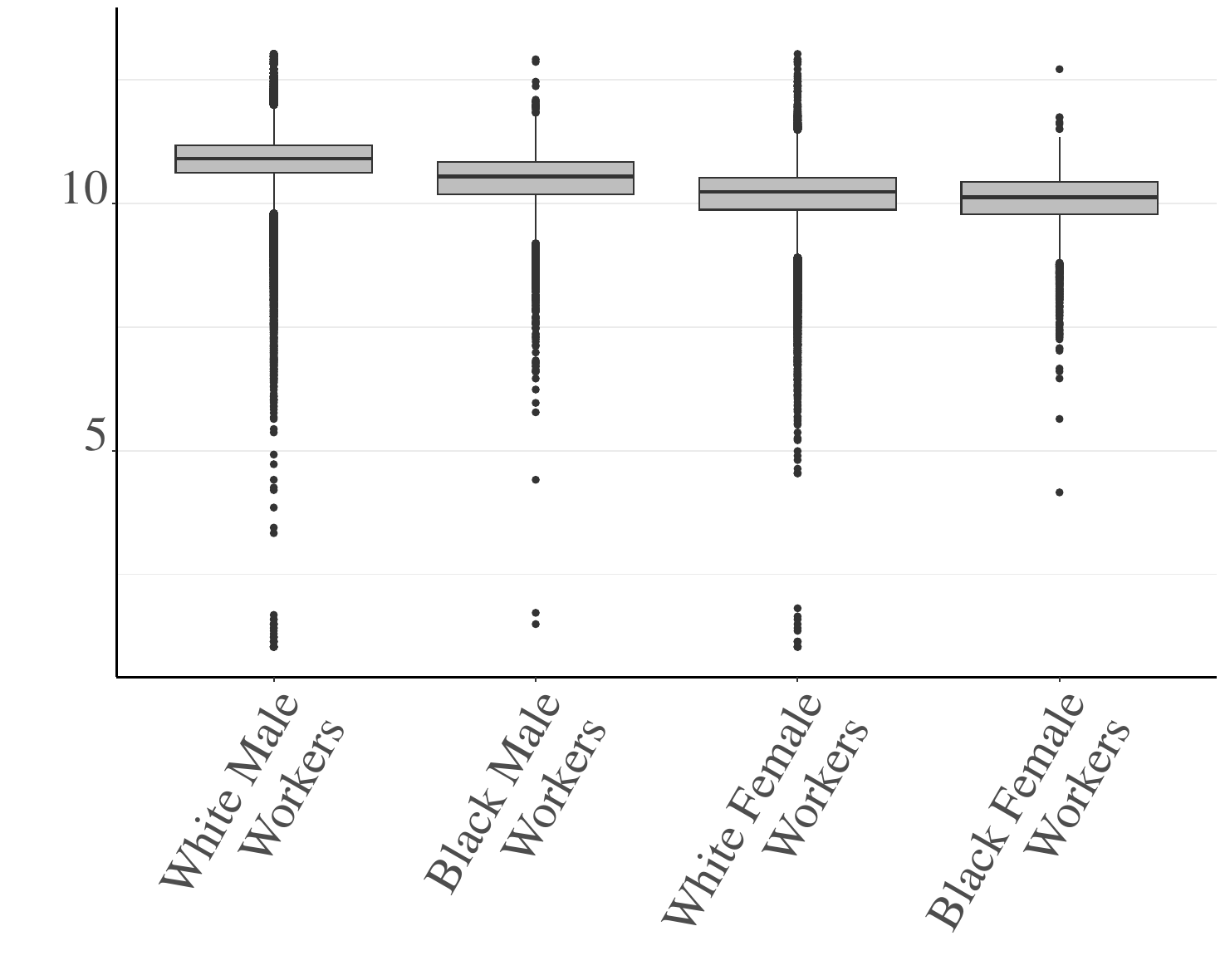} 
			\end{subfigure} 
			\begin{subfigure}{0.465\textwidth}
				\centering
				\caption{Treated industries, Year $\geq 1967$}
				\includegraphics[width=0.95\textwidth]{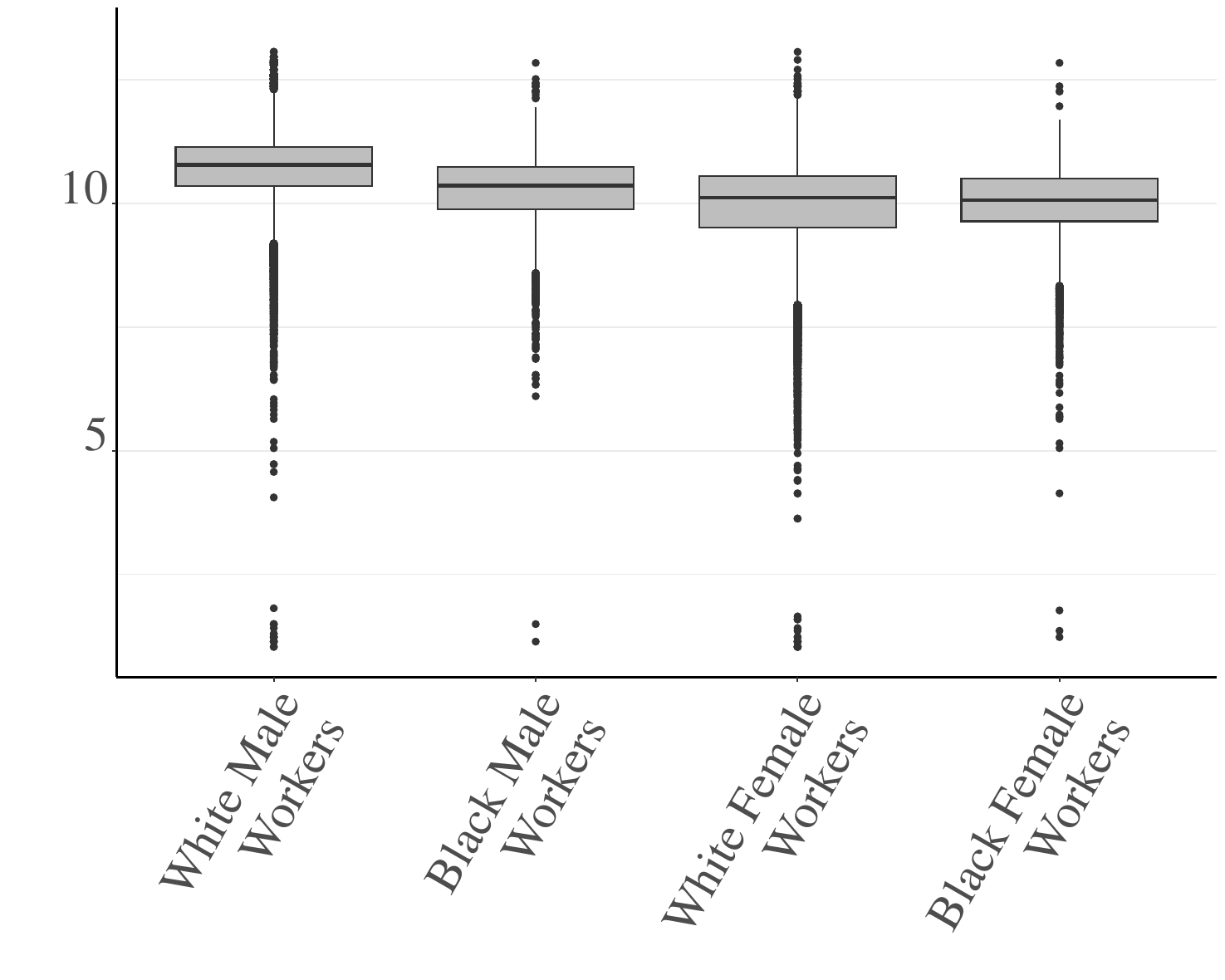} \\
			\end{subfigure}

			\begin{minipage}{\linewidth} 
				\footnotesize
				\textit{Notes:}  
				The data sources are 
				the Current Population Survey 1962, 1964--1981.
				Sample includes black or white adults aged 25--65,
				who worked more than 13 weeks last year,
				worked three hours last week,
				do not live in group quarters,
				are not self-employed and  not unpaid family worker
				with no missing industry or occupation code within the treated and control industries.
			\end{minipage}  
		\end{figure}
		
		\begin{figure}[!htb]
			\centering
			\caption{Aggregate Employment Shares}
			\label{fig:summary_stat}
			\begin{subfigure}{0.465\textwidth}
				\centering
				\caption{Black proportions, by industries}
				\includegraphics[width=0.95\textwidth]{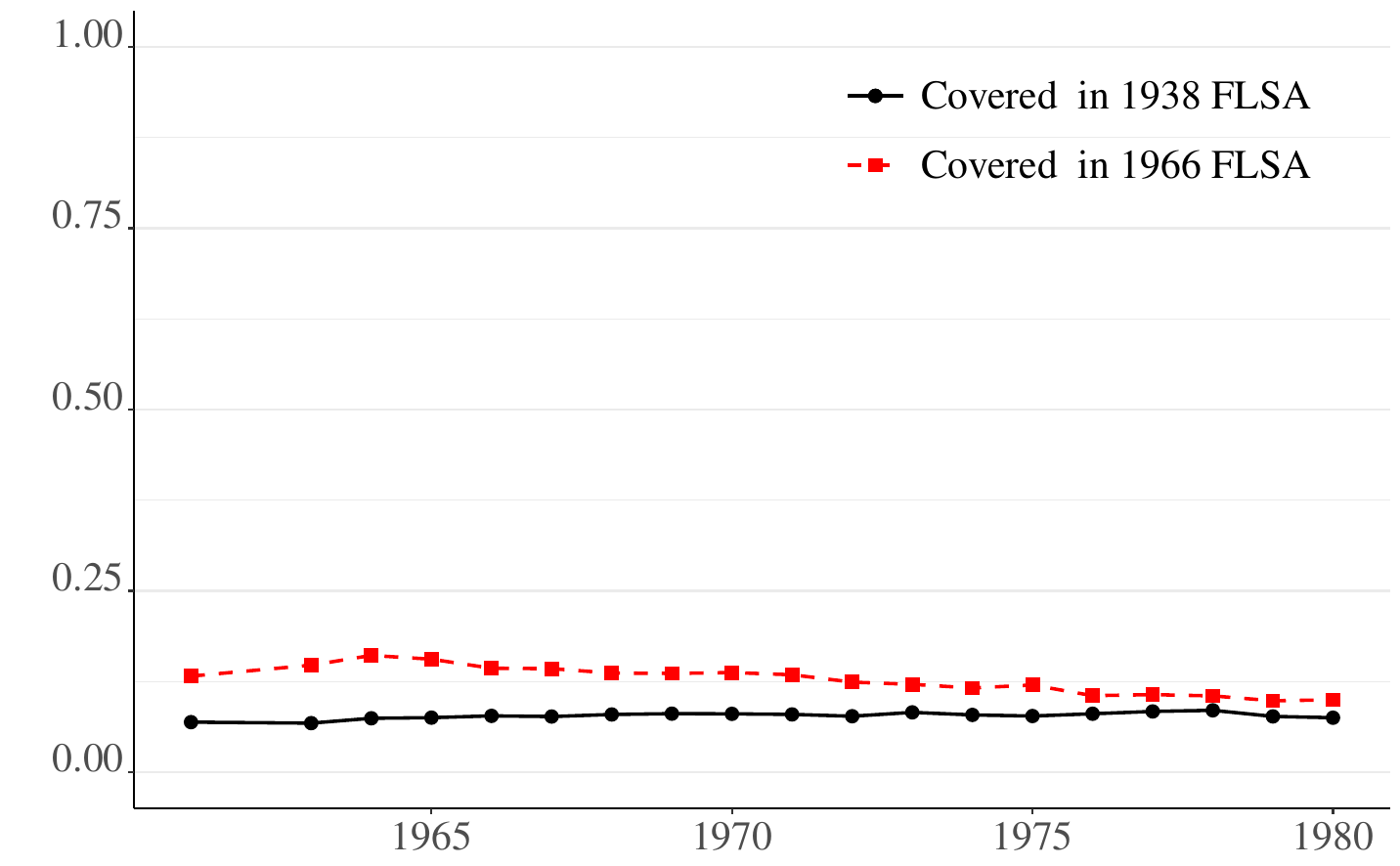} 
			\end{subfigure}
			\begin{subfigure}{0.465\textwidth}
				\centering
				\caption{Female proportions, by industries}
				\includegraphics[width=0.95\textwidth]{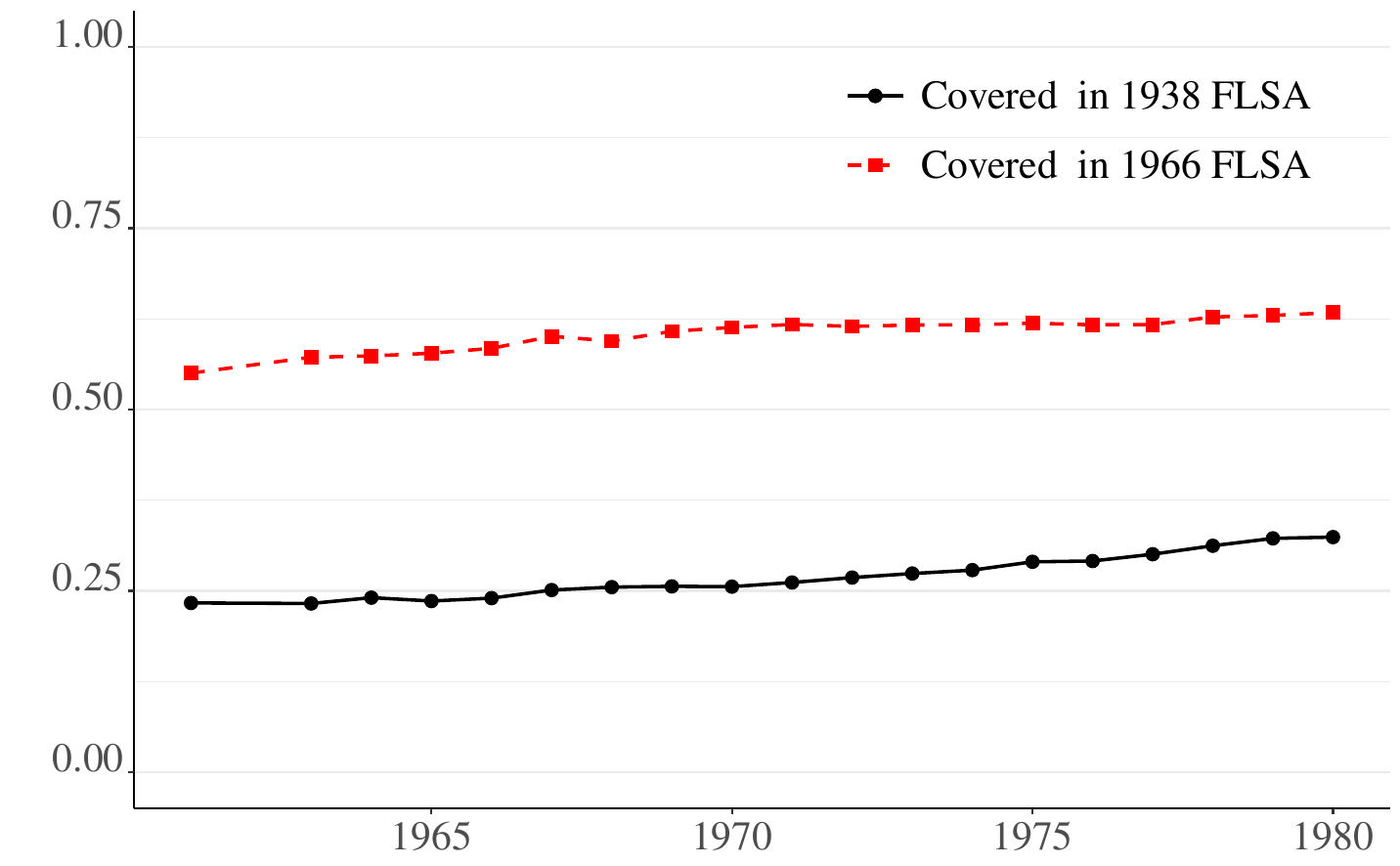} 
			\end{subfigure} \\
			\begin{subfigure}{0.465\textwidth}
				\centering
				\caption{Control industries, by race and gender}
				\includegraphics[width=0.95\textwidth]{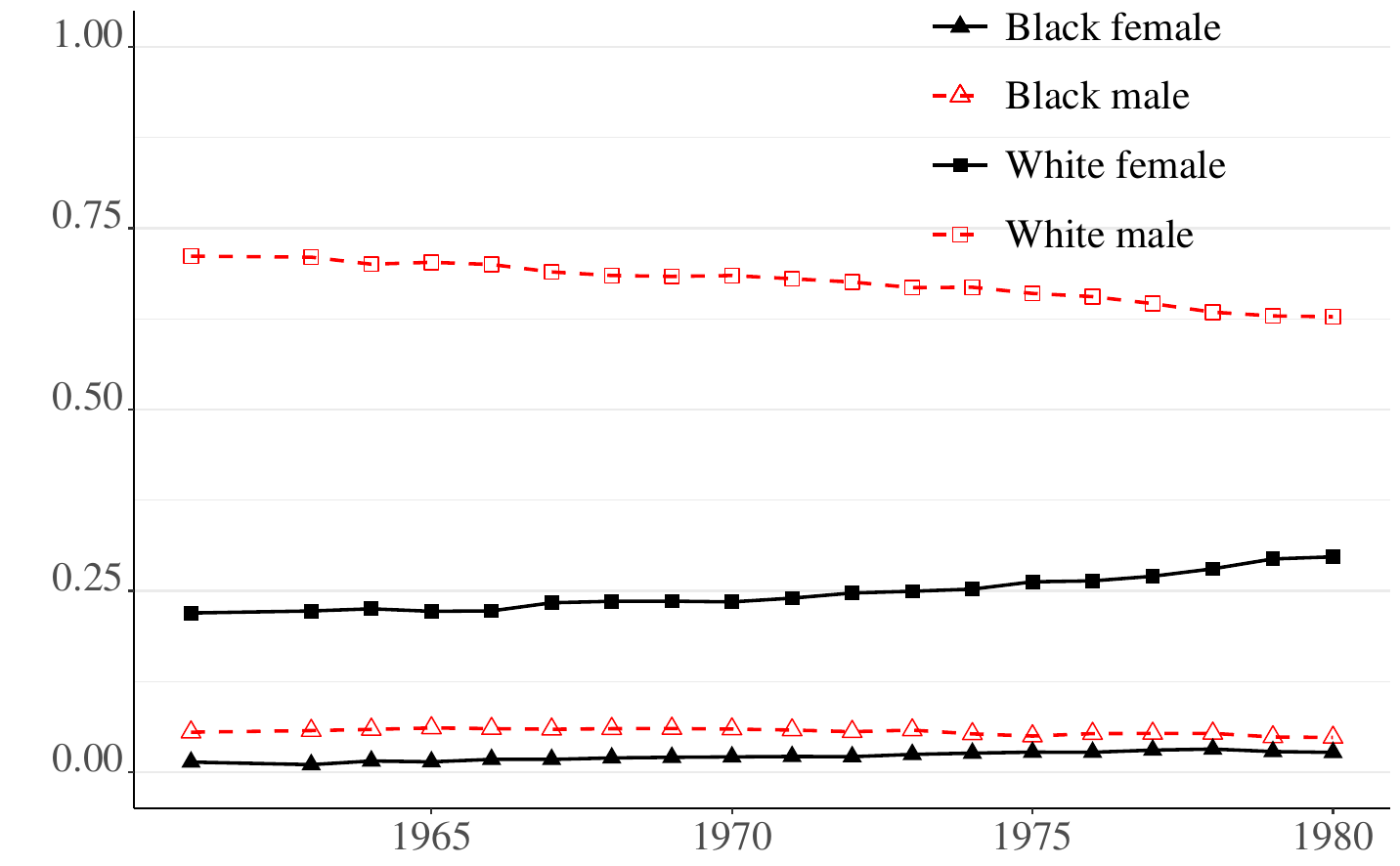} 
			\end{subfigure}
			\begin{subfigure}{0.465\textwidth}
				\centering
				\caption{Treated industries, by race and gender}
				\includegraphics[width=0.95\textwidth]{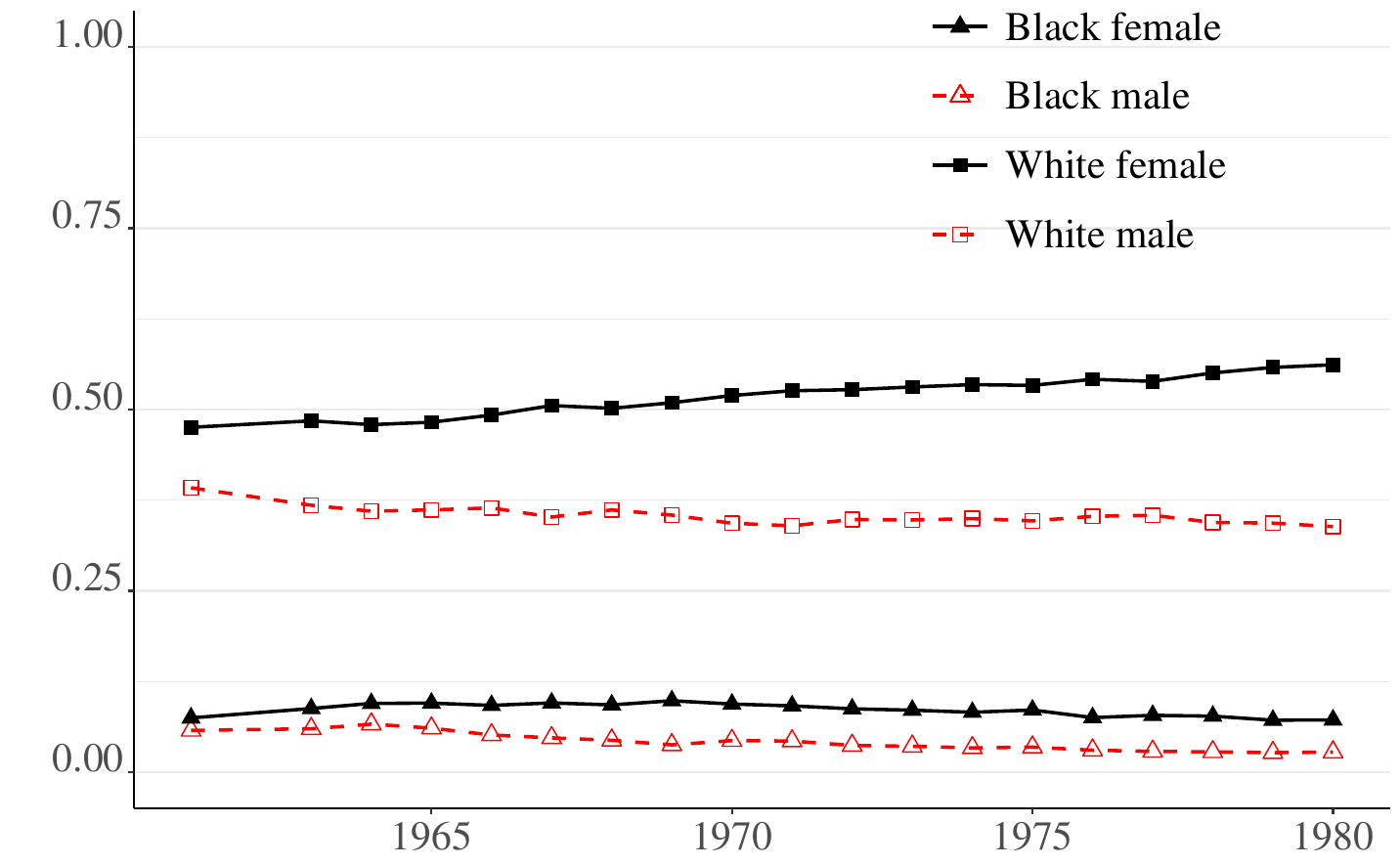} 
			\end{subfigure} 
			
			\begin{minipage}{\linewidth} 
				\footnotesize
				\textit{Notes:}  
				The data sources are 
				the Current Population Survey 1962, 1964--1981.
				Sample includes black or white adults aged 25--65,
				who worked more than 13 weeks last year,
				worked three hours last week,
				do not live in group quarters,
				are not self-employed and  not unpaid family worker
				with no missing industry or occupation code within the treated and control industries.
			\end{minipage}  
		\end{figure}

		\begin{figure}[!htb]
			\centering
			\caption{Education and Experience by Race and Gender}
			\label{fig:educ_exp}
			\begin{subfigure}{0.465\textwidth}
				\centering
				\caption{Years of education}
				\includegraphics[width=0.95\textwidth]{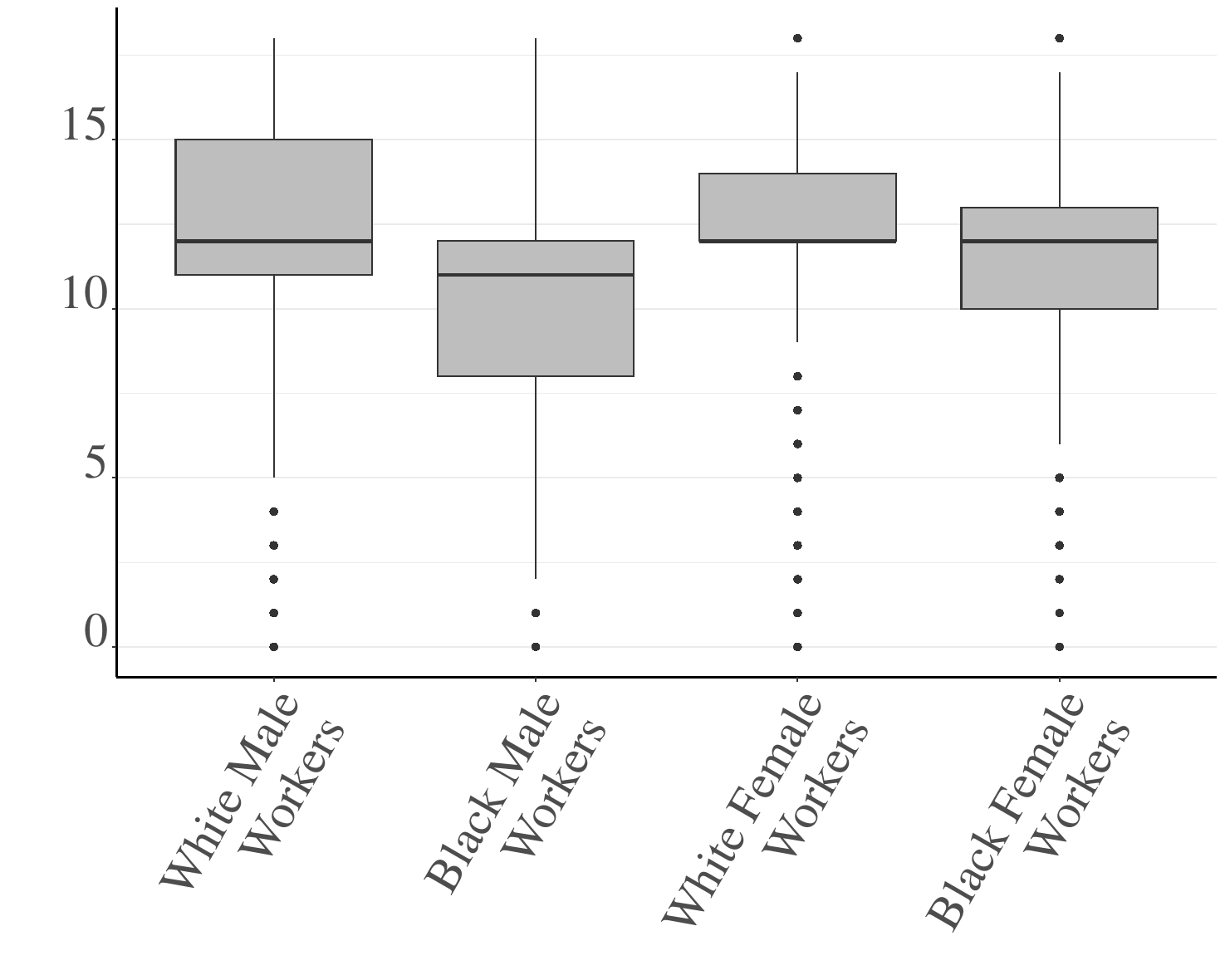} 
			\end{subfigure}
			\begin{subfigure}{0.465\textwidth}
				\centering
				\caption{Years of experience}
				\includegraphics[width=0.95\textwidth]{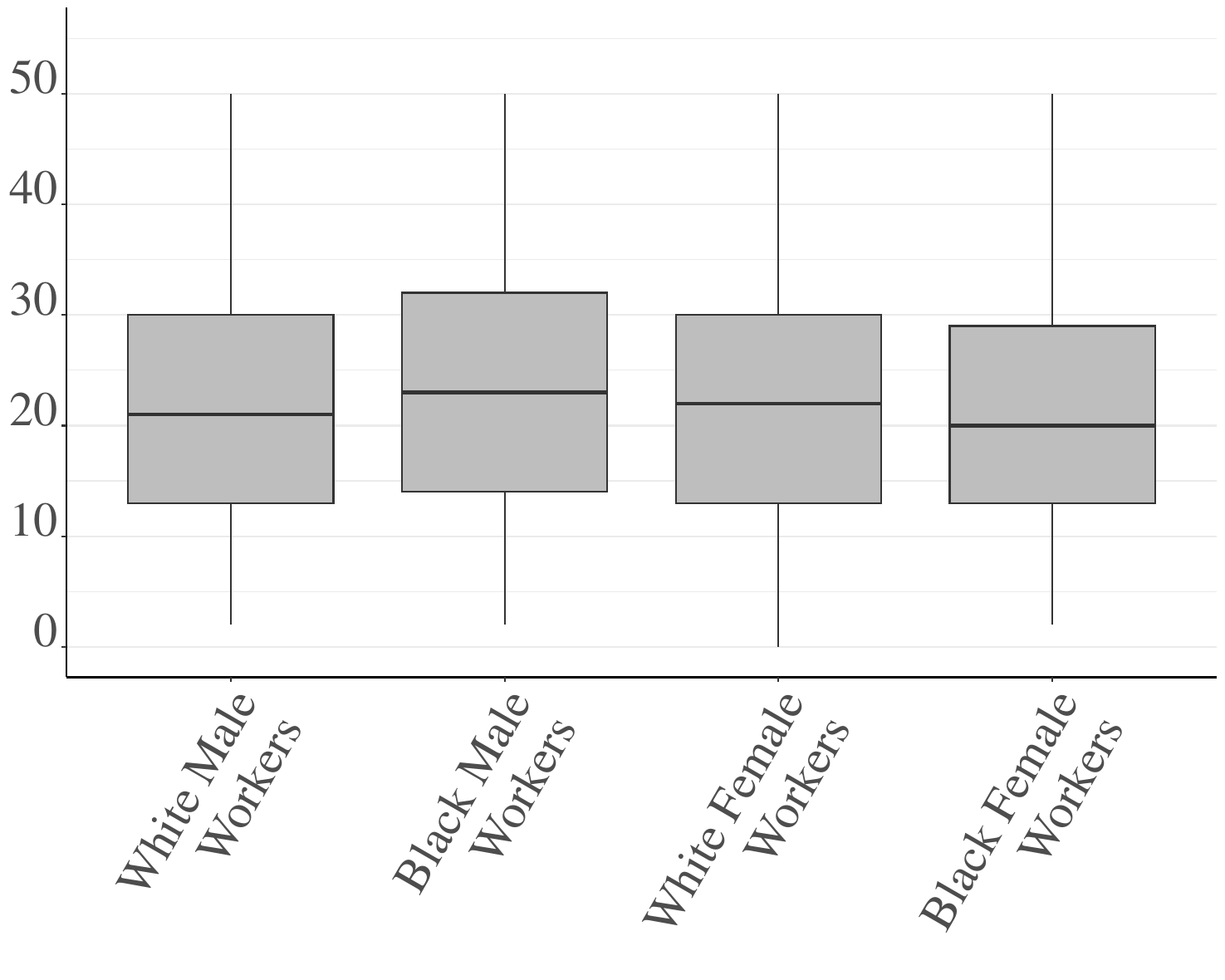} \\
			\end{subfigure} 
			
			\begin{minipage}{\linewidth} 
				\footnotesize
				\textit{Notes:}  
				The data sources are 
				the Current Population Survey 1962, 1964--1981.
				Sample includes black or white adults aged 25--65,
				who worked more than 13 weeks last year,
				worked three hours last week,
				do not live in group quarters,
				are not self-employed and  not unpaid family worker
				with no missing industry or occupation code within the treated and control industries.
			\end{minipage}  
		\end{figure}

		\clearpage

		\newpage
		
		\subsection{Estimation with Alternative Model}
		
		Different from \cite{derenoncourt2021minimum}, which incorporate the additive fixed effects to document the industrial and temporal heterogeneities, we consider the interactive fixed effects in our model specification (\ref{quantileDiD.eq2}) to allow for more flexible latent heterogeneities. A natural question to ask is whether it is necessary to consider the interactive fixed effects or if the conventional additive fixed effects model is the correct specification. To answer this question, we provide the estimation results based on the model with additive fixed effects as follows. We consider quantile random coefficient model in (\ref{Q_y}) with 
		$\alpha_{st}(u) = [\alpha_{1st}(u), \dots, \alpha_{Jst}(u)]'$,
		where, for $j=1, \dots, J$, 
		\begin{eqnarray}
			\label{quantileDiD_ae}
			\alpha_{jst}(u)
			= \beta_{j}(u) + 
			\delta_{jt}(u)    d_{st}
			+
			{f}_{jt}(u) + {\lambda}_{js}(u) +
			\eta_{jst}(u). 
		\end{eqnarray}
		
		In Figure \ref{fig:marginal_ae}, we present time-varying policy effects $\delta_{jt}(u)$ in (\ref{quantileDiD_ae}) with $95\%$ confidence intervals for the categorical individual-level covariates. The estimated policy effects are generally smaller than those from the model with IFE (see Figure~\ref{fig:marginal}), suggesting that omitting the IFE might lead to a downward bias in the estimation of regression coefficients. Especially, at $0.1$th quantile, the estimated policy effects for both black and female are not significantly different from 0 under model (\ref{quantileDiD_ae}), which is different from the results under model (\ref{quantileDiD.eq2}) and deviates from the existing literature. 
		
		Given the questionable estimation results, it is reasonable to question the correctness of using the simple additive fixed effects in the quantile random coefficient regression and use a more flexible IFE structure instead. In addition, we would like to mention that, based on the factor selection criteria given in Section \ref{sec:estimation}, the number of latent factors is distinct across quantiles $\tau$ and individual regressors $j$, ranging from 1 to 3. The variation in the number of latent factors also suggests that the model cannot be fully characterized with additive fixed effects.

		\begin{figure}[!htb]
			\centering
			\caption{Time-Varying Policy Effect Estimates ($\delta_{jt}$) with the Alternative Model (\ref{quantileDiD_ae})}
			\label{fig:marginal_ae}
			\begin{subfigure}{\textwidth}
				\centering
				\caption{Intercept (Baseline: White, Male, Full-Time Workers) }
				\includegraphics[width=\textwidth]{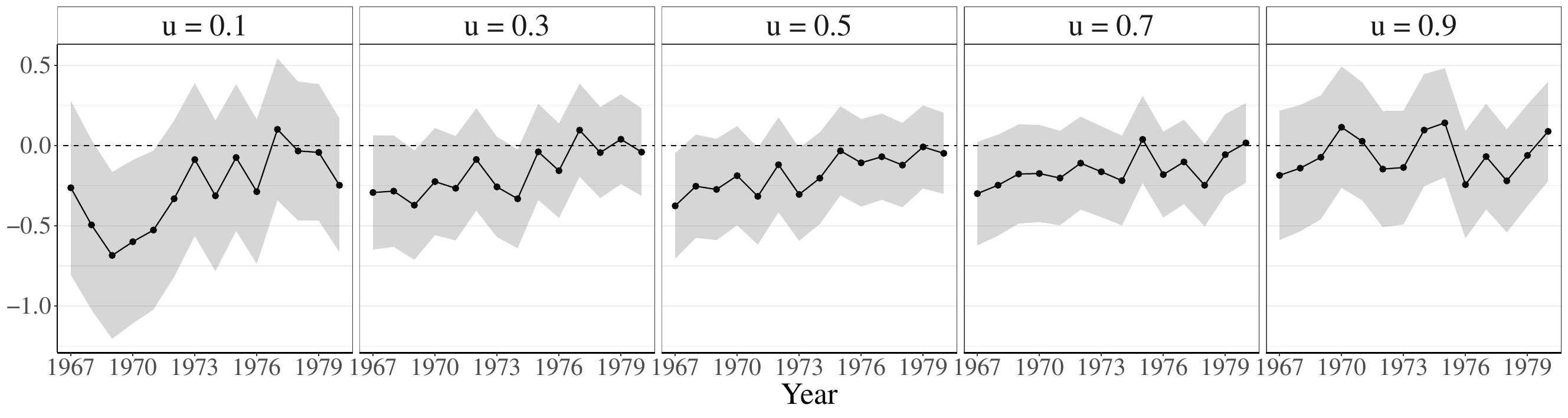} 
			\end{subfigure}\\
			\begin{subfigure}{\textwidth}
				\centering
				\caption{Black}
				\includegraphics[width=\textwidth]{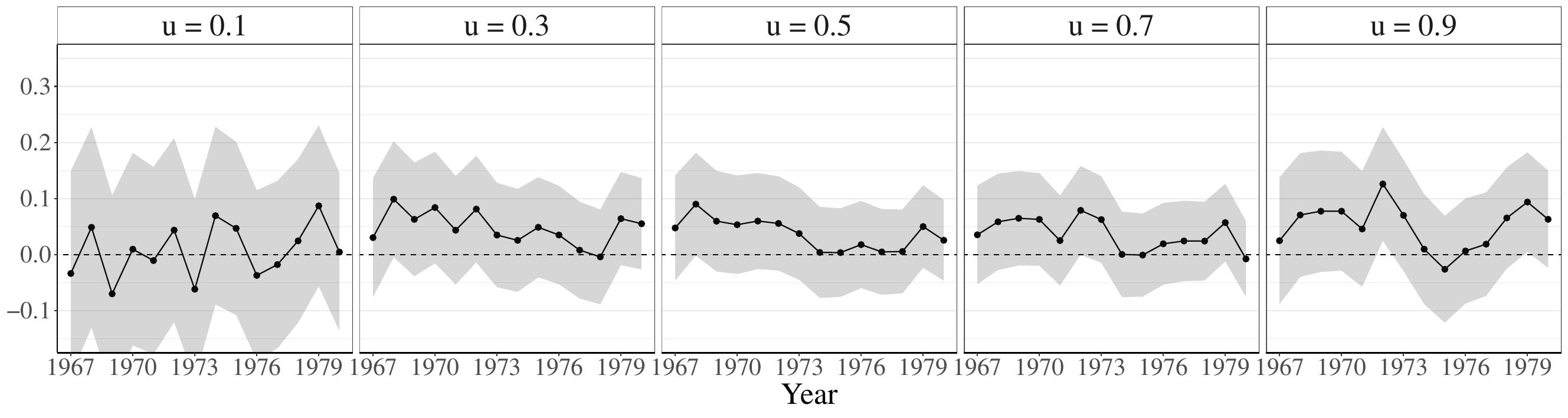} 
			\end{subfigure}\\ 
			\begin{subfigure}{\textwidth}
				\centering
				\caption{Female}
				\includegraphics[width=\textwidth]{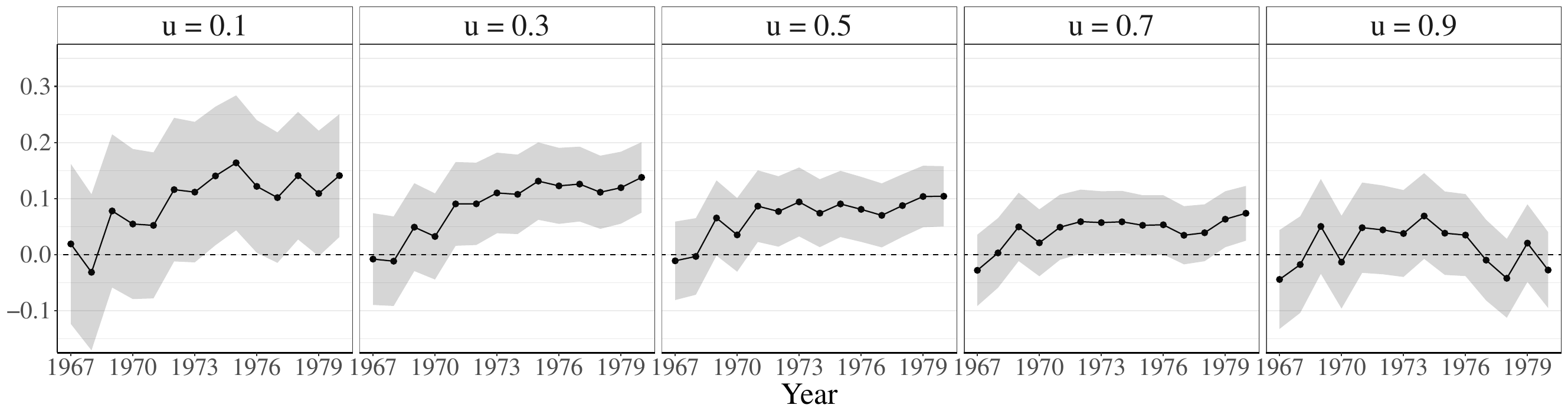}  
			\end{subfigure} 
			\begin{minipage}{0.95\linewidth}
				\footnotesize
				\textit{Notes:} Panels (a)-(c) present the estimated time-varying marginal policy effect
				$\delta_{jt}(u)$
				for $t=1967,\dots,1980$ from the alternative model (\ref{quantileDiD_ae}). From left to right, figures correspond to the estimates at quantiles $u = 0.1, 0.3, 0.5, 0.7, 0.9$. Point estimates are plotted with solid black lines, and the pointwise $95\%$ confidence interval are shown as grey shaded area.
			\end{minipage}
		\end{figure}
		\clearpage

		\section{Technical Assumptions}\label{sec:tech_assumptions}
		In this section, we provide the technical assumptions required for establishing the asymptotic results, together with the corresponding justifications.
		
		\begin{assumption} \label{ass.mixing}
			\begin{enumerate}[label=(\roman*)]
				\item  for $s=1,..., S$, the random sequence $\{{\ell}_{j s t}(u):=(d_{st}, x_{s t}', {f}_{jt}'(u), \eta_{j s t}(u)): t \geq 1\}$ is a stationary and $\alpha$-mixing process with mixing coefficient $a_{s}(\tau)$ and $\tau>0$. Furthermore, there exists a positive coefficient function $a(\tau)$ such that
				$\sup_{s} a_{s}(\tau) \leq a(\tau)$
				and 
				$\sum_{t\neq l}^{T} a(|t-l|)^{\delta/(4+\delta)}=O(T)$
				for such $\delta>0$ that $\sup_{s,t} \mathbb{E}\left[\|\ell_{j s t}(u)\|^{4+\delta}\right]<\infty$.
				\item For any cross groups $s$ and $g$ with $s \neq g$, the random sequence $\{\left(\ell_{j s t}(u), \ell_{j g t}(u)\right): t \geq 1\}$ is also an
				$\alpha$-mixing process with mixing coefficient $a_{sg}(\tau)$, which satisfies
				$\sum_{s \neq g}^{S}a_{sg}(0)^{\delta/(4+\delta)}  \allowbreak =O(S)$ and
				$ \sum_{s\neq g}^{S} \sum_{t \neq l}^{T} a_{sg}(|t-l|)^{\delta/(4+\delta)} = O(ST)$.
				\item
				For any cross groups $s, g, k, m = 1,...,S$ where $s \neq g \neq k \neq m$, the random sequence $\{(\ell_{j s t}(u), \ell_{j g t}(u), \ell_{j k t}(u), \ell_{j m t}(u)): t \geq 1 \}$ is an $\alpha$-mixing process with mixing coefficient $a_{s g k m}(\tau)$ such that
				$
				\sum_{s,g, k, m=1}^{S} \sum_{t \neq l}^{T}
				a_{s g k m}(|t-l|)^{\delta/(4+\delta)}=O\left(S^{2} T\right).
				$
				
			\end{enumerate}
		\end{assumption}
		
		Assumption~\ref{ass.mixing} uses the notation of ``$\alpha$-mixing'' for panel data \citep[e.g.,][]{dong2015semiparametric,jiang2020recursive} to capture the temporal dependence exhibited in large panels and controls for the cross-sectional dependency by regulating the mixing coefficient across sections in a concise manner.
		Alternatively, one can assume the high-order moment conditions employed by \cite{bai2009panel}.

		\begin{assumption}\label{ass.reg_1}
			For any fixed $u\in \mathcal{U}$, as $S,T \rightarrow \infty$ jointly,  
			\begin{enumerate}[label=(\roman*)] 
				\item 
				$ S^{-1}\sum_{s=1}^{S} R_{js}(u)^2 >0 $,
				where
				$$R_{js}(u) := d_{s} - S^{-1}\sum_{g=1}^{S}\omega_{j,sg}(u)d_g,
				\quad
				\omega_{j,sg}(u) := {\lambda}_{jg}(u)'
				\big(
				S^{-1}\Lambda_j(u)' \Lambda_j(u)
				\big)^{-1} {\lambda}_{js}(u),   $$
				\item  the eigenvalues of the following quantities: 
				\begin{align*}
					& \text{(a) }
					(ST)^{-1} \sum_{s=1}^{S} X_s'X_s,
					\quad
					\text{(b) }
					(ST)^{-1}\Big\{\sum_{s=1}^{S} X_s'X_s - \Big(\sum_{s=1}^{S}d_s\Big)^{-1}\sum_{s,g=1}^S X_s'D_sD_g'X_g\Big\},\\
					& 
					\text{(c) }
					\inf_{F: T^{-1}F'F = I_{r}} (ST)^{-1}\sum_{s=1}^{S} X_s' M_{F}X_s ,
				\end{align*}
				are bounded away from zero with probability one, where 
				$M_{F}:= I_T -T^{-1}FF'$.
			\end{enumerate}
		\end{assumption}

		Assumption~\ref{ass.reg_1} is a technical assumption which guarantees that the inverse matrix in the initial and recursive estimators of the regression coefficients are well-defined, so that the estimation in the second step is valid. Similar assumption are adopted in \cite[][etc]{bai2009panel,jiang2020recursive}.

		\begin{assumption}
			\label{ass.tech}
			For each recursive step $m \geq 1$,
			there exist 
			positive definite matrices
			$\Sigma_{FF_j}$ and
			$\Sigma_{0}$ such that, as $S, T \to \infty$,
			\begin{enumerate}[label=(\roman*)]
				\item  
				$
				\widehat{F}^{(m)}_j(u) \in 
				\{
				F \in \mathbb{R}^{T \times r}: T^{-1}F'F = I_r,
				T^{-2}F_j(u)' FF' F_j(u) \rightarrow \Sigma_{FF_j}
				\}
				$;
				\item
				$
				(\widehat{\delta}^{(m)}_{j}(u), \widehat{\beta}^{(m)}_j(u) ) \in 
				\{ (\delta,\beta) \in \mathbb{R}^{T-T_0+1} \times \mathbb{R}^{K} : F_j^{(m)}(u)' \widehat{L}_j(\delta,\beta)F_j^{(m)}(u) \overset{p}{\rightarrow} \Sigma_{0}\}
				$, where $\widehat{L}_j(\delta,\beta)$ is defined as (\ref{eq:W}).

			\end{enumerate}
		\end{assumption}

		Assumption~\ref{ass.tech}.(i) is required for deriving a closed-form expression for the recursive formula of $(\widehat{\delta}^{(m)}_{j}(u)$, $\widehat{\beta}^{(m)}_j(u) )$, so that the Central Limit Theorem (CLT) can be established accordingly, and (ii) is a technical assumption required in the derivations, which ensures the invertibility of $\widehat{V}^{(m)}_j(u) := \text{diag}\big(\widehat{\rho}^{(m)}_{j,1}(u),...,\widehat{\rho}^{(m)}_{j,r}(u)\big)$.

		\begin{assumption} \label{ass.norm}
			For $u \in \mathcal{U}$, the following rates hold:
			\begin{enumerate}[label=(\roman*)]
				\item
				$ \mathbb{E}\|\sum_{s=1}^{S}X_sX_s'\|^2 = O(S^2T^2)$, 
				\item $\mathbb{E}\|\sum_{s=1}^{S}X_s' F_j(u) {\lambda}_{js}(u)\|^{2} = O(ST)$,
				\item for all $t \geq T_0$, $ \mathbb{E}\|\sum_{s=1}^{S} d_sf_{jt}(u)' {\lambda}_{js}(u)\|^{2} = O(S) $ uniformly.
				
			\end{enumerate}
			
		\end{assumption}
		
		Assumption~\ref{ass.norm} guarantees the desirable rates of the regression coefficients. In particular, 
		condition (i) follows trivially when ${\max}_{s,g,t,l} |\mathbb{E}(x'_{st}x_{gt}x_{sl}'x_{gl}) | \leq C < \infty$. Conditions (ii) and (iii) hold when the factor component is independent of the group-level regressors and satisfies $\mathbb{E}[{\lambda}_{js}(u)'f_{jt}(u)]=0$.
		
		\section{Proofs of The Main Results}\label{sec:proof}
		In this section, we provides the proof of the results in the main text.  
		In Section~\ref{A.proof2}, we first derive the convergence rates for the parameter of interest in Theorem~\ref{t5}, and then present the proofs of the asymptotic distributions for the regression coefficients (Theorem~\ref{t6}). Finally, we provide  the consistency of our proposed estimators for the asymptotic bias and covariance, stated as Proposition~\ref{clt_estimate}. 
		In Section~\ref{sec:identification_proof}, we provide the proof of the results on treatment parameters, including the identification result (Theorem~\ref{theorem:identification}) and the CLT (Theorem~\ref{CLT_inequality}).
		
		In what follows, we use a few additional notation.
		For a square matrix $A$,
		let $\text{tr}(A)$ denote the trace operation of $A$
		and let $\rho_{\max}(A)$ and $\rho_{\min}(A)$ denote
		the maximum and minimum eigenvalue of $A$, respectively.
		Notice that $\rho_{\max}(A) \leq \text{tr}(A)$
		for every symmetric positive semi-definite matrix. Thus, $\|X\|^2\leq \text{tr}(X'X)$
		for  matrix $X$, and we repeatedly use this inequality.
		We denote $c$ and $C$ as strictly positive constants that depend only on $c_M,\ C_M$, whose values can change at each appearance.
		Let ${x}_{st}, {x}_{s(k)}$ and $x_{st(k)}$ represent different vectors related to the regressor $X$. Precisely, ${x}_{st}$ is a $K\times 1$ vector, ${x}_{s(k)}$ is a $T \times 1$ vector and $x_{st(k)}$ is a scalar for any $(s,t,k) \in \{1,...,S\} \times \{1,...,T\} \times \{1,...,K\}$. 
		Let $\delta_{ST} \equiv \min(\sqrt{S},\sqrt{T})$. Let $e_k$ denote the a column vector, whose entries are all $0$ except for the $k$-th entry; the dimension of $e_k$ varies along with the context.
		
		In addition, we use the following facts: for all $(s,t) \in \{1,...,S\} \times \{1,...,T\} $, $\|X_s\| = O_p(\sqrt{T})$, $\|F_j(u)\| = O_p(\sqrt{T})$, $\|\widehat{F}_j^{(m)}(u)\| = O_p(\sqrt{T})$, $\|\Lambda_j(u)\| = O_p(\sqrt{S})$, $\|\widehat{\Lambda}_j^{(m)}(u)\| = O_p(\sqrt{S})$ under Assumption~\ref{ass.cov}. Also, $\|M_{F_j}(u)\| = O_p(1)$ since the largest
		eigenvalue of $M_{F_j}(u)$ is 1 as $M_{F_j}(u)$ is a projection matrix, and similarly, $\|M_{\widehat{F}_j}^{(m)}(u)\| = O_p(1)$. In addition, we define the orthogonal projection matrix $P_{F_j}(u) := {F_j}(u)({F_j}(u)'{F_j}(u))^{-1}{F_j}(u)$.
		\subsection{Proofs of the Asymptotic Properties of Regression Coefficients}\label{A.proof2}

		Before proceeding to the proof, since the convergence of $\widehat{\delta}^{(m)}_{jt}(u)$ and $\widehat{\beta}^{(m)}_{j}(u)$ are different, there is a need to partial-out the expression of $\widehat{\delta}^{(m)}_{jt}(u)$ and $\widehat{\beta}^{(m)}_{j}(u)$ from that of $\widehat{\beta}^{(m)}_j(u) := [\widehat{\delta}^{(m)}_{jT_0}(u),...,\widehat{\delta}^{(m)}_{jT}(u),\widehat{\beta}^{(m)}_j(u)]'$ obtained in Section~\ref{sec:estimation}. Recall that the estimator $\widehat{\beta}^{(m)}_j(u)$ is the minimizer of the SSR of Equation~(\ref{SSR}). Therefore, by setting the first order derivative of SSR with respect to $\beta^{(m)}_j(u)$ to zero, we obtain the following relationships between $\widehat{\delta}^{(m)}_{jt}(u)$ and $\widehat{\beta}^{(m)}_{j}(u)$:
		\begin{align}\label{B.eq1}
			\begin{cases}
				\widehat{\delta}_{jt}^{(0)}(u) = \big(\sum_{s=1}^{S}d_s^2\big)^{-1}\sum_{s=1}^{S}d_s\big(\widehat{\alpha}_{jst}(u) - x_{st}'\widehat{\beta}_{j}^{(0)}(u)\big), \quad t= T_0,...,T,
				\\
				\widehat{\beta}_{j}^{(0)}(u) = \big(\sum_{s=1}^{S}X_s'X_s\big)^{-1}\sum_{s=1}^{S}X_s'\big(\widehat{A}_{js}(u) - D_s\widehat{\delta}_{j}^{(0)}(u) \big),
			\end{cases}
		\end{align}
		and, for $m \geq 1$, given $(\widehat{F}_j^{(m)}(u),\widehat{\Lambda}^{(m)}_j(u))$,  
		{\small
			\begin{gather}\label{B.eq2}
				\begin{cases}
					\widehat{\delta}_{jt}^{(m)}(u) = \big(\sum_{s=1}^{S}d_s^2\big)^{-1}
					\sum_{s=1}^{S}d_s \big(\widehat{\alpha}_{jst}(u) - x_{st}'\widehat{\beta}_{j}^{(m)}(u) - \widehat{f}^{(m)}_{jt}(u)'\widehat{\lambda}^{(m)}_{js}(u)\big),  
					\\
					\widehat{\beta}_{j}^{(m)}(u) = \big(\sum_{s=1}^{S}X_s'X_s\big)^{-1}
					\sum_{s=1}^{S}X_s'\big(\widehat{A}_{js}(u) - D_s\widehat{\delta}_{j}^{(m)}(u) - \widehat{F}^{(m)}_j(u)\widehat{\lambda}^{(m)}_{js}(u) \big).
				\end{cases}
			\end{gather}
		}
		Using the above relationships, we first establish the convergence rate of the regression coefficients in Theorem~\ref{t5}.

		\bigskip
		\noindent
		\textbf{Proof of Theorem~\ref{t5}}
		Since $u$ and $j$ are fixed, we suppress $u$ and $j$ throughout the following proof. 
		The proof is completed by induction. We first show that (i) $\sqrt{ST}\|\widehat{\beta}^{(0)}- {\beta}\| = O_P(1) $, $\sqrt{S}|\widehat{\delta}^{(0)}_{t} - {\delta}_{t} | = O_P(1)$, and $\sqrt{S/T}|\sum_{t=1}^{T_0}\widehat{\delta}^{(0)}_{t} - {\delta}_{t} | = O_P(1)$. Then, using the iterative formula and the rates of the previous estimators, we show that suppose the above rates hold at $(m-1)^{\text{th}}$ iteration, then the rates hold at $m^{\text{th}}$ iteration for all $m \geq 1$.

		\textbf{(i)} We start with $\widehat{\beta}^{(0)} - {\beta}$. By simple algebra on (\ref{B.eq1}) and the linear model (\ref{alpha_st}) for $\alpha_{st}$, we obtain
		\begin{align}\label{t5p.eq1}
			\begin{cases}
				\widehat{\delta}_{t}^{(0)} - \delta_t = \big(\sum_{s=1}^{S}d_s^2\big)^{-1}\sum_{s=1}^{S}d_s\big((\widehat{\alpha}_{st}-\alpha_{st}) - x_{st}'(\widehat{\beta}^{(0)} - \beta) + f_t'\lambda_s + \eta_{st}\big),
				\\
				\widehat{\beta}^{(0)} - \beta = \big(\sum_{s=1}^{S}X_s'X_s\big)^{-1}\sum_{s=1}^{S}X_s'\big((\widehat{A}_{s}-A_s) - D_s(\widehat{\delta}^{(0)}-\delta) + F\lambda_s + \eta_s \big).
			\end{cases}
		\end{align}
		By further substituting the expression of $\widehat{\delta}_{t}^{(0)} - \delta_t$ ($t = T_0,...,T$) into $\widehat{\beta}^{(0)} - \beta$, we obtain
		\begin{align*}
			\Big[ \sum_{s=1}^{S} & X_s'X_s 
			-\Big(\sum_{s=1}^{S}d_s^2\Big)^{-1}\sum_{s,g=1}^{S}X_s'D_sD_g'X_g\Big]
			\big(\widehat{\beta}^{(0)} - \beta \big) \\
			& = \sum_{s=1}^{S}X_s'[(\widehat{A}_s - A_s) + F\lambda_s + \eta_s] 
			- \Big(\sum_{s=1}^{S}d_s^2\Big)^{-1}\sum_{s,g=1}^{S}X_s'D_sD_g'[(\widehat{A}_g - A_g) + F\lambda_g + \eta_g].
		\end{align*}
		
		An application of the triangle inequality together with Lemma~\ref{l1} yields
		$
		\big \|
		\sum_{s=1}^{S}
		X_s'
		\big (
		\widehat{{A}}_{s}
		-
		A_{s}
		\big )
		\big \|
		= O_P((ST)^{1/4}).
		$
		Also,
		$
		(ST)^{-1}
		\sum_{s=1}^{S}  X_s' F \lambda_{s}
		= O_P\big(\sqrt{ST}\big)
		$
		and (\ref{al5.eq2}) in
		Lemma~\ref{al.5}
		shows 
		that
		$
		(ST)^{-1/2}
		\sum_{s=1}^{S}  X_s' \eta_{s}
		=
		O_P(1)$.
		In addition, recall that $D_s = d_s[e_{T_0},...,e_T]$ and $d_s \in \{0,1\}$, we then have
		\begin{align*}
			\Big\|\sum_{s,g=1}^{S}X_s'D_sD_g'(\widehat{A}_g - A_g) \Big\|
			= \sup_{g,t}|\widehat{\alpha}_{gt} - \alpha_{gt}| \cdot \sum_{s,g=1}^{S}\sum_{t=T_0}^{T}\|x_{st}d_sd_g \| = O_P(S^{5/4}T^{1/4}),
		\end{align*}
		following from Lemma~\ref{l1}, 
		\begin{align*}
			\Big\|\sum_{s,g=1}^{S}X_s'D_sD_gF\lambda_g  \Big\|
			\leq \Big(\sum_{t=T_0}^{T}\Big\|\sum_{s=1}^{S}d_sx_{st}\Big\|^2\Big)^{1/2}\Big(\sum_{t=T_0}^{T}\Big\|\sum_{g=1}^{S}d_sf_t'\lambda_g\Big\|^2\Big)^{1/2} = O_P(ST),
		\end{align*}
		due to the assumption that $\mathbb{E}\|\sum_{g=1}^{S}d_sf_t'\lambda_g\|^2 = O_P(S)$ for all $t \geq T_0$ uniformly, and,
		\begin{align*}
			\Big\|\sum_{s,g=1}^{S}X_s'D_sD_g\eta_g\Big\|
			= \Big\|\sum_{s,g=1}^{S}\sum_{t=T_0}^{T}d_s d_g x_{st}\eta_{gt}\Big\| = O_P(S\sqrt{T})
		\end{align*}
		by $\alpha$-mixing. Finally, under Assumption~\ref{ass.reg_1}, we have $\|(\sum_{s=1}^{S}d_s^2)^{-1}\| = O_P(S^{-1}) $ and $\|[\sum_{s=1}^{S}X_s'X_s 
		-(\sum_{s=1}^{S}d_s^2)^{-1}\sum_{s,g=1}^{S}X_s'D_sD_g'X_g]^{-1}\| = O_P((ST)^{-1}) $. Thus, we collect the rates for all terms and obtain $\sqrt{ST}\|\widehat{\beta}^{(0)}_j(u) - {\beta}_j(u)\| = O_P(1) $.
		
		\textbf{(ii)} We then consider $\widehat{\delta}^{(0)}_{t} - {\delta}_{t}$ for a given $t \geq T_0$. We check all the terms of the expression of $\widehat{\delta}^{(0)}_{t} - {\delta}_{t}$ in (\ref{t5p.eq1}) as follows: $|\sum_{s=1}^{S}d_s(\widehat{\alpha}_{st}-\alpha_{st})| \leq \sup_{s,t}|\widehat{\alpha}_{st}-\alpha_{st}|\cdot\sum_{s=1}|d_s| = O_P(S^{1/4}T^{-3/4})$ using Lemma~\ref{l1}, $|\sum_{s=1}^{S}d_sf_t'\lambda_s| = O_P(\sqrt{S})$ by assumption, $|\sum_{s=1}^{S}d_s\eta_{st}| = O_P(\sqrt{S})$ by $\alpha$-mixing, and $|\sum_{s=1}^{S}d_s x_{st}'(\widehat{\beta}^{(0)} - \beta)| \leq   \|\widehat{\beta}^{(0)}_j(u) - {\beta}_j(u)\|\sum_{s=1}^{S} |x_{st}|  = O_P(\sqrt{S/T})$. 
		Finally, since $|(\sum_{s=1}^{S}d_s^2)^{-1}\| = O_P(S^{-1})$ under Assumption~\ref{ass.reg_1}, collecting the rates of all terms yields $\sqrt{S}|\widehat{\delta}^{(0)}_{jt}(u) - {\delta}_{jt}(u)| = O_P(1)$.
		
		In addition, we show that $|\sum_{t=T_0}^{T}\widehat{\delta}^{(0)}_{jt}(u) - {\delta}_{jt}(u)| = O_P(1)$ as follows. Since
		\begin{align*} \label{t5p.eq4}
			\sum_{t=T_0}^{T}\widehat{\delta}^{(0)}_{t} - {\delta}_{t} = 
			\big(\sum_{s=1}^{S}d_s^2\big)^{-1}\sum_{s=1}^{S}\sum_{t=T_0}^{T}d_s\big((\widehat{\alpha}_{st}-\alpha_{st}) - x_{st}'(\widehat{\beta}^{(0)} - \beta) + f_t'\lambda_s + \eta_{st}\big),
		\end{align*}
		applying similar arguments as those in (ii), it is easy to check that $|\sum_{s=1}^{S}\sum_{t=T_0}^{T}d_s(\widehat{\alpha}_{st}-\alpha_{st}) | = O_P((ST)^{1/4})$, $|\sum_{s=1}^{S}\sum_{t=T_0}^{T}d_s x_{st}'(\widehat{\beta}^{(0)} - \beta) | = O_P(\sqrt{ST})$, $|\sum_{s=1}^{S}\sum_{t=T_0}^{T}d_sf_t'\lambda_s| = O_P(\sqrt{ST})$, and $|\sum_{s=1}^{S}\sum_{t=T_0}^{T}d_s\eta_{st}| = O_P(\sqrt{ST})$. Therefore, we have
		\begin{align}
			\bigg|\sum_{t=T_0}^{T}\widehat{\delta}^{(0)}_{t}  - {\delta}_{t}\bigg| 
			& \leq
			O_P(S^{-1})
			\cdot \bigg|\sum_{s=1}^{S}\sum_{t=T_0}^{T}d_s\big((\widehat{\alpha}_{st}-\alpha_{st}) - x_{st}'(\widehat{\beta}^{(0)} - \beta) + f_t'\lambda_s + \eta_{st}\big)\bigg|   \nonumber\\
			& = O_P(1).
		\end{align}
		We also note that a similar argument yields that  $\| \sum_{t=T_0}^{T}\widehat{f}^{(1)}_t(\widehat{\delta}^{(0)}_t - \delta_t)\| = O_P(1) $.
		
		Now, for a given $m\geq 1$, suppose that $\sqrt{ST}\|\widehat{\beta}^{(m-1)}- {\beta}\| = O_P(1) $, $\sqrt{S}|\widehat{\delta}^{(m-1)}_{t} - {\delta}_{t} | = O_P(1)$, $|\sum_{t=T_0}^{T}\widehat{\delta}^{(m-1)}_{t} - {\delta}_{t} | = O_P(1)$ and $\| \sum_{t=T_0}^{T}\widehat{f}^{(m)}_t(\widehat{\delta}^{(m-1)}_t - \delta_t)\| = O_P(1) $.
		
		\textbf{(iii)} We want to show  $\sqrt{ST}\|\widehat{\beta}^{(m)} - {\beta}\| = O_P(1) $. By simple algebra on (\ref{B.eq2}) and the linear model for $\alpha_{st}$, we obtain
		\begin{align} \label{t5p.eq2}
			\Big[ \sum_{s=1}^{S}X_s'X_s 
			& -\Big(\sum_{s=1}^{S}d_s^2\Big)^{-1}\sum_{s,g=1}^{S}X_s'D_sD_g'X_g\Big]
			\big(\widehat{\beta}^{(m)} - \beta \big) \nonumber\\
			= & \sum_{s=1}^{S}X_s'\big[(\widehat{A}_s - A_s) + \big(F\lambda_s-\widehat{F}^{(m)}\widehat{\lambda}^{(m)}_s\big) + \eta_s\big] \nonumber\\
			& - \Big(\sum_{s=1}^{S}d_s^2\Big)^{-1}\sum_{s,g=1}^{S}X_s'D_sD_g \big[(\widehat{A}_g - A_g) + \big(F\lambda_g-\widehat{F}^{(m)}\widehat{\lambda}^{(m)}_g\big) + \eta_g \big], 
		\end{align}
		and 
		\begin{align}\label{t5p.eq3}
			\widehat{\delta}_{t}^{(m)} - \delta_t = \big(\sum_{s=1}^{S}d_s^2\big)^{-1}\sum_{s=1}^{S}d_s\big((\widehat{\alpha}_{st}-\alpha_{st}) - x_{st}'(\widehat{\beta}^{(m)} - \beta) + \big(f_t'\lambda_s - \widehat{f}^{(m)\prime}_t\widehat{\lambda}^{(m)}_s\big) + \eta_{st}\big).
		\end{align}
		
		To find the rate of $\widehat{\beta}^{(m)}_j(u) - {\beta}_j(u)$, we derive the rates for all terms on the right-hand side of (\ref{t5p.eq2}) below. As show in part (i), we have $\|\sum_{s=1}^{S}X_s'[(\widehat{A}_s - A_s)  + \eta_s]\| = O_P(\sqrt{ST})$. We now consider the term $\sum_{s=1}^{S}X_s'(F\lambda_s-\widehat{F}^{(m)}\widehat{\lambda}^{(m)}_s)$. Using the expression that $\widehat{\lambda}^{(m)}_s = T^{-1}\widehat{F}^{(m)\prime}(\widehat{A}_s - D_s\widehat{\delta}^{(m)} - X_s\widehat{\beta}^{(m)})$, we obtain
		\allowdisplaybreaks
		\begin{align}\label{t5p.eq5}
			\sum_{s=1}^{S}X_s'(F\lambda_s-\widehat{F}^{(m)}\widehat{\lambda}^{(m)}_s) 
			= & \sum_{s=1}^{S}X_s'(F-\widehat{F}^{(m)}H^{(m)})\lambda_s + \sum_{s=1}^{S}X_s'\widehat{F}^{(m)}(H^{(m)}\lambda_s - \widehat{\lambda}^{(m)}_s) \nonumber\\
			= 
			& -\sum_{s=1}^{S}X_s'P^{(m)}_{\widehat{F}}(\widehat{A}_s - A_s)
			- \sum_{s=1}^{S}X_s'P^{(m)}_{\widehat{F}}\eta_s \nonumber\\
			& +\sum_{s=1}^{S}X_s'P^{(m)}_{\widehat{F}}X_s(\widehat{\beta}^{(m-1)} - \beta) 
			+ \sum_{s=1}^{S}X_s'P^{(m)}_{\widehat{F}}D_s(\widehat{\delta}^{(m-1)} - \delta) \nonumber\\
			& + \sum_{s=1}^{S}X_s'M^{(m)}_{\widehat{F}}(F-\widehat{F}^{(m)}H^{(m)})\lambda_s, 
		\end{align}
		where 
		\begin{align}
			H_j^{(m)}(u) = &  \widehat{\Upsilon}_j^{(m)}(u)\left({K}_j^{(m)}(u)\right)^{-1}, \label{prop3.eq1}
		\end{align}
		in which $\widehat{\Upsilon}_j^{(m)}(u)$ is the $r \times r$ diagonal matrix with diagonal elements being the $r$ largest eigenvalues of $\widehat{L}_j\big(\widehat{\delta}^{(m-1)}_j(u),\widehat{\beta}^{(m-1)}_j(u)\big)$, defined in (\ref{eq:W}), in descending order, and
		\begin{align*}
			K^{(m)}_j(u) = &  \bigg(\frac{\Lambda_j(u)'\Lambda_j(u)}{S}\bigg)\bigg(\frac{F_j(u)'\widehat{F}^{(m)}_j(u)}{T}\bigg).
		\end{align*}
		Now, we consider each term in the second equation of (\ref{t5p.eq5}). For the first three terms, by H\"{o}lder's and triangle inequalities, it is straightforward to show
		\begin{align*}
			& \Big\|\sum_{s=1}^{S}X_s'P^{(m)}_{\widehat{F}}(\widehat{A}_s - A_s)\Big\| 
			\leq \sup_{s} \|\widehat{A}_s - A_s\| \cdot \sum_{s=1}^{S}\|X_s\| \cdot \|P^{(m)}_{\widehat{F}}\| = O_P((ST)^{1/4}),\\
			& \Big\|\sum_{s=1}^{S}X_s'P^{(m)}_{\widehat{F}}X_s(\widehat{\beta}^{(m-1)} - \beta) \Big\| 
			\leq
			\sum_{s=1}^{S}\|X_s\|^2 \cdot \|P^{(m)}_{\widehat{F}}\| \cdot \|\widehat{\beta}^{(m-1)} - \beta\| = O_P(\sqrt{ST}), 
		\end{align*}
		and $\|\sum_{s=1}^{S}X_s'P^{(m)}_{\widehat{F}}\eta_s\| =  O_P(\sqrt{ST})$ by $\alpha$-mixing property. 
		Using the definition of $D_s$, the forth term is of order $O_P(T)$ as follows
		\allowdisplaybreaks
		\begin{align*}
			\Big\|\sum_{s=1}^{S}X_s'P^{(m)}_{\widehat{F}}D_s(\widehat{\delta}^{(m-1)} - \delta) \Big\| 
			& = \Big\|\frac{1}{T}\sum_{s=1}^{S}X_s'\widehat{F}^{(m)} d_s \sum_{t=T_0}^{T}\widehat{f}^{(m)}_t(\widehat{\delta}^{(m-1)}_t - \delta_t) \Big\| \\
			& \leq \frac{1}{T}\sum_{s=1}^{S}\|X_s\| \cdot \|\widehat{F}^{(m)}\| \cdot \Big\| \sum_{t=T_0}^{T}\widehat{f}^{(m)}_t(\widehat{\delta}^{(m-1)}_t - \delta_t) \Big\| \\
			& = O_P(T),
		\end{align*}
		provided $\| \sum_{t=T_0}^{T}\widehat{f}^{(m)}_t(\widehat{\delta}^{(m-1)}_t - \delta_t)\| = O_P(1) $.
		In addition, Lemma~\ref{t2} yields that the last term 
		\begin{align*}
			\sum_{s=1}^{S}X_s' M_{\widehat{F}}^{(m)} F{\lambda}_s
			= & - \sum_{s=1}^{S}X_s' M_{\widehat{F}}^{(m)} \Big(I_1^{(m)}+I_2^{(m)} + I_3^{(m)}\Big) \widehat{F}^{(m)} \left(K^{(m)}\right)^{-1} \lambda_s \\
			& + O_P (\sqrt{S}) \cdot \bigg|\sum_{t=T_0}^{T} \widehat{\delta}^{(m-1)}_t - \delta_t \bigg|^2 
			+O_P(T\sqrt{S}) \cdot \Big\|\widehat{\beta}^{(m-1)} - \beta\Big\|^2 \nonumber\\
			&
			+ O_P (\sqrt{ST}) \cdot \bigg|\sum_{t=T_0}^{T} \widehat{\delta}^{(m-1)}_t - \delta_t\bigg| \cdot  \Big\|\widehat{\beta}^{(m-1)} - \beta\Big\| \nonumber\\
			&
			+ O_P(\sqrt{T}) \cdot 
			\bigg(\Big\|\widehat{{\beta}}^{(m-1)}-{\beta}\Big\|
			+ \bigg|\sum_{t=T_0}^{T}{\delta}_{t}- \widehat{{\delta}}^{(m-1)}_{t}\bigg|\bigg)+ O_P(\sqrt{ST}),
		\end{align*}
		where, according to the proof of Lemma~\ref{t2}, the leading term associated with $I_1^{(m)}$ is of order $O_P (S) \cdot|\sum_{t=T_0}^{T} \widehat{\delta}^{(m-1)}_t - \delta_t|^2 
		+O_P(ST) \cdot\|\widehat{\beta}^{(m-1)} - \beta\|^2
		+ O_P (S\sqrt{T}) \cdot |\sum_{t=T_0}^{T} \widehat{\delta}^{(m-1)}_t - \delta_t| \cdot  \|\widehat{\beta}^{(m-1)} - \beta\|$, and the leading terms associated with $I_2^{(m)}$ and $I_3^{(m)}$ are of order $O_P(\sqrt{ST}) \cdot \big(\|\widehat{{\beta}}^{(m-1)}-{\beta}\| + |\sum_{t=T_0}^T\widehat{\delta}^{(m)}_t - \delta_t| \big)$.
		Thus, given the rates: $\sqrt{ST}\|\widehat{\beta}^{(m-1)}- {\beta}\| = O_P(1) $, $\sqrt{S}|\widehat{\delta}^{(m-1)}_{t} - {\delta}_{t} | = O_P(1)$, $|\sum_{t=T_0}^{T}\widehat{\delta}^{(m-1)}_{t} - {\delta}_{t} | = O_P(1)$ and $\| \sum_{t=T_0}^{T}\widehat{f}^{(m)}_t(\widehat{\delta}^{(m-1)}_t - \delta_t)\| = O_P(1) $, we conclude that
		\begin{align*}
			\bigg\|\sum_{s=1}^{S}X_s'M^{(m)}_{\widehat{F}}(F-\widehat{F}^{(m)}H^{(m)})\lambda_s\bigg\| = O_P(\sqrt{ST}).
		\end{align*}

		Finally, given the rates of all terms on the right-hand side  of (\ref{t5p.eq5}), we conclude that $\|\sum_{s=1}^{S}X_s'(F\lambda_s -\widehat{F}^{(m)}\widehat{\lambda}^{(m)}_s) \| = O_P(\sqrt{ST})$ since $T/S \to \kappa$.
		
		Collecting all terms so far, we obtain that the first term on the right-hand side of (\ref{t5p.eq2})
		\begin{align*}
			\bigg\| \sum_{s=1}^{S}X_s'\big[(\widehat{A}_s - A_s) + \big(F\lambda_s-\widehat{F}^{(m)}\widehat{\lambda}^{(m)}_s\big) + \eta_s\big] \bigg\| = O_P(\sqrt{ST}),
		\end{align*}
		and similar arguments yield that 
		\begin{align*}
			\bigg\|\sum_{s,g=1}^{S}X_s'D_sD_g \big[(\widehat{A}_g - A_g) + \big(F\lambda_g-\widehat{F}^{(m)}\widehat{\lambda}^{(m)}_g\big) + \eta_g \big] \bigg\| = O_P(S^{3/2}T^{1/2}).
		\end{align*}
		In addition, as $(ST)^{-1}\big[\sum_{s=1}^{S}X_s'X_s 
		-\Big(\sum_{s=1}^{S}d_s^2\Big)^{-1}\sum_{s,g=1}^{S}X_s'D_sD_g'X_g\big]$ and $S^{-1}\sum_{s=1}^{S}d_s^2$ are invertible and their inverse are bounded above according to Assumption~\ref{ass.reg_1}, from (\ref{t5p.eq2}) we obtain that
		\begin{align*} 
			\sqrt{ST}  \big(\widehat{\beta}^{(m)} & - \beta \big) \\
			= &
			\bigg[ \frac{1}{ST}\sum_{s=1}^{S}X_s'X_s -\bigg(\frac{1}{S}\sum_{s=1}^{S}d_s^2\bigg)^{-1} \frac{1}{S^2T}\sum_{s,g=1}^{S}X_s'D_sD_g'X_g\bigg]^{-1} \\
			& \cdot \bigg\{ \frac{1}{\sqrt{ST}}
			\sum_{s=1}^{S}X_s'\big[(\widehat{A}_s - A_s) + \big(F\lambda_s-\widehat{F}^{(m)}\widehat{\lambda}^{(m)}_s\big) + \eta_s\big] \nonumber\\
			& - \Big(\frac{1}{S}\sum_{s=1}^{S}d_s^2\Big)^{-1}\frac{1}{S^{3/2}T^{1/2}}\sum_{s,g=1}^{S}X_s'D_sD_g \big[(\widehat{A}_g - A_g) + \big(F\lambda_g-\widehat{F}^{(m)}\widehat{\lambda}^{(m)}_g\big) + \eta_g \big] \bigg\} \nonumber\\
			=& O_P(1).
		\end{align*}
		
		\textbf{(iv)} Finally, we consider the terms associated with $\widehat{\delta}_{t}^{(m)} - \delta_{t}$ for a given $m \geq 1$. We start by deriving the rates for all terms in (\ref{t5p.eq3}).  Recall from part (ii) that $|\sum_{s=1}^{S}d_s[(\widehat{\alpha}_{st}-\alpha_{st})| = O_P(S^{1/4}T^{-3/4}) $, $|\sum_{s=1}^{S}d_s\eta_{st}| = O_P(\sqrt{S}) $ and $|\sum_{s=1}^{S}d_s x_{st}'(\widehat{\beta}^{(m)} - \beta)| = O_P(\sqrt{S/T})$ given that $\|\widehat{\beta}^{(m)} - \beta\| = O_P(\sqrt{ST})$. Below we consider consider the term $|\sum_{s=1}^{S}d_s(f_t'\lambda_s - \widehat{f}^{(m)\prime}_t\widehat{\lambda}^{(m)}_s)|$. 
		
		Using the estimation expression of $\widehat{\lambda}^{(m)}_s$, we obtain that
		\begin{align}\label{t5p.eq6}
			\sum_{s=1}^{S}d_s(f_t'\lambda_s - \widehat{f}^{(m)\prime}_t\widehat{\lambda}^{(m)}_s)
			= & -\frac{1}{T}\sum_{s=1}^{S}d_s \widehat{f}_t^{(m)\prime}\widehat{F}^{(m)\prime}(\widehat{A}_s - A_s)
			- \frac{1}{T}\sum_{s=1}^{S}d_s \widehat{f}_t^{(m)\prime}\widehat{F}^{(m)\prime}\eta_s \nonumber\\
			& +\frac{1}{T}\sum_{s=1}^{S}d_s\widehat{f}_t^{(m)\prime}\widehat{F}^{(m)\prime}X_s(\widehat{\beta}^{(m-1)} - \beta) \nonumber\\
			& + \frac{1}{T}\sum_{s=1}^{S}d_s \widehat{f}_t^{(m)\prime}\widehat{F}^{(m)\prime}D_s(\widehat{\delta}^{(m-1)} - \delta) \nonumber\\
			& + \sum_{s=1}^{S}d_s e_t'M^{(m)}_{\widehat{F}}(F-\widehat{F}^{(m)}H^{(m)})\lambda_s. 
		\end{align}  
		
		Using triangle- and Cauchy--Schwartz inequalities, it is easy to check the rates of the first four terms of (\ref{t5p.eq6}) as follows
		\begin{align*}
			&\bigg|\frac{1}{T}\sum_{s=1}^{S}d_s \widehat{f}_t^{(m)\prime}\widehat{F}^{(m)\prime}(\widehat{A}_s - A_s)\bigg| \\
			& \qquad \leq 
			\frac{1}{T}\sup_{s}(\widehat{A}_s - A_s) \sum_{s=1}^{S}|d_s| \cdot  \|\widehat{f}_t^{(m)} \| \cdot \|\widehat{F}^{(m)}\| 
			= O_P(S^{-1/4}T^{-3/4}),\\
			& \bigg|\frac{1}{T}\sum_{s=1}^{S}d_s \widehat{f}_t^{(m)\prime}\widehat{F}^{(m)\prime}\eta_s\bigg|
			= \frac{1}{T}\bigg|\sum_{s=1}^{S}\sum_{l=1}^{T}d_s \widehat{f}_t^{(m)\prime}\widehat{f}^{(m)}_l\eta_{sl}\bigg| = O_P\bigg(\sqrt{\frac{S}{T}}\bigg),\\
			& \bigg|\frac{1}{T}\sum_{s=1}^{S}d_s \widehat{f}_t^{(m)\prime}\widehat{F}^{(m)\prime}X_s(\widehat{\beta}^{(m-1)} - \beta)\bigg| \\
			& \qquad 
			\leq \frac{1}{T} \bigg(\sum_{s=1}^{S}|d_s|^{2}\bigg)^{1/2} \cdot  \bigg(\sum_{s=1}^{S}\|X_s\|^{2}\bigg)^{1/2} \cdot  \|\widehat{f}_t^{(m)} \| \cdot \|\widehat{F}^{(m)}\| \cdot \|\widehat{\beta}^{(m-1)} - \beta)\|
			= O_P\bigg(\sqrt{\frac{S}{T}}\bigg),\\
			& \bigg|\frac{1}{T}\sum_{s=1}^{S}d_s \widehat{f}_t^{(m)\prime}\widehat{F}^{(m)\prime}D_s(\widehat{\delta}^{(m-1)} - \delta) \bigg| \\
			& \qquad 
			\leq \frac{1}{T} \sum_{s=1}^{S}|d_s|  \cdot  \|\widehat{f}_t^{(m)} \| \cdot \bigg\| \sum_{l=T_0}^{T}\widehat{f}_l^{(m)}(\widehat{\delta}^{(m-1)}_l - \delta_{l}) \bigg| 
			= O_P\bigg(\frac{S}{T}\bigg),
		\end{align*} 
		given that $\sup_s \widehat{A}_s - A_s = O_P(S^{-3/4}T^{-1/4})$ from Lemma~\ref{l1}, $\| \sum_{l=T_0}^{T}\widehat{f}_l^{(m)}(\widehat{\delta}^{(m-1)}_l - \delta_{l}) | = O_P(1) $ and $\|\widehat{\beta}^{(m-1)} - \beta\| = O_P((ST)^{-1/2})$ from induction assumption. Lastly, according to Lemma~\ref{prop.iter}, we have
		\begin{align*}
			\sum_{s=1}^{S} d_s e_t'M^{(m)}_{\widehat{F}}(F-\widehat{F}^{(m)}H^{(m)})\lambda_s 
			= & 
			\frac{1}{S} \sum_{s,g=1}^{S}\sum_{l=T_0}^{T}\omega_{sg} d_s d_g (\widehat{\delta}^{(m-1)}_t - \delta_t) -  \frac{1}{S} \sum_{s,g=1}^{S}\omega_{sg}d_s\eta_{gt}  \\
			& - \frac{1}{ST}\sum_{s,g=1}^{S}d_s \mathbb{E}[{\eta}_{gt} {\eta}_{g}']  \widehat{F}^{(m)} (K^{(m)})^{-1}\lambda_s 
			+ O_P\bigg(\sqrt{\frac{S}{T}}\bigg),
		\end{align*}
		where the remaining terms on the right-hand side have been shown to be $O_P(\sqrt{S})$ in Lemma~\ref{prop.iter}. 
		Therefore, given the rates of all terms of (\ref{t5p.eq3}), we finally obtain that 
		\begin{align} \label{t5p.eq7}
			\sqrt{S}\big(\widehat{\delta}_{t}^{(m)} - \delta_t\big) 
			= & \bigg(\frac{1}{S}\sum_{s=1}^{S}d_s^2\bigg)^{-1} \cdot \frac{1}{\sqrt{S}}\bigg[ 
			\sum_{s=1}^{S}\Big(d_s - \frac{1}{S}\sum_{g=1}^{S}\omega_{sg}d_g\Big)\eta_{st} \nonumber\\
			& -\frac{1}{S} \sum_{s,g=1}^{S}\sum_{l=T_0}^{T}\omega_{sg} d_s d_g (\widehat{\delta}^{(m-1)}_t - \delta_t)  \nonumber\\
			& - \frac{1}{ST}\sum_{s,g=1}^{S}d_s \mathbb{E}[{\eta}_{gt} {\eta}_{g}']  \widehat{F}^{(m)} (K^{(m)})^{-1}\lambda_s \bigg] + o_P(1),
		\end{align}
		and the leading terms on the right-hand side have shown to be $O_P(1)$, and thus, we conclude that $\sqrt{S}|\widehat{\delta}_{t}^{(m)} - \delta_t| = O_P(1)$.
		
		To complete the proof of induction, it requires to further show
		$|\sum_{t=T_0}^{T}\widehat{\delta}_{t}^{(m)} - \delta_t| = O_P(1)$ and $|\sum_{t=T_0}^{T}\widehat{f}^{(m+1)\prime}_t(\widehat{\delta}_{t}^{(m)} - \delta_t)| = O_P(1)$, which can be proof in the same manner as above, and thus is omitted here. Finally, combining parts (i) to (iv), we complete the proof of Theorem~\ref{t5}.
		\hfill$\square$
		
		\bigskip
		Given the convergence rate, the following proof establishes the limiting distribution for the regression coefficients.
		
		\noindent
		\textbf{Proof of Theorem~\ref{t6}}
		To complete the proof, we first show that, for any given $u \in \mathcal{U}$, $j=1,...,J$ and $t \geq T_0$, the converged estimator $\widehat{\delta}_{jt}(u)$ has an asymptotic linear expansion around the true parameter, and then establish the asymptotic normality for the joint policy parameter $(\widehat{\delta}_t(u_1)',\widehat{\delta}_t(u_2)')'$.

		Recall from the proof of Theorem~\ref{t5}, we derived an asymptotic representation for the iterative estimator $\widehat{\delta}^{(m)}_{jt}(u)$ for $m \geq 1$ as (\ref{t5p.eq7}), where the subscript $j$ and quantile $u$ is omitted. Therefore, when the estimators converges, we have
		\begin{align*}
			\sqrt{S} & \big(\widehat{\delta}_{jt}(u) - \delta_{jt}(u)\big) \\
			= & \bigg(\frac{1}{S}\sum_{s=1}^{S}d_s^2\bigg)^{-1} \cdot \frac{1}{\sqrt{S}}\bigg[ 
			\sum_{s=1}^{S}R_{js}(u)\eta_{jst}(u)
			-\frac{1}{S} \sum_{s,g=1}^{S}\sum_{l=T_0}^{T}\omega_{j,sg}(u) d_s d_g \big(\widehat{\delta}_{jt}(u) - \delta_{jt}(u)\big)  \nonumber\\
			& - \frac{1}{ST}\sum_{s,g=1}^{S}d_s \mathbb{E}[{\eta}_{jgt}(u) {\eta}_{jg}(u)']  \widehat{F}_j(u) \bigg(\frac{F_j(u)'\widehat{F}_j(u)}{T}\bigg)^{-1}\bigg(\frac{\Lambda_j(u)'\Lambda_j(u)}{S}\bigg)^{-1}\lambda_{js}(u) \bigg] + o_P(1),
		\end{align*}
		where $R_{js}(u)$ and $\omega_{j,sg}(u)$ are defined above Assumption~\ref{ass.reg_1}.
		Equivalently, we have
		\begin{align*}
			\sqrt{S} & \big(\widehat{\delta}_{jt}(u) - \delta_{jt}(u)\big) \\
			= & \bigg(\frac{1}{S}\sum_{s=1}^{S} R_{js}(u)^2\bigg)^{-1} \cdot \frac{1}{\sqrt{S}}\bigg[ 
			\sum_{s=1}^{S}R_{js}(u)\eta_{jst}(u)  \nonumber\\
			& - \frac{1}{ST}\sum_{s,g=1}^{S}d_s \mathbb{E}[{\eta}_{jgt}(u) {\eta}_{jg}(u)']  \widehat{F}_j(u) \bigg(\frac{F_j(u)'\widehat{F}_j(u)}{T}\bigg)^{-1}\bigg(\frac{\Lambda_j(u)'\Lambda_j(u)}{S}\bigg)^{-1}\lambda_{js}(u) \bigg] \\
			& + o_P(1),
		\end{align*}
		given the fact that 
		$$\frac{1}{S}\sum_{s=1}^{S}d_s^2 -\frac{1}{S^2}\sum_{s,g=1}^{S}\sum_{l=T_0}^{T}\omega_{j,sg}(u) d_s d_g 
		= \frac{1}{S}\sum_{s=1}^{S} R_{js}(u)^2 >0 $$
		according to the definition of $\omega_{j,sg}(u)$ and Assumption~\ref{ass.reg_1}.(i).
		
		Next, we claim that 
		\begin{align*}
			\frac{1}{S^{3/2}T} 
			&\sum_{s,g=1}^{S}d_s \mathbb{E}[{\eta}_{jgt}(u) {\eta}_{jg}(u)']  \widehat{F}_j(u) \bigg(\frac{F_j(u)'\widehat{F}_j(u)}{T}\bigg)^{-1}\bigg(\frac{\Lambda_j(u)'\Lambda_j(u)}{S}\bigg)^{-1}\lambda_{js}(u) \\
			&\overset{p}{\to}
			\frac{1}{S^{3/2}T}\sum_{s,g=1}^{S}d_s \mathbb{E}[{\eta}_{jgt}(u) {\eta}_{jg}(u)']  F_j(u) \bigg(\frac{\Lambda_j(u)'\Lambda_j(u)}{S}\bigg)^{-1}\lambda_{js}(u). 
		\end{align*}
		To prove this claim, it is sufficient to show that 
		\begin{align*}
			\bigg\| \widehat{F}_j(u) & \bigg(\frac{F_j(u)'\widehat{F}_j(u)}{T}\bigg)^{-1}\bigg(\frac{\Lambda_j(u)'\Lambda_j(u)}{S}\bigg)^{-1} -  F_j(u) \bigg(\frac{\Lambda_j(u)'\Lambda_j(u)}{S}\bigg)^{-1}\bigg\|\\
			& \leq 
			\bigg\| \widehat{F}_j(u)  -  F_j(u)\bigg(\frac{F_j(u)'\widehat{F}_j(u)}{T}\bigg)^{-1} \bigg\| \cdot 
			\bigg\|\bigg(\frac{F_j(u)'\widehat{F}_j(u)}{T}\bigg)^{-1}\bigg\|
			\cdot 
			\bigg\|\bigg(\frac{\Lambda_j(u)'\Lambda_j(u)}{S}\bigg)^{-1} \bigg\| \\
			& \leq
			\bigg\| \widehat{F}_j(u)  -  P_{F_j}(u) \widehat{F}_j(u) \bigg\| \cdot 
			\bigg\|\bigg(\frac{F_j(u)'\widehat{F}_j(u)}{T}\bigg)^{-1}\bigg\|
			\cdot 
			\bigg\|\bigg(\frac{\Lambda_j(u)'\Lambda_j(u)}{S}\bigg)^{-1} \bigg\| \\
			& \leq
			\sqrt{T} \|P_{\widehat{F}_j}(u) - P_{F_j}(u) \| \cdot  \bigg\|\bigg(\frac{F_j(u)'\widehat{F}_j(u)}{T}\bigg)^{-1}\bigg\|
			\cdot 
			\bigg\|\bigg(\frac{\Lambda_j(u)'\Lambda_j(u)}{S}\bigg)^{-1} \bigg\|\\
			& = o_P(\sqrt{T}), 
		\end{align*}
		since $\|P_{\widehat{F}_j}(u) - P_{F_j}(u) \| = o_P(1)$, which can be shown similar to (\ref{al6.eq3}) of Lemma~\ref{al.6}. Thus,
		\begin{align*}
			\bigg\|\frac{1}{S^{3/2}T} 
			&\sum_{s,g=1}^{S}d_s \mathbb{E}[{\eta}_{jgt}(u) {\eta}_{jg}(u)']  \widehat{F}_j(u) \bigg(\frac{F_j(u)'\widehat{F}_j(u)}{T}\bigg)^{-1}\bigg(\frac{\Lambda_j(u)'\Lambda_j(u)}{S}\bigg)^{-1}\lambda_{js}(u) \\
			& \qquad \qquad -
			\frac{1}{S^{3/2}T}\sum_{s,g=1}^{S}d_s \mathbb{E}[{\eta}_{jgt}(u) {\eta}_{jg}(u)']  F_j(u) \bigg(\frac{\Lambda_j(u)'\Lambda_j(u)}{S}\bigg)^{-1}\lambda_{js}(u) \bigg\| \\
			\leq & \frac{1}{S^{3/2}T} \cdot 
			\bigg\|\sum_{s=1}d_s\lambda_{js}(u)\bigg\|
			\cdot 
			\bigg\|\sum_{g=1}^{S}\mathbb{E}[{\eta}_{jgt}(u) {\eta}_{jg}(u)'] \bigg\|\\
			& \cdot 
			\bigg\| \widehat{F}_j(u) \bigg(\frac{F_j(u)'\widehat{F}_j(u)}{T}\bigg)^{-1}\bigg(\frac{\Lambda_j(u)'\Lambda_j(u)}{S}\bigg)^{-1} -  F_j(u) \bigg(\frac{\Lambda_j(u)'\Lambda_j(u)}{S}\bigg)^{-1}\bigg\| \\
			= & \frac{1}{S^{3/2}T} \cdot O_P(S) \cdot O_P(S) \cdot o_P(\sqrt{T})
			= o_P(1).
		\end{align*}
		Therefore, we obtain the asymptotic expansion of $\widehat{\delta}_{jt}(u)$ as 
		\begin{align}\label{t5p.eq8}
			\sqrt{S} & \big(\widehat{\delta}_{jt}(u) - \delta_{jt}(u)\big)  \nonumber\\
			= & - \bigg(\frac{1}{S}\sum_{s=1}^{S} R_{js}(u)^2\bigg)^{-1} \frac{1}{S^{3/2}T}\sum_{s,g=1}^{S}d_s \mathbb{E}[{\eta}_{j,gt}(u) {\eta}_{jg}(u)']  F_j(u) \bigg(\frac{\Lambda_j(u)'\Lambda_j(u)}{S}\bigg)^{-1}\lambda_{js}(u)  \nonumber\\
			& +
			\bigg(\frac{1}{S}\sum_{s=1}^{S} R_{js}(u)^2\bigg)^{-1} 
			\cdot \frac{1}{\sqrt{S}}
			\sum_{s=1}^{S}R_{js}(u)\eta_{jst}(u) + o_P(1),
		\end{align}
		where the first term on the right-hand side has the probability limit $B_{jt}(u)$ defined in Theorem~\ref{t6} 
		
		Combining (\ref{t5p.eq8}) with the definition of $\widehat{\delta}_t(u)$, we obtain that 
		\begin{align*}
			\sqrt{S} & \big(\widehat{\delta}_{t}(u) - \delta_{t}(u)\big) = B_{t}(u) + K_{t}(u) + o_P(1),
		\end{align*}
		where $B_{t}(u)$ is defined in Theorem~\ref{t6}, and $K_{t}(u)$ is defined above Assumption~\ref{ass.clt_delta}. Finally, combining the representation with Assumption~\ref{ass.clt_delta}, we establish the joint CLT in Theorem~\ref{t6}.
		\hfill$\square$


		\bigskip
		
		\noindent
		\textbf{Proof of Corollary \ref{clt_estimate}}
		\textbf{(i)} To show $\widehat{B}_t(u) \pto B_t(u)$, it is sufficient to show that $\widehat{B}_{jt}(u) \pto B_{jt}(u)$ for any given $j=1,...,J $ and $u \in \mathcal{U}$.
		
		We first note that, under the assumption of no cross-sectional dependence, the asymptotic bias $B_{jt}(u)$ is reduced to 
		\begin{align*}
			B_{jt}(u) = 
			\underset{S,T \to \infty}{\plim}
			& - \bigg(\frac{1}{S}\sum_{s=1}^{S} R_{js}(u)^2\bigg)^{-1} \\
			& \cdot \frac{1}{S^{3/2}T}\sum_{s,g=1}^{S}d_s \mathbb{E}[{\eta}_{j,gt}(u)^2]  f_{jt}(u)' \bigg(\frac{\Lambda_j(u)'\Lambda_j(u)}{S}\bigg)^{-1}\lambda_{js}(u).
		\end{align*}
		To shown the consistency of $\widehat{B}_{jt}(u)$, it is sufficient to prove the following two claims:
		\allowdisplaybreaks
		\begin{align}
			& \bigg\| \frac{1}{S} \sum_{s=1}^{S} R_{js}(u)^2 
			- \frac{1}{S}\sum_{s=1}^{S} \widehat{R}_{js}(u)^2\bigg\|= o_P(1), \label{c4.5p_eq0}\\
			& \bigg\| \frac{1}{S^{3/2}T}\sum_{s,g=1}^{S}d_s \mathbb{E}[{\eta}_{j,gt}(u)^2]  f_{jt}(u)' \bigg(\frac{\Lambda_j(u)'\Lambda_j(u)}{S}\bigg)^{-1}\lambda_{js}(u) \nonumber\\
			& \qquad -\frac{1}{S^{3/2}T}\sum_{s,g=1}^{S}d_s \big(\widehat{\eta}_{j,gt}(u)\big)^2 \widehat{f}_{jt}(u)' \bigg(\frac{\widehat{\Lambda}_j(u)'\widehat{\Lambda}_j(u)}{S}\bigg)^{-1}\widehat{\lambda}_{js}(u)\bigg\| = o_P(1). \label{c4.5p_eq00}
		\end{align}
		
		We start with the proof of (\ref{c4.5p_eq0}).
		Using the identity $a^2 - b^2 = (a+b)(a-b)$ and H\"{o}lder's inequality, we obtain that 
		\begin{align} \label{c4.5p_eq1}
			\bigg\| \frac{1}{S}  \sum_{s=1}^{S} R_{js}(u)^2 
			- \frac{1}{S}\sum_{s=1}^{S} \widehat{R}_{js}(u)^2\bigg\| 
			\leq & \frac{1}{S} \bigg[\sum_{s=1}^{S} \bigg\|\frac{1}{S} \sum_{g=1}^{S}\big(\widehat{\omega}_{j,sg}(u) + \omega_{j,sg}(u)\big)d_g - 2d_s  \bigg\|^2\bigg]^{1/2}  \nonumber\\
			& \cdot  \bigg[\sum_{s=1}^{S} \bigg\|\frac{1}{S} \sum_{g=1}^{S}\big(\widehat{\omega}_{j,sg}(u) - \omega_{j,sg}(u)\big)d_g  \bigg\|^2\bigg]^{1/2}.
		\end{align}
		It is straightforward to show that
		{\small
			\begin{align}\label{c4.5p_eq2}
				\bigg[\sum_{s=1}^{S}& \bigg\|\frac{1}{S} \sum_{g=1}^{S}\big(\widehat{\omega}_{j,sg}(u) + \omega_{j,sg}(u)\big)d_g - 2d_s  \bigg\|^2 \bigg]^{1/2} \nonumber\\
				= & \frac{1}{S}\bigg\|\sum_{g=1}^{S}\bigg[\widehat{\Lambda}_{j}(u) \bigg(\frac{\widehat{\Lambda}_{j}(u)'\widehat{\Lambda}_{j}(u)}{S}\bigg)^{-1}\widehat{\lambda}_{jg}(u)+ {\Lambda}_{j}(u) \bigg(\frac{{\Lambda}_{j}(u)'{\Lambda}_{j}(u)}{S}\bigg)^{-1}{\lambda}_{jg}(u)\bigg]d_g - 2D_s  \bigg\| \nonumber\\
				= & \frac{1}{S}\bigg\{\bigg\|\widehat{\Lambda}_{j}(u) \bigg(\frac{\widehat{\Lambda}_{j}(u)'\widehat{\Lambda}_{j}(u)}{S}\bigg)^{-1}\sum_{g=1}^{S}\widehat{\lambda}_{jg}(u)d_g\bigg\| 
				+ \bigg\|{\Lambda}_{j}(u) \bigg(\frac{{\Lambda}_{j}(u)'{\Lambda}_{j}(u)}{S}\bigg)^{-1}\sum_{g=1}^{S}{\lambda}_{jg}(u)d_g \bigg\| \bigg\} \nonumber\\
				= & O_P(\sqrt{S}), 
		\end{align}}
		since
		\begin{align*}
			\bigg\|\Lambda_{j}(u) \bigg(\frac{\Lambda_{j}(u)'\Lambda_{j}(u)}{S}\bigg)^{-1}\sum_{g=1}^{S}{\lambda}_{jg}(u)d_g \bigg\| 
			& \leq \|\Lambda_{j}(u)\| \cdot 
			\bigg\|\bigg(\frac{\Lambda_{j}(u)'\Lambda_{j}(u)}{S}\bigg)^{-1}\bigg\| 
			\cdot \Big\|\sum_{g=1}^{S}{\lambda}_{jg}(u)d_g \Big\| \\
			& = O_P(S^{3/2})
		\end{align*}
		and the same rate holds for the first term in the second last equation.
		
		Given the expression of $\omega_{j,sg}(u)$ and the identity that $\widehat{a}\widehat{b}\widehat{c} - abc  = (\widehat{a} - a)\widehat{b}\widehat{c} + a(\widehat{b}-b)\widehat{c} + ab(\widehat{c}- c)$, we write
		\begin{align*}
			\bigg[\sum_{s=1}^{S} & \bigg\|\frac{1}{S} \sum_{g=1}^{S}\big(\widehat{\omega}_{j,sg}(u) + \omega_{j,sg}(u)\big)d_g - 2d_s  \bigg\|^2\bigg]^{1/2} \\
			= & \frac{1}{S}\bigg\|\sum_{g=1}^{S}\bigg[\widehat{\Lambda}_{j}(u) \bigg(\frac{\widehat{\Lambda}_{j}(u)'\widehat{\Lambda}_{j}(u)}{S}\bigg)^{-1}\widehat{\lambda}_{jg}(u)
			- {\Lambda}_{j}(u) \bigg(\frac{{\Lambda}_{j}(u)'{\Lambda}_{j}(u)}{S}\bigg)^{-1}{\lambda}_{jg}(u)\bigg]d_g  \bigg\| \\
			\leq & \frac{1}{S}\bigg\{
			\bigg\|\Big(\widehat{\Lambda}_{j}(u) - \Lambda_j(u)H_j(u)'\Big) \bigg(\frac{\widehat{\Lambda}_{j}(u)'\widehat{\Lambda}_{j}(u)}{S}\bigg)^{-1}\sum_{g=1}^{S}\widehat{\lambda}_{jg}(u)d_g \bigg\| \\
			& + \bigg\|{\Lambda}_{j}(u)H_j(u)'\bigg[\bigg(\frac{\widehat{\Lambda}_{j}(u)'\widehat{\Lambda}_{j}(u)}{S}\bigg)^{-1} \\
			& \qquad \qquad - \big(H_j(u)'\big)^{-1}\bigg(\frac{{\Lambda}_{j}(u)'{\Lambda}_{j}(u)}{S}\bigg)^{-1} \big(H_j(u)\big)^{-1}\bigg]\sum_{g=1}^{S}\widehat{\lambda}_{jg}(u)d_g   \bigg\|  \\
			& + \bigg\|{\Lambda}_{j}(u) \bigg(\frac{{\Lambda}_{j}(u)'{\Lambda}_{j}(u)}{S}\bigg)^{-1}\big(H_j(u)\big)^{-1}\bigg[\sum_{g=1}^{S}\widehat{\lambda}_{jg}(u) - H_j(u){\lambda}_{jg}(u)\bigg]d_g\bigg\|\bigg\}.
		\end{align*}
		Combining (\ref{fl1.eq3}) of Lemma~\ref{factor.l1} with (\ref{al6.eq1}) and Theorem~\ref{t5}, we obtain that $ \|\widehat{\Lambda}_{j}(u) - \Lambda_j(u)H_j(u)'\| = O_P(1)$, and 
		$$\Big\|\big(S^{-1}\widehat{\Lambda}_{j}(u)'\widehat{\Lambda}_{j}(u)\big)^{-1} - \big(H_j(u)'\big)^{-1}\big(S^{-1}{\Lambda}_{j}(u)'{\Lambda}_{j}(u)\big)^{-1} (H_j(u)\big)^{-1} \Big\|  = O_P(\delta_{ST}^{-1}).$$ Thus,
		by Cauchy--Schwartz inequality, we obtain that the first term in the last inequality is
		\begin{align*}
			\bigg\|\Big(\widehat{\Lambda}_{j}(u) & - \Lambda_j(u)H_j(u)'\Big) \bigg(\frac{\widehat{\Lambda}_{j}(u)'\widehat{\Lambda}_{j}(u)}{S}\bigg)^{-1}\sum_{g=1}^{S}\widehat{\lambda}_{jg}(u)d_g \bigg\| \\
			= & \Big\|\widehat{\Lambda}_{j}(u) - \Lambda_j(u)H_j(u)'\Big\| \cdot \bigg\| \bigg(\frac{\widehat{\Lambda}_{j}(u)'\widehat{\Lambda}_{j}(u)}{S}\bigg)^{-1}
			\bigg\| \cdot \bigg\|\sum_{g=1}^{S}\widehat{\lambda}_{jg}(u)d_g \bigg\| = O_P(S),
		\end{align*}
		and the same rate applies to the second term. For the third term,
		\begin{align*}
			\bigg\| & {\Lambda}_{j}(u)  \bigg(\frac{{\Lambda}_{j}(u)'{\Lambda}_{j}(u)}{S}\bigg)^{-1}\big(H_j(u)\big)^{-1}\bigg[\sum_{g=1}^{S}\widehat{\lambda}_{jg}(u) - H_j(u){\lambda}_{jg}(u)\bigg]d_g\bigg\| \\
			& \leq
			\|{\Lambda}_{j}(u)\|  \cdot 
			\bigg\|\bigg(\frac{{\Lambda}_{j}(u)'{\Lambda}_{j}(u)}{S}\bigg)^{-1}\bigg\|
			\cdot \Big\|\big(H_j(u)\big)^{-1}\Big\|  \\
			& \qquad \qquad \cdot \bigg(\sum_{g=1}^{S}\|\widehat{\lambda}_{jg}(u) - H_j(u){\lambda}_{jg}(u)\|^2 \bigg)^{1/2}
			\cdot \bigg(\sum_{g=1}^{S} |d_g|^{2} \bigg)^{1/2} \\
			& \leq \|{\Lambda}_{j}(u)\|  \cdot 
			\bigg\|\bigg(\frac{{\Lambda}_{j}(u)'{\Lambda}_{j}(u)}{S}\bigg)^{-1}\bigg\|
			\cdot \Big\|\big(H_j(u)\big)^{-1}\Big\| 
			\cdot \Big\|\widehat{\Lambda}_{j}(u) - {\Lambda}_{j}(u)H_j(u)'\Big\|^2
			\cdot \bigg(\sum_{g=1}^{S} |d_g|^{2} \bigg)^{1/2} \\
			& = O_P(S).
		\end{align*}
		Collecting all terms, we have
		\begin{align*}
			\bigg[\sum_{s=1}^{S} & \bigg\|\frac{1}{S} \sum_{g=1}^{S}\big(\widehat{\omega}_{j,sg}(u) + \omega_{j,sg}(u)\big)d_g - 2d_s  \bigg\|^2\bigg]^{1/2} = O_P(1).
		\end{align*}
		And thus, together with (\ref{c4.5p_eq1}) and (\ref{c4.5p_eq2}), we prove the claim (\ref{c4.5p_eq0}).
		
		To prove the second claim (\ref{c4.5p_eq00}), we consider two terms
		{\small
			\begin{align*}
				\frac{1}{S^{3/2}T}\sum_{s,g=1}^{S}d_s \mathbb{E}[{\eta}_{j,gt}(u)^2] \bigg(f_{jt}(u)' \bigg(\frac{\Lambda_j(u)'\Lambda_j(u)}{S}\bigg)^{-1}\lambda_{js}(u) - \widehat{f}_{jt}(u)' \bigg(\frac{\widehat{\Lambda}_j(u)'\widehat{\Lambda}_j(u)}{S}\bigg)^{-1}\widehat{\lambda}_{js}(u)
				\bigg) ,
		\end{align*}}
		and
		\begin{align*}
			\frac{1}{S^{3/2}T}\sum_{s,g=1}^{S}d_s \Big(\mathbb{E}[{\eta}_{j,gt}(u)^2] - \widehat{\eta}_{j,gt}(u)^2 \Big)
			\widehat{f}_{jt}(u)' \bigg(\frac{\widehat{\Lambda}_j(u)'\widehat{\Lambda}_j(u)}{S}\bigg)^{-1}\widehat{\lambda}_{js}(u).
		\end{align*}
		The first term is $o_P(1)$ by expanding the terms related to factors and loadings and apply Lemma~\ref{al.6} and \ref{factor.l1}. For the second term, it is easy to show that $\sum_{s=1}^{S}\big(\eta_{j,st}(u)^2 - \widehat{\eta}^{(m)}_{j,st}(u)^2\big)
		= O_P(\sqrt{S})$
		and 
		$\sum_{s=1}^{T}\E[\eta_{j,st}(u)^2] - \eta_{j,st}(u)^2 = O_P(\sqrt{S})$, which leads to 
		{\small
			\begin{align*}
				\bigg\|\frac{1}{S^{3/2}T} & \sum_{s,g=1}^{S}d_s \Big(\mathbb{E}[{\eta}_{j,gt}(u)^2] - \widehat{\eta}_{j,gt}(u)^2 \Big)
				\widehat{f}_{jt}(u)' \bigg(\frac{\widehat{\Lambda}_j(u)'\widehat{\Lambda}_j(u)}{S}\bigg)^{-1}\widehat{\lambda}_{js}(u) \bigg\| \\
				\leq & 
				\frac{1}{S^{3/2}T}
				\cdot \bigg\|\sum_{g=1}^{S} \mathbb{E}[{\eta}_{j,gt}(u)^2] - \widehat{\eta}_{j,gt}(u)^2\bigg\|
				\cdot 
				\|\widehat{f}_{jt}(u)\| 
				\cdot
				\Bigg\|\bigg(\frac{\widehat{\Lambda}_j(u)'\widehat{\Lambda}_j(u)}{S}\bigg)^{-1}\Bigg\|
				\cdot \Big\|\sum_{s=1}^{S}d_s\widehat{\lambda}_{js}(u)\Big\|\\
				= & o_P(1).
		\end{align*}}
		Then, given these two claims, it is straightforward that $\widehat{B}_t(u) \pto B_t(u)$.
		
		\textbf{(ii)} 
		We first note that under the assumption of no cross-sectional correlation, the $(j,k)^{\text{th}}$ entry of $\Sigma_{t}(u_1,u_2)$ is given by
		\begin{align*}
			{\sigma}_{t,jk}(u_1,u_2) = \plim_{S,T \to \infty}
			& \bigg(\frac{1}{S}\sum_{s=1}^{S} R_{js}(u_1)^2\bigg)^{-1} 
			\bigg(\frac{1}{S}\sum_{s=1}^{S} R_{ks}(u_2)^2\bigg)^{-1} \\
			&
			\cdot \frac{1}{S}
			\sum_{s=1}^{S}R_{js}(u_1)R_{ks}(u_2)\eta_{jst}(u_1)\eta_{kst}(u_2).
		\end{align*}
		Thus, to show $\widehat{\Sigma}_t(u_1,u_2) \pto \Sigma_t(u_1,u_2) $, it is sufficient to show $\widehat{\sigma}_{t,jk}(u_1,u_2) \pto {\sigma}_{t,jk}(u_1,u_2)$, where $\widehat{\sigma}_{t,jk}(u_1,u_2) $ is the $(j,k)^{\text{th}}$ entry of $\widehat{\Sigma}_{t}(u_1,u_2)$.
		Given (\ref{c4.5p_eq0}), it remains to show 
		\begin{align}\label{c4.5p_eq4}
			\frac{1}{S}
			\sum_{s=1}^{S} & \widehat{R}_{js}(u_1)
			\widehat{R}_{ks}(u_2) \widehat{\eta}_{j,st}(u_1)\widehat{\eta}_{k,st}(u_2) 
			\pto
			\frac{1}{S}
			\sum_{s=1}^{S}\R_{js}(u_1)R_{ks}(u_2)\eta_{j,st}(u_1)\eta_{k,st}(u_2). 
		\end{align}
		Using the identity that $\widehat{a}\widehat{b}- ab  = (\widehat{a} - a)\widehat{b} + a(\widehat{b}-b)$, we first note that
		\begin{align*}
			\bigg|\frac{1}{S} 
			\sum_{s=1}^{S} & \Big[\widehat{R}_{js}(u_1)\widehat{R}_{ks}(u_2) - R_{js}(u_1)R_{ks}(u_2)\Big]
			\eta_{jst}(u_1)\eta_{kst}(u_2)\bigg| =o_P(1),
		\end{align*}
		whose proof is similar to (\ref{c4.5p_eq0}).
		It remains to consider the term
		\begin{align*}
			\bigg|\frac{1}{S} & 
			\sum_{s=1}^{S} R_{js}(u_1)R_{ks}(u_2) \big( \widehat{\eta}_{jst}(u_1)\widehat{\eta}_{kst}(u_2)- \eta_{jst}(u_1)\eta_{kst}(u_2) \big)\bigg|.
		\end{align*}
		From 
		\begin{align*}
			\bigg|\frac{1}{S} & 
			\sum_{s=1}^{S} d_s^2 \big( \widehat{\eta}_{jst}(u_1)\widehat{\eta}_{k,st}(u_2)- \eta_{jst}(u_1)\eta_{kst}(u_2) \big)\bigg| \\
			\leq & 
			\bigg|\frac{1}{S} 
			\sum_{s=1}^{S} d_s \big( \widehat{\eta}_{jst}(u_1)- \eta_{jst}(u_1) \big)\widehat{\eta}_{kst}(u_2)\bigg| 
			+
			\bigg|\frac{1}{S} 
			\sum_{s=1}^{S} d_s \big( \widehat{\eta}_{kst}(u_2)- \eta_{k,st}(u_2) \big)\eta_{jst}(u_1)\bigg|,
		\end{align*}
		and 
		the expression of $\widehat{\eta}_{jst}(u)$, 
		we write
		{\small
			\begin{align*}
				\bigg| 
				\sum_{s=1}^{S} & d_s \big( \widehat{\eta}_{kst}(u_2)- \eta_{kst}(u_2) \big)\eta_{j,st}(u_1)\bigg| \\
				\leq &  
				\bigg|\sum_{s=1}^{S} d_s \eta_{jst}(u_1) (\widehat{\alpha}_{jst}(u) - \alpha_{jst}(u))\bigg| 
				+
				\bigg|\sum_{s=1}^{S} d_s \eta_{jst}(u_1)d_s(\delta_{jt}(u) - \widehat{\delta}_{jt}(u))\bigg|\\
				& +
				\bigg|\sum_{s=1}^{S} d_s \eta_{jst}(u_1) x_{st}'(\beta_{j}(u) - \widehat{\beta}_{j}(u)) \bigg|
				+
				\bigg|\sum_{s=1}^{S} d_s \eta_{jst}(u_1) (f_{jt}(u)'\lambda_{js}(u) - \widehat{f}_{jt}(u)'\widehat{\lambda}_{js}(u)])\bigg|\\
				= &  o_P(1),
		\end{align*}}
		whose proof are similar to part (iv) in the proof of Theorem~\ref{t5}. Given a similar reason, we also have 
		$|S^{-1}\sum_{s=1}^{S} d_s ( \widehat{\eta}_{jst}(u_1)- \eta_{jst}(u_1) )\widehat{\eta}_{kst}(u_2)|=o_P(1)$. In addition, replacing $d_s$ with $S^{-1}\sum_{g=1}^{S}\omega_{j,sg}(u)d_g$ does not affect the convergence rate. Therefore, we obtain that
		\begin{align*}
			\bigg|\frac{1}{S} & 
			\sum_{s=1}^{S} R_{js}(u_1)R_{ks}(u_2) \big( \widehat{\eta}_{jst}(u_1)\widehat{\eta}_{kst}(u_2)- \eta_{jst}(u_1)\eta_{kst}(u_2) \big)\bigg| = o_P(1),
		\end{align*}
		which completes the proof of (\ref{c4.5p_eq4}). Thus, the proof of (ii) is complete.
		\hfill$\square$

		\subsection{Proof of Results on Treatment Parameters}\label{sec:identification_proof}

		In this section, we prove the theorems related to the treatment parameters, discussed in Section~\ref{identification}.
		
		\bigskip
		\noindent
		\textbf{Proof of Theorem \ref{theorem:identification}}
		Let $u \in \mathcal{U}$, $z \in \mathcal{Z}$ and $t \geq T_0$
		be fixed.
		Under the potential outcome framework (\ref{Q_std})-(\ref{eq:delta1}),
		we can write 
		$\Delta_{t}^{AQTT}(u|z) =
		z'
		\E[
		\alpha_{st}^{1}(u)
		-
		\alpha_{st}^{0}(u)
		|
		d_{s}=1
		]
		$.
		Under Assumption~\ref{a2.3}.(i),
		it follows from (\ref{eq:delta1}) that
		$
		\E[
		\alpha_{jst}^{1}(u)
		-
		\alpha_{jst}^{0}(u)
		|
		d_{s}=1
		]
		=
		\E[
		\Delta_{jst}(u)
		|
		d_{s}=1
		]
		$, for $j=1,...,J$.
		Letting $\Delta_{st}(u) := [\Delta_{1st}(u),...,\Delta_{Jst}(u)]'$,
		we can show
		\begin{eqnarray*}
			\Delta_{t}^{AQTT}(u|z) &=& z'\E[  \Delta_{st}(u)  |  d_{s}=1  ], \\
			\dt{\Delta}_t^{W}(u_1,u_2|z) &=& z' \E[  \Delta_{st}(u_2) -  \Delta_{st}(u_1)  |  d_{s}=1 ], \\
			\dt{\Delta}_t^{B}(u|z_1,z_2) &=& (z_2- z_1)' \E[  \Delta_{st}(u)  |  d_{s}=1 ].
		\end{eqnarray*}
		Moreover, 
		we have
		$\alpha_{jst}(u) =
		(1 - d_{st}) \alpha_{jst}^{0}(u) + d_{st} \alpha_{jst}^{1}(u)$
		under the potential outcome framework.
		Therefore,
		we can rewrite (\ref{eq:delta1}) as
		\begin{align}
			\label{indentification_proof.eq1}
			\alpha_{jst}(u) = d_{st}\E[  \Delta_{jst}(u)  |  d_{s}=1  ] + x_{st}'\beta_j(u) + f_{jt}(u)'\lambda_{js}(u) + \eta_{jst}(u)
		\end{align}
		where
		$\eta_{jst}(u) := d_{st}(\Delta_{jst}(u) - \E[  \Delta_{jst}(u)  |  d_{s}=1  ]) + (1-d_{st})\eta^{0}_{jst}(u) + d_{st} \eta^{1}_{jst}(u)$. Then, it is clear that (\ref{indentification_proof.eq1}) coincides with (\ref{alpha_st}) by defining $\delta_{jt}(u) :=  \E[  \Delta_{jst}(u)  |  d_{s}=1  ]$.
		
		To prove the theorem, it remains to check that 
		$\delta_{jt}(u)  $
		can be identified from model (\ref{Q_y})-(\ref{alpha_st})
		for each $j = 1, \dots, J$
		and $t \ge T_{0}$.
		
		It is know that, for model (\ref{Q_y})-(\ref{alpha_st}), under Assumption~\ref{a2.1}, we can estimate the quantile regression coefficients $\alpha_{jst}(u)$ from model (\ref{Q_y}). Also, the
		factors and factor loadings are identifiable under Assumption~\ref{a2.2} using the argument
		of
		\cite{bai2009panel}.
		Thus, we treat the quantile regression coefficients, factors and loadings as known objects
		in the rest of the proof.
		Then, since $\E[\eta_{jst}(u)| d_{st},X_s] = 0$ under Assumption~\ref{a2.3}, taking the expectation on both sides of (\ref{alpha_st}) leads to the normal equations
		\begin{align}
			& \E[d_{st}(\alpha_{jst} -f_{jt}(u)'\lambda_{js}(u))] = \E[d_{st}^2]\delta_{jt}(u) + \E[d_{st}x_{st}']\beta_{j}(u) , \label{indentification_proof.eq2}\\
			& \E[x_{st}(\alpha_{jst} -f_{jt}(u)'\lambda_{js}(u))] = \E[d_{st}x_{st}]\delta_{jt}(u) + \E[x_{st}x_{st}']\beta_{j}(u) . \label{indentification_proof.eq3}
		\end{align}
		Solving (\ref{indentification_proof.eq3}) with respect to $\beta_{j}(u)$, we have 
		$$\beta_{j}(u) = \E[x_{st}x_{st}']^{-1} \{\E[x_{st}(\alpha_{jst} -f_{jt}(u)'\lambda_{js}(u))] - \E[d_{st}x_{st}]\delta_{jt}(u)\}$$ given that $\E[x_{st}x_{st}']^{-1} $ is invertible under Assumption~\ref{a2.2}.(ii). 
		Substituting the solution into (\ref{indentification_proof.eq2}), we obtain
		\begin{align*}
			\delta_{jt}(u) = \E[d_{st}\Pi_{st}]^{-1}\E[\Pi_{st}(\alpha_{jst}(u) -  f_{jt}(u)'\lambda_{js}(u))],
		\end{align*}
		where $\Pi_{st} := d_{st} - \E[d_{st}x_{st}']\E[x_{st}x_{st}']^{-1}x_{st} $. Moreover, since $\text{Var}(x_{st}) = \E[x_{st}x_{st}']- \E[x_{st}]\E[x_{st}'] >0$, it follows that 
		$$
		\E[d_{st}\Pi_{st}] = \E[\E[d_{st}\Pi_{st}] | d_{st} ] 
		= \mathbb{P}(d_{st}= 1 ) (1 - \E[x_{st}']\E[x_{st}x_{st}']^{-1}\E[x_{st} ] ) >0, 
		$$
		that is,  
		$\E[d_{st}\Pi_{st}]$ is invertible.
		
		Because $\delta_{jt}(u)$ is identifiable in model (\ref{Q_y})-(\ref{alpha_st}), the policy parameters $\Delta_{t}^{AQTT}(u|z)$, $\dt{\Delta}_{t}^{B}(u|z_1,z_2)$ and $\dt{\Delta}_t^{W}(u_1,u_2|z)$ are identifiable from the arguments at the beginning of the proof.
		\hfill$\square$
		\vspace{0.5cm}
		
		\bigskip\noindent
		\textbf{Proof of Theorem~\ref{CLT_inequality}}
		Recall from Theorem~\ref{theorem:identification} that 
		$\Delta_{t}^{AQTT}(u|z)=z'\delta_{\cdot t}(u)$, which is a linear combination of $\delta_{\cdot t}(u)$. Thus, following the joint-CLT result in Theorem~\ref{t6}, 
		it is straightforward to check that, given $z$, the asymptotic bias and variance  are 
		\begin{align*}
			\lim_{S,T\to \infty}\E\Big[\sqrt{S}  \Big (    \widehat{{\Delta}}^{AQTT}_t(u|z)  - {\Delta}^{AQTT}_t(u|z)  \Big )\Big]
			= & \lim_{S,T\to \infty}\E\Big[\sqrt{S}z'\Big(\widehat{\delta}_{\cdot t}(u) -\delta_{\cdot t}(u) \Big)\Big] \\
			= & z'B_t(u),\\
			\lim_{S,T \to \infty}\text{Var}\Big[\sqrt{S}  \Big (    \widehat{{\Delta}}^{AQTT}_t(u|z)  - {\Delta}^{AQTT}_t(u|z)\Big)\Big]
			= & \lim_{S,T\to \infty}\text{Var}\Big[\sqrt{S}z'\Big(\widehat{\delta}_{\cdot t}(u) -\delta_{\cdot t}(u) \Big)\Big] \\
			= & z'\Sigma_t(u,u)z.
		\end{align*}
		As $\widehat{\dt{\Delta}}_{t}{\hspace{-0.1cm}}^{B}(u| z_1,z_2) $ and $\widehat{\dt{\Delta}}_{t}{\hspace{-0.1cm}}^{W}(u_{1}, u_{2}| z)$ are also linear combinations of $\delta_{\cdot t}$, following similar arguments, we establish the corresponding claims in Theorem~\ref{CLT_inequality}.
		\hfill$\square$

		\section{Technical Details for the Main Proofs}\label{online.theorem}
		
		In this section, we provide the techical details necessary for the proofs in Section~\ref{sec:proof}.

		\begin{lemma}
			\label{l1}
			Under Assumption~\ref{a2.1}, \ref{ass.dens} and \ref{ass.growth}, for fixed $u \in \mathcal{U}$, and all $(s,t) \in \{1,...,S\} \times \{1,...,T\}$,
			\begin{align*}
				\underset{s,t}{\sup}\|\widehat{\alpha}_{st}(u) - {\alpha}_{st}(u)\|
				= O_P\left((ST)^{-3/4}\right).
			\end{align*}
			
		\end{lemma}
		
		\noindent
		\textbf{Proof}
		To prove the claim, it is enough to show that, for some $c>0$, there exists a fixed and sufficiently large $M>0$, such that 
		\begin{align*}
			\mathbb{P}\bigg(\underset{s,t}{\sup}\|\widehat{\alpha}_{st}(u) - {\alpha}_{st}(u)\|>\frac{M}{(ST)^{3/4}}\bigg) \rightarrow 0.
		\end{align*}

		Under Assumption~\ref{ass.growth} that $(ST)^{3/4}(\ln(N_{\min})/N_{\min})^{1/2} \leq C_M$, we choose $M$ such that $M>\big(2C_M^2/(3c)\big)^{1/2}$, where $c$ is the constant in Lemma~\ref{bl.nonasymptotic} independent of $S,T,N_{\min}$. Then,
		we have
		\begin{align*}
			\mathbb{P}\bigg(\underset{s,t}{\sup}\|\widehat{\alpha}_{st}(u) - {\alpha}_{st}(u)\|>\frac{M}{(ST)^{3/4}}\bigg)
			& \leq 
			\sum_{s=1}^{S}\sum_{t=1}^{T}\mathbb{P}\bigg(\|\widehat{\alpha}_{st}(u) - {\alpha}_{st}(u)\|>\frac{M}{(ST)^{3/4}}\bigg)\\
			& \leq 
			\sum_{s=1}^{S}\sum_{t=1}^{T}Ce^{-cM^2N_{st}/(ST)^{3/2}}\\
			& 
			\leq CSTe^{-cM^2N_{\min}/(ST)^{3/2}}\\
			& 
			\leq
			C \cdot \bigg(\frac{C_M^2 N_{\min} }{\ln(N_{\min})}\bigg)^{2/3}N_{\min}^{-cM^2/C_M^{2}}\\
			&
			\rightarrow 0,
		\end{align*}
		where the second inequality follows from the non-asymptotic upper bound given in Lemma~\ref{bl.nonasymptotic}, the last inequality follows from Assumption~\ref{ass.growth}, which converges to $0$ since $2/3-cM^2/C_M^2 <0$.	
		\hfill$\square$

		\begin{lemma}\label{al.5}
			Under Assumption \ref{ass.cov} and \ref{ass.mixing}, we have the following estimations, for any fixed $u \in \mathcal{U}$, $k=1,...,K$ and $j=1,...,J$,
			\begin{align}
				& \left\|\frac{1}{S T} \sum_{s=1}^{S} {\eta}_{js}(u) {\eta}_{js}'(u)\right\|
				=O_P\left(\max \left(\frac{1}{\sqrt{S}}, \frac{1}{\sqrt{T}}\right)\right), \label{al5.eq1}\\
				& \left\|\frac{1}{S T} \sum_{s=1}^{S} {x}_{s(k)} {\eta}_{js}'(u)\right\|=O_P\left(\frac{1}{\sqrt{ST}}\right), \label{al5.eq2}\\
				& \left\|\frac{1}{S T} \sum_{s=1}^{S}  F_{j}(u){\lambda}_{js}(u) {\eta}_{js}'(u)\right\|=O_P\left(\frac{1}{\sqrt{S}}\right). \label{al5.eq3}
			\end{align}
		\end{lemma}

		\noindent
		\textbf{Proof}
		Since $u$ and $j$ are fixed, we suppress $u$ and $j$ throughout the following proof.
		For (\ref{al5.eq1}),
		\begin{align*}
			\mathbb{E} \bigg[\bigg\|
			& \frac{1}{S T} \sum_{s=1}^{S} {\eta}_{s} {\eta}_{s}^{\prime}\bigg\|^{2} \bigg]
			\leq  \mathbb{E}\left[\frac{1}{S^{2} T^{2}} \sum_{s, g=1}^{S} \sum_{t, l=1}^{T} \eta_{s t} \eta_{g t} \eta_{s l} \eta_{g l}\right] \\
			=& \frac{1}{S^{2} T^{2}} \sum_{s=1}^{S} \sum_{t, l=1}^{T} \mathbb{E}\left[\eta_{s t}^{2} \eta_{s l}^{2}\right]
			+\frac{1}{S^{2} T^{2}} \sum_{s \neq g}^{S} \sum_{t=1}^{T} \mathbb{E}\left[\eta_{s t}^{2} \eta_{g t}^{2}\right]
			+\frac{1}{S^{2} T^{2}} \sum_{s \neq g}^{S} \sum_{t \neq l}^{T} \mathbb{E}\left[\eta_{s t} \eta_{g t} \eta_{s l} \eta_{g l}\right] \\
			=& O\left(\frac{1}{S}\right) 
			+O\left(\frac{1}{T}\right)
			+O\left(\frac{1}{ST}\right) 
			= O\left(\max \left(\frac{1}{S}, \frac{1}{T}\right)\right) ,
		\end{align*}
		where the second last equality is obtained using Assumption~\ref{ass.cov}.(iv) and the fact that 
		\begin{align*}
			\bigg|
			\frac{1}{S^{2} T^{2}} &
			\sum_{s \neq g}^{S} \sum_{t \neq l}^{T} \mathbb{E}\left[\eta_{s t} \eta_{g t} \eta_{s l} \eta_{g l}\right]\bigg| \\
			& =  \frac{1}{S^{2} T^{2}} \sum_{s \neq g}^{S} \sum_{t \neq l}^{T}
			\Big|\mathbb{E}\left[\eta_{s t} \eta_{g t} \eta_{s l} \eta_{g l}\right]
			-\mathbb{E}[\eta_{s t}] \mathbb{E}[\eta_{g t}] \mathbb{E}[\eta_{s l}] \mathbb{E}[\eta_{g l}]\Big|  \\
			& \leq  C
			\sum_{s\neq g}^{S} \sum_{t \neq l}^{T} 
			a_{sg}(|t-l|)^{\frac{\delta}{\delta+4}}
			\left(\mathbb{E}\left[|\eta_{st}|^{\delta+4}\right]\right)^{\frac{1}{\delta+4}} 
			\left(\mathbb{E}\left[|\eta_{gt}|^{\delta+4}\right]\right)^{\frac{1}{\delta+4}} \\
			& \qquad \qquad 
			\cdot
			\left(\mathbb{E}\left[|\eta_{sl}|^{\delta+4}\right]\right)^{\frac{1}{\delta+4}} 
			\left(\mathbb{E}\left[|\eta_{gl}|^{\delta+4}\right]\right)^{\frac{1}{\delta+4}}  \\
			& \leq  C \sum_{s\neq g}^{S} \sum_{t \neq l}^{T}   a_{sg}(|t-l|)^{\frac{\delta}{\delta+4}}
			= O \left(\frac{1}{ST}\right)
			,
		\end{align*}
		under Assumption~\ref{ass.mixing}.(i) and (ii). 
		Therefore,
		\begin{align*}
			\left\|\frac{1}{S T} \sum_{s=1}^{S} {\eta}_{s} {\eta}_{s}'\right\|
			=O_P\left(\max \left(\frac{1}{\sqrt{S}}, \frac{1}{\sqrt{T}}\right)\right).
		\end{align*}

		For (\ref{al5.eq2}),
		\begin{align}\label{al5p.eq1}
			\mathbb{E}\bigg[\bigg\|
			& \frac{1}{S T} \sum_{s=1}^{S} {x}_{s(k)} {\eta}_{s}'\bigg\|^2\bigg] 
			\leq 
			\frac{1}{S^2 T^2} \sum_{s,g=1}^{S}\sum_{t,l=1}^{T}
			\mathbb{E}\left[\eta_{sl}x_{st(k)}x_{gt(k)}\eta_{gl}\right]  \nonumber\\
			= &
			\frac{1}{S^2 T^2} \sum_{s =1}^{S}\sum_{t=1}^{T} \mathbb{E}\left[\eta_{st}^2x_{st(k)}^2\right]
			+
			\frac{1}{S^2 T^2} \sum_{s \neq g}^{S}\sum_{t \neq l}^{T} \mathbb{E}\left[\eta_{sl}x_{st(k)}x_{gt(k)}\eta_{gl}\right] \nonumber\\
			& 
			+
			\frac{1}{S^2 T^2}\sum_{s =1}^{S}\sum_{t \neq l}^{T} \mathbb{E}\left[\eta_{sl}x_{st(k)}x_{st(k)}\eta_{sl}\right]
			+ 
			\frac{1}{S^2 T^2} \sum_{s \neq g}^{S}\sum_{t=1}^{T} \mathbb{E}\left[\eta_{st}x_{st(k)}x_{gt(k)}\eta_{gt}\right] \nonumber\\
			=&  O \left(\frac{1}{ST}\right),
		\end{align}
		where the last equality comes from the following facts that the first term in the first equality is
		\begin{align*}
			\frac{1}{S^2 T^2} \sum_{s =1}^{S}\sum_{t=1}^{T} \mathbb{E}\left[\eta_{st}^2x_{st(k)}^2\right] 
			\leq \frac{1}{S^2 T^2} \sum_{s =1}^{S}\sum_{t=1}^{T} \left(\mathbb{E}\left[\eta_{st}^4\right]\right)^{1/2} \cdot \left(\mathbb{E}\left[x_{st(k)}^4\right]\right)^{1/2} 
			= O\left(\frac{1}{ST}\right),
		\end{align*}
		under Assumption~\ref{ass.cov}.(i) and (iv), and the second term in the first equality of (\ref{al5p.eq1}) is
		\begin{align*}
			\frac{1}{S^2 T^2} 
			&
			\sum_{s \neq g}^{S}\sum_{t \neq l}^{T} \mathbb{E}\left[\eta_{sl}x_{st(k)}x_{gt(k)}\eta_{gl}\right] \\
			= & 
			\frac{1}{S^2 T^2} \sum_{s \neq g}^{S}\sum_{t \neq l}^{T} 
			\Big(\mathbb{E}\left[\eta_{st}x_{st(k)}x_{gt(k)}\eta_{gl}\right] - \mathbb{E}\left[\eta_{sl}x_{st(k)}\right] \cdot
			\mathbb{E}\left[x_{gt(k)}\eta_{gl}\right] \Big) \\
			\leq & \frac{1}{S^2 T^2} \sum_{s \neq g}^{S}\sum_{t \neq l}^{T}
			10 a_{sg}(|t-l|)^{\frac{\delta}{\delta+2}}
			\left(\mathbb{E}\left[|\eta_{sl}x_{st(k)}|^{\delta+2}\right]\right)^{\frac{1}{\delta+2}}
			\left(\mathbb{E}\left[|x_{gt(k)}\eta_{gl}|^{\delta+2}\right]\right)^{\frac{1}{\delta+2}} \\
			=& O\left(\frac{1}{ST}\right),
		\end{align*}
		where  the first equality holds since $ \mathbb{E}\left[\eta_{st}x_{st(k)}\right] = \mathbb{E}[x_{st(k)}]\mathbb{E}(\eta_{st}|x_{st(k)})] = 0$ under Assumption~\ref{ass.cov}.(iv), the last inequality holds due to Lemma~\ref{al.4} and  Assumption~\ref{ass.cov}.(iv), and the last equality follows from Assumption \ref{ass.mixing}(i) and \ref{ass.mixing}(ii). The third term in the first equality of (\ref{al5p.eq1}) is
		\begin{align*}
			\frac{1}{S^2 T^2}
			&
			\sum_{s =1}^{S}\sum_{t \neq l}^{T} \mathbb{E}\left[\eta_{st}x_{st(k)}x_{sl(k)}\eta_{sl}\right] \\
			= & 
			\frac{1}{S^2 T^2} \sum_{s =1 }^{S}\sum_{t \neq l}^{T} 
			\Big(\mathbb{E}\left[\eta_{st}x_{st(k)}x_{sl(k)}\eta_{sl}\right] - \mathbb{E}\left[\eta_{st}x_{st(k)}\right] \cdot
			\mathbb{E}\left[x_{sl(k)}\eta_{sl}\right] \Big) \\
			\leq & \frac{1}{S^2 T^2} \sum_{s =1}^{S}\sum_{t \neq l}^{T}
			10 a(|t-l|)^{\frac{\delta}{\delta+2}}
			\left(\mathbb{E}\left[|\eta_{st}x_{st(k)}|^{\delta+2}\right]\right)^{\frac{1}{\delta+2}}
			\left(\mathbb{E}\left[|x_{sl(k)}\eta_{sl}|^{\delta+2}\right]\right)^{\frac{1}{\delta+2}} \\
			=& O\left(\frac{1}{ST}\right),
		\end{align*}
		where the last equality follows from Assumption~\ref{ass.mixing}.(i). And lastly, the fourth term in the first equality of (\ref{al5p.eq1}) is 
		\begin{align*}
			\frac{1}{S^2 T^2}
			&
			\sum_{s \neq g}^{S}\sum_{t = 1}^{T} \mathbb{E}\left[\eta_{st}x_{st(k)}x_{gl(k)}\eta_{gl}\right] \\
			= & 
			\frac{1}{S^2 T^2} \sum_{s \neq g}^{S}\sum_{t =1}^{T} 
			\Big(\mathbb{E}\left[\eta_{st}x_{st(k)}x_{gt(k)}\eta_{gt}\right] - \mathbb{E}\left[\eta_{st}x_{st(k)}\right] \cdot
			\mathbb{E}\left[x_{gt(k)}\eta_{gt}\right] \Big) \\
			\leq & \frac{1}{S^2 T^2} \sum_{s \neq g}^{S}\sum_{t =1}^{T}
			10 a_{sg}(0)^{\frac{\delta}{\delta+2}}
			\left(\mathbb{E}\left[|\eta_{st}x_{st(k)}|^{\delta+2}\right]\right)^{\frac{1}{\delta+2}}
			\left(\mathbb{E}\left[|x_{gt(k)}\eta_{gt}|^{\delta+2}\right]\right)^{\frac{1}{\delta+2}} \\
			=& O\left(\frac{1}{ST}\right)
		\end{align*}
		where the last equality follows from Assumption \ref{ass.mixing}(i) and (ii).

		Therefore, we have
		\begin{align*}
			\left\|\frac{1}{S T} \sum_{s=1}^{S} x_{s(k)} {\eta}_{s}'\right\|
			=O \left(\frac{1}{ST}\right).
		\end{align*}

		For (\ref{al5.eq3}), we have
		\begin{align*}
			\mathbb{E}\bigg[\bigg\|
			& \frac{1}{S T} \sum_{s=1}^{S}  F {\lambda}_{s}{\eta}_{s}'\bigg\|^2\bigg]
			\leq  \frac{1}{S^2 T^2}\sum_{s,g=1}^{S}  
			\mathbb{E}\left[\text{tr}\left({\eta}_{s}{\lambda}_{s}' F' F {\lambda}_{g}{\eta}_{g}' \right)\right] \\
			= & \frac{1}{S^2 T} \sum_{i=1}^{r}\sum_{s=1}^{S}\sum_{t=1}^{T}
			\mathbb{E}\left[\eta_{st}^2 \lambda_{si}^2 \right] 
			+ \frac{1}{S^2 T} \sum_{i=1}^{r}\sum_{s \neq g}^{S}\sum_{t=1}^{T}
			\mathbb{E}\left[\eta_{gt}\eta_{st}\lambda_{si} \lambda_{gi} \right] 
			=   O \left(\frac{1}{S}\right), 
		\end{align*}
		where the last equality holds due to the facts that
		\begin{align*}
			\frac{1}{S^2 T} \sum_{i=1}^{r}\sum_{s=1}^{S}\sum_{t=1}^{T}
			\mathbb{E}\left[\eta_{st}^2 \lambda_{si}^2 \right] 
			\leq \frac{1}{S^2 T} \sum_{i=1}^{r}\sum_{s=1}^{S}\sum_{t=1}^{T}
			\left(\mathbb{E}\left[\eta_{st}^4\right]\right)^{1/2} 
			\left(\mathbb{E}\left[\lambda_{si}^4\right]\right)^{1/2}
			= O\left(\frac{1}{S}\right),
		\end{align*}
		under Assumption~\ref{ass.cov}, and
		\begin{align*}
			\frac{1}{S^2 T} &
			\sum_{i=1}^{r}\sum_{s \neq g}^{S}\sum_{t=1}^{T}
			\mathbb{E}\left[\eta_{gt}\eta_{st}\lambda_{si} \lambda_{gi} \right] \\
			= &  
			\frac{1}{S^2 T} \sum_{i=1}^{r}\sum_{s \neq g}^{S}\sum_{t=1}^{T}
			\Big(
			\mathbb{E}\left[\eta_{st}\lambda_{si} \eta_{gt}\lambda_{gi} \right]
			- 
			\mathbb{E}\left[\eta_{st}\lambda_{si} \right]
			\cdot
			\mathbb{E}\left[\eta_{gt}\lambda_{gi} \right] \Big) \\
			\leq & \frac{1}{S^2 T} \sum_{i=1}^{r}\sum_{s \neq g}^{S}\sum_{t =1}^{T}
			10 a_{sg}(0)^{\frac{\delta}{\delta+2}}
			\left(\mathbb{E}\left[|\eta_{st}\lambda_{si}|^{\delta+2}\right]\right)^{\frac{1}{\delta+2}}
			\left(\mathbb{E}\left[|\eta_{gt}\lambda_{gi}|^{\delta+2}\right]\right)^{\frac{1}{\delta+2}} 
			= O\left(\frac{1}{S}\right),
		\end{align*}
		under Assumption~\ref{ass.mixing}.(i) and (ii).
		
		Therefore,
		\begin{align*}
			\left\|\frac{1}{S T} \sum_{s=1}^{S}  F{\lambda}_{s}{\eta}_{s}'\right\|=O_P\left(\frac{1}{\sqrt{S}}\right).
		\end{align*}
		\hfill$\square$
		
		\begin{lemma} \label{prop3}
			Under  assumptions of Theorem~\ref{t5}, we have for any fixed $u\in \mathcal{U}$, $j=1,...,J$ and $m \geq 1$,

			\begin{align*}
				\frac{1}{\sqrt{T}} & \left(\widehat{F}_j^{(m)}(u)H_j^{(m)}(u) - F_j(u)\right) \\
				=& 
				\frac{1}{\sqrt{T}} \left[I_1^{(m)}(u) + I_2^{(m)}(u) + I_3^{(m)}(u)\right]  \widehat{F}^{(m)}_j(u)\left(K_j^{(m)}(u)\right)^{-1}
				+
				O_P(\delta_{ST}^{-1})  \\
				&
				+ o_P\left(\frac{1}{\sqrt{ST}}\right)
				\cdot
				\bigg(\Big\|\widehat{{\beta}}^{(m-1)}_j(u)-{\beta}_j(u)\Big\|
				+\bigg |\sum_{t=T_0}^{T} \widehat{{\delta}}^{(m-1)}_{jt}(u) - {\delta}_{jt}(u)\bigg| \bigg),
			\end{align*}
			where 
			\begin{align*} 
				& I_1^{(m)} 
				:=  \frac{1}{ST} \sum_{s=1}^{S}\Big[D_s (\widehat{\delta}^{(m-1)}_j(u) - \delta_j(u))(\widehat{\delta}^{(m-1)}_j(u) - \delta_j(u))'D_s' \\
				& \qquad \qquad \qquad \quad
				+  D_s (\widehat{\delta}^{(m-1)}_j(u) - \delta_j(u))(\widehat{\beta}^{(m-1)}_j(u) - \beta_j(u))'X_s' \\
				& \qquad \qquad \qquad \quad
				+ X_s (\widehat{\beta}^{(m-1)}_j(u) - \beta_j(u))(\widehat{\delta}^{(m-1)}_j(u) - \delta_j(u))'D_s' \\
				& \qquad \qquad \qquad \quad
				+ X_s (\widehat{\beta}^{(m-1)}_j(u) - \beta_j(u))(\widehat{\beta}^{(m-1)}_j(u) - \beta_j(u))'X_s'\Big], \nonumber\\
				& I_2^{(m)} (u)
				:=  \frac{1}{ST} \sum_{s=1}^{S} \Big[
				D_s\left({\delta}_j(u)- \widehat{{\delta}}^{(m-1)}_j(u)\right){\lambda}_{js}(u)' F_j(u)' \\
				& \qquad \qquad \qquad \quad 
				+ X_s\left({\beta}_j(u)- \widehat{{\beta}}^{(m-1)}_j(u)\right){\lambda}_{js}(u)' F_j(u)' \\
				& \qquad \qquad \qquad \quad
				+ F_j(u) {\lambda}_{js}(u) \left({\delta}_j(u)- \widehat{{\delta}}^{(m-1)}_j(u)\right)' D_s'\\
				& \qquad \qquad \qquad \quad
				+ F_j(u) {\lambda}_{js}(u) \left({\beta}_j(u)- \widehat{{\beta}}^{(m-1)}_j(u)\right)' X_s'
				\Big], \nonumber\\
				& I_3^{(m)}(u)
				:=   \frac{1}{ST} \sum_{s=1}^{S}
				\bigg[D_s\left({\delta}_j(u)- \widehat{{\delta}}^{(m-1)}_j(u)\right){\eta}_{js}(u)' 
				+ X_s\left({\beta}_j(u)- \widehat{{\beta}}^{(m-1)}_j(u)\right){\eta}_{js}(u)' \\
				& \qquad \qquad \qquad \quad
				+
				{\eta}_s  \left({\delta}- \widehat{{\delta}}^{(m-1)}\right)' D_s'
				+
				{\eta}_s  \left({\beta}- \widehat{{\beta}}^{(m-1)}\right)' X_s'
				\bigg], 
			\end{align*}
			and 
			\begin{align*}
				K^{(m)}_j(u) :=   \left(\frac{\Lambda_j(u)'\Lambda_j(u)}{S}\right)\left(\frac{F_j(u)'\widehat{F}^{(m)}_j(u)}{T}\right),
				\quad 
				H_j^{(m)}(u) := \widehat{\Upsilon}_j^{(m)}(u)\left({K}_j^{(m)}(u)\right)^{-1}, 
			\end{align*}
			in which $\widehat{\Upsilon}_j^{(m)}(u)$ is the $r \times r$ diagonal matrix with diagonal elements being the $r$ largest eigenvalues of $\widehat{L}_j\big(\widehat{\delta}^{(m-1)}_j(u),\widehat{\beta}^{(m-1)}_j(u)\big)$, defined in (\ref{eq:W}), in descending order.
		\end{lemma}
		
		\noindent
		\textbf{Proof}
		Since $u$ and $j$ are fixed in Lemma~\ref{prop3}, we suppress $u$ and $j$ throughout the following proof.
		
		It is easy to show that a square matrix A is invertible if $\det(AA') = \det(A) \det(A') \neq 0$, 
		and thus, 
		$T^{-1}F'\widehat{F}^{(m)}$ is invertible and bounded under Assumption \ref{ass.tech}.(i). Additionally, $S^{-1}\Lambda'\Lambda$ is invertible and bounded under Assumption~\ref{ass.cov}.(ii), 
		and thus, $K^{(m)}$ is invertible, and 
		\begin{align}
			\label{prop3p.eq0}
			\left\|\Big(K^{(m)}\Big)^{-1}\right\| 
			\leq \left\|\bigg(\frac{\Lambda'\Lambda}{S}\bigg)^{-1}\right\| 
			\cdot \left\|\bigg(\frac{F'\widehat{F}^{(m)}}{T}\bigg)^{-1}\right\| 
			= O_P(1).
		\end{align}
		Then, given the definition of $H^{(m)}$, to obtain the asymptotic expression of $\widehat{F}^{(m)} H^{(m)} - F$, it is enough to analyze $\widehat{F}^{(m)} \widehat{\Upsilon}^{(m)} - F K^{(m)}$, and then $\widehat{F}^{(m)} H^{(m)} - F = \big(\widehat{F}^{(m)} \widehat{\Upsilon}^{(m)} - F K^{(m)}\big) \big(K^{(m)}\big)^{-1}$.

		With ($\widehat{\delta}^{(m-1)}$,$\widehat{\beta}^{(m-1)}$) obtained at the $(m-1)^{\text{th}}$ step, we have the estimator $\big(\widehat{F}^{(m)}, \widehat{\Lambda}^{(m)}\big)$ via the PCA. Thus, $\widehat{F}^{(m)}$ statisfies
		\begin{equation} \label{eq.vj}
			\widehat{F}^{(m)} \widehat{\Upsilon}^{(m)} = \widehat{L}^{(m)}\widehat{F}^{(m)}.
		\end{equation}
		Moreover,  $\widehat{L}^{(m)}=(ST)^{-1}\sum_{s=1}^{S}
		\big (
		\widehat{A}_{s} - X_s \widehat{{\beta}}^{(m-1)}\big)
		\big(
		\widehat{A}_{s} - X_s \widehat{{\beta}}^{(m-1)} \big)'$  can be deposed into eight terms by substituting $\widehat{A}_s$ in $\widehat{L}^{(m)}$ with $\widehat{{A}}_s =  X_s{\beta} +  F {\lambda}_{s} +  {\eta}_{s} + (\widehat{{A}}_s - {A}_s)$ as follows:
		\allowdisplaybreaks
		\begin{align*}
			& I_1^{(m)} 
			=  \frac{1}{ST} \sum_{s=1}^{S}\Bigg[D_s (\widehat{\delta}^{(m-1)} - \delta)(\widehat{\delta}^{(m-1)} - \delta)'D_s'
			+  D_s (\widehat{\delta}^{(m-1)} - \delta)(\widehat{\beta}^{(m-1)} - \beta)'X_s' \\
			& \qquad \qquad \qquad \quad
			+ X_s (\widehat{\beta}^{(m-1)} - \beta)(\widehat{\delta}^{(m-1)} - \delta)'D_s' 
			+ X_s (\widehat{\beta}^{(m-1)} - \beta)(\widehat{\beta}^{(m-1)} - \beta)'X_s'\Bigg]
			, \\
			& I_2^{(m)}
			= \frac{1}{ST} \sum_{s=1}^{S} 
			\bigg[D_s\left({\delta}- \widehat{{\delta}}^{(m-1)}\right){\lambda}_s' F'
			+
			X_s\left({\beta}- \widehat{{\beta}}^{(m-1)}\right){\lambda}_s' F' \\
			& \qquad \qquad \qquad \quad
			+
			F {\lambda}_s \Big({\delta}- \widehat{{\delta}}^{(m-1)}\Big)' D_s'
			+
			F {\lambda}_s \Big({\beta}- \widehat{{\beta}}^{(m-1)}\Big)' X_s'
			\bigg], \\
			& I_3^{(m)}
			= \frac{1}{ST} \sum_{s=1}^{S}
			\bigg[D_s\left({\delta}- \widehat{{\delta}}^{(m-1)}\right){\eta}_s' 
			+ X_s\left({\beta}- \widehat{{\beta}}^{(m-1)}\right){\eta}_s' \\
			& \qquad \qquad \qquad \quad
			+
			{\eta}_s  \left({\delta}- \widehat{{\delta}}^{(m-1)}\right)' D_s'
			+
			{\eta}_s  \left({\beta}- \widehat{{\beta}}^{(m-1)}\right)' X_s'
			\bigg], \\
			& I_4^{(m)}
			= \frac{1}{ST} \sum_{s=1}^{S} 
			\bigg[F {\lambda}_s {\eta}_s' 
			+
			{\eta}_s  {\lambda}_s'  F'
			+
			{\eta}_s  {\eta}_s'
			\bigg],\\
			& I_5^{(m)}
			= \frac{1}{ST} \sum_{s=1}^{S} 
			\bigg[ \left(\widehat{{A}}_s - {A}_s \right) \left({\delta}- \widehat{{\delta}}^{(m-1)}\right)' D_s'
			+
			\left(\widehat{{A}}_s - {A}_s \right) \left({\beta}- \widehat{{\beta}}^{(m-1)}\right)' X_s'\\
			& \qquad \qquad \qquad \quad
			+
			D_s \left({\delta}- \widehat{{\delta}}^{(m-1)}\right) \left(\widehat{{A}}_s - {A}_s \right)'
			+
			X_s \left({\beta}- \widehat{{\beta}}^{(m-1)}\right) \left(\widehat{{A}}_s - {A}_s \right)'
			\bigg], \\
			& I_{6}^{(m)}
			= \frac{1}{ST} \sum_{s=1}^{S} 
			\bigg[
			\left(\widehat{{A}}_s - {A}_s \right) {\eta}_s' 
			+
			{\eta}_s \left(\widehat{{A}}_s - {A}_s \right)'
			+
			\left(\widehat{{A}}_s - {A}_s \right) {\lambda}_s' F'
			+
			F {\lambda}_s \left(\widehat{{A}}_s - {A}_s \right)'
			\bigg] ,\\
			& I_{7}^{(m)}
			= \frac{1}{ST} \sum_{s=1}^{S} \left(\widehat{{A}}_s - {A}_s \right)\left(\widehat{{A}}_s - {A}_s \right)', \\
			& I_{8}^{(m)}
			= \frac{1}{ST} \sum_{s=1}^{S} F {\lambda}_s {\lambda}_s' F' .
		\end{align*}
		Then, since $I_{8}^{(m)} \widehat{F}^{(m)} = F K^{(m)}$ according to the definition of $K^{(m)}$, we have
		\begin{align} 
			\label{prop2p.eq0}
			\widehat{F}^{(m)} \widehat{\Upsilon}^{(m)} - F K^{(m)}= \sum_{h=1}^{7}I_h^{(m)} \widehat{F}^{(m)},
		\end{align}
		where $I_1^{(m)}$, $I_2^{(m)}$ and $I_3^{(m)}$ are the leading terms, while the rest of the terms are negligible in the limit. To show this, we start by showing the rate of convergence for $I_1^{(m)}$. By triangle and Cauchy-Schwartz inequalities, we have
		\begin{align*}
			\|I_1^{(m)}\| 
			\leq  & \frac{1}{ST} \Bigg[\bigg\|\sum_{s=1}^{S}D_s (\widehat{\delta}^{(m-1)} - \delta)(\widehat{\delta}^{(m-1)} - \delta)'D_s'\bigg\| \\
			& \quad
			+ 2 \bigg\|\sum_{s=1}^{S}D_s (\widehat{\delta}^{(m-1)} - \delta)(\widehat{\beta}^{(m-1)} - \beta)'X_s'\bigg\| \\
			& \quad + \bigg\|\sum_{s=1}^{S}W_s (\widehat{\beta}^{(m-1)} - \beta)(\widehat{\beta}^{(m-1)} - \beta)'X_s'\bigg\|\Bigg] \\
			\leq  & \frac{1}{ST} \Bigg[ \bigg\|\sum_{s=1}^{S}\sum_{t,l=T_0}^{T} d_s e_t (\widehat{\delta}^{(m-1)}_t - \delta_t) (\widehat{\delta}^{(m-1)}_l - \delta_l) e_l' \bigg\| \\
			& \quad + 2 \bigg\|\sum_{s=1}^{S}X_s\bigg\| \cdot \bigg\|\sum_{t=T_0}^{T} e_t (\widehat{\delta}^{(m-1)}_t - \delta_t) \bigg\| \cdot \|\widehat{\beta}^{(m-1)} - \beta\| \\
			& \quad 
			+ \sum_{s=1}^{S}\|X_s\|^2 \cdot \|\widehat{\beta}^{(m-1)} - \beta\|^2 \Bigg] \\
			\leq & \frac{1}{ST} \Bigg[ \sum_{s=1}^{S}|d_s|\cdot \bigg|\sum_{t=T_0}^{T} \widehat{\delta}^{(m-1)}_t - \delta_t\bigg|^2 
			+ \sum_{s=1}^{S}\|X_s\|^2 \cdot \|\widehat{\beta}^{(m-1)} - \beta\|^2 \\
			& \quad 
			+ 2 \bigg\|\sum_{s=1}^{S}X_s\bigg\| \cdot \bigg|\sum_{t=T_0}^{T} \widehat{\delta}^{(m-1)}_t - \delta_t\bigg| \cdot \|\widehat{\beta}^{(m-1)} - \beta\| 
			\Bigg] \\
			= & O_P \bigg(\frac{1}{T}\bigg) \cdot \bigg|\sum_{t=T_0}^{T} \widehat{\delta}^{(m-1)}_t - \delta_t\bigg|^2 
			+O_P(1) \cdot\|\widehat{\beta}^{(m-1)} - \beta\|^2\\
			& \quad
			+ O_P \bigg(\frac{1}{\sqrt{T}}\bigg) \cdot \bigg|\sum_{t=T_0}^{T} \widehat{\delta}^{(m-1)}_t - \delta_t\bigg| \cdot  \|\widehat{\beta}^{(m-1)} - \beta\|,
		\end{align*}
		For $I_2^{(m)}$, a similar argument yields that 
		\begin{align*}
			\|I_2^{(m)}\| 
			\leq  & \frac{1}{ST} \Bigg[\bigg\|\sum_{s=1}^{S}D_s (\widehat{\delta}^{(m-1)} - \delta)\lambda_s'F'\bigg\| 
			+ \bigg\|\sum_{s=1}^{S}W_s (\widehat{\beta}^{(m-1)} - \beta)\lambda_s'F'\bigg\| \Bigg] \\
			\leq & \frac{1}{ST} \Bigg[ \bigg\|\sum_{s=1}^{S}d_s\lambda_s'F'\bigg\|\cdot \bigg|\sum_{t=T_0}^{T} \widehat{\delta}^{(m-1)}_t - \delta_t\bigg|
			+ \bigg\|\sum_{s=1}^{S}W_s\lambda_s'F'\bigg\| \cdot  \|\widehat{\beta}^{(m-1)} - \beta\| \Bigg] \\
			= &  O_P\left(\frac{1}{\sqrt{ST}}\right) \cdot 
			\bigg(\Big\|\widehat{{\beta}}^{(m-1)}-{\beta}\Big\|
			+\bigg |\sum_{t=T_0}^{T}{\delta}_{t}- \widehat{{\delta}}^{(m-1)}_{t}\bigg| \bigg),
		\end{align*}
		under Assumption~\ref{ass.norm} that $\mathbb{E}\|\sum_{s=1}^{S}W_s\lambda_s'F'\|^2 = O_P(ST)$ and $\mathbb{E}\|\sum_{s=1}^{S}d_s\lambda_s'F'\|^2 = O_P(ST)$. 
		
		For $I_3^{(m)}$ and $I_4^{(m)}$, by triangle and Cauchy-Schwartz inequalities and the property of $\alpha$-mixing varaibles (see Lemma~\ref{al.5}) we have
		\begin{align*}
			\left\|I_{3}^{(m)}\right\| 
			\leq & 2 \bigg(
			\sum_{k=1}^{K} \bigg\|\frac{1}{S T} \sum_{s=1}^{S} {x}_{s(k)} {\eta}_{s}^{\prime}\bigg\|^2\bigg)^{1/2} 
			\bigg(\sum_{k=1}^{K}\left|{\beta}_{k}- \widehat{{\beta}}^{(m-1)}_{k}\right|^2\bigg)^{1/2} \\
			&
			+ \frac{2}{S T}
			\bigg\| \sum_{s=1}^{S} {d}_{s} {\eta}_{s}^{\prime}\bigg\| \cdot  
			\bigg |\sum_{t=T_0}^{T}{\delta}_{t}- \widehat{{\delta}}^{(m-1)}_{t}\bigg| \\
			\leq & O_P\left(\frac{1}{\sqrt{ST}}\right) \cdot 
			\bigg(\Big\|\widehat{{\beta}}^{(m-1)}-{\beta}\Big\|
			+\bigg |\sum_{t=T_0}^{T}{\delta}_{t}- \widehat{{\delta}}^{(m-1)}_{t}\bigg| \bigg),\\
			\left\|I_{4}^{(m)}\right\| 
			\leq  & 2 
			\left\|\frac{1}{S T} \sum_{s=1}^{S} F {\lambda}_s {\eta}_s'\right\| 
			+
			\left\|\frac{1}{S T} \sum_{s=1}^{S}  {\eta}_s {\eta}_s'\right\| 
			= O_P\left(\delta_{ST}^{-1}\right).
		\end{align*}
		
		For $I_5^{(m)}$ to $I_7^{(m)}$, we apply Cauchy-Schwartz inequality and have
		{\small
			\begin{align} 
				\label{l8p.eq4}
				\left\|I_5^{(m)}\right\| 
				& \leq 
				\frac{2}{ST}\ \underset{s}{\sup}\left\|\widehat{{A}}_s - {A}_s \right\|
				\cdot 
				\bigg( \sum_{s=1}^{S}  \left\|X_{s} \right\| \cdot \left\|{\beta}- \widehat{{\beta}}^{(m-1)}\right\|
				+ 
				\bigg\| \sum_{s=1}^{S} {d}_{s}\bigg\| \cdot  
				\bigg |\sum_{t=T_0}^{T}{\delta}_{t}- \widehat{{\delta}}^{(m-1)}_{t}\bigg| \bigg) \nonumber\\
				& = O_P\left( (ST)^{-3/4}\right) \cdot 
				\bigg(\Big\|\widehat{{\beta}}^{(m-1)}-{\beta}\Big\|
				+\bigg |\sum_{t=T_0}^{T}{\delta}_{t}- \widehat{{\delta}}^{(m-1)}_{t}\bigg| \bigg), \nonumber\\
				\left\|I_{6}^{(m)}\right\| 
				& \leq \frac{2}{ST}
				\ \underset{s}{\sup}\left\|\widehat{{A}}_s - {A}_s \right\|
				\cdot \sum_{s=1}^{S} \big(\left\| \eta_{s} \right\| + \|\lambda_s'F\|\big)
				= O_P\left( (ST)^{-3/4}\right) , \\
				\left\|I_{7}^{(m)}\right\| 
				& \leq \frac{1}{ST}\sum_{s=1}^{S}\left\|\widehat{A}_s - A_s\right\|^2
				=  \frac{1}{ST} 
				\sum_{s=1}^{S}\sum_{t=1}^{T}\left\|\widehat{\alpha}_{st} - {\alpha}_{st}\right\|^2
				= O_P\left( (ST)^{-3/2}\right),\nonumber
		\end{align}}
		since $\sup_s\|\widehat{{A}}_s - {A}_s \|^2 = \sum_{t=1}^{T}\sup_{s,t}\|\widehat{\alpha}_{st} - \alpha_{st} \|^2 = O_P(S^{-3/2}T^{-1/2})$ according to Lemma~\ref{l1} and $\sum_{s=1}^{S} \|X_{s} \| = O_P(S\sqrt{T})$,  $\sum_{s=1}^{S} \|\eta_{s} \| = O_P(S\sqrt{T})$ and  $\sum_{s=1}^{S} \|\lambda_{s}'F' \| = O_P(S\sqrt{T})$ under Assumption~\ref{ass.cov}.

		Finally, combining (\ref{prop3p.eq0})-(\ref{prop2p.eq0}) with $\widehat{F}^{(m)} = O_P(\sqrt{T})$ and $\big(K^{(m)}\big)^{-1} = O_P(1)$, we obtain 
		\begin{align*}
			\frac{1}{\sqrt{T}}  \left(\widehat{F}^{(m)} H^{(m)} - F\right)  
			& = \frac{1}{\sqrt{T}}\sum_{h=1}^{7}I_h^{(m)} \widehat{F}^{(m)}\left(K^{(m)}\right)^{-1} \\
			& = \frac{1}{\sqrt{T}} \left[ I_1^{(m)}  + I_2^{(m)} + I_3^{(m)}\right] \widehat{F}^{(m)}\left(K^{(m)}\right)^{-1} +  O_P(\delta_{ST}^{-1}) \\
			& \quad + o_P\left(\frac{1}{\sqrt{ST}}\right) 
			\cdot \bigg(\Big\|\widehat{{\beta}}^{(m-1)}-{\beta}\Big\| +\bigg |\sum_{t=T_0}^{T} \widehat{{\delta}}^{(m-1)}_{t} - {\delta}_{t}\bigg| \bigg),
		\end{align*}
		and furthermore,
		{\small
			\begin{align} \label{al6.eq1}
				\big\| \widehat{F}^{(m)} H^{(m)} - F  \big\| 
				& = O_P \bigg(\frac{1}{\sqrt{T}}\bigg) \cdot \bigg|\sum_{t=T_0}^{T} \widehat{\delta}^{(m-1)}_t - \delta_t\bigg|^2 
				+O_P(\sqrt{T}) \cdot\|\widehat{\beta}^{(m-1)} - \beta\|^2 \nonumber\\
				& \quad
				+ O_P (1) \cdot \bigg|\sum_{t=T_0}^{T} \widehat{\delta}^{(m-1)}_t - \delta_t\bigg| \cdot  \|\widehat{\beta}^{(m-1)} - \beta\| \nonumber\\
				& \quad 
				+ O_P\left(\frac{1}{\sqrt{S}}\right) \cdot 
				\bigg(\Big\|\widehat{{\beta}}^{(m-1)}-{\beta}\Big\|
				+\bigg |\sum_{t=T_0}^{T}{\delta}_{t}- \widehat{{\delta}}^{(m-1)}_{t}\bigg| \bigg)
				+ O_P(1). 
		\end{align}}
		
		\hfill$\square$

		\begin{lemma} \label{t2}
			Under assumptions of Theorem~\ref{t6}, for any fixed $u\in \mathcal{U}$, $j=1,...,J$ and recursive step $m \geq 1$, 
			\begin{align*}
				\sum_{s=1}^{S} & X_s' M_{\widehat{F}_j}^{(m)}(u) \big(F_j(u) - \widehat{F}^{(m)}_j(u)H^{(m)}_j(u)\big){\lambda}_{js}(u)\\
				=
				&
				-\sum_{s=1}^{S} X_s' M_{\widehat{F}_j}^{(m)}(u) \Big(I_1^{(m)}+I_2^{(m)}\Big)\widehat{F}_j^{(m)}(u) \left(K_j^{(m)}(u)\right)^{-1}{\lambda}_{js}(u) \\
				& +O_P (\sqrt{S}) \cdot \bigg|\sum_{t=T_0}^{T} \widehat{\delta}^{(m-1)}_{jt}(u) - \delta_{jt}(u) \bigg|^2
				+O_P(T\sqrt{S}) \cdot \Big\|\widehat{\beta}^{(m-1)}_j(u) - \beta_j(u)\Big\|^2 \nonumber\\
				&
				+ O_P (\sqrt{ST}) \cdot \bigg|\sum_{t=T_0}^{T} \widehat{\delta}^{(m-1)}_{jt}(u) - \delta_{jt}(u) \bigg| \cdot  \Big\|\widehat{\beta}^{(m-1)}_j(u) - \beta_j(u)\Big\| \nonumber\\
				&
				+ O_P(\sqrt{T}) \cdot 
				\bigg(\Big\|\widehat{\beta}^{(m-1)}_j(u) - \beta_j(u)\Big\|
				+ \bigg|\widehat{\delta}^{(m-1)}_{jt}(u) - \delta_{jt}(u) \bigg|\bigg)
				+ O_P(\sqrt{ST})
			\end{align*}
			where $I_1^{(m)}$, $I_2^{(m)}$, $I_3^{(m)}$, $H^{(m)}_j(u)$ and $K^{(m)}_j(u)$ are defined in Lemma~\ref{prop3}.
		\end{lemma}
		
		\noindent
		\textbf{Proof}
		Since $u$ and $j$ are fixed, for notational simplicity, we suppress $u$ and $j$ throughout the following proof.
		
		By definition of $M^{(m)}_{\widehat{F}}$, we have
		$
		\sum_{s=1}^{S}X_s' M_{\widehat{F}}^{(m)} F{\lambda}_s
		=\sum_{s=1}^{S}X_s' M_{\widehat{F}}^{(m)} \big(F - \widehat{F}^{(m)} H^{(m)}\big){\lambda}_s
		$. 
		Also, recall from the proof of Lemma~\ref{prop3}, we have $ \widehat{F}^{(m)} H^{(m)} - F = \sum_{h=1}^{7}I_h^{(m)}\widehat{F}^{(m)} \left(K^{(m)}\right)^{-1}$ with  $I_h^{(m)},\ h=1,...,7$ defined above (\ref{prop2p.eq0}), so that we can write
		\begin{align}\label{t1p.eq4}
			\sum_{s=1}^{S}X_s'M_{\widehat{F}}^{(m)} \big(F - \widehat{F}^{(m)} H^{(m)}\big)\lambda_s
			= & - \sum_{h=1}^{7}\sum_{s=1}^{S}X_s' M_{\widehat{F}}^{(m)} I_h^{(m)} \widehat{F}^{(m)} \left(K^{(m)}\right)^{-1} \lambda_s \nonumber\\
			=: & \sum_{h=1}^{7} J_h^{(m)},
		\end{align}
		where $J_1^{(m)}$, $J_2^{(m)}$ and $J_3^{(m)}$ are the leading terms in the asymptotic expression, while the remaining terms are negligible. We first compute the norm of the leading terms. Given the rates $I_1^{(m)}$, $I_2^{(m)}$ and $I_3^{(m)}$ of given in the proof of Lemma~\ref{prop3}, by Cauchy-Swartz and H\"{o}lder's inequality, we have
		\begin{align*}
			\left\|J_{1}^{(m)}\right\| 
			& \leq 2
			\bigg(\sum_{s=1}^{S}\|X_s\|^{2}\bigg)^{1/2}
			\cdot 
			\bigg(\sum_{s=1}^{S}\|\lambda_s\|^{2}\bigg)^{1/2}
			\cdot
			\left\|M_{\widehat{F}}^{(m)}\right\| 
			\cdot \left\|I_{1}^{(m)}\right\|  
			\cdot \left\|\widehat{F}^{(m)}\right\|
			\cdot \left\| \left(K^{(m)}\right)^{-1}\right\| \\
			& = 
			O_P (S) \cdot \bigg|\sum_{t=T_0}^{T} \widehat{\delta}^{(m-1)}_t - \delta_t\bigg|^2 
			+O_P(ST) \cdot\|\widehat{\beta}^{(m-1)} - \beta\|^2\\
			& \quad
			+ O_P (S\sqrt{T}) \cdot \bigg|\sum_{t=T_0}^{T} \widehat{\delta}^{(m-1)}_t - \delta_t\bigg| \cdot  \|\widehat{\beta}^{(m-1)} - \beta\|
			,
		\end{align*}
		where $\big\|(K^{(m)})^{-1}\big\| = O_P(1)$ according to (\ref{prop3p.eq0}), $\sum_{s=1}^{S}\|X_s\|^2 = O_P(ST)$, $\sum_{s=1}^{S}\|\lambda_s\|^2 = O_P(S)$ and $\|\widehat{F}^{(m)}\| = O_P(\sqrt{T})$ under Assumption~\ref{ass.cov}. Similar arguments yield that 
		$J_2^{(m)}$ and $J_3^{(m)}$ are of rate $O_P(\sqrt{ST}) \cdot \big(\|\widehat{{\beta}}^{(m-1)}-{\beta}\| + |\sum_{t=T_0}^T\widehat{\delta}^{(m)}_t - \delta_t| \big)$.

		For $J_4^{(m)}$, we plug-in the expression of $I_4^{(m)}$ to obtain
		{\small
			\begin{align*}
				J_4^{(m)} 
				= 
				-
				& 
				\frac{1}{ST}\sum_{s,g=1}^{S}
				X_s' M_{\widehat{F}}^{(m)} 
				\bigg\{
				F {\lambda}_g {\eta}_g' 
				+{\eta}_g {\lambda}_g'F'  
				+
				\left({\eta}_{g} {\eta}_{g}' - \mathbb{E}[{\eta}_{g} {\eta}_{g}'] \right)
				+ 
				\mathbb{E}[{\eta}_{g} {\eta}_{g}'] 
				\bigg\} \widehat{F}^{(m)}  \left(K^{(m)}\right)^{-1}
				\lambda_s,
		\end{align*}}
		and below we evaluate the asymptotic bound for each of the terms expanding by terms in $I_4^{(m)}$. 
		The first term has order $O_P(\sqrt{ST})\cdot \|\widehat{F}^{(m)} H^{(m)} - F \|$ due to
		\begin{align*}
			\bigg\|\frac{1}{ST} &  \sum_{s=1}^{S}  M_{\widehat{F}}^{(m)} 
			F {\lambda}_s {\eta}_s' \widehat{F}^{(m)}  \left(K^{(m)}\right)^{-1}\bigg\| \\
			&  \leq  \left\|M_{\widehat{F}}^{(m)}\right\| \cdot
			\left\|\widehat{F}^{(m)} H^{(m)} - F  \right\| 
			\cdot \bigg\|\frac{1}{S T} \sum_{s=1}^{S} F {\lambda}_s {\eta}_s'\bigg\| 
			\cdot \left\| \left(K^{(m)}\right)^{-1}\right\| \\
			&
			= \left\|\widehat{F}^{(m)} H^{(m)} - F  \right\|
			\cdot O_P \left(\frac{1}{\sqrt{S}}\right) \cdot O_P(1),
		\end{align*}
		where the first inequality follows from Cauchy-Swartz inequality and the property that $ M_{\widehat{F}}^{(m)} 
		F =  M_{\widehat{F}}^{(m)} 
		\big(F -\widehat{F}^{(m)} H^{(m)} \big) $, and the second last equality follows from (\ref{al5.eq3}) in Lemma~\ref{al.5}, (\ref{prop3p.eq0}), and (\ref{al6.eq1}). 
		For the second term, we plug-in the expression of $K^{(m)}$ and then apply an $\alpha$-mixing argument similar to (\ref{al5.eq2}) in Lemma~\ref{al.5} to obtain
		\begin{align*}
			\frac{1}{ST} \sum_{s,g=1}^{S} X_s'M_{\widehat{F}}^{(m)} {\eta}_g {\lambda}_g'
			F'  \widehat{F}^{(m)}  \left(K^{(m)}\right)^{-1} \lambda_s
			&  = \sum_{g=1}^{S}\bigg(\frac{1}{S}\sum_{s=1}^{S}\omega_{sg}X_s'M_{\widehat{F}}\bigg)\eta_g 
			= O_P \left(\sqrt{ST}\right).
		\end{align*}
		The third term is shown to be 
		$O_P(\sqrt{T}) \cdot 
		(\|\widehat{F}^{(m)}-F \left(H^{(m}\right)^{-1}\|^2 
		+\|\widehat{F}^{(m)} - F \left(H^{(m}\right)^{-1}\| )
		+ O_P(\delta_{ST})$ in Lemma~\ref{bl.2} 
		and the fourth term has order $O_P\big(S\big)$ given that $\sum_{s=1}^{S}\mathbb{E}[\eta_s \eta_s'] = O(S)$ according to Assumption~\ref{ass.cov}.(ii).
		Collecting the terms in $J_4^{(m)}$ and using the rates of $\|\widehat{F}^{(m)} H^{(m)} - F  \|$ obtained from (\ref{al6.eq1}), we have
		\begin{align*}
			\left\|J_4^{(m)}\right\| 
			= & O_P (\sqrt{S}) \cdot \bigg|\sum_{t=T_0}^{T} \widehat{\delta}^{(m-1)}_t - \delta_t \bigg|^2 
			+O_P(T\sqrt{S}) \cdot \Big\|\widehat{\beta}^{(m-1)} - \beta\Big\|^2 \\
			&
			+ O_P (\sqrt{ST}) \cdot \bigg|\sum_{t=T_0}^{T} \widehat{\delta}^{(m-1)}_t - \delta_t\bigg| \cdot  \Big\|\widehat{\beta}^{(m-1)} - \beta\Big\| \\
			&
			+ O_P(\sqrt{T}) \cdot 
			\bigg(\Big\|\widehat{{\beta}}^{(m-1)}-{\beta}\Big\|
			+ \bigg|\sum_{t=T_0}^{T}{\delta}_{t}- \widehat{{\delta}}^{(m-1)}_{t}\bigg|\bigg)
			+
			O_P(\sqrt{ST})
			.
		\end{align*}

		For $\sum_{h=5}^{7}J_h^{(m)}$, it is easy to check that
		{\small
			\begin{align*}
				\sum_{h=5}^{7} \left\|J_h^{(m)}\right\|   
				& \leq 
				\bigg(\sum_{s=1}^{S}\|X_s\|^{2}\bigg)^{1/2}
				\cdot 
				\bigg(\sum_{s=1}^{S}\|\lambda_s\|^{2}\bigg)^{1/2}
				\cdot
				\sum_{h=5}^{7}\left\|I_h^{(m)}\right\|   
				\cdot \left\|\widehat{F}^{(m)}\right\|
				\cdot \left\|M_{\widehat{F}}^{(m)}\right\| 
				\cdot \left\|\left(K^{(m)}\right)^{-1}\right\| \\
				& = O_P \big((ST)^{1/4}\big) \cdot
				\bigg( 1 + \bigg|\sum_{t=T_0}^{T}{\delta}_{t}- \widehat{{\delta}}^{(m-1)}_{t}\bigg| + \Big\|\widehat{{\beta}}^{(m-1)} - {\beta}\Big\| \bigg),
		\end{align*}}
		given the rates of $I_h^{(m)}$ for $h=5,6,7$ from (\ref{l8p.eq4}).
		Finally, collecting the rate of convergence for $J_h^{(m)},\ h=1,...,9$ , we obtain
		\begin{align}\label{t2p.eq3}
			\sum_{s=1}^{S}X_s' M_{\widehat{F}}^{(m)} F{\lambda}_s
			= & - \sum_{s=1}^{S}X_s' M_{\widehat{F}}^{(m)} \Big(I_1^{(m)}+I_2^{(m)} + I_3^{(m)}\Big) \widehat{F}^{(m)} \left(K^{(m)}\right)^{-1} \lambda_s \nonumber\\
			& + O_P (\sqrt{S}) \cdot \bigg|\sum_{t=T_0}^{T} \widehat{\delta}^{(m-1)}_t - \delta_t \bigg|^2 
			+O_P(T\sqrt{S}) \cdot \Big\|\widehat{\beta}^{(m-1)} - \beta\Big\|^2 \nonumber\\
			&
			+ O_P (\sqrt{ST}) \cdot \bigg|\sum_{t=T_0}^{T} \widehat{\delta}^{(m-1)}_t - \delta_t\bigg| \cdot  \Big\|\widehat{\beta}^{(m-1)} - \beta\Big\| \nonumber\\
			&
			+ O_P(\sqrt{T}) \cdot 
			\bigg(\Big\|\widehat{{\beta}}^{(m-1)}-{\beta}\Big\|
			+ \bigg|\sum_{t=T_0}^{T}{\delta}_{t}- \widehat{{\delta}}^{(m-1)}_{t}\bigg|\bigg)
			+ O_P(\sqrt{ST}), 
		\end{align}
		which completes the proof.	
		\hfill$\square$

		\begin{lemma} \label{prop.iter}
			Suppose the assumptions of Theorem~\ref{t6} hold, and $\sqrt{ST}\|\widehat{\beta}^{(m-1)}- {\beta}\| = O_P(1) $, $\sqrt{S}|\widehat{\delta}^{(m-1)}_{t} - {\delta}_{t} | = O_P(1)$, $|\sum_{t=T_0}^{T}\widehat{\delta}^{(m-1)}_{t} - {\delta}_{t} | = O_P(1)$ and $\| \sum_{t=T_0}^{T}\widehat{f}^{(m)}_t(\widehat{\delta}^{(m-1)}_t - \delta_t)\| = O_P(1) $.
			Then, as $T,S \rightarrow \infty$, we have, for any fixed $u\in\mathcal{U}$, $j=1,...,J$, $m \geq 0$, 
			\begin{align*}
				\sum_{s=1}^{S} & d_s e_t'M^{(m)}_{\widehat{F}_j}(u)(F_j(u)-\widehat{F}_j^{(m)}(u)H^{(m)}_j(u))\lambda_{js}(u) \\
				= & 
				\frac{1}{S} \sum_{s,g=1}^{S}\sum_{l=T_0}^{T}\omega_{sg} d_s d_g (\widehat{\delta}^{(m-1)}_{jt}(u) - \delta_{jt}(u)) -  \frac{1}{S} \sum_{s,g=1}^{S}\omega_{sg}d_s\eta_{jgt}(u)  \\
				& - \frac{1}{ST}\sum_{s,g=1}^{S}d_s \mathbb{E}[{\eta}_{jgt}(u) {\eta}_{jg}(u)']  \widehat{F}^{(m)}_j(u) (K^{(m)}_j(u))^{-1}\lambda_{js}(u) 
				+ O_P\bigg(\sqrt{\frac{S}{T}}\bigg),
			\end{align*}
			where $\omega_{j,sg}(u)$ is defined in Assumption~\ref{ass.reg_1}, and  $H^{(m)}_j(u)$ and $K^{(m)}_j(u)$ are defined in Lemma~\ref{prop3}.
			Moreover, the leading terms on the right-hand side are $O_P(\sqrt{S})$.
		\end{lemma}

		\noindent
		\textbf{Proof}
		Since $u$ and $j$ are fixed, for notational simplicity, we suppress $u$ and $j$ throughout the following proof. 
		
		Similar to Lemma~\ref{t2}, we can write 
		$$\sum_{s=1}^{S}d_s e_t'M^{(m)}_{\widehat{F}}(F-\widehat{F}^{(m)}H^{(m)})\lambda_s = -\sum_{h=1}^{7}\sum_{s=1}^{S}d_s e_t'M^{(m)}_{\widehat{F}}I_h^{(m)}\widehat{F}^{(m)}(K^{(m)})^{-1}\lambda_s.$$
		For similar reason to the proof of Lemma~\ref{t2}, the terms associated with $I_5^{(m)},\ I_6^{(m)},\ I_7^{(m)}$ are negligible, proof omitted for concise. Below, we consider the terms associated with $I_1^{(m)},...,I_4^{(m)}$, respectively.

		Using $e_t'M^{(m)}_{\widehat{F}} = e_t' - T^{-1}\widehat{f}^{(m)\prime}_t \widehat{F}^{(m)\prime}$, we write
		\begin{align*}
			\sum_{s=1}^{S}d_s e_t'M^{(m)}_{\widehat{F}}I_1^{(m)}\widehat{F}^{(m)}(K^{(m)})^{-1}\lambda_s 
			= & \sum_{s=1}^{S}d_s e_t'I_1^{(m)}\widehat{F}^{(m)}(K^{(m)})^{-1}\lambda_s \\
			& - \frac{1}{T}\sum_{s=1}^{S}d_s \widehat{f}^{(m)\prime}_t \widehat{F}^{(m)\prime}I_1^{(m)}\widehat{F}^{(m)}(K^{(m)})^{-1}\lambda_s. 
		\end{align*}
		For the first term on the right-hand side of the equation, we substitute-in the expression of $I_1^{(m)}$ and further obtain that 
		\begin{align*}
			\sum_{s=1}^{S}d_s e_t'I_1^{(m)}\widehat{F}^{(m)}(K^{(m)})^{-1}\lambda_s 
			& =  \frac{1}{ST} \Bigg[\sum_{s,g=1}^{S}\sum_{l=T_0}^{T}d_s (\widehat{\delta}^{(m-1)}_t - \delta_t) (\widehat{\delta}^{(m-1)}_l - \delta_l)\widehat{f}^{(m)\prime}_l(K^{(m)})^{-1}\lambda_s \\
			& \quad + \sum_{s,g=1}^{S}d_s d_g (\widehat{\delta}^{(m-1)}_t - \delta_t)(\widehat{\beta}^{(m-1)} - \beta)'X_g'\widehat{F}^{(m)}(K^{(m)})^{-1}\lambda_s \\
			& \quad + \sum_{s,g=1}^{S}\sum_{l=T_0}^{T}d_s w_{gt}' (\widehat{\beta}^{(m-1)} - \beta)(\widehat{\delta}^{(m-1)}_l - \delta_l)\widehat{f}^{(m)\prime}_l(K^{(m)})^{-1}\lambda_s \\
			& \quad + \sum_{s,g=1}^{S}d_s w_{gt}' (\widehat{\beta}^{(m-1)} - \beta)(\widehat{\beta}^{(m-1)} - \beta)'X_g'\widehat{F}^{(m)}(K^{(m)})^{-1}\lambda_s \Bigg], 
		\end{align*}
		where all terms are of $O_P(S/T)$ under the assumptions that $\sqrt{ST}\|\widehat{\beta}^{(m-1)}- {\beta}\| = O_P(1) $, $\sqrt{S}|\widehat{\delta}^{(m-1)}_{t} - {\delta}_{t} | = O_P(1)$, $|\sum_{t=T_0}^{T}\widehat{\delta}^{(m-1)}_{t} - {\delta}_{t} | = O_P(1)$ and $\| \sum_{t=T_0}^{T}\widehat{f}^{(m)}_t(\widehat{\delta}^{(m-1)}_t - \delta_t)\| = O_P(1) $. Furthermore, using the rate of $I_1^{(m)}$ from Lemma~\ref{prop3}, the rate of the second term
		\begin{align*}
			\bigg|\frac{1}{T} & \sum_{s=1}^{S}d_s \widehat{f}^{(m)\prime}_t \widehat{F}^{(m)\prime}I_1^{(m)}\widehat{F}^{(m)}(K^{(m)})^{-1}\lambda_s\bigg| \\
			\leq &  \frac{1}{T} \bigg(\sum_{s=1}^{S}|d_s|^{2}\bigg)^{1/2} \cdot  \bigg(\sum_{s=1}^{S}\|\lambda_s\|^{2}\bigg)^{1/2} \cdot \|\widehat{f}^{(m)}_t\| \cdot \|\widehat{F}^{(m)}\| \cdot \|I_1^{(m)}\| \cdot \|\widehat{F}^{(m)}\| \cdot \|(K^{(m)})^{-1}\| \\
			= & O_P\bigg(\frac{S}{T}\bigg).
		\end{align*} 
		Therefore, we obtain that 
		\begin{align*}
			\bigg|\sum_{s=1}^{S}d_s e_t'M^{(m)}_{\widehat{F}}I_1^{(m)}\widehat{F}^{(m)}(K^{(m)})^{-1}\lambda_s \bigg|  = O_P\bigg(\frac{S}{T}\bigg).
		\end{align*}
		Under similar derivations, we shown that $|\sum_{s=1}^{S}d_s e_t'M^{(m)}_{\widehat{F}}I_3^{(m)}\widehat{F}^{(m)}(K^{(m)})^{-1}\lambda_s |  = O_P(\sqrt{S/T} )$.
		
		For the term associated with $I_2^{(m)}$, again, we expand it as
		\begin{align*}
			\sum_{s=1}^{S}d_s e_t'M^{(m)}_{\widehat{F}}I_2^{(m)}\widehat{F}^{(m)}(K^{(m)})^{-1}\lambda_s 
			= & \sum_{s=1}^{S}d_s e_t'I_2^{(m)}\widehat{F}^{(m)}(K^{(m)})^{-1}\lambda_s \\
			& - \frac{1}{T}\sum_{s=1}^{S}d_s \widehat{f}^{(m)\prime}_t \widehat{F}^{(m)\prime}I_2^{(m)}\widehat{F}^{(m)}(K^{(m)})^{-1}\lambda_s,
		\end{align*}
		and the second term is $O_P(\sqrt{S/T})$ given the rate of $I_2^{(m)}$ from Lemma~\ref{prop3}. Now we expand the first term by the expression of $I_2^{(m)}$:
		\begin{align*}
			\sum_{s=1}^{S} & d_s e_t'I_2^{(m)}\widehat{F}^{(m)}(K^{(m)})^{-1}\lambda_s \\
			& =  - \frac{1}{S} \sum_{s,g=1}^{S}\sum_{l=T_0}^{T}\omega_{sg} d_s d_g (\widehat{\delta}^{(m-1)}_t - \delta_t) 
			- \frac{1}{S}\sum_{s,g=1}^{S}\omega_{sg} d_s x_{gt}' (\widehat{\beta}^{(m-1)} - \beta)\\
			& \quad - \frac{1}{ST}\sum_{s,g=1}^{S}\sum_{l=T_0}^{T}d_s f_t' \lambda_g (\widehat{\delta}^{(m-1)}_l - \delta_l)\widehat{f}^{(m)\prime}_l(K^{(m)})^{-1}\lambda_s \\
			& \quad - \frac{1}{ST}\sum_{s,g=1}^{S}d_s f_t' \lambda_g (\widehat{\beta}^{(m-1)} - \beta)'X_g'\widehat{F}^{(m)}(K^{(m)})^{-1}\lambda_s.
		\end{align*}
		Using the induction assumptions and H\"{o}lder's and Cauchy-Schwartz inequalities, we obtain
		\begin{align*}
			&\bigg|\frac{1}{S} \sum_{s,g=1}^{S}\sum_{l=T_0}^{T}\omega_{sg} d_s d_g (\widehat{\delta}^{(m-1)}_t - \delta_t) \bigg| 
			= O_P(\sqrt{S}),\\
			&\bigg|\frac{1}{S}\sum_{s,g=1}^{S}\omega_{sg} d_s x_{gt}' (\widehat{\beta}^{(m-1)} - \beta)\bigg\|
			= O_P\bigg(\sqrt{\frac{S}{T}}\bigg),\\
			&\bigg|\frac{1}{ST}\sum_{s,g=1}^{S}\sum_{l=T_0}^{T}d_s f_t' \lambda_g (\widehat{\delta}^{(m-1)}_l - \delta_l)\widehat{f}^{(m)\prime}_l(K^{(m)})^{-1}\lambda_s\bigg| 
			= O_P\bigg({\frac{S}{T}}\bigg),\\
			&\bigg|\frac{1}{ST}\sum_{s,g=1}^{S}d_s f_t' \lambda_g (\widehat{\beta}^{(m-1)} - \beta)'X_g'\widehat{F}^{(m)}(K^{(m)})^{-1}\lambda_s\bigg| 
			= O_P\bigg(\sqrt{\frac{S}{T}}\bigg).
		\end{align*}
		Therefore, we conclude
		\begin{align*}
			\sum_{s=1}^{S}d_s e_t'M^{(m)}_{\widehat{F}}I_2^{(m)}\widehat{F}^{(m)}(K^{(m)})^{-1}\lambda_s 
			= & - \frac{1}{S} \sum_{s,g=1}^{S}\sum_{l=T_0}^{T}\omega_{sg} d_s d_g (\widehat{\delta}^{(m-1)}_t - \delta_t) + O_P\bigg(\sqrt{\frac{S}{T}} \bigg).
		\end{align*}
		
		For the term associated with $I_4^{(m)}$, we expand the terms using the expression of $I_4^{(m)}$ and obtain
		\begin{align*}
			\sum_{s=1}^{S} & d_s e_t'M^{(m)}_{\widehat{F}}I_4^{(m)}\widehat{F}^{(m)}(K^{(m)})^{-1}\lambda_s \\
			= & 
			\frac{1}{ST}\sum_{s,g=1}^{S}
			d_s e_t' M_{\widehat{F}}^{(m)} 
			\bigg\{
			F {\lambda}_g {\eta}_g' 
			+{\eta}_g {\lambda}_g'F'  
			+
			\left({\eta}_{g} {\eta}_{g}' - \mathbb{E}[{\eta}_{g} {\eta}_{g}'] \right)
			+ 
			\mathbb{E}[{\eta}_{g} {\eta}_{g}'] 
			\bigg\} \widehat{F}^{(m)}  \left(K^{(m)}\right)^{-1}
			\lambda_s,
		\end{align*}
		where, by standard arguments, we show that 
		$|(ST)^{-1}\sum_{s,g=1}^{S}d_s e_t' M_{\widehat{F}}^{(m)}({\eta}_{g} {\eta}_{g}' - \mathbb{E}[{\eta}_{g} {\eta}_{g}'] ) \widehat{F}^{(m)} \allowbreak (K^{(m)})^{-1}\lambda_s| = O_P(S/T)$,
		and 
		$|(ST)^{-1}\sum_{s,g=1}^{S}d_s e_t' M_{\widehat{F}}^{(m)}F {\lambda}_g {\eta}_g' \widehat{F}^{(m)} (K^{(m)})^{-1}\lambda_s| \allowbreak = O_P(\sqrt{S/T})$. In addition, 
		\begin{align*}
			\frac{1}{ST}\sum_{s,g=1}^{S}d_s e_t' M_{\widehat{F}}^{(m)} {\eta}_g {\lambda}_g'F'  \widehat{F}^{(m)}  \left(K^{(m)}\right)^{-1}\lambda_s
			= \frac{1}{S} \sum_{s,g=1}^{S}\omega_{sg}d_s\eta_{gt} 
			- \frac{1}{ST}\sum_{s,g=1}^{S}\omega_{sg}d_s\widehat{f}_t'\widehat{F}\eta_{g}, 
		\end{align*}
		where the first term is $O_P(\sqrt{S})$, while the second term is $O_P(\sqrt{S/T})$. 
		The last term 
		{\small
			\begin{align*}
				\frac{1}{ST} & \sum_{s,g=1}^{S}d_s e_t' M_{\widehat{F}}^{(m)}\mathbb{E}[{\eta}_{g} {\eta}_{g}']  \widehat{F}^{(m)} (K^{(m)})^{-1}\lambda_s  \\
				& = \frac{1}{ST}\sum_{s,g=1}^{S}d_s \mathbb{E}[{\eta}_{gt} {\eta}_{g}']  \widehat{F}^{(m)} (K^{(m)})^{-1}\lambda_s  
				- \frac{1}{ST^2}\sum_{s,g=1}^{S}d_s \widehat{f}^{(m)\prime}_t \widehat{F}^{(m)}\mathbb{E}[{\eta}_{g} {\eta}_{g}']  \widehat{F}^{(m)} (K^{(m)})^{-1}\lambda_s,  
		\end{align*}}
		where the first term is $O_P(S/\sqrt{T})$ and the second term is $O_P(S/T)$. Therefore, collecting all terms of $\sum_{s=1}^{S} d_s e_t'M^{(m)}_{\widehat{F}}I_4^{(m)}\widehat{F}^{(m)}(K^{(m)})^{-1}\lambda_s$, we obtain
		\begin{align*}
			\sum_{s=1}^{S} & d_s e_t'M^{(m)}_{\widehat{F}}I_4^{(m)}\widehat{F}^{(m)}(K^{(m)})^{-1}\lambda_s \\
			& = \frac{1}{S} \sum_{s,g=1}^{S}\omega_{sg}d_s\eta_{gt}  + \frac{1}{ST}\sum_{s,g=1}^{S}d_s \mathbb{E}[{\eta}_{gt} {\eta}_{g}']  \widehat{F}^{(m)} (K^{(m)})^{-1}\lambda_s + O_P\bigg(\sqrt{\frac{S}{T}}\bigg),
		\end{align*}
		where the first two terms on the right-hand side are $O_P(\sqrt{S})$ and $O_P(S/\sqrt{T})$, respectively. 
		
		Collecting all terms associated with $I_h^{(m)}$, we obtain
		\begin{align*}
			\sum_{s=1}^{S} d_s e_t'M^{(m)}_{\widehat{F}}(F-\widehat{F}^{(m)}H^{(m)})\lambda_s 
			= & 
			\frac{1}{S} \sum_{s,g=1}^{S}\sum_{l=T_0}^{T}\omega_{sg} d_s d_g (\widehat{\delta}^{(m-1)}_t - \delta_t) -  \frac{1}{S} \sum_{s,g=1}^{S}\omega_{sg}d_s\eta_{gt}  \\
			& - \frac{1}{ST}\sum_{s,g=1}^{S}d_s \mathbb{E}[{\eta}_{gt} {\eta}_{g}']  \widehat{F}^{(m)} (K^{(m)})^{-1}\lambda_s 
			+ O_P\bigg(\sqrt{\frac{S}{T}}\bigg),
		\end{align*}
		where the remaining terms on the right-hand side have been shown to be $O_P(\sqrt{S})$ given that $T/S \to \kappa$. 
		\hfill$\square$

		\section{Preliminary Lemmas} \label{B.lemma}
		This section presents the preliminary lemmas and the proofs that are required in Section~\ref{online.theorem}. 
		\begin{lemma}[Theorem~3 of \cite{chetverikov2016iv}]
			\label{bl.nonasymptotic}
			Under Assumption~\ref{a2.1} and \ref{ass.dens}, there exist constants $\overline{c}, c, C >0$, which are independent of $s,t,N_{st},S,T,N_{ST}$, such that for all $(s,t) \in \{1,...,S\} \times \{1,...,T\}$, $u \in \mathcal{U}$ and $x \in (0, \overline{c})$,
			\begin{align*}
				\mathbb{P}(\|\widehat{\alpha}_{st}(u) - \alpha_{st}(u)\|> x) 
				\leq
				Ce^{-cx^2N_{st}}.
			\end{align*}
		\end{lemma}
		
		\begin{lemma}[Lemma A.1 of \cite{gao2007nonlinear}] \label{al.4}
			Suppose that $\left\{M_{m}^{n}:-\infty<m \leq n<+\infty\right\}$ are the $\sigma$-fields generated by a stationary and $\alpha$-mixing process $\left\{\xi_{i}\right\}_{-\infty}^{+\infty}$ with the mixing coefficient $a(i) .$ For some positive integers $m$, let $\delta_i \in M_{s_i}^{t_{i}}$ where $s_{1}<t_{1}<s_{2}<t_{2}<\cdots<s_{m}<t_{m}$ and assume that $t_{i}-s_{i} \geq \tau$ for all $i$ and some $\tau>0$. Assume further that, for some $p_{i}>1, \mathbb{E}\left|\delta_{i}\right|^{p_{i}}<+\infty,$ for which $Q:=\sum_{i=1}^{\ell} \frac{1}{p_{i}}<1$.
			Then we have 
			$$\left|\mathbb{E}\left(\Pi_{i=1}^{\ell} \delta_{i}\right)-\Pi_{i=1}^{\ell} \mathbb{E}\left(\delta_{i}\right)\right|<10(\ell-1) a(\tau)^{1-Q} \Pi_{i=1}^{\ell}\left(\mathbb{E}\left|\delta_{i}\right|^{p_{i}}\right)^{\frac{1}{p_{i}}}.
			$$
		\end{lemma}
		

		
		\begin{lemma}  \label{lemma:m4-bound}
			If $\E \|x_{st}\|^{4} \le C$ 
			for any $1 \le s \le S$ and $1 \le t \le T$,
			then 
			$\max_{1 \le s \le S} \|X_{s}\| = o_p\big( (ST)^{1/2} \big)$.
		\end{lemma}
		
		\noindent
		\textbf{Proof}
		Let $M>0$ be an arbitrary constant.
		
		By Markov inequality, we have
		\begin{eqnarray*}
			P
			\Big (
			\max_{1 \le s \le S}
			\|X_s\|
			>
			M (ST)^{1/2}
			\Big )
			\le
			\sum_{s=1}^{S}
			P
			\Big (
			\rho_{\max}(X_{s}X_{s}')
			> 
			M^{2} ST
			\Big )
			\le
			\frac{
				\sum_{s=1}^{S}
				\E [\rho_{\max}(X_{s}X_{s}')^2]}{M^{4} (ST)^{2}}.
		\end{eqnarray*}
		In addition, under Assumption~\ref{ass.cov}.(i), we have
		\begin{eqnarray*}
			\sum_{s=1}^{S}
			\E
			\big [
			\rho_{\max}(X_{s}X_{s}')^2
			\big ]
			\le 
			\sum_{s=1}^{S}
			\E
			\big [
			\mathrm{tr}(X_{s}'X_{s})^{2}
			\big]
			\le 
			\sum_{s=1}^{S}
			\E
			\bigg [
			\bigg (
			\sum_{t=1}^{T}
			\| x_{st}\|^2
			\bigg )^{2}
			\bigg ]
			= O(S T^{2}),
		\end{eqnarray*}
		which implies 
		$
		(ST)^{-1/2}
		\max_{1 \le s \le S}
		\sqrt{\rho_{\max}(X_{s}X_{s}')}
		= o_P(1)
		$.
		\hfill$\square$ 
		
		\begin{lemma}\label{bl.2}
			Under the assumptions of Theorem~\ref{t5}, for any fixed $u \in \mathcal{U}$, $j=1,...,J$ and $m \geq 1$, 
			\begin{align*}
				\frac{1}{S T} &  \sum_{s,g=1}^{S} X_s' M_{\widehat{F}_j}^{(m)}(u)
				\left({\eta}_{js}(u) {\eta}_{js}(u)' - \mathbb{E}[{\eta}_{js}(u) {\eta}_{js}(u)'] \right)
				\widehat{F}^{(m)}_j(u) \left(K^{(m)}_j(u)\right)^{-1} \lambda_{js}(u) \\
				= & O_P(\sqrt{T}) \cdot 
				\bigg(\bigg\|\widehat{F}^{(m)}_j(u)-F_j(u) \left(H^{(m}(u)\right)^{-1}\bigg\|^2 + 
				\Big\|\widehat{F}^{(m)}_j(u) - F_j(u) \left(H^{(m}(u)\right)^{-1}\Big\| \bigg) \\
				& + O_P\left(\delta_{ST}\right).
			\end{align*}
		\end{lemma}
		
		\noindent
		\textbf{Proof}
		Since $u$ and $j$ are fixed, we suppress $u$ and $j$ throughout the following proof.
		
		We first plug-in $M_{\widehat{F}}^{(m)} = I_T - T^{-1}\widehat{F}^{(m)}\widehat{F}^{(m)\prime}$ to obtain
		\begin{align*}
			\frac{1}{S T} \sum_{s,g=1}^{S} 
			&
			X_s' M_{\widehat{F}}^{(m)}
			\left({\eta}_{s} {\eta}_{s}' - \mathbb{E}[{\eta}_{s} {\eta}_{s}'] \right)
			\widehat{F}^{(m)} \left(K^{(m)}\right)^{-1} \lambda_s \\
			= & \frac{1}{ST} \sum_{s,g=1}^{S} X_{s}^{\prime} 
			\left({\eta}_{g} {\eta}_{g}' - \mathbb{E}[{\eta}_{g} {\eta}_{g}'] \right)
			\widehat{F}^{(m)} \left(K^{(m)}\right)^{-1} {\lambda}_{s} \\
			& -  \frac{1}{S T} \sum_{s=1}^{S} \frac{X_s' \widehat{F}^{(m)}}{T} 
			\bigg(\sum_{g=1}^{S} \widehat{F}^{(m)\prime}
			\left({\eta}_{g} {\eta}_{g}' - \mathbb{E}[{\eta}_{g} {\eta}_{g}'] \right)
			\widehat{F}^{(m)} \left(K^{(m)}\right)^{-1}
			\bigg) {\lambda}_{s} \\
			= & 
			\frac{1}{ST} \sum_{s,g=1}^{S} X_{s}^{\prime} 
			\left({\eta}_{g} {\eta}_{g}' - \mathbb{E}[{\eta}_{g} {\eta}_{g}'] \right)
			\left(\widehat{F}^{(m)} - F \left(H^{(m}\right)^{-1}\right) 
			\left(K^{(m)}\right)^{-1} {\lambda}_{s} \\
			& + \frac{1}{ST} \sum_{s,g=1}^{S} X_{s}^{\prime} 
			\left({\eta}_{g} {\eta}_{g}' - \mathbb{E}[{\eta}_{g} {\eta}_{g}'] \right)
			F \left(H^{(m}\right)^{-1} 
			\left(K^{(m)}\right)^{-1} {\lambda}_{s} \\
			& -  \frac{1}{S T} \sum_{s=1}^{S} \frac{X_s' \widehat{F}^{(m)}}{T} 
			\bigg(\sum_{g=1}^{S} \widehat{F}^{(m)\prime}
			\left({\eta}_{g} {\eta}_{g}' - \mathbb{E}[{\eta}_{g} {\eta}_{g}'] \right)
			\widehat{F}^{(m)} \left(K^{(m)}\right)^{-1}
			\bigg) {\lambda}_{s} ,
		\end{align*}
		and then we evaluate the three terms in the last equality separately. 
		
		Apply Cauchy-Schwartz inequality to the first term, we obtain
		\begin{align*}
			\bigg\|\frac{1}{ST}
			&
			\sum_{s,g=1}^{S} X_{s}^{\prime} 
			\left({\eta}_{g} {\eta}_{g}' - \mathbb{E}[{\eta}_{g} {\eta}_{g}'] \right)
			\left(\widehat{F}^{(m)} - F \left(H^{(m}\right)^{-1}\right) 
			\left(K^{(m)}\right)^{-1} {\lambda}_{s}\bigg\| \\
			\leq &
			\frac{1}{\sqrt{S}} 
			\left\|\widehat{F}^{(m)} - F \left(H^{(m}\right)^{-1}\right\|
			\cdot \left\|\left(K^{(m)}\right)^{-1}\right\| \\
			& \cdot 
			\left(
			\sum_{s=1}^{S} \left\|\frac{1}{T\sqrt{S}}\sum_{g=1}^{S} {X}_{s}^{\prime} 
			\left({\eta}_{g} {\eta}_{g}' - \mathbb{E}[{\eta}_{g} {\eta}_{g}'] \right) \right\|^2
			\right)^{1/2}
			\left(
			\sum_{s=1}^{S} \left\|{\lambda}_{s} \right\|^2\right)^{1/2}\\
			= & \frac{1}{\sqrt{S}} \cdot  
			O_P\Big(\Big\|\widehat{F}^{(m)} - F \left(H^{(m}\right)^{-1}\Big\|\Big)
			\cdot O_P(1) 
			\cdot
			O_P(1) \cdot O_P(\sqrt{S}) 
			\\
			= & O_P \Big( \Big\|\widehat{F}^{(m)} - F \left(H^{(m}\right)^{-1}\Big\| \Big),
		\end{align*}
		due to Assumption~\ref{ass.cov}.(ii), (\ref{al6.eq1}), (\ref{prop3p.eq0}), and the fact that, for any fixed $s=1,...,S$,
		\allowdisplaybreaks
		\begin{align} \label{t2p.eq2}
			\mathbb{E}\left[\left\|\frac{1}{T\sqrt{S}}\sum_{g=1}^{S} {X}_{s}^{\prime} 
			\left({\eta}_{g} {\eta}_{g}' - \mathbb{E}[{\eta}_{g} {\eta}_{g}'] \right) \right\|^2 \right] 
			= & \frac{C}{ST^2} \cdot O(ST^2) = O(1), 
		\end{align}
		according to Lemma~\ref{al.4} under Assumption~\ref{ass.cov} and Assumption~\ref{ass.mixing}.
		
		For the second term, similarly we have
		\begin{align*}
			\bigg\|\frac{1}{ST}
			&
			\sum_{s,g=1}^{S} X_{s}^{\prime} 
			\left({\eta}_{g} {\eta}_{g}' - \mathbb{E}[{\eta}_{g} {\eta}_{g}'] \right)
			F \left(H^{(m}\right)^{-1} 
			\left(K^{(m)}\right)^{-1} {\lambda}_{s}\bigg\| \\
			\leq & \frac{1}{\sqrt{S}} 
			\cdot \left\|\left(H^{(m}\right)^{-1}\right\| 
			\cdot \left\|\left(K^{(m)}\right)^{-1}\right\|\\
			& \cdot 
			\Bigg(\sum_{s=1}^{S} \bigg\|\frac{1}{T\sqrt{S}} 
			\sum_{g=1}^{S} \sum_{t,l=1}^{T} {x}_{st} 
			\left(\eta_{gt} \eta_{gl} - \mathbb{E}[\eta_{gt} \eta_{gl}] \right)
			{f}_l' \bigg\|^2\bigg)^{1/2}
			\bigg(\sum_{s=1}^{S}\|{\lambda}_s\|^2\bigg)^{1/2} \\
			= & O_P \big(\sqrt{S}\big).
		\end{align*}
		We consider the third term as follows
		\begin{align*}
			\bigg\|\frac{1}{S T}
			&
			\sum_{s=1}^{S} \frac{X_s' \widehat{F}^{(m)}}{T} 
			\bigg(\sum_{g=1}^{S} \widehat{F}^{(m)\prime}
			\left({\eta}_{g} {\eta}_{g}' - \mathbb{E}[{\eta}_{g} {\eta}_{g}'] \right)
			\widehat{F}^{(m)} \left(K^{(m)}\right)^{-1}
			\bigg) {\lambda}_{s}
			\bigg\| \\
			\leq & 
			\bigg(\sum_{s=1}^{S}\|X_s\|^2\bigg)^{1/2}
			\cdot 
			\bigg(\sum_{s=1}^{S}\|\lambda_s\|^2\bigg)^{1/2}
			\cdot \big\|\widehat{F}^{(m)}\big\|
			\cdot \left\|\left(K^{(m)}\right)^{-1}  \right\|\\
			& 
			\cdot \left\| \frac{1}{S T^2} \sum_{g=1}^{S} \widehat{F}^{(m)\prime}
			\left({\eta}_{g} {\eta}_{g}' - \mathbb{E}[{\eta}_{g} {\eta}_{g}'] \right) \widehat{F}^{(m)} \right\| \\
			=&  O_P(ST) \cdot \bigg\| \frac{1}{S T^2} \sum_{g=1}^{S} \widehat{F}^{(m)\prime}
			\left({\eta}_{g} {\eta}_{g}' - \mathbb{E}[{\eta}_{g} {\eta}_{g}'] \right) \widehat{F}^{(m)} \bigg\|\\
			= & O_P(\sqrt{T}) \cdot 
			\left\|\widehat{F}^{(m)}-F \left(H^{(m}\right)^{-1}\right\|^2
			+ O_P(\sqrt{T}) \cdot 
			\left\|\widehat{F}^{(m)}-F \left(H^{(m}\right)^{-1}\right\|
			+ O_P\left(\delta_{ST}\right),
		\end{align*}
		where the last equality follows from the following result
		{\small
			\begin{align*}
				& \frac{1}{S T^2}
				\sum_{g=1}^{S} \widehat{F}^{(m)\prime}
				\left({\eta}_{g} {\eta}_{g}' - \mathbb{E}[{\eta}_{g} {\eta}_{g}'] \right) \widehat{F}^{(m)} \\
				& = O_P\bigg(\frac{1}{T\sqrt{S}} \bigg) \cdot 
				\left\|\widehat{F}^{(m)}-F \left(H^{(m}\right)^{-1}\right\|^2
				+ O_P\bigg(\frac{1}{T\sqrt{S}} \bigg) \cdot 
				\left\|\widehat{F}^{(m)}-F \left(H^{(m}\right)^{-1}\right\|
				+ O_P\left(\frac{1}{T\sqrt{S}}\right)
				,
		\end{align*}}
		whose proof is given as follows.
		By adding and subtracting $F (H^{(m)})^{-1}$ we have
		\begin{align*}
			\frac{1}{S T^2}
			&
			\sum_{g=1}^{S} \widehat{F}^{(m)\prime}
			\left({\eta}_{g} {\eta}_{g}' - \mathbb{E}[{\eta}_{g} {\eta}_{g}'] \right) \widehat{F}^{(m)} \\
			= &
			\frac{1}{ST^2} \sum_{g=1}^{S} 
			\left(\widehat{F}^{(m)}-F \left(H^{(m)}\right)^{-1}\right)' 
			\left({\eta}_{g}' {\eta}_{g} - \mathbb{E}[{\eta}_{g}' {\eta}_{g}] \right)
			\left(\widehat{F}^{(m)}-F \left(H^{(m)}\right)^{-1}\right) \\
			& +
			\frac{1}{ST^2} \sum_{g=1}^{S} \left(H^{(m}\right)^{-1}F' \left({\eta}_{g} {\eta}_{g}' - \mathbb{E}[{\eta}_{g} {\eta}_{g}'] \right) \left(\widehat{F}^{(m)}-F \left(H^{(m}\right)^{-1}\right) \\
			& + \frac{1}{ST^2} \sum_{g=1}^{S} 
			\left(\widehat{F}^{(m)}-F \left(H^{(m)}\right)^{-1}\right)' 
			\left(\eta_{g}' \eta_{g} - \mathbb{E}[{\eta}_{g}' {\eta}_{g}] \right)
			F \left(H^{(m)\prime}\right)^{-1} \\
			& + \frac{1}{ST^2} \sum_{g=1}^{S} \left(H^{(m)}\right)^{-1}F' \left({\eta}_{g} {\eta}_{g}' - \mathbb{E}[{\eta}_{g} {\eta}_{g}'] \right) F \left(H^{(m)}\right)^{-1}.
		\end{align*}
		The first term
		\begin{align*}
			\bigg\|\frac{1}{ST^2} 
			&
			\sum_{g=1}^{S} 
			\left(\widehat{F}^{(m)}-F \left(H^{(m}\right)^{-1}\right)' 
			\left({\eta}_{g}' {\eta}_{g} - \mathbb{E}[{\eta}_{g}' {\eta}_{g}] \right)
			\left(\widehat{F}^{(m)}-F \left(H^{(m}\right)^{-1}\right)\bigg\| \\
			\leq & \frac{1}{T\sqrt{S}}
			\left\|\widehat{F}^{(m)}-F \left(H^{(m}\right)^{-1}\right\|^2
			\cdot \bigg\|\frac{1}{T\sqrt{S}}\sum_{g=1}^{S} \left({\eta}_{g}' {\eta}_{g} - \mathbb{E}[{\eta}_{g}' {\eta}_{g}] \right)\bigg\| \\
			= & O_P\bigg(\frac{1}{T\sqrt{S}} \bigg) \cdot 
			\left\|\widehat{F}^{(m)}-F \left(H^{(m}\right)^{-1}\right\|^2,
		\end{align*}
		using the fact that
		$
		\big\|\sum_{g=1}^{S} {\eta}_{g} {\eta}_{g}' - \mathbb{E}[{\eta}_{g} {\eta}_{g}'] \big\| = O_P\big(T\sqrt{S}\big),
		$
		which can be shown in a similar way to the proof of (\ref{t2p.eq2}). 
		Similar arguments yield that the second and third term are
		$O_P ( \|\widehat{F}^{(m)} - F \left(H^{(m}\right)^{-1}\| ) + O_P(\max(S^{-1}T^{-1/2},S^{-1/2}T^{-1}))$, and the fourth term is
		$O_P(S^{-1/2}T^{-1})$.

		Collecting all the terms so far, we obtain
		\begin{align*}
			\frac{1}{S T} \sum_{s,g=1}^{S} 
			&
			X_s' M_{\widehat{F}}^{(m)}
			\left({\eta}_{s} {\eta}_{s}' - \mathbb{E}[{\eta}_{s} {\eta}_{s}'] \right)
			\widehat{F}^{(m)} \left(K^{(m)}\right)^{-1} \lambda_s \\
			& =O_P(\sqrt{T}) \cdot 
			\bigg(\bigg\|\widehat{F}^{(m)}-F \left(H^{(m}\right)^{-1}\bigg\|^2 
			+\Big\|\widehat{F}^{(m)} - F \left(H^{(m}\right)^{-1}\Big\| \bigg)
			+ O_P\left(\delta_{ST}\right).
		\end{align*}
		\hfill$\square$

		\begin{lemma}\label{al.6}
			Under assumptions of Theorem~\ref{t5}, for any fixed $u \in \mathcal{U}$, $j=1,...,J$ and $m \geq 1$, 
			\begin{align}
				& \left\|\frac{1}{T}F_j'(u)\left(\widehat{F}_j^{(m)}(u)H_j^{(m)}(u) - F_j(u)\right) \right\|
				= O_P\left(\frac{1}{\sqrt{ST}}\right) 
				\cdot \bigg\| \widehat{F}^{(m)}  - F \big(H^{(m)}\big)^{-1}\bigg\|,  \label{al6.eq4}\\
				& \left\| I_r - \big(H^{(m)}_j(u)'\big)^{-1}\big(H^{(m)}_j(u)\big)^{-1} \right\|  \label{al6.eq2} 
				\\
				& \qquad  = \frac{1}{T}\left\|\widehat{F}^{(m)} - F \big(H^{(m)}\big)^{-1}\right\|^2 
				+ O_P\Big(\frac{1}{\sqrt{ST}}\Big) 
				\cdot \left\|\widehat{F}^{(m)} - F \big(H^{(m)}\big)^{-1}\right\|,\nonumber\\
				& \left\|P_{\widehat{F}_j}^{(m)}  - P_{F_j} \right\| \label{al6.eq3}\\
				& \qquad
				= \frac{1}{T}\left\|\widehat{F}^{(m)} - F \big(H^{(m)}\big)^{-1}\right\|^2 
				+ O_P\Big(\frac{1}{\sqrt{ST}}\Big) 
				\cdot \left\|\widehat{F}^{(m)} - F \big(H^{(m)}\big)^{-1}\right\| , \nonumber
			\end{align}
			where $H_j^{(m)}$ is defined in Lemma~\ref{prop3}, and $\| \widehat{F}^{(m)}  - F \big(H^{(m)}\big)^{-1}\| $ is given by (\ref{al6.eq1}).
		\end{lemma}

		\noindent
		\textbf{Proof}
		Since $u$ and $j$ are fixed, we suppress $u$ and $j$ throughout the following proof.

		For (\ref{al6.eq4}), according to the proof of Lemma~\ref{prop3}, we have the following decomposition:
		\begin{align*}
			\frac{1}{T}F' \left(\widehat{F}^{(m)} H^{(m)} - F  \right) 
			= \sum_{h=1}^{7} \frac{1}{T} F'I_h^{(m)}\widehat{F}^{(m)}\left(K^{(m)}\right)^{-1}
		\end{align*}
		where $I_h^{(m)},\ h=1,...,7$ are defined above (\ref{prop2p.eq0}) 
		Using the rate of $\|I_h^{(m)}\|$ from the proof of Lemma~\ref{prop3} and (\ref{prop3p.eq0}), it is straightforward to check that for $h=1$,
		{\small
			\begin{align*}
				\left\|\frac{1}{T} F'I_h^{(m)}\widehat{F}^{(m)}\left(K^{(m)}\right)^{-1}\right\|
				\leq  \left\|\frac{F'}{\sqrt{T}} \right\| \cdot \left\|I_1^{(m)}\right\| 
				\cdot \left\|\frac{\widehat{F}^{(m)}}{\sqrt{T}}\right\|\cdot \left\|\left(K^{(m)}\right)^{-1}\right\|
				= o_P(1) \cdot \left\|\widehat{{\beta}}^{(m-1)}-{\beta}\right\|.
		\end{align*}}
		Applying similar arguments, we obtain $\big\|T^{-1}F'I_h^{(m)}\widehat{F}^{(m)}\big(K^{(m)}\big)^{-1}\big\| = O_P( \big\|\widehat{{\beta}}^{(m-1)}-{\beta}\big\|)$ for $h=2,3,5,6,7$. Thus, it remains to check the rate  of
		\begin{align}
			\label{al6p.eq1}
			\frac{1}{T} 
			F'I_4^{(m)}\widehat{F}^{(m)}\left(K^{(m)}\right)^{-1} 
			= &
			\frac{1}{ST^2} \sum_{s=1}^{S} 
			F' 
			\bigg[F {\lambda}_s {\eta}_s' 
			+
			{\eta}_s  {\lambda}_s'  F'
			+
			{\eta}_s  {\eta}_s'
			\bigg]
			\widehat{F}^{(m)}\left(K^{(m)}\right)^{-1} \nonumber\\
			= &  
			\frac{1}{ST} \sum_{s=1}^{S} {\lambda}_s {\eta}_s' \widehat{F}^{(m)} \left(K^{(m)}\right)^{-1} 
			+ 
			\frac{1}{ST} \sum_{s=1}^{S} F'{\eta}_s  {\lambda}_s'  \left(\frac{\Lambda'\Lambda}{S}\right)^{-1} \nonumber\\
			&
			+ 
			\frac{1}{S T^2} \sum_{s=1}^{S} F' {\eta}_s  {\eta}_s'  \widehat{F}^{(m)}  \cdot \left(K^{(m)}\right)^{-1} . 
		\end{align}
		The first term in the last equality of (\ref{al6p.eq1})
		\begin{align*}
			\bigg\|\frac{1}{ST}
			& 
			\sum_{s=1}^{S} {\lambda}_s {\eta}_s' \widehat{F}^{(m)} \left(K^{(m)}\right)^{-1} \bigg\| \\
			\leq & \bigg[ \bigg\| \frac{1}{ST}\sum_{s=1}^{S} {\lambda}_s {\eta}_s' \left(\widehat{F}^{(m)}  - F \big(H^{(m)}\big)^{-1} \right) \bigg\| 
			+ \bigg\| \frac{1}{ST}\sum_{s=1}^{S} {\lambda}_s {\eta}_s' F \big(H^{(m)}\big)^{-1} \bigg\| \bigg]
			\cdot \left\| \Big(K^{(m)}\Big)^{-1} \right\| \\
			\leq & \bigg[\bigg\| \frac{1}{ST}\sum_{s=1}^{S} {\lambda}_s {\eta}_s' \bigg\| 
			\cdot \bigg\| \widehat{F}^{(m)}  - F \big(H^{(m)}\big)^{-1}\bigg\|
			+ \bigg\| \frac{1}{ST}\sum_{s=1}^{S} {\lambda}_s {\eta}_s' F \bigg\|
			\cdot \bigg\| \big(H^{(m)}\big)^{-1} \bigg\| \bigg]
			\cdot O_P(1) \\
			= & O_P\left(\frac{1}{\sqrt{ST}}\right) 
			\cdot \bigg\| \widehat{F}^{(m)}  - F \big(H^{(m)}\big)^{-1}\bigg\|
			+ O_P\left(\frac{1}{\sqrt{ST}}\right)
			\cdot O_P(1) \\
			= & O_P\left(\frac{1}{\sqrt{ST}}\right) 
			\cdot \bigg\| \widehat{F}^{(m)}  - F \big(H^{(m)}\big)^{-1}\bigg\|,
		\end{align*}
		where the second last equality follows from (\ref{al6.eq1}) and Lemma~\ref{al.5}.
		Following Lemma~\ref{al.5}, the second term in the last equality of (\ref{al6p.eq1}) becomes
		\begin{align*}
			\bigg\|\frac{1}{ST} \sum_{s=1}^{S} F'{\eta}_s  {\lambda}_s'  \left(\frac{\Lambda'\Lambda}{S}\right)^{-1} \bigg\|
			\leq \left\| \frac{1}{ST} \sum_{s=1}^{S} F'{\eta}_s  {\lambda}_s'\right\|
			\cdot \left\| \left(\frac{\Lambda'\Lambda}{S}\right)^{-1} \right\|  
			= O_P\left(\frac{1}{\sqrt{ST}}\right).
		\end{align*}
		The third term in the last equality of (\ref{al6p.eq1}) is
		\begin{align*}
			\bigg\| \frac{1}{S T^2}
			&
			\sum_{s=1}^{S} F' {\eta}_s  {\eta}_s'  \widehat{F}^{(m)}  \cdot \left(K^{(m)}\right)^{-1}  \bigg\| \\
			\leq &  \bigg\| \frac{1}{S T^2} \sum_{s=1}^{S} F' {\eta}_s  {\eta}_s'  \left(\widehat{F}^{(m)} - F \big(H^{(m)}\big)^{-1}\right) \bigg\| 
			\cdot \left\|\left(K^{(m)}\right)^{-1} \right\| \\
			& + \bigg\| \frac{1}{S T^2} \sum_{s=1}^{S} F' {\eta}_s  {\eta}_s'  F\big(H^{(m)}\big)^{-1} \bigg\| 
			\cdot \left\|\left(K^{(m)}\right)^{-1} \right\|\\
			\leq & \bigg\{\left\|\frac{F'}{\sqrt{T}}\right\| 
			\cdot \bigg\| \frac{1}{ST} \sum_{s=1}^{S}{\eta}_s  {\eta}_s' \bigg\|
			\cdot \frac{1}{\sqrt{T}}\left\|\widehat{F}^{(m)} - F \big(H^{(m)}\big)^{-1}\right\| \\
			& +
			\bigg\| \frac{1}{S T^2} \sum_{s=1}^{S} F' {\eta}_s  {\eta}_s'  F  \bigg\|
			\cdot \left\|\big(H^{(m)}\big)^{-1}\right\| \bigg\}
			\cdot \left\|\left(K^{(m)}\right)^{-1} \right\| \\
			= & O_P\Big(\frac{1}{\sqrt{ST}}\Big) 
			\cdot \left\|\widehat{F}^{(m)} - F \big(H^{(m)}\big)^{-1}\right\|
			+ O_P\left(\frac{1}{T}\right) \\
			= & O_P\Big(\frac{1}{\sqrt{ST}}\Big) 
			\cdot \left\|\widehat{F}^{(m)} - F \big(H^{(m)}\big)^{-1}\right\|,
		\end{align*}
		due to Lemma~\ref{al.5} and the following fact that
		\begin{align*}
			\mathbb{E}\bigg[\bigg\| 
			& \frac{1}{S T^2} \sum_{s=1}^{S} F' {\eta}_s  {\eta}_s'  F  \bigg\| \bigg]
			\leq 
			\mathbb{E} \bigg[ \frac{1}{S T^2}\bigg\| \sum_{s=1}^{S} F' {\eta}_s  \bigg\|^2 \bigg] 
			\leq  \frac{1}{S T^2}\sum_{s=1}^{S} \sum_{t,l=1}^{T} \mathbb{E} \left[ \eta_{st} {f}_{t}' {f}_{l} \eta_{sl}  \right] \\
			= &  \frac{1}{S T^2}\sum_{s=1}^{S} \sum_{t = l}^{T} \mathbb{E} \left[ \eta_{st}^2 {f}_{t}' {f}_{t}  \right] 
			+  \frac{1}{S T^2}\sum_{s=1}^{S} \sum_{t \neq l}^{T} \mathbb{E} \left[ \eta_{st}{f}_{t}' {f}_{l} \eta_{sl}  \right] \\
			\leq & O_P\left(\frac{1}{T}\right) 
			+ \frac{C}{S T^2}\sum_{i=1}^{r}\sum_{s=1}^{S} \sum_{t \neq l}^{T} 
			a_{s}(|t-l|)^{\frac{\delta}{\delta+2}}
			\left(\mathbb{E}\left[|\eta_{st}f_{ti}|^{\delta+2}\right]\right)^{\frac{1}{\delta+2}}
			\left(\mathbb{E}\left[|\eta_{sl}f_{li}|^{\delta+2}\right]\right)^{\frac{1}{\delta+2}} \\
			= &  O_P\left(\frac{1}{T}\right) ,
		\end{align*}
		by Lemma~\ref{al.4}.
		Collecting all the terms, we have
		\begin{align*}
			\left\| \frac{1}{T}F' \left(\widehat{F}^{(m)} H^{(m)} - F  \right)  \right\|
			= O_P\Big(\frac{1}{\sqrt{ST}}\Big) 
			\cdot \left\|\widehat{F}^{(m)} - F \big(H^{(m)}\big)^{-1}\right\|.
		\end{align*}
		
		For (\ref{al6.eq2}), by multiplying $\big(H^{(m)\prime}\big)^{-1}$ and $\big(H^{(m)}\big)^{-1}$ on both sides of $T^{-1} F' (\widehat{F}^{(m)} H^{(m)} - F )$, respectively, we obtain
		\begin{align*} 
			\bigg\|\frac{1}{T} \big(H^{(m)\prime}\big)^{-1}F' \widehat{F}^{(m)} 
			- \big(H^{(m)\prime}\big)^{-1}\big(H^{(m)}\big)^{-1} \bigg\| 
			=  \left\|\widehat{F}^{(m)} - F \big(H^{(m)}\big)^{-1}\right\|,
		\end{align*}
		since $\left\|F\right\| = O_P(\sqrt{T})$ under Assumption~\ref{ass.cov}.(ii), and
		$
		\big\|\big(H^{(m)}\big)^{-1}\big\| \leq \big\|\big(\Upsilon^{(m)}\big)^{-1}\big\| \cdot \big\|{K}^{(m)}\big\| = O_P(1).
		$
		Moreover, we have
		\begin{align*} 
			\bigg\|I_r 
			& - \frac{1}{T} \big(H^{(m)\prime}\big)^{-1}F' \widehat{F}^{(m)} \bigg\|
			= \frac{1}{T} \bigg\|\widehat{F}^{(m)\prime} \left( \widehat{F}^{(m)} - F\big(H^{(m)\prime}\big)^{-1}\right)  \bigg\| \nonumber\\
			& \leq  \frac{1}{T} \left\| \widehat{F}^{(m)} - F\big(H^{(m)\prime}\big)^{-1}  \right\|^2
			+ \left\|\big(H^{(m)}\big)^{-1}\right\| \cdot \left\| \frac{1}{T} F'\left( \widehat{F}^{(m)} - F\big(H^{(m)\prime}\big)^{-1}\right)  \right\| \nonumber \\
			&  = \frac{1}{T}\left\|\widehat{F}^{(m)} - F \big(H^{(m)}\big)^{-1}\right\|^2 
			+ O_P\Big(\frac{1}{\sqrt{ST}}\Big) 
			\cdot \left\|\widehat{F}^{(m)} - F \big(H^{(m)}\big)^{-1}\right\| .
		\end{align*}
		Therefore, by triangle inequality, we show that
		\begin{align*}
			\bigg\|I_r &
			- \big(H^{(m)\prime}\big)^{-1}\big(H^{(m)}\big)^{-1} \bigg\| \\
			& \leq  \left\|I_r - \frac{1}{T} \big(H^{(m)\prime}\big)^{-1}F' \widehat{F}^{(m)} \right\|
			+ \left\|\frac{1}{T} \big(H^{(m)\prime}\big)^{-1}F' \widehat{F}^{(m)} 
			- \big(H^{(m)\prime}\big)^{-1}\big(H^{(m)}\big)^{-1} \right\| \\
			& = \frac{1}{T}\left\|\widehat{F}^{(m)} - F \big(H^{(m)}\big)^{-1}\right\|^2 
			+ O_P\Big(\frac{1}{\sqrt{ST}}\Big) 
			\cdot \left\|\widehat{F}^{(m)} - F \big(H^{(m)}\big)^{-1}\right\| .
		\end{align*}
		
		For (\ref{al6.eq3}), since
		\begin{align*}
			\left\|P_{\widehat{F}}^{(m)}  - P_{F} \right\|^2
			\leq \text{tr}\left(P_{\widehat{F}}^{(m)}  - P_{F}\right)^2
			= 2 \text{tr}\left(I_r - \frac{1}{T} {F}^{\prime}P_{ \widehat{F}^{(m)}}{F}\right),
		\end{align*}
		it is suffices to examine 
		\begin{align*}
			I_r &
			- \frac{1}{T} {F}^{\prime}P_{ \widehat{F}^{(m)}}{F} 
			=  I_r - \frac{F' \widehat{F}^{(m)}}{T} \frac{ \widehat{F}^{(m)\prime} F}{T} \\
			= &  I_r
			- \left[\frac{1}{T}\widehat{F}^{(m)\prime}\left(F-\widehat{F}^{(m)} \big(H^{(m)}\big)^{-1}\right) + \big(H^{(m)}\big)^{-1}\right]' \\
			&
			\cdot \left[\frac{1}{T}\widehat{F}^{(m)\prime}\left(F-\widehat{F}^{(m)}\big(H^{(m)}\big)^{-1}\right) + \big(H^{(m)}\big)^{-1}\right] \\
			= & I_r - \big(H^{(m)\prime}\big)^{-1}\big(H^{(m)}\big)^{-1}
			-
			\frac{1}{T^2}\left(F-\widehat{F}^{(m)} \big(H^{(m)}\big)^{-1}\right)'\widehat{F}^{(m)}\widehat{F}^{(m)\prime}\left(F-\widehat{F}^{(m)}\big(H^{(m)}\big)^{-1}\right)
			\\
			& - \big(H^{(m)\prime}\big)^{-1}\cdot \frac{1}{T}\widehat{F}^{(m)\prime}\left(F-\widehat{F}^{(m)} \big(H^{(m)}\big)^{-1}\right)
			- \frac{1}{T}\left(F-\widehat{F}^{(m)} \big(H^{(m)}\big)^{-1}\right)'\widehat{F}^{(m)}\cdot \big(H^{(m)}\big)^{-1}
			\\
			= &
			\frac{1}{T}\left\|\widehat{F}^{(m)} - F \big(H^{(m)}\big)^{-1}\right\|^2 
			+ O_P\Big(\frac{1}{\sqrt{ST}}\Big) 
			\cdot \left\|\widehat{F}^{(m)} - F \big(H^{(m)}\big)^{-1}\right\|  ,
		\end{align*}
		where the second equality follows from the fact that $T^{-1}\widehat{F}^{(m)\prime}\widehat{F}^{(m)} = I_r$, and the last equality follows from (\ref{al6.eq4})-(\ref{al6.eq2}) and $\big\|(H^{(m)})^{-1}\big\| = O_P(1)$. 	
		\hfill$\square$


		\begin{lemma} \label{factor.l1}
			Under the assumptions of Theorem~\ref{t5}, we have for any fixed $u\in \mathcal{U}$, $j=1,...,J$, $m \geq 1$, $s=1,...,S$,
			{\small
				\begin{align}
					& \left\|\widehat{{\lambda}}_{js}^{(m)}(u) - H_j^{(m)}(u)  {\lambda}_{js}(u) \right\| \label{fl1.eq2}\\
					& 
					\quad = O_P\bigg(\frac{1}{\sqrt{T}}\bigg) \cdot \bigg|\sum_{t=T_0}^{T}{\delta}_t - \widehat{{\delta}}^{(m-1)}_t\bigg|
					+ 
					O_P\left( \left\|{\beta} - \widehat{{\beta}}^{(m-1)}\right\|  \right)
					+ O_P\left(\frac{1}{\sqrt{T}}\right) 
					\cdot \bigg\| \widehat{F}^{(m)}  - F \big(H^{(m)}\big)^{-1}\bigg\| , \nonumber \\
					&
					\frac{1}{S}\left\|{\widehat{\Lambda}}_{j}^{(m)}(u) -  {\Lambda}_{j}(u)H^{(m)}_j(u) '\right\|^2   \label{fl1.eq3}\\
					& = \quad O_P\bigg(\frac{1}{T}\bigg) \cdot \bigg|\sum_{t=T_0}^{T}{\delta}_t - \widehat{{\delta}}^{(m-1)}_t\bigg|^2
					+ 
					O_P\left( \left\|{\beta} - \widehat{{\beta}}^{(m-1)}\right\| ^2 \right)
					+ O_P\left(\frac{1}{T}\right) 
					\cdot \bigg\| \widehat{F}^{(m)}  - F \big(H^{(m)}\big)^{-1}\bigg\|^2,
					\nonumber\\
					& 
					\left\|\bigg(\frac{{\widehat{\Lambda}}_{j}^{(m)}(u)'\widehat{\Lambda}_{j}(u)}{S}\bigg)^{-1} -  \big(H^{(m)}_j(u)'\big)^{-1}\bigg(\frac{{{\Lambda}}_{j}^{(m)}(u)'{\Lambda}_{j}(u)}{S}\bigg)^{-1}\big(H^{(m)}_j(u)\big)^{-1}\right\| \label{fl1.eq4}\\
					& \quad
					= O_P\bigg(\frac{1}{T}\bigg) \cdot \bigg|\sum_{t=T_0}^{T}{\delta}_t - \widehat{{\delta}}^{(m-1)}_t\bigg|^2
					+ 
					O_P\left( \left\|{\beta} - \widehat{{\beta}}^{(m-1)}\right\| ^2 \right)
					+ O_P\left(\frac{1}{T}\right) 
					\cdot \bigg\| \widehat{F}^{(m)}  - F \big(H^{(m)}\big)^{-1}\bigg\|^2,  \nonumber
			\end{align}}
			where $H_j^{(m)}$ is defined in Lemma~\ref{prop3}. 
			
		\end{lemma}
		
		\noindent
		\textbf{Proof}
		Since $u$ and $j$ are fixed in Lemma~\ref{factor.l1}, we suppress $u$ and $j$ throughout the following proof.
		
		As $\widehat{\lambda}_s^{(m)}$ is estimated via PCA, we have $\widehat{\lambda}_s^{(m)} =  T^{-1}\widehat{F}^{(m)\prime} \big(\widehat{{A}}_s - D_s\widehat{\delta}^{(m-1)} - X_s \widehat{{\beta}}^{(m-1)}\big)$. Thus, for (\ref{fl1.eq2}), substituting $\widehat{{A}}_s = D_s \delta + X_s {\beta}  + F {\lambda}_s + {\eta}_s + (\widehat{{A}}_s - {A}_s)$ into the expression of $\widehat{{\lambda}}_s^{(m)}$, we obtain
		\begin{align*}
			\widehat{{\lambda}}_s^{(m)} 
			& =  \frac{1}{T}\widehat{F}^{(m)\prime} \left(\widehat{{A}}_s - D_s\widehat{\delta}^{(m-1)} - X_s \widehat{{\beta}}^{(m-1)}\right) \\
			& =   \frac{1}{T}\widehat{F}^{(m)\prime} D_s \left({\delta} - \widehat{{\delta}}^{(m-1)}\right) 
			+ \frac{1}{T}\widehat{F}^{(m)\prime} X_s \left({\beta} - \widehat{{\beta}}^{(m-1)}\right) \\
			&\quad + \frac{1}{T}\widehat{F}^{(m)\prime} F {\lambda}_s + \frac{1}{T}\widehat{F}^{(m)\prime} {\eta}_s 
			+ \frac{1}{T}\widehat{F}^{(m)\prime} \left(\widehat{{A}}_s - {A}_s\right) \\
			& =   \frac{1}{T}\widehat{F}^{(m)\prime} D_s \left({\delta} - \widehat{{\delta}}^{(m-1)}\right) 
			+ \frac{1}{T}\widehat{F}^{(m)\prime} X_s \left({\beta} - \widehat{{\beta}}^{(m-1)}\right) 
			+ \frac{1}{T}\widehat{F}^{(m)\prime} \left[F - \widehat{F}^{(m)} H^{(m)}\right] {\lambda}_s \\
			& 
			\quad
			+ \frac{1}{T}\widehat{F}^{(m)\prime}  \widehat{F}^{(m)} H^{(m)} {\lambda}_s 
			+ \frac{1}{T}\left[\widehat{F}^{(m)} - F \big(H^{(m)}\big)^{-1}\right]' {\eta}_s \\
			& \quad 
			+ \frac{1}{T} \left[F \big(H^{(m)}\big)^{-1}\right]' {\eta}_s 
			+ \frac{1}{T}\widehat{F}^{(m)\prime} \left(\widehat{{A}}_s - {A}_s\right).
		\end{align*}
		Therefore,
		{\small
			\begin{align*}
				\Big\|\widehat{{\lambda}}_s^{(m)} 
				&
				- H^{(m)} {\lambda}_s \Big\| \\
				\leq & 
				\left\|\frac{1}{T}\widehat{F}^{(m)\prime} D_s \left({\delta} - \widehat{{\delta}}^{(m-1)}\right) \right\|
				+\left\|\frac{1}{T}\widehat{F}^{(m)\prime} X_s \left({\beta} - \widehat{{\beta}}^{(m-1)}\right) \right\| \\
				&
				+ \left\|\frac{1}{T}\widehat{F}^{(m)\prime} \left[F - \widehat{F}^{(m)} H^{(m)}\right] {\lambda}_s \right\| 
				+ \left\|\frac{1}{T}\left[\widehat{F}^{(m)} - F \big(H^{(m)}\big)^{-1}\right]' {\eta}_s \right\| \\
				&
				+ \left\|\frac{1}{T} \left[F \big(H^{(m)}\big)^{-1}\right]' {\eta}_s \right\|
				+ \left\|\frac{1}{T}\widehat{F}^{(m)\prime} \left(\widehat{{A}}_s - {A}_s\right)\right\| \\
				\leq & 
				\left\|\frac{\widehat{F}^{(m)\prime}}{T}\right\| \cdot |d_s| \cdot \left\|\sum_{t=T_0}^{T}{\delta}_t - \widehat{{\delta}}^{(m-1)}_t\right\| 
				+ \left\|\frac{\widehat{F}^{(m)\prime}}{\sqrt{T}}\right\| \cdot \left\|\frac{X_s }{\sqrt{T}}\right\| \cdot \left\|{\beta} - \widehat{{\beta}}^{(m-1)}\right\| \\
				&
				+ \left\|\frac{1}{T} \widehat{F}^{(m)\prime} \left[F - \widehat{F}^{(m)} H^{(m)}\right]\right\| \cdot \left\| {\lambda}_s \right\|  + \left\|\frac{1}{\sqrt{T}}\left[\widehat{F}^{(m)} - F \big(H^{(m)}\big)^{-1}\right]' \right\| \cdot \left\|\frac{ {\eta}_s}{\sqrt{T}} \right\| \\
				&
				+ \left\| \big(H^{(m)\prime}\big)^{-1} \right\|
				\cdot \bigg\|\frac{1}{T}\sum_{t=1}^{T} {f}_t \eta_{st} \bigg\|
				+ \frac{1}{\sqrt{T}}\left\|\frac{\widehat{F}^{(m)\prime}}{\sqrt{T}}\right\| \cdot \left\|\widehat{{A}}_s - {A}_s\right\| \\
				=\ & O_P\bigg(\frac{1}{\sqrt{T}}\bigg) \cdot \bigg|\sum_{t=T_0}^{T}{\delta}_t - \widehat{{\delta}}^{(m-1)}_t\bigg|
				+ 
				O_P\left( \left\|{\beta} - \widehat{{\beta}}^{(m-1)}\right\|  \right)
				+ O_P\left(\frac{1}{\sqrt{ST}}\right) 
				\cdot \bigg\| \widehat{F}^{(m)}  - F \big(H^{(m)}\big)^{-1}\bigg\| \\
				& + O_P\left(\frac{1}{\sqrt{T}}\right) 
				\cdot \bigg\| \widehat{F}^{(m)}  - F \big(H^{(m)}\big)^{-1}\bigg\| 
				+ O_P\left(\frac{1}{\sqrt{T}}\right)
				+ O_P\left(\frac{1}{S\sqrt{T}}\right) \\
				=\ & O_P\bigg(\frac{1}{\sqrt{T}}\bigg) \cdot \bigg|\sum_{t=T_0}^{T}{\delta}_t - \widehat{{\delta}}^{(m-1)}_t\bigg|
				+ 
				O_P\left( \left\|{\beta} - \widehat{{\beta}}^{(m-1)}\right\|  \right)
				+ O_P\left(\frac{1}{\sqrt{T}}\right) 
				\cdot \bigg\| \widehat{F}^{(m)}  - F \big(H^{(m)}\big)^{-1}\bigg\|  ,
		\end{align*}}
		where the second last equality follows from Assumption~\ref{ass.cov}, Lemma~\ref{l1}, Lemma~\ref{al.6} and the fact that
		\begin{align*}
			\mathbb{E}\left[\left\|\frac{1}{T} F' {\eta}_s \right\|^2 \right] 
			= & \frac{1}{T^2}\sum_{t,l=1}^{T}\sum_{k=1}^{r}\mathbb{E}[\eta_{st}f_{t(k)}f_{l(k)}\eta_{sl}] \\
			= & \frac{1}{T^2}\sum_{t=1}^{T}\sum_{k=1}^{r}\mathbb{E}[\eta_{st}^2f_{t(k)}^2]
			+ \frac{1}{T^2}\sum_{t \neq l}^{T}\sum_{k=1}^{r}\mathbb{E}[\eta_{st}f_{t(k)}f_{l(k)}\eta_{sl}]  \\
			\leq & O_P\left(\frac{1}{T}\right) + \frac{1}{T^2}\sum_{t \neq l}^{T}\sum_{k=1}^{r}
			10 a(|t-l|)^{\frac{\delta}{\delta+2}}
			\left(\mathbb{E}\left[|\eta_{st}f_{t(k)}|^{\delta+2}\right]\right)^{\frac{1}{\delta+2}}
			= O\left(\frac{1}{T}\right),
		\end{align*}
		following Lemma~\ref{al.4}.
		
		Apply similar argument to (\ref{fl1.eq2}), we can show 
		{\small
			\begin{align*}
				\frac{1}{S} & \left\|{\widehat{\Lambda}}^{(m)} -   {\Lambda}H^{(m)\prime}\right\|^2
				= \frac{1}{S}\sum_{s=1}^{S}\left\|\widehat{{\lambda}}_{s}^{(m)} - H^{(m)}  {\lambda}_{js}(u) \right\|^2 \\
				& =  O_P\bigg(\frac{1}{T}\bigg) \cdot \bigg|\sum_{t=T_0}^{T}{\delta}_t - \widehat{{\delta}}^{(m-1)}_t\bigg|^2
				+ 
				O_P\left( \left\|{\beta} - \widehat{{\beta}}^{(m-1)}\right\| ^2 \right)
				+ O_P\left(\frac{1}{T}\right) 
				\cdot \bigg\| \widehat{F}^{(m)}  - F \big(H^{(m)}\big)^{-1}\bigg\|^2.
		\end{align*}}

		For the last claim, we write
		{\small
			\begin{align*}
				\bigg\| & \bigg(\frac{\widehat{\Lambda}^{(m)\prime}\widehat{\Lambda}^{(m)}}{S}\bigg)^{-1} - \big(H^{(m)\prime}\big)^{-1}\bigg(\frac{\Lambda'\Lambda}{S}\bigg)^{-1}\big(H^{(m)}\big)^{-1}\bigg\| \\
				=  & 
				\bigg\| \bigg(\frac{\widehat{\Lambda}^{(m)\prime}\widehat{\Lambda}^{(m)}}{S}\bigg)^{-1} - \bigg(\frac{(\Lambda H^{(m)\prime})'(\Lambda H^{(m)\prime})}{S}\bigg)^{-1}\bigg\| \\
				\leq & 
				\bigg\| \frac{\widehat{\Lambda}^{(m)\prime}\widehat{\Lambda}^{(m)}}{S}-\frac{(\Lambda H^{(m)\prime})'(\Lambda H^{(m)\prime})}{S}\bigg\| 
				\cdot 
				\bigg\| \bigg(\frac{\widehat{\Lambda}^{(m)\prime}\widehat{\Lambda}^{(m)}}{S}\bigg)^{-1} \bigg\|
				\cdot 
				\bigg\|\big(H^{(m)\prime}\big)^{-1}\bigg(\frac{\Lambda'\Lambda}{S}\bigg)^{-1}\big(H^{(m)}\big)^{-1}\bigg\| \\
				\leq & 
				O_P\bigg(\frac{1}{T}\bigg) \cdot \bigg|\sum_{t=T_0}^{T}{\delta}_t - \widehat{{\delta}}^{(m-1)}_t\bigg|^2
				+ 
				O_P\left( \left\|{\beta} - \widehat{{\beta}}^{(m-1)}\right\| ^2 \right)
				+ O_P\left(\frac{1}{T}\right) 
				\cdot \bigg\| \widehat{F}^{(m)}  - F \big(H^{(m)}\big)^{-1}\bigg\|^2,
		\end{align*}}
		where the last line follows from (\ref{fl1.eq3}).	
		\hfill$\square$


{\footnotesize
\bibliographystyle{chicago}
\bibliography{bibfile} 

}

\end{document}